\documentclass[a4paper, oneside, american]{memoir}
\aliaspagestyle{part}{empty}

\usepackage{style}       
\usepackage{mnfrontpage} 
\usepackage{epigraph}

\setsecnumdepth{subsection}

\title{Equilibration and Typicality in Quantum Processes}
\author{Pedro Figueroa Romero}
\kind{\vspace{18ex}School of Physics and Astronomy\\
Monash University\\[3ex]
A thesis submitted for the degree of\\
\emph{Doctor of Philosophy}\\[3ex]
{September 2020}} 

\begin{document}
    \frontmatter
    \mnfrontpage
    \null\vfill
\thispagestyle{empty}
\noindent I certify that I have made all reasonable efforts to secure copyright permissions for third-party content included in this thesis and have not knowingly added copyright content to my work without the owner's permission.

\copyright\quad Pedro Figueroa Romero (2020)
    \clearpage
    \thispagestyle{empty}\mbox{}
    \clearpage
    \chapter{Abstract}
Forgetfulness is a common occurrence in natural processes. After all, if each carbon atom remembered its detailed past, then each of these would have a unique behaviour and there would be no sense in classifying atoms and molecules. Moreover, without forgetfulness, repeatability would be impossible. Despite this, small systems constantly leak information about their state to their surroundings, and quantum mechanics tells us that this information can never be deleted, so that it invariably returns to influence their future behaviour.

How can physical nature be forgetful if it is not allowed to forget? Precisely, in the theory of open quantum systems, memory is the rule and forgetfulness the exception. This conundrum is not dissimilar from that of the emergence of the laws of thermodynamics purely from underlying quantum mechanical laws, which dates from the inception of quantum mechanics itself, and is still fertile ground for important foundational and practical questions.

In this thesis, we address the emergence of forgetfulness---more precisely called \emph{Markovianity}---through a generalization of the concepts of equilibration on average and typicality, previously employed to study the emergence of statistical mechanics from quantum mechanics: the first shows how time-dependent quantities of quantum systems evolve towards fixed values and stay close to them for most times, while the second one argues that small subsystems of a composite are in equilibrium for almost all states of the whole. Using the process tensor framework for quantum stochastic processes, we introduce an extended notion of equilibration, characterizing the conditions under which a multitime quantum process can be approximately described by an equilibrium one. Furthermore, without resorting to the Born-Markov assumption of weak coupling, we show that Markovian processes are typical, and prove that there are physical non-Markovian processes that with high probability look almost Markovian for all orders of correlations, in a phenomenon that we call \emph{Markovianization}.

The results within this thesis bridge the aforementioned gap between what we see in the real world and what idealised physical theories say. The main implications of these findings are twofold: foundationally, they give birth to almost Markovian processes from closed quantum dynamics, while for applications and experiments, they pave the way to predict and quantify the rate at which memory effects become relevant.
    \newpage
\chapter{Declaration}
This thesis contains no material which has been accepted for the award of any other degree or diploma at any university or equivalent institution and, to the best of my knowledge and belief, contains no material previously published or written by another person, except where due reference is made in the text of the thesis.\\[3em]

Melbourne, September 2020
    \chapter{Acknowledgements}
I have grown incommensurately in the last three and a half years, and despite reaching this point in such dire times (amidst the \textsc{covid}-19 pandemic), I can only be grateful for all the people who supported me in one way or another.

I am deeply thankful to both Kavan Modi and Felix Pollock for giving me the chance to do this research, for their patience, enthusiasm, kindness and guidance; they are two of the brightest minds that I have been fortunate to learn from and I hold both in the highest regard as scientists and as human beings. I was also very fortunate to have been part of the unique Monash Quantum Information Science ($\mathsf{MonQIS}$) group, I am grateful for the support and insightful conversations with all its members, as well as for the memes and the shenanigans. My special thanks to Francesco Campaioli, Mathias J{\o}rgensen, Roberto Mu\~{n}oz and Magdalini Zonnios for their support, friendship and good times. Thanks as well to Cody Duncan, honorary classical $\mathsf{MonQIS}$ member, for helping me proofreading my early drafts.

I specially thank Cristian Sierra, who showed up when I was in the middle of my PhD and is now a dearest friend and brother for life, he was there to help me get through the darkest days as well as to have fun on the sunny ones, I thank him for sharing both the mundane and the outstanding with me.

I am sincerely grateful to the whole School of Physics \& Astronomy for the accommodating and amicable atmosphere, together with all the academic and administrative facilities and support. Financially, this work would not have been possible without the Monash  Graduate  Scholarship (MGS) and the Monash International Postgraduate Research Scholarship (MIPRS).

At this stage of my life, I am now fairly convinced that luck plays a higher role in life than most talents one could put forth, and in this sense, I am incredibly lucky for having my family, \emph{mam\'{a}, pap\'{a} y mi hermana Ale}, and their support from so far away for all this time. To be writing these lines was literally unbelievable ten years ago as a \emph{high-school} dropout, but I never lacked their unconditional support and this is only possible due to them.

Finally, in this same spirit, I am grateful to Felicia Grant for filling my life with light and colour, bringing warmth and safety in these uncertain times, and being my constant reminder of what really matters in life.
    \chapter{Publications}
The main results in this thesis have been presented in the following publications:

\begin{itemize}[label={}]
  \item~\cite{FigueroaRomero2020equilibration}
  \emph{Equilibration on average in quantum processes with finite temporal resolution}\\Pedro Figueroa-Romero, Felix A. Pollock and Kavan Modi\\
  Phys. Rev. E 102, 032144.
  
  \item~\cite{FigueroaRomero2019almostmarkovian} \emph{Almost Markovian processes from closed dynamics}\\Pedro Figueroa-Romero, Kavan Modi, and Felix A. Pollock\\Quantum 3, 136 (2019).
  
  \item~\cite{FigueroaRomero2020makovianization} \emph{Markovianization by design}\\Pedro Figueroa-Romero, Felix A. Pollock and Kavan Modi\\Preprint: arXiv:2004.07620 [quant-ph].\\
\end{itemize}

By date of writing and publication, these are ordered as \cite{FigueroaRomero2019almostmarkovian}, \cite{FigueroaRomero2020equilibration}, \cite{FigueroaRomero2020makovianization}, but for sake of presentation we follow the order cited above in the corresponding Chapters \ref{sec:equilibration}, \ref{sec:typicality} and \ref{sec: Markovianization by design}.
    \newpage
    \thispagestyle{empty}
    \null\vspace{0.5in}
    \hfil \emph{No hay nada m\'{a}s vertiginoso que mirar atr\'{a}s.} \hfil
    \vfil
    \newpage
{
    \cleartorecto
    \setcounter{tocdepth}{2}
    \hypersetup{linkcolor=black, linktoc=all}
    \tableofcontents*    
    \cleartorecto
    \listoffigures*      
}
    \cleartorecto
    \printglossary[type=\acronymtype, nonumberlist]
    \chapter{Introduction}
Thermodynamics is without a doubt the most resilient, universal and, in a sense, \emph{strange} theory of physics. It was originally devised as a phenomenological theory, surviving all major revolutions in science and serving as the highest authority for most technological developments. Similarly, the concept of energy pervades the social imaginarium (albeit often in pseudo-scientific ways) and there's arguably no more popular physical law than the first law of thermodynamics, which, together with the remaining three laws (or two, depending who one asks), no one has yet been able to contradict in a real laboratory.

On the other hand, quantum mechanics is one of the fundamental pillars on which any physical theory must stand. By circa 1930, the mathematical formalism of quantum mechanics was established thoroughly by John von Neumann in Ref.~\cite{vonNeumann}, and given that statistical mechanics---the mathematical spinal cord of thermodynamics---was already a mature field of research, it is no surprise that he would quickly approach a quantum theory of statistical mechanics and raise foundational questions that remained largely unsolved until very recently.

A characteristic feature of quantum mechanics is that of generating a distinct type of correlations which is non-local in nature, and this was just recently seen to be key in explaining the quantum emergence of statistical mechanics. Correlations, however, can also be temporal, and it turns out that many foundational questions regarding these correlations in quantum mechanics can be posed in an analogous way to those regarding the foundations of statistical mechanics. Namely, how do quantum systems become \emph{forgetful}, i.e. behaving independently of their detailed past? Or how resilient is a system in thermal equilibrium when temporal correlations are present and when the act of observing the system itself disturbs it?

The overarching goal of this thesis is to provide an in-depth and novel account of the relationship between the phenomenon of equilibration, which is a pillar in the foundations of statistical mechanics, and general quantum processes that incorporate a full account of correlations and memory effects; namely, under what conditions do quantum processes with seemingly negligible temporal correlations arise, given that processes with non-vanishing correlations are the norm? And furthermore, how do we account for the pervasiveness of these in nature? In investigating these questions, we find a non-trivial interplay between complexity and randomness, further characterizing how equilibration comes about in quantum processes.

\section{Outline}
The rest of text is organised as follows: Part I consists of Chapters~\ref{sec:qm101} to~\ref{sec:processes}, containing an overview and explanation of existing literature and no original results, followed by Part II, which consists of Chapters~\ref{sec:equilibration} to~\ref{sec: Markovianization by design} and discusses the original results obtained during this PhD.{\hypersetup{linkcolor=black, linktoc=all}
\begin{description}
    \item[\cref{sec:qm101}] We briefly introduce the concepts and the mathematical framework of quantum mechanics such as quantum states, observables, measurements, closed system evolution and distinguishability of quantum states.

    \item[\cref{sec:statmech}] We first motivate this chapter with a brief discussion of the second law of thermodynamics and why a derivation purely from quantum mechanics is needed. This is followed by a discussion of the main results in the literature showing how quantum systems dynamically fulfil the second law given solely the quantum mechanical framework, in a concept known as equilibration on average. Finally, we discuss how the backbone of the second law, namely the fundamental postulate of statistical mechanics, emerges naturally through entanglement in quantum mechanics without any a-priori assumptions by appealing to a notion of typicality.

    \item[\cref{sec:processes}] We begin by describing more general quantum operations with the concept of quantum channels and three of their main representations, namely dilations, the operator sum representation and the Choi Jamio\l{}kowski isomorphism. We then describe open quantum dynamics and the problem of initial correlations, together with a resolution known as the superchannel. Following this, we motivate the generalization encoding the initial correlations problem via multiple interventions, whereby temporal correlations between more than two points become relevant, in turn leading to the generalization of the superchannel known as the process tensor. We describe how the process tensor framework generalizes the concept of classical stochastic processes and the concept of Markovianity, contrasting with several different approaches in the literature which have proved to be problematic. We finally discuss how the process tensor framework naturally provides an unambiguous measure of non-Markovianity.
    
    \item[\cref{sec:equilibration}] In this chapter we blend the concepts in the first chapters to study the conditions under which a process with a finite temporal resolution can be approximately described by an equilibrium one, which is equivalent to having operations being implemented with a fuzzy clock or to having a system with uniformly fluctuating energies. We first define what we mean by an equilibrium process and by a fuzzy clock, and we then derive a generalization of the concept of equilibration on average to one which can be operationally assessed at multiple times, placing an upper-bound on a new observable distinguishability measure comparing a multitime process with a fuzzy clock against a fixed equilibrium one. We will see that the conditions for equilibration to occur can be extended, with genuine multitime contributions depending on the fuzzy process and the amount of disturbance of the observer’s operations on it. We finally motivate a parallel between the concepts of subsystem equilibration, thermalization and the emergence of statistical mechanics, and process equilibration, \emph{Markovianization} and the emergence of memorylessness in nature. This chapter is based on Ref.~\cite{FigueroaRomero2020equilibration}
    
    \item[\cref{sec:typicality}] By bridging the ideas from the discussion of the emergence of statistical mechanics and the postulate of equal a-priori probabilities to the realm of quantum processes, we argue that this naturally leads to the question of the emergence of memoryless processes, known as Markovian, purely from the rules of quantum mechanics. Motivated by the results on typicality for quantum states, we are able to formally prove that a quantum process drawn uniformly at random will be almost memoryless with high probability whenever it is undergone within a large-dimensional environment. We argue that our results have a parallel interpretation to the case of the emergence of statistical mechanics in the sense of replacing ad-hoc assumptions and approximations to render quantum processes memoryless. We finally discuss the limitations of the typicality approach, the most contentious of which is further discussed in the following chapter. This chapter is based on Ref.~\cite{FigueroaRomero2019almostmarkovian}
    
    \item[\cref{sec: Markovianization by design}] Almost all quantum processes drawn at random within a large environment will be almost Markovian. However, nature seldom behaves randomly. In this chapter we identify a class of physically motivated quantum processes --known as unitary designs-- that satisfy a statement known as a large deviations bound, quantifying the probability that these differ greatly from their Markovian counterparts. We show that, similar to the way that quantum systems thermalize, quantum processes \emph{Markovianize} in the sense that they can converge towards Markovian processes in the correct limits, and in particular as the overall complexity of the interactions in the whole system increases. We further exemplify our result making use of an efficient construction of an approximate unitary design with an $n$-qubit quantum circuit mediated by two-qubit interactions only, showing how seemingly simple systems can speedily become forgetful. We finally discuss potential applications as well as further open questions related to the forgetfulness of nature. This chapter is based on Ref.~\cite{FigueroaRomero2020makovianization}.
\end{description}}
    \mainmatter         
    \setcounter{chapter}{-1}
    \part{Background}
    \chapter{Quantum Mechanics 101}
\label{sec:qm101}
\setlength{\epigraphwidth}{0.7\textwidth}
\epigraph{\emph{Quantum theory is a set of rules allowing the computation of probabilities for the outcomes of tests which follow specified preparations.}}{Asher Peres~(\cite{peres2006quantum})}

In this chapter we briefly introduce the basic notation and mathematical concepts from quantum mechanics required for all the remaining topics covered this PhD. These can be consulted in standard textbooks, such as Ref.~\cite{nielsen2000quantum,bengtsson2006geometry,breuer2002theory}.

\section{Quantum systems and quantum states}\label{sec: quantum states and measurements}
The essential ingredients of quantum mechanics are quantum states, the transformations between them and the measurement outcomes that we observe.

The space of definite states of every quantum system is isomorphic to the space of rays in a Hilbert space $\mscr{H}$, that is, the equivalence class of proportional vectors in $\mscr{H}$. Henceforth we will restrict ourselves to finite, $d$-dimensional Hilbert spaces over the set of complex numbers $\mbb{C}$ with inner product $\langle\phi|\psi\rangle$ for $|\phi\rangle$, $|\psi\rangle$ vectors in $\mscr{H}$. A bipartite quantum system $\mathsf{AB}$ comprising systems $\mathsf{A}$ and $\mathsf{B}$ is associated with a tensor product $\mscr{H}_\mathsf{AB}=\mscr{H}_{\mathsf{A}}\otimes\mscr{H}_\mathsf{B}$, which is such that, if $\{|\alpha\rangle\}_{\alpha=1}^{d_\mathsf{A}}$ and $\{|\beta\rangle\}_{\beta=1}^{d_\mathsf{B}}$ are bases for $\mathsf{A}$ and $\mathsf{B}$, respectively, then any vector in $\mscr{H}_{\mathsf{AB}}$ can be represented by $|\Psi\rangle=\sum_{\alpha,\beta}\psi_{\alpha\beta}|\alpha\beta\rangle$ for coefficients $\psi_{\alpha\beta}\in\mbb{C}$ and where $|\alpha\beta\rangle:=|\alpha\rangle\otimes|\beta\rangle$. Any $n$-partite space $\mscr{H}\cong\mscr{H}_1\otimes\cdots\otimes\mscr{H}_n$ then is built similarly by extension.

The most general state of a quantum system is then specified by a density operator $\rho$, which is an element of the space of bounded operators $\mscr{B}(\mscr{H})$\footnote{ In general these also must have a finite trace; this is ensured for finite-dimensional systems.} and which additionally is set to satisfy
\begin{equation}
    \rho=\rho^\dg\,\text{(Hermiticity)},\quad\rho\geq0\,\text{(positivity)}\quad\text{and}\quad\tr(\rho)=1\,\text{(unit trace)},
    \label{eq: def quantum state}
\end{equation}
where  Hermiticity means $\rho$ equals its conjugate transpose, here denoted by $\dg$, which implies in turn that all its eigenvalues are real. We denote the subset of density operators on a given space by $\$(\mscr{H})\subset\mscr{B}(\mscr{H})$. Positivity explicitly means positive semi-definite, $\langle\varphi|\rho|\varphi\rangle\geq0$ for any vector $|\varphi\rangle\in\mscr{H}$, which implies that all of the eigenvalues of $\rho$ are non-negative; we will commonly refer to this property \emph{for quantum states} simply as positivity. Together with unit trace, these will ensure all probabilities corresponding to the outcomes of a measurement of a quantum system are real, positive and add up to unity. We will commonly refer to density operators simply as quantum states.

A quantum state $\rho\in\$(\mscr{H})$ is called pure if there exists a vector $|\psi\rangle$ such that $\rho=|\psi\rangle\!\langle\psi|$, or equivalently if its rank, i.e. the dimension of its image, is equal to one. From the spectral theorem then it follows that every quantum state is a convex mixture of pure states, $\rho=\sum_ip_i|\psi_i\rangle\!\langle\psi_i|$, with $p_i\geq0$ and $\sum{p_i}=1$. The purity of a quantum state is given by
\begin{equation}
    \f{1}{d}\leq\tr(\rho^2)\leq1,
\end{equation}
which is known as the purity of the state $\rho$. The upper-bound is saturated when the state is pure, whilst the lower bound is reached for the so-called maximally mixed state, given by $\rho=\mbb1/d$, where $\mbb1$ is the identity operator on $\mscr{H}$, here the $d\times{d}$ identity matrix. This is interpreted as the state of maximal ignorance, as the system has equal probability to be in any possible pure state.

Furthermore, any mixed state can be expressed as a reduced state of a pure state in a larger Hilbert space: this is known as purification. Here \emph{reduced} means a state of a subset of degrees of freedom of the full system, i.e. for a quantum state in a bipartite system, $\rho\in\$(\mscr{H}_\mathsf{A}\otimes\mscr{H}_\mathsf{B})$, we define
\begin{align}
    \rho_\mathsf{A}&:=\tr_\mathsf{B}(\rho):=\sum_{\beta=1}^{d_\mathsf{B}}(\mbb1_\mathsf{A}\otimes\langle\beta|)\,\rho\,(\mbb1_\mathsf{A}\otimes|\beta\rangle),
    \label{eq: partial trace}
\end{align}
as the reduced state on space $\mathsf{A}$, where $\tr_\mathsf{B}(\cdot)$ is called a partial trace, defined as a trace over subspace $\mathsf{B}$ and with $\mbb1_\mathsf{A}$ the identity operator solely on $\mathsf{A}$. Similarly, $\rho_\mathsf{B}=\tr_\mathsf{A}(\rho)$ is the reduced state of $\rho$ on subspace $\mathsf{B}$. It is clear by inspection that the resulting $\rho_\mathsf{A},\,\rho_\mathsf{B}$ are legitimate quantum states. Thus purification means that every mixed quantum state $\rho\in\$(\mscr{H})$ can be expressed as $\rho=\tr_\mathsf{\Gamma}[|\Psi\rangle\!\langle\Psi|]$ for some pure state $|\Psi\rangle\in\mscr{H}\otimes\mscr{H}_\mathsf{\Gamma}$. The system $\mathsf{\Gamma}$ is usually referred to as an ancillary space or just an ancilla.\footnote{ \emph{Ancilla} is the Latin term for maidservant; despite other (mainly negative) connotations it potentially carries, it is now standard in quantum information science as a synonym of \emph{auxiliary}.} This is easily seen by the so-called Schmidt decomposition, which ensures that any bipartite Hilbert space vector can be written in the form $|\Psi\rangle=\sum_{i=1}^D\sqrt{\varphi_i}|u_iv_i\rangle$ where here $|u_i\rangle$ and $|v_i\rangle$ are orthonormal states in the respective subsystems, $D=\min(d,d_\mathsf{\Gamma})$, and $\varphi_i$ are strictly positive coefficients such that $\sum\varphi_i=1$~\cite{peres2006quantum}. Then for such $|\Psi\rangle$ we obtain the reduced state $\rho=\sum\varphi_i|u_i\rangle\!\langle{u}_i|$, so we can always go in the opposite direction by decomposing any quantum state via the spectral theorem and using it to construct a pure state in an extended space incorporating an ancilla. Notice that in the case of a purification, $d\leq{d}_\mathsf{\Gamma}$ and that such purification will not be unique. This concept has far reaching consequences as will be shown below and can be clearly motivated physically by thinking of the ancillary space as an environment.

\section{Measurements and observables}
Whereas in classical mechanics we can describe the state of a system in a somehow passive way, in quantum mechanics we need access to the density operator through other operators that play an active role in a sense we will now describe. The concept of an observable is tightly related with that of a measurement: the most general measurements are represented by a finite ordered set $\{\mathrm{M}_i\}$ called a \gls{POVM}, where the elements $\mathrm{M}_i$ are such that
\begin{equation}
    \sum_i\mathrm{M}_i=\mbb1,\quad\text{with}\quad \mathrm{M}_i=\mathrm{M}_i^\dg\quad\text{and}\quad\mathrm{M}_i\geq0,
\end{equation}
i.e. Hermitian positive semidefinite operators forming a partition of the identity operator in $\mscr{H}$. A \gls{POVM} measurement applied to a state $\rho$ produces the $i\textsuperscript{th}$ outcome with probability $\tr(\mathrm{M}_i\rho)$, with the definition of a \gls{POVM} ensuring that these sum up to unity. A \gls{POVM} is called \emph{informationally complete} if its statistics fully determine the density matrix, which requires at least $d^2$ elements~\cite{bengtsson2006geometry}.

In particular, whenever the \gls{POVM} consists of $d$ elements with all being orthogonal \emph{projectors}, i.e. $\{\mathrm{M}_i=\Pi_i\}$ satisfying 
\begin{equation}
    \Pi_i\Pi_j=\Pi_i\delta_{ij},\qquad i,j=1,2,\ldots,d,
\end{equation}
this is called a \emph{projective measurement}. We then refer to a Hermitian operator as an observable whenever its real eigenvalues correspond to measurable outcomes. In particular then, an observable with spectral decomposition $A=\sum_{i=1}^d\alpha_i\Pi_i$ describes a projective measurement with probability $\tr(\Pi_i\rho)$ of obtaining the $i\textsuperscript{th}$ outcome $\alpha_i$.

Upon measuring, a quantum system will generally change its state. In general a \gls{POVM} is not enough to determine the post-measurement state. In Chapter~\ref{sec:processes} we will introduce the notion of so-called Kraus operators which will let us deal with this; in particular, for projective measurements, the projectors themselves are Kraus operators and after measurement the state is generally\footnote{ More generally, even in the case of orthogonal projectors as \gls{POVM} elements, projectors as Kraus operators need not be implied in the post-measurement state.} transformed to $\rho\to\rho^\prime=\sum_i\Pi_i\,\rho\,\Pi_i$, or in particular, if the $i$\textsuperscript{th} outcome is observed, then
\begin{equation}
    \rho\to\rho^\prime=\f{\Pi_i\,\rho\,\Pi_i}{\tr[\Pi_i\rho]},
    \label{eq: Born2}
\end{equation}
which is called a \emph{selective} measurement in Ref.~\cite{bengtsson2006geometry}. Selective measurements are repeatable in the sense that if they are performed again the post measurement state remains the same; this is not true for general \gls{POVM}s.

Finally, we will denote the \emph{expectation value} of $A$ on the state $\rho$ by
\begin{equation}
    \langle{A}\rangle_\rho:=\tr(A\rho)=\sum_i\alpha_i\tr(\Pi_i\rho),
\end{equation}
so that $\tr(\Pi_i\rho)$ is the corresponding probability for the $i\textsuperscript{th}$ outcome $\alpha_i$, which in essence constitutes Born's rule~\cite{Gleason1975,Busch_Gleason}. In particular, the \emph{measurement statistics} of a \gls{POVM} will refer to the vector of probabilities $\tr(\Pi_i\rho)$.

\section{Closed system dynamics}\label{sec: closed system dynamics}
We now discuss the dynamical picture for quantum systems. A \emph{closed} system described by $\rho$ at a given time $t\in\mbb{R}^+$, where $\mbb{R}^+$ denotes the set of positive real numbers, will evolve unitarily according to the Schr\"{o}dinger equation,\footnote{ In general such evolution for any operator is referred to as the von Neumann or the quantum Liouville equation.}
\begin{equation}
    i\f{\partial}{\partial{t}}\rho(t)=[H,\rho(t)],
    \label{eq: Schroedinger}
\end{equation}
where $H\in\mscr{B}(\mscr{H})$ is the observable known as the Hamiltonian of the system, and where we set units $\hbar=1$. This operator can be said to be generating the dynamics of $\rho$, and whenever it does not depend on time, as we will consider throughout this thesis, it gives rise to the unitary time-evolution operator
\begin{equation}
    U(t)=\exp(-iHt).
\end{equation}

In general, any operator $V\in\mscr{B}(\mscr{H})$ such that $VV^\dg=V^\dg{V}=\mbb1$ is called unitary. We can see that given a state $\rho$, any $\sigma=V\,\rho\,V^\dg$ remains a density operator. We thus have the solution to Eq.~\eqref{eq: Schroedinger} as
\begin{equation}
    \rho(t)=U(t)\,\rho(0)\,U^\dg(t),
    \label{eq: time evolved state}
\end{equation}
where we will usually denote the initial state at time $t=0$ simply as $\rho(0)=\rho$. 

We will commonly write the Hamiltonian in a spectral decomposition
\begin{equation}
    H=\sum_{n=1}^\mathfrak{D}E_nP_n,
\end{equation}
where $P_n$ is the spectral projector onto the $n\textsuperscript{th}$ eigenspace of $H$ with energy (eigenvalue) $E_n$. Here $\mathfrak{D}=|\mathrm{spec}(H)|\leq{d}$ is the number of distinct energies $E_n$; if $H$ is degenerate with the $n$\textsuperscript{th} level having degeneracy $\ell$, then $P_n=\sum_{j=1}^\ell|n_j\rangle\!\langle{n}_j|$ with $\{|n_\ell\rangle\}$ a basis for the $n$\textsuperscript{th} energy eigenspace.

The expectation value of an observable $A$ with respect to a time-evolved state $\rho(t)$ can then be written as
\begin{align}
    \langle{A}\rangle_{\rho(t)}&=\tr[A\rho(t)]=\tr[A\,U(t)\,\rho\,U^\dg(t)]=\tr[U^\dg(t)\,A\,U(t)\,\rho],
    \label{eq: temporal expectation value}
\end{align}
by the cyclic property of the trace. This can be equivalently thought of as the expectation value $\langle{A}(t)\rangle_\rho$ of a time-evolved operator $A(t):=U^\dg(t)\,A\,U(t)$ at time $t$ on the state $\rho$. This is commonly known as the Heisenberg picture, with the observable $A$ evolving according to
\begin{equation}
    -i\f{\partial}{\partial{t}}A(t)=[H,A(t)],
    \label{eq: von Neumann}
\end{equation}
so that whenever $[H,A(t)]=0$, the expectation value in Eq.~\eqref{eq: temporal expectation value} is constant and, similarly to the classical case, $A$ is called a conserved quantity.

\section{Distinguishability of quantum states: trace distance}\label{sec: trace distance qm101}
Another important aspect we will require is to be able to quantify how different, or how distinguishable, two quantum states are. In particular, a distinguishability measure called the trace distance will be central to most of our discussions. We can put the task of distinguishing a pair of quantum states in terms of measuring the \emph{distance} between them, i.e. of quantifying how close or far are two states from each other.

Let us first define exactly what we mean by a distance measure.

\begin{definition}[Distance measure]
\label{def: distance measure}
Let $\mscr{V}$ be a vector space and $x,y\in\mscr{V}$ any two points. A distance measure $\Delta:\mscr{V}\times\mscr{V}\to\mbb{R}^+_0$ between $x$ and $y$ satisfies:
\begin{compactenum}[\itshape i.]
    \item Positivity: $\Delta(x,y)\geq0$ with equality if and only if $x=y$.
    \item Symmetry: $\Delta(x,y)=\Delta(y,x)$.
    \item Triangle inequality: $\Delta(x,y)\leq\Delta(x,v)+\Delta(v,y)$ for any $v\in\mscr{V}$.
\end{compactenum}
\end{definition}

While there can be a plethora of valid distance measures, not necessarily all of these will constitute a distinguishability measure. By distinguishability measure we specifically refer to a distance measure with an operational meaning, i.e. one that can be ultimately phrased and quantified via measurements. In both the classical and quantum cases, there is no unique way of quantifying distinguishability; however, while the classical case depends just on the statistical state, in the quantum one this depends both on the way the state is measured and the quantum states. That is, considering a \gls{POVM}, say $\{\mathrm{M}_i\}$, the relevant quantity is given by the probabilities $\tr[\mathrm{M}_i\rho]$, as discussed in Section~\ref{sec: quantum states and measurements}.

One way to do this, for any two quantum states $\rho,\sigma\in\$(\mscr{H})$, is to quantify the probability of error in guessing which one of the two is the given state for the system in a single measurement.

Consider then a distance $D_{\{M_i\}}$ defined by
\begin{equation}
    D_{\{\mathrm{M}_i\}}(\rho,\sigma):=\f{1}{2}\sum_i|\tr\mathrm{M}_i(\rho-\sigma)|,
\end{equation}
with the $1/2$ being a normalization factor. This corresponds to the distinguishability between $\rho$ and $\sigma$ given the \gls{POVM} $\{\mathrm{M}_i\}$, as it compares the probabilities for each outcome on either state given a measurement. Specifically, considering a system which was prepared in either state $\rho$ or state $\sigma$, we care about guessing which of these the system is actually in, not minding both destroying the actual state and accidentally guessing the wrong outcome. To this end, the best strategy we can take, given the $i\textsuperscript{th}$ outcome, is to guess that the state is $\rho$ if $\tr[\mathrm{M}_i\rho]\geq\tr[\mathrm{M}_i\sigma]$ and to guess it is in $\sigma$ otherwise. Then we automatically can be right half of the time, with the probability of success in correctly guessing the correct state being
\begin{equation}
    \mbb{P}_{\{\mathrm{M}_i\}}^\text{success}=\f{1}{2}-\f{1}{2}\sum_i|\tr[\mathrm{M}_i\rho]-\tr[\mathrm{M}_i\sigma]|=\f{1}{2}[1-{D}_{\{\mathrm{M}_i\}}(\rho,\sigma)].
    \label{eq: prob success tr dist}
\end{equation}

If we take a subset of all possible \gls{POVM}s, say $\mbb{M}$, we could also define
\begin{equation}
    D_\mbb{M}:=\max_{\{\mathrm{M}_i\}\in\mbb{M}}D_{\{\mathrm{M}_i\}}
\end{equation}
as a distinguishability restricted to such subset. It can be readily verified that this is a legitimate distance measure, and furthermore we have the hierarchy
\begin{equation}
    0\leq{D}_{\mbb{M}}(\rho,\sigma)\leq{D}(\rho,\sigma)\leq1,
\end{equation}
where we have defined
\begin{equation}
    {D}(\rho,\sigma):=\f{1}{2}\|\rho-\sigma\|_1,
    \label{eq: trace distance def}
\end{equation}
as the so-called trace distance, where $\|\cdot\|_1$ is the Schatten 1-norm or also sometimes called trace-norm, which can be defined as the case $p=1$ of the family of norms
\begin{equation}
    \|X\|_p:=\tr\left[|X|^p\right]^{1/p},
    \label{eq: Schatten p-norm}
\end{equation}
called the Schatten $p$-norms. This is because the trace distance would give the distinguishability measure with the optimal of all possible measurements~\cite{bengtsson2006geometry}.

Importantly, the Schatten norms satisfy the hierarchy
\begin{equation}
    \|\cdot\|_1\geq\|\cdot\|_2\geq\ldots\geq\|\cdot\|,
    \label{eq: Schatten hierarchy}
\end{equation}
where here $\|X\|_\infty:=\|X\|$ will be referred to as the \emph{operator norm}, and corresponds to the largest singular value\footnote{ The singular values of a matrix $X$ are the square roots of the eigenvalues of $X^\dg{X}$. Thus, for a Hermitian matrix, the singular values are the absolute values of its eigenvalues.} of $X$.

There are several reasons why the trace distance is important and usually preferred among other state distinguishability measures, which are nevertheless legitimate in their own right. While some of these will become apparent when we present so-called quantum maps in Section~\ref{sec: process tensor}, its operational relevance is overall what makes it a suitable distinguishability measure. Specifically, we can see from Eq.~\eqref{eq: prob success tr dist} that the trace distance is precisely the one that maximizes the probability of success with the optimal amongst all possible measurements. Similarly, other scenarios where the trace distance is relevant can be seen e.g. in Ref.~\cite{Watrous}.

\section{Entanglement}
The history of entanglement is well-known, with the discussion beginning when the famous EPR paper~\cite{EPR_original} came to light, and after which Schr\"{o}dinger, in correspondence with Einstein~\cite{schrodinger_1935}, coined the term \emph{entanglement} to describe the new kind of correlation. Schr\"{o}dinger would later add that entanglement is not \emph{one} but rather \emph{the} characteristic trait of quantum mechanics. A full in-depth, geometrical discussion of entanglement can be seen e.g. in Ref.~\cite{bengtsson2006geometry}.

While the consequences of bipartite entanglement are far reaching (often capturing the popular imagination as well), it is a concept that has a very simple definition. Consider a bipartite space $\mscr{H}\cong\mscr{H}_\mathsf{A}\otimes\mscr{H}_\mathsf{B}$, then a state $\rho\in\$(\mscr{H})$ is called \emph{separable} if either it is a product state, $\rho=\rho_\mathsf{A}\otimes\rho_\mathsf{B}$, or if it can be written as a convex combination of product states, $\rho=\sum{p}_i\rho_\mathsf{A}^{(i)}\otimes\rho_\mathsf{B}^{(i)}$ with $\sum{p}_i=1$. Otherwise, the state $\rho$ is called \emph{entangled}.

The simplest example is that of a pair of qubits, with both $\mscr{H}_\mathsf{A}$ and $\mscr{H}_\mathsf{B}$ being two-dimensional. Written in the so-called computational basis $\{|0\rangle,|1\rangle\}$, we can find an orthogonal basis with this property (as we will see, in its most extreme form) called the Bell basis, given by the four vectors
\begin{equation}
    |\varphi^{\pm}\rangle:=\f{1}{\sqrt{2}}(|00\rangle\pm|11\rangle),\qquad
    |\vartheta^{\pm}\rangle:=\f{1}{\sqrt{2}}(|01\rangle\pm|10\rangle),
\end{equation}
which we can see are not separable. We can also see that $\tr_\mathsf{B}(|\varphi^+\rangle\!\langle\varphi^+|)=\f{1}{2}(|0\rangle\!\langle0|+|1\rangle\!\langle1|)=\mbb{1_\mathsf{A}}/2$, and similarly for all other reduced states, so that despite the whole being pure, this reduced state is maximally mixed. In other words, despite having full certainty of the global state, we get maximal ignorance in either subpart. This also naturally leads to a relation between the purity and how entangled a bipartite state is: the more entangled, the lower the purity of the reduced states.

We can generalize to bipartite systems of dimension $\mathsf{d}=d_\mathsf{A}=d_\mathsf{B}$ through the state that projects into the vector
\begin{equation}
    |\mathsf{\Psi}\rangle=\f{1}{\sqrt{\mathsf{d}}}\sum_{i=1}^\mathsf{d}|ii\rangle,
    \label{eq: def maximally entangled state}
\end{equation}
which we will generically call \emph{the} maximally entangled state, for which effectively we can readily see that $\tr_\mathsf{A}(|\mathsf{\Psi}\rangle\!\langle\mathsf{\Psi}|)=\tr_\mathsf{B}(|\mathsf{\Psi}\rangle\!\langle\mathsf{\Psi}|)=\mbb{1}/\mathsf{d}$. Notice that if the dimensions are different, at most the smaller subsystem can be maximally mixed.

With this we can now introduce an analogue measure of entanglement given by the von Neumann entropy of the reduced states
\begin{equation}
    S(\rho_\mathsf{A})=-\tr[\rho_\mathsf{A}\log\rho_\mathsf{A}],
    \label{eq: von Neumann entropy}
\end{equation}
and similarly for $\rho_\mathsf{B}$, where $\log$ can be taken to be base 2, either known as entanglement entropy. Consider a Schmidt decomposition of a state $|\phi\rangle\in\mscr{H}_\mathsf{A}\otimes\mscr{H}_\mathsf{B}$,
\begin{equation}
    |\phi\rangle=\sum_{i=1}^D\sqrt{\lambda_i}|u_iv_i\rangle,
\end{equation}
with $|u_i\rangle$, $|v_i\rangle$ orthonormal states in each system, $D=\min(d_\mathsf{A},d_\mathsf{B})$ and positive coefficients $\sum\lambda_i=1$. Then, writing $\Phi:=|\phi\rangle\!\langle\phi|$, the reduced state on either system is diagonal, $\Phi_\mathsf{A}=\sum\lambda_i|u_i\rangle\!\langle{u}_i|$ and $\Phi_\mathsf{B}=\sum\lambda_i|v_i\rangle\!\langle{v}_i|$, thus, as the logarithm of a diagonal matrix is the matrix of logarithms of its entries, it follows that
\begin{align}
    S(\Phi_\mathsf{A})&=-\sum_{i=1}^D\lambda_i\log(\lambda_i)=S(\Phi_\mathsf{B}),
\end{align}
which is known as the Shannon entropy of the distribution given by the eigenvalues $\lambda_i$, and indeed it turns out this is maximized for $\lambda_i=1/D$, i.e. when the distribution in the smaller subspace is uniform.
\chapter{The foundations of Statistical Mechanics}
\label{sec:statmech}
\setlength{\epigraphwidth}{0.45\textwidth}
\epigraph{\emph{If physical theories were people, thermodynamics would be the village witch.}}{Goold et al.~(\cite{Goold_2016})}
\section{Equilibrium and the second law of thermodynamics}

Thermodynamics---the branch of physics that deals with the different manifestations of energy and the relation between them---is often regarded as a cornerstone in physics, and anyone that aims to gain a deep knowledge about the nature of reality will certainly have to master its concepts. It is, however, a different kind of theory from, say quantum mechanics, in the sense that it pervades all of physics without necessarily being a \emph{fundamental} theory of physical reality in the same way that we regard quantum mechanics as fundamental.

Quantum mechanics indeed can be seen as a fundamental theory of physics at small spatial scales which in the suitable limit will contain the physics at ordinary macroscopic scales. That is, classical physics, and all of its predictions, are in principle attainable from quantum mechanics taken in the appropriate limit. In particular, this means that the laws of thermodynamics should \emph{emerge} from the microscopic physics given by quantum mechanics.

Historically, this was acknowledged since the inception of quantum theory itself by its founding fathers, the first one being perhaps Erwin Schr\"{o}dinger, who invoked what he calls a \emph{statistical hypothesis}~\cite{Schrodinger} (with English translation in Ref.~\cite{schrodinger2003}) about initial energy level populations in two weakly interacting systems. Schr\"{o}dinger's aim was to describe the long time behaviour of such systems, and he found that the states satisfying his hypothesis are well described by thermal states (satisfy a canonical distribution). Similarly, John von Neumann in Ref.~\cite{vonNeumann} (with translation in Ref.~\cite{vonNeumann2010}) sets out to explain how the irreversible behaviour of entropy emerges from quantum mechanics and how ensemble properties can be assumed in macroscopic (real and imperfect) physical systems.

Particularly puzzling is the behaviour found in the second law of thermodynamics, where the underlying quantum dynamics should give rise to an ever-increasing entropy and evolve towards an equilibrium configuration. The dynamics described by quantum mechanics by the Schr\"{o}dinger equation is unitary, which implies that the information of the system in question is conserved throughout its evolution. This has important consequences for the second law of thermodynamics, as it implies that the dynamics is reversible and shouldn't \emph{forget} its initial configuration when converging to a fixed state.

While the works of Schr\"{o}dinger and von Neumann made progress in addressing such issues, how the second law emerges from unitary quantum dynamics is far from trivial and remained largely unresolved for years. It has only been recently that such question has come back with renewed wave of interest, given not only major advances in experimental and computational techniques but also other unifying theoretical ideas that put fundamental problems such as this one in a new light and render them tractable~\cite{Gogolin_2016, DAlessio_2016, Goold_2016}.

The breadth of topics related to the foundations of statistical mechanics is quite large, with many other areas of physics being affected by them~\cite{Bocchieri_1959, Peres_1984, Jensen_1985, Deutsch_ETH, Srednicki_1994,Srednicki_1999, Rigol_2012, Murthy_2019}; for the purpose of this thesis, however, in this section and then further, we will focus on two main ideas that directly addressed the issues arising from the quantum mechanical foundations of statistical mechanics; these can be described as kinematic and dynamical, and are known as \emph{typicality} and \emph{equilibration}, respectively. The first one refers to an idea that anticipates that almost all quantum systems will be almost in equilibrium, so that most evolutions will carry quantum systems to equilibrium and stay close to it for most times. The second is precisely concerned with the characterization of such evolutions and thus how is equilibrium is achieved. Modern approaches further distinguish as \emph{thermalization} a more restrictive case which can be regarded as a proper thermodynamic equilibration~\cite{Gogolin_2016}, but this is outside the scope of this thesis.

While each of the laws of thermodynamics enjoy their fair share of popularity, the second law is particularly celebrated for its consequences, e.g. that of providing an arrow of time, banning perpetual motion machines, or giving us the conclusion that even our universe will meet death one day. It is one of the most fundamental principles in science and its significance stretches to practically everything we can think of in physical reality. The second law is concerned primarily with the direction in which physical processes can occur: in a nutshell, it is a statement about evolution of physical systems always proceeding towards a fixed state of equilibrium. 

The statement of the second law of thermodynamics usually invokes the concept of \emph{entropy}.\footnote{ The word was introduced by Rudolf Clausius in 1865 as a composition of terms standing for \emph{energy} and \emph{transformation}. The modern interpretation of entropy is that of an information measure and physics is rather one particular application; the scope of what can be said about entropy far exceeds what is covered in this thesis.} While in classical thermodynamics it was introduced as a quantity related to the efficiency of thermodynamic processes, its interpretation in statistical mechanics (classical or quantum) relates to how probable a macrostate, i.e. a global property of a system, is. In a nutshell, the second law states that the entropy of a thermally isolated system can only increase, so that systems can only evolve towards the most probable macrostate. While this agrees with our experience of reality, reconciling this phenomenological statement with microscopic reversible laws has been, and continues to be, a challenge.

While the quantitative description of the second law is given by the so-called $H$-theorem, first derived by Ludwig Boltzmann in 1872~\cite{Boltzmann} for an ideal gas, it was soon realized that this description was not in itself a proof but a statement that followed as a consequence of an implicit assumption. Such an assumption is known as the \emph{equal a-priori probabilities} assumption, otherwise known as the fundamental postulate of statistical mechanics:
\begin{definition}[Postulate of equal a-priori probabilities~\cite{kittel1980thermal}]
\label{postulate: equal a priori}
A closed system is equally likely to be in any of the microstates accessible to it. We refer to this either as the fundamental postulate of statistical mechanics or the postulate of equal a-priori probabilities.
\end{definition}
This is seemingly a perfectly reasonable assumption to make, however, it is still an arbitrary one that needs to be put in by hand. Any successful explanation of the second law purely from quantum mechanics must then account for this fundamental postulate of statistical mechanics.

The second key element that a quantum emergence of the second law must fulfil is that of explaining how is equilibrium reached dynamically, i.e. how is it possible that closed quantum dynamics, as given by the Schr\"{o}dinger equation, becomes forgetful in the sense of converging to a fixed equilibrium state. This puzzle can be posed simply by saying that the dynamics of closed quantum systems is \emph{unitary}; this implies that dynamics will be recurrent and time-reversal invariant, as explained below. Thus, if quantum systems equilibrate, it must be said precisely in what sense and in which way they do.

 In the following sections we present some of the results that have resolved to a great extent both explaining the equal a-priori probabilities postulate as well as the dynamic emergence of equilibrium, and which will be useful further for the main results of this thesis. We highlight throughout that the foundations of statistical mechanics is a vibrantly active topic of research forming part of so-called \emph{quantum thermodynamics}, with interest from all computational, applied and foundational fronts, thus some cutting-edge topics that are also of a high relevance, such as the \emph{eigenstate thermalization hypothesis}~\cite{Srednicki_1994,Srednicki_1999,Rigol_2012,DAlessio_2016,Murthy_2019} or equilibration timescales, are not discussed in this thesis.

\section{Equilibration on average}\label{sec: equilibration on average}
One of the very first problems that we are faced with when trying to approach the question of convergence towards equilibrium in quantum mechanics is that of recurrences, i.e. that quantum states evolving unitarily eventually return to their initial states. This was proved in Ref.~\cite{Bocchieri_1957}, and despite this being the case as well in classical mechanics with so-called Poincar\'{e} recurrences, the quantum discussion can be seen to have some differences~\cite{Peres_recurrence}. For finite dimensional systems, however, such as the ones we will consider here, this picture is intuitively clear since there are finitely many mutually distinguishable states towards which evolution can occur; this can be shown quantitatively as in Ref.~\cite{Wallace_recurrence} and recurrence timescales can be discussed as well~\cite{Bhattacharyya_recurrence, Gimeno_2017}.

Similarly, time-reversal-invariance is clear from Eq.~\eqref{eq: Schroedinger} and Eq.~\eqref{eq: von Neumann}. Genuinely equilibrating closed quantum systems thus cannot exist. A way to make sense of a quantum evolution towards equilibrium is by showing that indeed quantum systems approach an equilibrium state, albeit just staying \emph{close to it for most of the time}, with recurrences being rare. This is in essence what is meant by \emph{equilibration on average}, which we should note, however, is much more generic to what one usually associates with evolution towards \emph{thermal} equilibrium, known as thermalization. Specifically, we use the following general definition.
\begin{definition}[Equilibration on average~\cite{Gogolin_2016}]
\label{def: equilibration on average}
A time dependent property \emph{equilibrates on average} if its value remains close to a given equilibrium value for most times.
\end{definition}

In this chapter we will focus mainly on equilibration on average for the expectation value of observables in closed systems.

\subsection{The equilibrium state}
To make a formal statement about equilibration on average, we first need to define what we mean by an equilibrium value. As equilibration on average assesses how close a time-dependent property is from equilibrium \emph{for most times}, a natural candidate for an equilibrium state is the time-averaged state. Specifically, given an initial state $\rho\in\$(\mscr{H})$, an observable $A=\mscr{B}(\mscr{H})$, and denoting time-averaging by an overline, if the time-average of the expectation value of an observable $\overline{\langle{A}\rangle_{\rho(t)}}=\overline{\tr[A\rho(t)]}$ equilibrates, then it should do so to $\langle{A}\rangle_{\overline{\rho(t)}}=\tr[A\overline{\rho(t)}]$. This seems like an obvious observation, but this nevertheless motivates the definition
\begin{equation}
    \omega:=\lim_{T\to\infty}\overline{\rho}^T,
    \label{eq: definition omega}
\end{equation}
as the equilibrium state of the system, where
\begin{equation}
    \overline{\rho}^T:=\f{1}{T}\int_0^T\rho(t)\,dt,
    \label{eq: finite time avg}
\end{equation}
is the uniform time-average of $\rho(t)$ over a finite-interval $[0,T]$. For quantum systems with a time-independent Hamiltonian, the limit in Eq.~\eqref{eq: definition omega} is well-defined~\cite{Gogolin_2011} in essence because of the unitary evolution~\cite{Besicovitch}, and we get
\begin{align}
    \omega&=\lim_{T\to\infty}\f{1}{T}\int_0^T\ex^{-iHt}\rho\,\ex^{iHt}dt\nonumber\\
    &=\lim_{T\to\infty}\f{1}{T}\int_0^T\sum_{n,m=1}^\mathfrak{D}\ex^{-it(E_n-E_m)}P_n\,\rho\,P_m\,dt\nonumber\\
    &=\sum_{n,m=1}^\mathfrak{D}\left(\lim_{T\to\infty}\f{1}{T}\int_0^T\ex^{-it(E_n-E_m)}dt\right)\,P_n\,\rho\,P_m\nonumber\\
    &=\sum_{n=1}^\mathfrak{D}P_n\,\rho\,P_n,
\end{align}
as $\lim_{T\to\infty}\f{1}{T}\int_0^T\ex^{-it(E_n-E_m)}dt=\delta_{nm}$, where $\delta_{nm}$ refers to a Kronecker delta, equal to $1$ if $n=m$ and equal to $0$ otherwise. That is, the state $\omega$ corresponds to the completely dephased state of $\rho$ with respect to the Hamiltonian $H$. We illustrate this in Fig.~\ref{fig: omega dephased}.

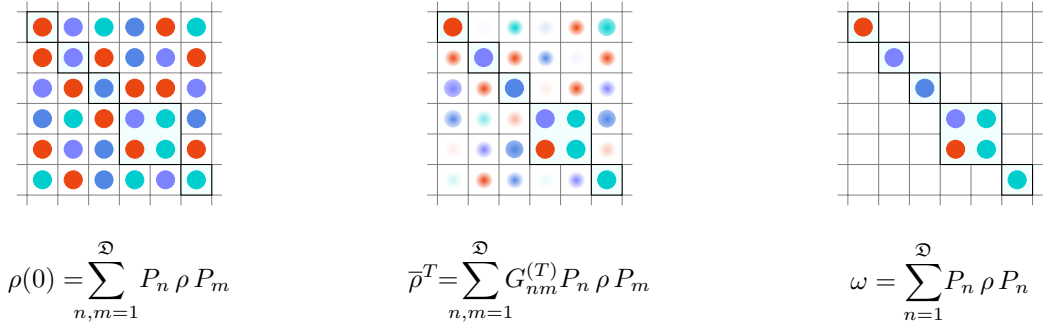
\begin{figure}[t]
\makebox[\textwidth][c]
{
    \begin{tikzpicture}[scale=1.35]
    \begin{scope}
    \draw[step=0.3cm,gray,very thin] (-1,-1) grid (1,1);
    \draw[fill=C3!5!white] (-0.9,0.9) rectangle (-0.6,0.6);
    \draw[fill=C3!5!white] (-0.6,0.6) rectangle (-0.3,0.3);
    \draw[fill=C3!5!white] (-0.3,0.3) rectangle (0,0);
    \draw[fill=C3!5!white] (0,0) rectangle (0.6,-0.6);
    \draw[fill=C3!5!white] (0.6,-0.6) rectangle (0.9,-0.9);
    \draw[fill=C2, draw=C2] (-0.75,0.75) circle[radius=2.5pt];
    \draw[fill=C4, draw=C4] (-0.45,0.45) circle[radius=2.5pt];
    \draw[fill=C1, draw=C1] (-0.15,0.15) circle[radius=2.5pt];
    \draw[fill=C4, draw=C4] (0.15,-0.15) circle[radius=2.5pt];
    \draw[fill=C3, draw=C3] (0.45,-0.45) circle[radius=2.5pt];
    \draw[fill=C3, draw=C3] (0.75,-0.75) circle[radius=2.5pt];
    \draw[fill=C2, draw=C2] (-0.75,0.45) circle[radius=2.5pt];
    \draw[fill=C4, draw=C4] (-0.75,0.15) circle[radius=2.5pt];
    \draw[fill=C1, draw=C1] (-0.75,-0.15) circle[radius=2.5pt];
    \draw[fill=C2, draw=C2] (-0.75,-0.45) circle[radius=2.5pt];
    \draw[fill=C3, draw=C3] (-0.75,-0.75) circle[radius=2.5pt];
    
    \draw[fill=C2, draw=C2] (-0.45,0.15) circle[radius=2.5pt];
    \draw[fill=C3, draw=C3] (-0.45,-0.15) circle[radius=2.5pt];
    \draw[fill=C4, draw=C4] (-0.45,-0.45) circle[radius=2.5pt];
    \draw[fill=C2, draw=C2] (-0.45,-0.75) circle[radius=2.5pt];

    \draw[fill=C2, draw=C2] (-0.15,-0.15) circle[radius=2.5pt];
    \draw[fill=C1, draw=C1] (-0.15,-0.45) circle[radius=2.5pt];
    \draw[fill=C1, draw=C1] (-0.15,-0.75) circle[radius=2.5pt];
    
    \draw[fill=C2, draw=C2] (0.15,-0.45) circle[radius=2.5pt];
    \draw[fill=C3, draw=C3] (0.15,-0.75) circle[radius=2.5pt];
    
    \draw[fill=C4, draw=C4] (0.45,-0.75) circle[radius=2.5pt];
    \draw[fill=C2, draw=C2] (0.75,0.45) circle[radius=2.5pt];
    \draw[fill=C4, draw=C4] (0.75,0.15) circle[radius=2.5pt];
    \draw[fill=C1, draw=C1] (0.75,-0.15) circle[radius=2.5pt];
    \draw[fill=C2, draw=C2] (0.75,-0.45) circle[radius=2.5pt];
    \draw[fill=C3, draw=C3] (0.75,0.75) circle[radius=2.5pt];
    
    \draw[fill=C2, draw=C2] (0.45,0.15) circle[radius=2.5pt];
    \draw[fill=C3, draw=C3] (0.45,-0.15) circle[radius=2.5pt];
    \draw[fill=C4, draw=C4] (0.45,0.45) circle[radius=2.5pt];
    \draw[fill=C2, draw=C2] (0.45,0.75) circle[radius=2.5pt];

    \draw[fill=C2, draw=C2] (0.15,0.15) circle[radius=2.5pt];
    \draw[fill=C1, draw=C1] (0.15,0.45) circle[radius=2.5pt];
    \draw[fill=C1, draw=C1] (0.15,0.75) circle[radius=2.5pt];
    
    \draw[fill=C2, draw=C2] (-0.15,0.45) circle[radius=2.5pt];
    \draw[fill=C3, draw=C3] (-0.15,0.75) circle[radius=2.5pt];
    \draw[fill=C4, draw=C4] (-0.45,0.75) circle[radius=2.5pt];
    \node at (0,-1.75cm) {$\rho(0)=\!\!\!\!\displaystyle{\sum_{n,m=1}^{\mathfrak{D}}}{P}_n\,\rho\,P_m$};
    \end{scope}
    
    \begin{scope}[shift={(4cm,0)}]
    \draw[step=0.3cm,gray,very thin] (-1,-1) grid (1,1);
    \draw[fill=C3!5!white] (-0.9,0.9) rectangle (-0.6,0.6);
    \draw[fill=C3!5!white] (-0.6,0.6) rectangle (-0.3,0.3);
    \draw[fill=C3!5!white] (-0.3,0.3) rectangle (0,0);
    \draw[fill=C3!5!white] (0,0) rectangle (0.6,-0.6);
    \draw[fill=C3!5!white] (0.6,-0.6) rectangle (0.9,-0.9);
    \draw[fill=C2, draw=C2] (-0.75,0.75) circle[radius=2.5pt];
    \draw[fill=C4, draw=C4] (-0.45,0.45) circle[radius=2.5pt];
    \draw[fill=C1, draw=C1] (-0.15,0.15) circle[radius=2.5pt];
    \draw[fill=C4, draw=C4] (0.15,-0.15) circle[radius=2.5pt];
    \draw[fill=C3, draw=C3] (0.45,-0.45) circle[radius=2.5pt];
    \draw[fill=C3, draw=C3] (0.75,-0.75) circle[radius=2.5pt];
    \shade[inner color=C2, outer color=white] (-0.75,0.45) circle[radius=2.5pt];
    \shade[inner color=C4, outer color=C4!50!white] (-0.75,0.15) circle[radius=2.5pt];
    \shade[inner color=C1, outer color=C1!20!white] (-0.75,-0.15) circle[radius=2.5pt];
    \shade[inner color=C2!10!white, outer color=white] (-0.75,-0.45) circle[radius=2.5pt];
    \shade[inner color=C3!20!white, outer color=white] (-0.75,-0.75) circle[radius=2.5pt];
    
    \shade[inner color=C2, outer color=white] (-0.45,0.15) circle[radius=2.5pt];
    \shade[inner color=C3!50!white, outer color=white] (-0.45,-0.15) circle[radius=2.5pt];
    \shade[inner color=C4, outer color=white] (-0.45,-0.45) circle[radius=2.5pt];
    \shade[inner color=C2, outer color=white] (-0.45,-0.75) circle[radius=2.5pt];

    \shade[inner color=C2!40!white, outer color=white] (-0.15,-0.15) circle[radius=2.5pt];
    \shade[inner color=C1, outer color=C1!60!white] (-0.15,-0.45) circle[radius=2.5pt];
    \shade[inner color=C1, outer color=white] (-0.15,-0.75) circle[radius=2.5pt];
    
    \draw[fill=C2,draw=C2] (0.15,-0.45) circle[radius=2.5pt];
    \shade[inner color=C3!10!white, outer color=white] (0.15,-0.75) circle[radius=2.5pt];
    
    \shade[inner color=C4, outer color=white] (0.45,-0.75) circle[radius=2.5pt];
    \shade[inner color=C2, outer color=white] (0.75,0.45) circle[radius=2.5pt];
    \shade[inner color=C4, outer color=white] (0.75,0.15) circle[radius=2.5pt];
    \shade[inner color=C1, outer color=C1!50!white] (0.75,-0.15) circle[radius=2.5pt];
    \shade[inner color=C2!30!white, outer color=white] (0.75,-0.45) circle[radius=2.5pt];
    \shade[inner color=C3, outer color=C3!40!white] (0.75,0.75) circle[radius=2.5pt];
    
    \shade[inner color=C2, outer color=white] (0.45,0.15) circle[radius=2.5pt];
    \shade[inner color=C3, outer color=C3] (0.45,-0.15) circle[radius=2.5pt];
    \shade[inner color=C4!10!white, outer color=white] (0.45,0.45) circle[radius=2.5pt];
    \shade[inner color=C2, outer color=white] (0.45,0.75) circle[radius=2.5pt];

    \shade[inner color=C2!10!white, outer color=white] (0.15,0.15) circle[radius=2.5pt];
    \shade[inner color=C1, outer color=white] (0.15,0.45) circle[radius=2.5pt];
    \shade[inner color=C1!20!white, outer color=white] (0.15,0.75) circle[radius=2.5pt];
    
    \shade[inner color=C2, outer color=white] (-0.15,0.45) circle[radius=2.5pt];
    \shade[inner color=C3, outer color=white] (-0.15,0.75) circle[radius=2.5pt];
    \shade[inner color=C4!5!white, outer color=white] (-0.45,0.75) circle[radius=2.5pt];
    \node at (0,-1.75cm) {$\overline{\rho}^T\!\!=\!\!\!\!\displaystyle{\sum_{n,m=1}^{\mathfrak{D}}}\!\!{G}_{nm}^{(T)}{P}_n\,\rho\,P_m$};
    \end{scope}
    \begin{scope}[shift={(8cm,0)}]
    \draw[step=0.3cm,gray,very thin] (-1,-1) grid (1,1);
    \draw[fill=C3!5!white] (-0.9,0.9) rectangle (-0.6,0.6);
    \draw[fill=C3!5!white] (-0.6,0.6) rectangle (-0.3,0.3);
    \draw[fill=C3!5!white] (-0.3,0.3) rectangle (0,0);
    \draw[fill=C3!5!white] (0,0) rectangle (0.6,-0.6);
    \draw[fill=C3!5!white] (0.6,-0.6) rectangle (0.9,-0.9);
    \draw[fill=C2, draw=C2] (-0.75,0.75) circle[radius=2.5pt];
    \draw[fill=C4, draw=C4] (-0.45,0.45) circle[radius=2.5pt];
    \draw[fill=C1, draw=C1] (-0.15,0.15) circle[radius=2.5pt];
    \draw[fill=C4, draw=C4] (0.15,-0.15) circle[radius=2.5pt];
    \draw[fill=C3, draw=C3] (0.45,-0.45) circle[radius=2.5pt];
    \draw[fill=C3, draw=C3] (0.75,-0.75) circle[radius=2.5pt];
    \draw[fill=C2, draw=C2] (0.15,-0.45) circle[radius=2.5pt];
    \draw[fill=C3, draw=C3] (0.45,-0.15) circle[radius=2.5pt];
    \node at (0,-1.75cm) {$\omega=\displaystyle{\sum_{n=1}^{\mathfrak{D}}}{P}_n\,\rho\,P_n$};
    \end{scope}
    \end{tikzpicture}
}
    \caption[The equilibrium state as a dephasing of the initial state]{\textbf{(Motivated by Ref.~\cite{Gogolin_2016}.) Dephasing as the mechanism for equilibration:} An initial state $\rho(0)=\rho$ in the energy eigenbasis gets dephased with respect to the Hamiltonian $H$ when averaged over time, $\overline{\rho}^T$, where here we denote $G_{nm}^{(T)}=\overline{\ex^{-it(E_n-E_m)}}^T$, ultimately being totally dephased in the infinite time-average limit, corresponding to $\omega=\lim_{T\to\infty}\overline{\rho}^T$.}
    \label{fig: omega dephased}
\end{figure}

\begin{remark}
In Ref.~\cite{Gogolin_2011} (and with more detail in Ref.~\cite{Gogolin_2016}) it is shown that $\omega$ is the unique state that maximises the von Neumann entropy among all others with the same expectation on all conserved quantities. That is, for all $\sigma\in\$(\mscr{H})$ such that $\tr(O_i\sigma)=\tr(O_i\omega)$ on all observables $\{O_i:[O_i,H]=0\}$, the state $\omega$ is the one that maximizes the von Neumann entropy.

Specifically, the von Neumann entropy of a state $\rho\in\$(\mscr{H})$, defined previously by means of Eq.~\eqref{eq: von Neumann entropy},  is a concave function~\cite{peres2006quantum}, i.e. for any states $\{\sigma_i\in\$(\mscr{H})\}$ and positive real numbers $\{\lambda_i\in\mbb{R}^+:\sum\lambda_i=1\}$, we have $S\left(\sum\lambda_i\sigma_i\right)\geq\sum\lambda_i{S}(\sigma_i)$. It then follows from Theorem V.2.1\footnote{ There one has to note that a function $f$ is said to be operator concave if $-f$ is operator convex.} in Ref.~\cite{Bhatia_2013} that $S(\omega)\geq{S}(\rho)$, since the dephasing with respect to $H$ is a so-called \emph{pinching}.\footnote{ This is the argument given in Ref.~\cite{Gogolin_2011,Gogolin_2016}; however, it also can be seen to follow from the so-called \emph{data processing inequality}~\cite{Lindblad_1975}, which implies $S(\Phi(\rho))\geq{S}(\rho)$ for any unital quantum channel $\Phi$, of which the dephasing map with respect to $H$ is a particular case. These concepts, however, will be introduced until the following chapter.}

Moreover, given any $\sigma_1,\sigma_2\in\$(\mscr{H})$, these yield the same eigenvalues on all conserved quantities $O_i$ if and only if $\sum{P_n\sigma_1P_n}=\sum{P_n\sigma_2P_n}$, and uniqueness can be shown by optimization~\cite{Gogolin_2016}.
\end{remark}

This result is remarkable in that one can get a maximum entropy principle\footnote{ While L. Boltzmann in Ref.~\cite{Boltzmann} and W. Gibbs in Ref.~\cite{gibbs2014} obtained an ever-increasing entropy result by physical considerations, E.T. Jaynes derived a maximum entropy principle in Ref.~\cite{Jaynes1} for the classical case and in Ref.~\cite{Jaynes2} for the quantum one (which is already conceptually similar to modern ones such as Ref.~\cite{Rigol_2007}).} purely from unitary quantum dynamics, as opposed to obtaining it from a probabilistic interpretation imposed onto a classical statistical theory. The same principle is satisfied in the respective conditions for a thermodynamic equilibrium with $\rho_\mc{G}\approx\exp(-\beta{H})/\tr[\exp(-\beta{H})]$ for an inverse temperature $\beta$ such that $\tr(\rho_\mc{G}H)=E$~\cite{peres2006quantum}. This maximum entropy principle, however, still does not explain the equal a-priori probabilities postulate or how quantum states evolve towards equilibrium.

\subsection{Temporal fluctuations}\label{sec: fluctuations infinite}
We are interested then in the behaviour for most times of the difference between the time-evolved expected value of an arbitrary observable, $\langle{A}\rangle_{\rho(t)}=\tr[A\rho(t)]$, and that in equilibrium, $\langle{A}\rangle_\omega=\tr[A\omega]$. By definition, the infinite temporal average of the difference will be identically equal to zero. The results in Ref.~\cite{Reimann_2008, Linden_2009, Short_2011} arguably constitute the seminal modern approach to dynamical equilibration; their main focus is on the temporal fluctuations of the expectation value of an observable around equilibrium, i.e. the variance
\begin{equation}
    \overline{|\langle{A}\rangle_{\rho(t)}-\langle{A}\rangle_\omega|^2}=\overline{|\tr[A(\rho(t)-\omega)]|^2}.
\end{equation}

Consider first a system with a non-degenerate Hamiltonian $H=\sum{E}_n|n\rangle\!\langle{n}|$ (so that $\mathfrak{D}=d$) and a pure initial state $\Psi=|\psi\rangle\!\langle\psi|\in\$(\mscr{H})$, then we have
\begin{align}
    \Psi(t)-\omega&=\sum_{n,m=1}^d\ex^{-it(E_m-E_n)}\Psi_{mn}|m\rangle\!\langle{n}|-\sum_{n=1}^d\Psi_{nn}|n\rangle\!\langle{n}|\nonumber\\
    &=\sum_{n \neq m}^d\ex^{-it(E_m-E_n)}\Psi_{mn}|m\rangle\!\langle{n}|,
\end{align}
where $\Psi_{mn}=\langle{m}|\psi\rangle\!\langle\psi|n\rangle$, and so, for any general operator $A\in\mscr{B}(\mscr{H})$,
\begin{align}
    \overline{|\langle{A}\rangle_{\Psi(t)}-\langle{A}\rangle_\omega|^2}&=\overline{\left|\sum_{n \neq m}^d\ex^{-it(E_m-E_n)}\Psi_{mn}A_{nm}\right|^2}=\sum_{\substack{n \neq m \\ \nu \neq \mu}}^d\overline{\ex^{-it(E_m-E_\mu-E_n+E_\nu)}}\Psi_{mn}\Psi_{\nu\mu}A_{nm}A_{\mu\nu}^*.
    \label{eq: variance Short 1}
\end{align}

An additional assumption is now made in Ref.~\cite{Linden_2009, Short_2011}, labelled a non-degenerate gap condition, so that no gap between energy levels occur more than once in the energy spectrum. For the average in Eq.~\eqref{eq: variance Short 1} it implies that
\begin{equation}
    \begin{matrix}E_m-E_\mu=E_n-E_\nu\\m\neq{n},\,\mu\neq\nu\end{matrix}\qquad\text{iff}\qquad\begin{matrix}m=\mu\\n=\nu\end{matrix}\quad,
    \label{eq: non-degenerate gaps}
\end{equation}
so that
\begin{equation}
    \overline{\ex^{-it(E_m-E_\mu-E_n+E_\nu)}}=\delta_{m\mu}\delta_{n\nu}.
\end{equation}

This condition was already considered by von Neumann~\cite{vonNeumann2010}, who called it a non-resonance condition, and it can be motivated physically when considering a bipartition of the full system to ensure interaction between both parts~\cite{Linden_2009, Gogolin_2016}. Now, $|\Psi_{nm}|^2=\Psi_{nn}\Psi_{mm}$, hence
\begin{align}
    \overline{|\langle{A}\rangle_{\Psi(t)}-\langle{A}\rangle_\omega|^2}&=\sum_{m \neq n}^d\Psi_{mm}\Psi_{nn}|A_{nm}|^2\nonumber\\
    &\leq\sum_{m, n}^d\Psi_{mm}A_{nm}\Psi_{nn}A_{mn}^*\nonumber\\
    &=\tr\left[\left(\sum_{n=1}^d\Psi_{nn}|n\rangle\!\langle{n}|A\right)\,\left(\sum_{m=1}^d\Psi_{mm}|m\rangle\!\langle{m}|A^\dg\right)\right]\nonumber\\
    &=\tr[\omega\,A\,\omega\,A^\dg],
\end{align}
where in the second line we included all indices $m=n$. Now we can apply the Cauchy-Schwarz inequality, which reads $\tr[P^\dg{Q}]^2\leq\tr[P^\dg{P}]\tr[Q^\dg{Q}]$ for any $P,Q\in\mscr{M}_\ell$, with $\mscr{M}_\ell$ the space of complex $\ell\times\ell$ matrices; this gives
\begin{align}
    \overline{|\langle{A}\rangle_{\Psi(t)}-\langle{A}\rangle_\omega|^2}&\leq\left(\tr\left[AA^\dg\omega^2\right]\tr\left[A^\dg{A}\omega^2\right]\right)^{1/2}.
\end{align}

Now we can employ H\"{o}lder's inequality, which says that
\begin{equation}
    |\tr[P^\dg{Q}]|\leq\|P\|_a^{1/a}\|Q\|_b^{1/b},\qquad\text{with}\quad\f{1}{a}+\f{1}{b}=1,
    \label{eq: Holders inequality}
\end{equation}
for any $a,b\in\mbb{R}^+_0$, where $\mbb{R}^+_0$ denotes the set of non-negative real numbers, and where
$\|X\|_p$ is the Schatten $p$-norm, defined in Eq.~\eqref{eq: Schatten p-norm}. We can thus take $a=1$, $b\to\infty$, so that
\begin{equation}
    \overline{|\langle{A}\rangle_{\Psi(t)}-\langle{A}\rangle_\omega|^2}\leq\|A\|^2\tr(\omega^2),
    \label{eq: Short infinite obs}
\end{equation}
as all quantities are positive.

This is already in essence the derivation made in Ref.~\cite{Short_2011}. The purity of the equilibrium state is given by
\begin{align}
    \tr(\omega^2)&=\sum_{n=1}^d\langle{n}|\Psi|n\rangle\!\langle{n}|\Psi|n\rangle=\sum_{n=1}^d\left(\tr\left[|n\rangle\!\langle{n}|\Psi\right]\right)^2,
\end{align}
which, written as in the second equality, can be directly read as the sum of squares of probabilities for each of the energy eigenstates to be occupied by the initial state. Indeed, if we consider a general degenerate Hamiltonian $H=\sum{E}_n{P}_n$ and any given initial state $\rho$, we can define
\begin{align}
    d_\text{eff}^{-1}(\rho):=\sum_{n=1}^\mathfrak{D}(\tr[P_n\rho])^2,
    \label{eq: deff inverse}
\end{align}
which is the so-called inverse \emph{effective dimension} of the state $\rho$, also labelled inverse participation ratio~\cite{PhysRevE.85.060101, Calixto_2015}. It satisfies $1\leq{d}_\text{eff}\leq\mathfrak{D}\leq{d}$, with the lower bound saturated when the initial state is an energy eigenstate and either upper bound saturated when the occupation probability is the same across all eigenspaces or eigenstates. Notice that $d_\text{eff}(\rho(t))=d_\text{eff}(\rho)$ is independent of time, and it only corresponds to the purity of $\omega$ when either the Hamiltonian is non-degenerate or when the initial state is pure.

\begin{figure}[t]
    \centering
    \begin{tikzpicture}[scale=0.45]
    \draw[thick,<->] (0,-4) -- (0,6);
    \draw[thick, -] (0,0) -- (8,0);
    \draw[-, dashed] (8,0) -- (12,0);
    \draw[thick, ->] (12,0) -- (21,0);
    \draw [ultra thick, C2] plot [smooth, tension=1] coordinates { (0,2) (1.1,5) (2,0) (3,0.3) (4,0) (5,0.1) (6,0) (7,-0.2) (8,0) (9,0)};
    \draw [very thick, C2] plot [smooth, tension=1] coordinates { (12,0) (13,-0.1) (14,0.1) (15,0) (16,0.2) (17,0) (18,-3) (19,0) (20,0) (21,0)};
    \node[rotate=90, above, shift={(0.2,0)}] at (0,0.2) {\Large{$\langle{A}\rangle_{\rho(t)}-\langle{A}\rangle_\omega$}};
    \node[right] at (21,0) {\textsf{Time} $t$};
    \node[below] at (1,0) {$0$};
    \draw[-, dashed] (1,0) -- (1,5);
    \end{tikzpicture}
    \caption[Equilibration on average on expectation values]{\textbf{Equilibration on average on expectation values:} An observable $A$ on a space with state $\rho(t)$ equilibrates on average if its expectation value remains close to the one on equilibrium $\omega$ for most times; this is guaranteed whenever the fluctuations around equilibrium are suppressed, which occurs whenever the overlap of $\rho$ with the energy eigenstates is large.}
    \label{fig: Short equilibration}
\end{figure}
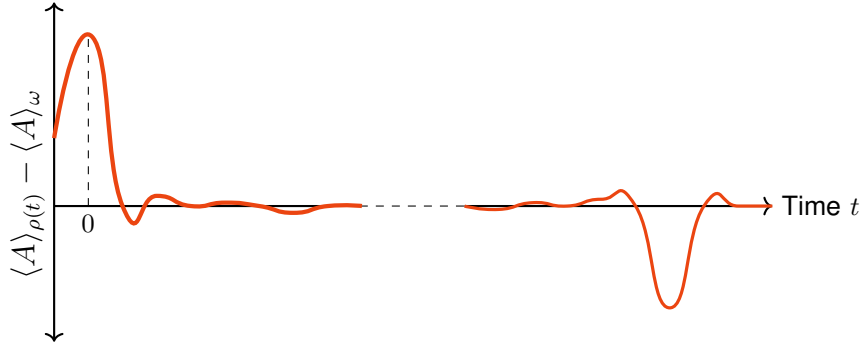

The extension of Eq.~\eqref{eq: Short infinite obs} to mixed states when the Hamiltonian is non-degenerate can be made by exploiting the fact that $|\rho_{nm}|^2\leq\rho_{nn}\rho_{mm}$, which follows from the positivity of $\rho$.\footnote{ Every principal $2\times2$ submatrix $\varrho_{ij}$ of $\rho$ is positive semi-definite~\cite{Horn_2012}, then for any such matrix, $\det(\varrho_{ij})=\rho_{ii}\rho_{jj}-|\rho_{ij}|^2\geq0$ from positivity of eigenvalues.} If the Hamiltonian is degenerate, $H=\sum_{n=1}^{\mathfrak{D}}{E}_nP_n$ with the $n$\textsuperscript{th} level having degeneracy $\ell$, but the initial state is pure, $\Psi=|\psi\rangle\!\langle\psi|$, we have
\begin{equation}
    |\psi(t)\rangle=\sum_{n=1}^{\mathfrak{D}}\sum_{j=1}^\ell\ex^{-itE_n}|n_j\rangle\!\langle{n}_j|\psi\rangle,
\end{equation}
so that the eigenbasis $\{|n_j\rangle\}$ for $H$ can be chosen such that the initial state $|\psi\rangle$ has an overlap with only a single eigenstate $|n_{j^\prime}\rangle$ for each energy level $E_n$. This means $|\psi\rangle$ will evolve as if acted on with a non-degenerate Hamiltonian $H^\prime=\sum_{n=1}^{\mathfrak{D}}E_n|n\rangle\!\langle{n}|$. In this case, it is clear that we still have $\tr(\omega^2)=d_\text{eff}^{-1}\left(\Psi\right)$.

More generally, if we consider both mixed states and degenerate energies, the general form of Eq.~\eqref{eq: Short infinite obs} can be seen to follow by purification~\cite{Short_2011}. That is, given $\rho\in\$(\mscr{H})$, take $\Psi\in\$(\mscr{H}\otimes\mscr{H})$ evolving under $\tilde{H}=H\otimes\mbb1$, i.e. with the ancilla evolving trivially, thus not affecting the spectrum or degeneracies of the original system. Then we have $\tr[A\rho]=\tr[(A\otimes\mbb1_{\mscr{H}})\Psi]$, i.e. the expectation values of $A$ on the original system coincide with those of $A\otimes\mbb1$ and also $\|A\otimes\mbb1\|=\|A\|$, i.e. the maximum singular value of $A$ on the original system also coincides with that of $A\otimes\mbb1$. Finally $d_\text{eff}(\Psi)=d_\text{eff}(\rho)$, as the full trace can be split as a composition of partial traces, i.e. $\tr_{\tilde{\mscr{H}}}=\tr_{\mscr{H}_1}\circ\tr_{\mscr{H}_2}=\tr_{\mscr{H}_2}\circ\tr_{\mscr{H}_1}$. Notice, however, that even though $\tr(\omega^{\prime\,2})=d_\text{eff}^{-1}(\rho)$, the purities do not coincide, and in fact $\tr(\omega^{\prime\,2})\geq\tr(\omega^2)$.

The result in Eq.~\eqref{eq: Short infinite obs} can thus more generally be written as
\begin{equation}
    \overline{|\langle{A}\rangle_{\rho(t)}-\langle{A}\rangle_\omega|^2}\leq\f{\|A\|^2}{d_\text{eff}(\rho)},
    \label{eq: Short infinite deff}
\end{equation}
which highlights both the role of $d_\text{eff}(\rho)$ in determining whether a system will equilibrate or not, as well as that of $\|A\|$ as a scale term for how well the observable can tell between the time-evolving state and equilibrium. That is, $d_\text{eff}(\rho)$ tells us if the temporal fluctuations around equilibrium will be suppressed or not, whilst the norm $\|A\|$ can be further restricted to an experimentally reasonable resolution~\cite{Reimann_2008}. We can illustrate the whole idea of observable equilibration on average as in Fig.~\ref{fig: Short equilibration}.

The punchline of Eq.~\eqref{eq: Short infinite deff} is that expectation values of observables in closed systems will equilibrate whenever the total occupation of the energy eigenstates by the initial state is small, or equivalently, when the overlap of the initial state with every single energy eigenstate is large so that Eq.~\eqref{eq: deff inverse} is small. That the effective dimension is typically large can be argued for small subsystems of large whole systems and, in particular, it is normally expected for macroscopic systems just by the strikingly large size these have. Taking into account that a macroscopic system has $\sim\mc{O}(10^{23})$ degrees of freedom, it's conceivable that an experimentalist will only be able to prepare a state that overlaps significantly only a very few of them so that the composite will still look quite mixed across all energy eigenstates.

Even when the experimentalist can prepare the initial state with low uncertainty and the levels can be occupied extremely unequally, it can be argued that at realistic scales the effective dimension will remain very large~\cite{Reimann_2012}. Notice that Eq.~\eqref{eq: Short infinite deff} applies the same if the observable $A$ acts only on a subpart, \gls{syst}, of a larger composite \gls{syst-env} because $\|A_\mathsf{S}\otimes\mbb1_\mathsf{E}\|=\|A_\mathsf{S}\|$ and $\tr[A_\mathsf{S}\tr_E(X)]=\tr[(A_\mathsf{S}\otimes\mbb1_\mathsf{E})X]$. Moreover, with the mathematical concepts that will be introduced in the next section, it was shown in Ref.~\cite{Linden_2009} that the effective dimension is typically large, i.e. for any pure state chosen at random from a large subspace of a Hilbert space the probability for the effective dimension to be small is exponentially small in the subspace dimension.

While these arguments make the case for the effective dimension being typically large, it seems clear that in general it would not be efficient to compute such a quantity. In Ref.~\cite{Brandao_return} a feasible way to decide if this is the case is presented for $k$-local Hamiltonians, i.e. Hamiltonians acting on a lattice of quantum systems with $H=\sum{h}_i$, with $h_i$ acting on the nearest $k$ sites to the $i\textsuperscript{th}$ one. There it is shown that indeeed the effective dimension is typically large for states with exponentially decaying correlations, i.e. with $\rho$ such that $|\langle{A}{B}\rangle_\rho-\langle{A}\rangle_\rho\langle{B}\rangle_\rho|$ maximized over operators $A$, $B$, decays exponentially in the distance between the support of the sites on which $A$ and $B$ act. Similarly, an equivalence of equilibrium ensembles, i.e. a statement that different macroscopic descriptions of the thermal equilibrium state lead to the same predictions, was shown in Ref.~\cite{Muller2015, Brandao_2015, Tasaki_2018}.

Finally, we can make a statistically relevant statement that fully evokes the definition of equilibration on average for expectation values of observables, such as that in Ref.~\cite{Reimann_2008}, by making use of Chebyshev's inequality, which states that $\mbb{P}[|X-\mu|\geq\kappa\,\sigma]\leq\kappa^{-2}$, holding for any $\kappa>0$ with $X$ a random variable with mean $\mu$ and variance $\sigma^2$. By picking $\kappa=\|A\|\left(d_\text{eff}^{1/3}\sigma\right)^{-1}$ with $\sigma:=\overline{|\tr[{A}({\rho(t)}-\omega)]|^2}$, together with the bound in Eq.~\eqref{eq: Short infinite deff}, this leads to
\begin{equation}
    \mbb{P}_t\left[\left|\langle{A}\rangle_{\rho(t)}-\langle{A}\rangle_\omega\right|\geq\f{\|A\|}{d_\text{eff}^{1/3}}\right]\leq\f{1}{d_\text{eff}^{1/3}},
\end{equation}
which explicitly states that for any time $t\in\mbb{R}^+$ drawn uniformly at random, the expectation value of $A$ with respect to $\rho(t)$ will be close to that with respect to $\omega$ for most times, whenever the effective dimension is large. This also quantitatively captures the notion mentioned previously and sketched in Fig.~\ref{fig: Short equilibration} that the departures from equilibrium are rare.

\subsection{Fluctuations within a finite time}
\label{sec: fluctuations finite time}
While the previous results show that dynamical equilibration is attained under remarkably mild conditions, two clear questions remain: one relates to the restriction on non-degeneracy of energy gaps and the second regards the timescale on which equilibration occurs. While the non-degenerate gaps condition ensures fully interacting systems regardless of how these are partitioned, it is still a restrictive one that could leave out physically relevant Hamiltonians. The second question is implicit in the previous result, as, while it ensures that equilibration for expectation values \emph{will} take place, it says nothing about \emph{when} it will take place.

One of the most significant steps taken in this direction was that in Ref.~\cite{Short_finite}: the non-degenerate energy gap condition was relaxed to one quantifying the number of degenerate gaps, and the averaging time window was restricted to be finite. As per Eq.~\eqref{eq: Short infinite deff}, this is expected to lead to an upper-bound that gives a correction with respect to the number of degenerate gaps and the width of the time-averaging window. Let us begin by considering again a pure state $\Psi\in\$(\mscr{H})$, and as per the argument above, a degenerate Hamiltonian $H=\sum{E}_nP_n$ which, however, has its eigenbasis chosen so that $\Psi$ evolves as if it was doing so under a non-degenerate $H$. Then we now care about an analogous quantity to the variance above, but now with a uniform time-average within an interval of width $T$, i.e.
\begin{equation}
    \overline{|\langle{A}\rangle_{\rho(t)}-\langle{A}\rangle_\omega|^2}^T=\sum_{\substack{n \neq m \\ \nu \neq \mu}}^\mathfrak{D}\overline{\ex^{-it(E_m-E_\mu-E_n+E_\nu)}}^T\Psi_{mn}\Psi_{\nu\mu}A_{nm}A_{\mu\nu}^*,
\end{equation}
where the finite time-average is defined in Eq.~\eqref{eq: finite time avg}. Now a way to simplify notation into a more familiar one is to label the energy gaps with $\ell:=(m,n)$, $\lambda:=(\mu,\nu)$, so that
\begin{equation}
    M^{(T)}_{\ell\lambda}:=\overline{\ex^{it(\mc{E}_\ell-\mc{E}_\lambda)}}^T
    \qquad\text{where}\qquad
    \begin{matrix}
    \mc{E}_\ell=E_m-E_n\\
    \mc{E}_\lambda=E_\mu-E_\nu\end{matrix}\,,
\end{equation}
and $v_\ell:=v_{(m,n)}=\Psi_{mn}A_{nm}$, $v_\lambda:=v_{(\mu,\nu)}=\Psi_{\mu\nu}A_{\nu\mu}$. Then it becomes clear that
\begin{align}
    \overline{|\langle{A}\rangle_{\rho(t)}-\langle{A}\rangle_\omega|^2}^T&=\sum_{\ell,\lambda}v_{\lambda}^*\, M^{(T)}_{\lambda\ell}\,v_\ell\nonumber\\
    &\leq\|M\|\|v\|_2^2=\|M\|\sum_{n\neq{m}}^\mathfrak{D}\Psi_{mn}A_{nm}\Psi_{nm}A_{mn}^*\nonumber\\
    &\leq\f{\|M\|\|A\|^2}{d_\text{eff}(\rho)},
\end{align}
where the second line is equivalent to $v^\dg{M}v\leq\lambda\|v\|_2^2$ for the Hermitian matrix $M$ with components $M_{\lambda\ell}^{(T)}$ and maximum singular value $\lambda$, and where in the third line $\|\cdot\|_2$ is the Schatten 2-norm defined in Eq.~\eqref{eq: Schatten p-norm}, with all remaining steps following as in Eq.~\eqref{eq: Short infinite deff}. This is effectively a correction to Eq.~\eqref{eq: Short infinite deff}, which reduces to it when $H$ has no degenerate energy gaps and when the infinite time limit is taken.

The components of $M$ are explicitly given by
\begin{equation}
    M_{\ell\lambda}^{(T)}=\begin{cases}1&\text{if}\qquad\mc{E}_\ell=\mc{E}_\lambda\\
    \displaystyle{\f{\exp[iT(\mc{E}_\ell-\mc{E}_\lambda)]-1}{iT(\mc{E}_\ell-\mc{E}_\lambda)}}&\text{otherwise}
    \end{cases},
    \label{eq: components M finite average}
\end{equation}
and we can use the hierarchy of Schatten norms in Eq.~\eqref{eq: Schatten hierarchy} to get
\begin{equation}
    \|M\|\leq\max_{\lambda}\sum_{\ell=1}^{\mathfrak{D}(\mathfrak{D}-1)} |M_{\ell\lambda}|,
    \label{eq: initial bound M finite average}
\end{equation}
where from Eq.~\eqref{eq: components M finite average} it follows that $|M_{\ell\lambda}|\leq1$, as all components have modulus less than one. Now the second key definition is to let $N(\varepsilon)$ be the maximum number of energy gaps in any interval of size $\varepsilon>0$,
\begin{equation}
    N(\varepsilon):=\max_E\{\ell:\ell\in\mscr{E},\,\mc{E}_\ell\in[E,\,E+\varepsilon]\},
    \label{eq: N(epsilon)}
\end{equation}
where we defined the set of labels $\mscr{E}:=\{(n,m):n,m\in\{1,2,\ldots,\mathfrak{D}\},n\neq{m}\}$. The maximum degeneracy of any energy gap is given by $\mathfrak{G}_\mc{E}:=\lim_{\varepsilon\to0^+}N(\varepsilon)$, with the non-degenerate case corresponding to $\mathfrak{G}_\mc{E}=1$. This implies that there are at most $N(\varepsilon)$ energy gaps such that
\begin{equation}
    (k-1/2)\varepsilon\leq\mc{E}_\ell-\mc{E}_\lambda<(k+1/2)\varepsilon,
    \label{eq: energy gaps with k}
\end{equation}
for any non-zero integer $k$, otherwise just taking $|M_{\ell\lambda}|\leq1$. Then Eq.~\eqref{eq: initial bound M finite average} can be bounded by,
\begin{equation}
    \|M\|\leq N(\varepsilon)\left\{1+2\sum_{k=1}^{\mathfrak{D}(\mathfrak{D}-1)/2}\f{2}{T(k-1/2)\varepsilon}\right\},
\end{equation}
where the first term comes from the $k=0$ contribution and the second term uses Eq.~\eqref{eq: energy gaps with k}; the sum is maximised by having as many values with small $|k|$ as possible, which gives the factor of $2$ at the front of the sum. Now the sum can be further bounded~\cite{Short_finite} so that
\begin{equation}
    \|M\|\leq{N}(\varepsilon)\left(1+\f{8\log_2\mathfrak{D}}{\varepsilon T}\right).
\end{equation}

This renders the final result for equilibration on average for expectation values of observables within a finite time interval:
\begin{theorem}[Observable equilibration in finite time~\cite{Short_2011}]\label{thm: Short finite time observable equilbration}
Given a quantum system in state $\rho(t)\in\$(\mscr{H})$ at time $t$, evolving via a time-independent Hamiltonian with $\mathfrak{D}$ distinct energies, then for any operator $A\in\mscr{B}(\mscr{H})$ and any energy $\varepsilon>0$ and time $T>0$,
\begin{equation}
    \overline{|\langle{A}\rangle_{\rho(t)}-\langle{A}\rangle_\omega|^2}^T\leq\f{\|A\|^2}{d_\text{eff}(\rho)}N(\varepsilon)\left(1+\f{8\log_2\mathfrak{D}}{\varepsilon{T}}\right),
    \label{eq: Short finite time}
\end{equation}
where $\|\cdot\|$ denotes largest singular value, and with inverse effective dimension $d_\text{eff}^{-1}(\varrho)$ and $N(\varepsilon)$, the maximum number of energy gaps in an interval or width $\varepsilon$, defined in Eq.~\eqref{eq: deff inverse} and in Eq.~\eqref{eq: N(epsilon)}, respectively.
\end{theorem}

It is entirely clear now that Eq.~\eqref{eq: Short finite time} reduces to Eq.~\eqref{eq: Short infinite deff} in the non-degenerate energy gaps case, $\mathfrak{G}_\mc{E}:=\lim_{\varepsilon\to0^+}N(\varepsilon)=1$ together with the infinite time-window limit $T\to\infty$. Thus the previous conditions on the effective dimension $d_\text{eff}$ as well as the resolution of $A$ for equilibration to occur remain, while now the correction factor is also required to be small. The first term only requires a low energy gap degeneracy, as opposed to restrict to no gap degeneracy at all, while the second already touches upon one of the most relevant problems for equilibration, namely determining a relevant timescale within which it will occur.

The biggest issue with having equilibration on average in an infinite time-average window is self evident in that equilibration could take the age of the universe to manifest, and while the prediction for its occurrence would still be correct, it would nevertheless be meaningless. By inspecting Eq.~\eqref{eq: Short finite time}, a time-averaging window of width
\begin{equation}
    T\gtrsim\f{\log_2\mathfrak{D}}{\varepsilon},
    \label{eq: Short timescale}
\end{equation}
will give equilibration provided the remaining quantities in the bound are small, becoming of the order of the infinite-time-average bound in Eq.~\eqref{eq: Short infinite deff}. Notice the minimum $\varepsilon$ is not an ideal choice: while this parameter can be picked arbitrarily, how large both $N(\varepsilon)$ and the bound in Eq.~\eqref{eq: Short timescale} have to be taken into account. This is problematic as it leads to a dependence in the system size and could be at odds with recurrence times~\cite{Hemmer_1958, Bhattacharyya_recurrence,Wallace_recurrence,Gimeno_2017}.

While this result offers an insight into the equilibration time scale problem, the question is far from settled as opposed to that about dynamical equilibration. While a general upper bound is much sought after, with promising candidates as in Ref.~\cite{GarciaPintos_2017,deOliveira_2018}, where it is claimed that for physical observables one such bound is independent of the system size,\footnote{ See too, however, the much recent Ref.~\cite{Heveling_2020}} the problem has so far proved too complex in general and only progress in some classes of systems has been made.\footnote{ A review of equilibration timescales and the classes of systems where results have been obtained can be found in Ref.~\cite{Wilming2018}.} Moreover, even performing classical simulations is out of reach to directly approach this, since it has to be done on large systems and long timescales. This is thus one of the most relevant open problems in equilibration to this day, and as a consequence is out of the scope of this thesis.

\subsection{Trace distance equilibration}\label{sec: tr dist equilibration}
The main issue that arises from approaching equilibration on average in terms of expectation values of observables is that, even when the expectation values with respect to two different states can be equal, this does not necessarily mean that observations cannot distinguish the states. This can be made explicit with an example given in Ref.~\cite{Short_2011}, where an observable yields an equal mixture of $-1$ and $+1$ outcomes with respect to one state and always $0$ with respect to another state, so that a measurement of the observable will \emph{always} distinguish them, despite their expectation values being identical. This can be phrased alternatively as saying that, even though an observable might not distinguish a pair of states for most times, this does not imply that for most times it can not distinguish them.

These issues can be addressed by considering the distinguishability between two quantum states, which can be phrased in terms of quantifying an operationally relevant distance measure between such pair of states, in our case between a time-evolved state and the equilibrium state. We thus follow Ref.~\cite{Short_2011, Short_finite} and employ the trace distance $D$, defined in Eq.~\eqref{eq: trace distance def}, between $\rho(t)$ and $\omega$, and so we say that a system equilibrates on average if $\overline{{D}(\rho(t),\omega)}^T\ll1$ over a finite time-window of width $T$. Some of the main results on this front are the ones in Ref.~\cite{Linden_2009,Short_2011, Short_finite} under the very same conditions that we specify in Theorem~\ref{thm: Short finite time observable equilbration} previously for observable equilibration; furthermore, the proof for observable equilibration can be used in a straightforward way to obtain
\begin{equation}
    \overline{{D}_\mbb{M}(\rho(t),\,\omega)}^T\leq\f{\varkappa(\mbb{M})}{4\sqrt{d_\text{eff}}}\sqrt{N(\varepsilon)f(\varepsilon T)},
\end{equation}
where we define
\begin{equation}
    f(\varepsilon T):=1+\f{8\log_2\mathfrak{D}}{\varepsilon T},
    \label{eq: Short def f(epsilon T)}
\end{equation}
and where where $\varkappa(\mbb{M})$ is the number of possible measurement outcomes in a finite subset of possible measurements $\mbb{M}$, each with a finite set of outcomes, and the remaining quantities as in Theorem~\ref{thm: Short finite time observable equilbration}. This result follows specifically because
\begin{align}
    \overline{{D}_\mbb{M}(\rho(t),\,\omega)}^T&\leq\f{1}{2}\sum_{i,\{\mathrm{M}_i\}\in\mbb{M}}\overline{|\tr[\mathrm{M}_i(\rho(t)-\omega)]|}^T\nonumber\\
    &\leq\f{1}{2}\sum_{i,\{\mathrm{M}_i\}\in\mbb{M}}\sqrt{\overline{|\tr[\mathrm{M}_i(\rho(t)-\omega)]|^2}^T},
    \label{eq: Short dist POVM}
\end{align}
where in the first line we upper bounded the definition over the maximum with a sum over all \gls{POVM}s in $\mbb{M}$ and in the second line we used Jensen's inequality,\footnote{ $\mbb{E}[f(X)]\leq{f}(\mbb{E}[X])$ with expectation $\mbb{E}$, a random variable $X$ and a concave function $f$.} with the remaining steps as in the derivation for Eq.~\eqref{eq: Short finite time}, where it is furthermore argued in Ref.~\cite{Short_2011} that $\sum_{i,\{\mathrm{M}_i\}\in\mbb{M}}\|\mathrm{M}_i\|\leq\varkappa(\mbb{M})/2$. The previous conditions for equilibration now remain, with the only difference being the term $\varkappa(\mbb{M})$, which ought only to be compared with $d_\text{eff}$. If we again consider the discussion above for realistic, macroscopic systems with $\mc{O}(10^{23})$ degrees of freedom---such that $d\simeq10^{10^{23}}$---even with e.g. $d_\text{eff}\leq{d}^{1/100}$, the number of possible outcomes $\varkappa(\mbb{M})$ will remain small compared to the effective dimension, despite it being in itself very possibly a large quantity.

Finally, a bound on the distinguishability with the trace distance was also obtained in Ref.~\cite{Linden_2009,Short_2011, Short_finite} for a subsystem  of a larger system-environment composite. Now the interest is on the distinguishability between the reduced states $\rho_\mathsf{S}(t):=\tr_\mathsf{E}(\rho(t))$ and $\omega_\mathsf{S}:=\tr_\mathsf{E}(\omega)$. The upper bound, with a setup as in Theorem~\ref{thm: Short finite time observable equilbration} on a composite Hilbert space $\mscr{H}_\mathsf{S}\otimes\mscr{H}_\mathsf{E}$ reads
\begin{equation}
    \overline{{D}(\rho_\mathsf{S}(t),\,\omega_\mathsf{S})}^T\leq\f{d_\mathsf{S}}{2\sqrt{d_\text{eff}}}\sqrt{N(\varepsilon)f(\varepsilon T)},
    \label{eq: Short trace dist finite}
\end{equation}
which, despite being very similar, cannot be obtained from Eq.~\eqref{eq: Short dist POVM} as the number of all possible measurements is infinite. It can, however, be obtained as well from Theorem~\ref{thm: Short finite time observable equilbration}; the key step in the proof in Ref.~\cite{Short_2011} is taking an orthonormal basis in \gls{syst} given by $d_\mathsf{S}^2$ operators $\{F_i\}$ such that $\rho(t)-\omega=\sum_i\lambda_i(t)F_i$. The explicit form of the operators is reminiscent of a Fourier transform and can be seen in Ref.~\cite{Short_2011}. However, we point out that these are non-Hermitian, so that when Theorem~\ref{thm: Short finite time observable equilbration} is applied, it is relevant that it applies to general linear bounded operators $A\in\mscr{B}(\mscr{H})$; the other relevant property is that $F_iF_j^\dg=F_i^\dg{F}_j=\delta_{ij}\mbb{1}_S/d_\mathsf{S}$, so that also $\|F_\mathsf{S}^\dg\otimes\mbb1_\mathsf{E}\|^2=1/d_\mathsf{S}$. With this, we have
\begin{align}
    \overline{\|\rho_\mathsf{S}(t)-\omega_\mathsf{S}\|_1}^T
    &\leq\sqrt{d_\mathsf{S}}\,\overline{\|\rho_\mathsf{S}(t)-\omega_\mathsf{S}\|_2}^T\nonumber\\
    &\leq\sqrt{d_\mathsf{S}}\sqrt{\sum_{i,j}\overline{\lambda_i(t)\lambda_j^*(t)}^T\tr[F_iF_j^\dg]}\nonumber\\
    &=\sqrt{d_\mathsf{S}}\sqrt{\sum_{i}\overline{|\lambda_i(t)|^2}^T}\nonumber\\
    &=\sqrt{d_\mathsf{S}}\sqrt{\sum_{i}\overline{|\tr\{[\rho(t)-\omega](F_i^\dg\otimes\mbb{1}_\mathsf{E})\}|^2}^T},
\end{align}
where in the first line the inequality $\|X\|_1\leq\sqrt{\mathrm{dim}(X)}\|X\|_2$ was used\footnote{ Let $X$ a matrix of dimension $n\times{n}$ and let $\{x_i\}$ be its eigenvalues, then by convexity of the square, ${\|X\|_1^2=n^2(\sum|x_i|/n)^2\leq{n}\sum|x_i|^2=n\|X\|_2^2}$.}, followed by Jensen's inequality for the square root in the second; the result then follows by applying Theorem~\ref{thm: Short finite time observable equilbration} with the operator $A=F_i^\dg\otimes\mbb{1}_\mathsf{E}$. The result in Eq.~\eqref{eq: Short trace dist finite} implies equilibration in small subsystems \gls{syst}, i.e. whenever we can access only a small subpart of the full system-environment composite with any observable no matter how exotic or unrealistic.

The result in Eq.~\eqref{eq: Short finite time} is remarkable in its generality and its reach, as it implies that small subsystems will equilibrate on average within some finite time-window with respect to any measurement we wish to make.

Here we have presented some of the main results for equilibration on average for expectation values and briefly discussed the distinguishability approach; however, a comprehensive discussion of dynamical equilibration with other approaches and further intersections with other topics related to the foundations of statistical mechanics can be found e.g. in Ref.~\cite{Gogolin_2011, Yukalov_2011, Wilming2018, deOliveira_2018}.

We now continue to discuss the second big question we posed at the beginning, namely, how the fundamental postulate of statistical mechanics is justified given only the rules of quantum mechanics.

\section{Typicality}\label{sec: state typicality}
While we came near to touching upon the emergence of the fundamental postulate of statistical mechanics around Eq.~\eqref{eq: von Neumann entropy}, with the maximum entropy principle satisfied by the equilibrium state, this has not been fully justified just yet. The approach of equilibrium on average can be classified as a dynamical one, in the sense that it tells us \emph{how} quantum states can evolve towards equilibrium and stay close to it for most times. However, the backbone of the quest to understand equilibration is precisely a notion of equilibrium itself and how it is justified. Furthermore, from concepts like the zeroth law of thermodynamics, the canonical ensemble of statistical mechanics, or the maximum entropy principle, this equilibrium seems to be the most likely state in most cases.

This is precisely the notion that is understood when one speaks of typicality, and it is the kind of notion that Schr\"{o}dinger and von Neumann had in mind when first trying to reconcile statistical mechanics with quantum mechanics. In Ref.~\cite{Schrodinger} (translated in Ref.~\cite{schrodinger2003}), Sch\"{o}dinger studied a pair of systems that are weakly coupled for long times, and found that by assuming a certain proportionality of the population of the initial state in the energy levels with respect to the degeneracy in the non-interacting levels---which he dubbed a \emph{statistical hypothesis}---the reduced states of a small subsystem are well described by a thermal state. Similarly, in Ref.~\cite{vonNeumann} (translated in Ref.~\cite{vonNeumann2010}), von Neumann presents his \emph{quantum ergodic theorem}, stating in essence that for non-degenerate, non-resonant  (i.e. with no degenerate gaps) Hamiltonians, for most decompositions of the Hilbert space\footnote{ Detail and other informal discussion can be seen in Ref.~\cite{Normal_typicality}.} and all initial states, for most times the evolving state of the system is macroscopically indistinguishable from a suitable microcanonical state.

This is in the very same spirit as most of the works that approached this question many years later. Some prominent examples are Ref.~\cite{Lloyd_1988, Goldstein_2006, Gemmer_2009}. Notice that even in the work of von Neumann, despite the statement resembling one about equilibration, it is really one about the suitable partitioning of the Hilbert space, holding for most of them under some assumptions about the Hamiltonian but for all initial states. In this thesis we will focus mainly on the result in Ref.~\cite{Popescu2006}, which not only addresses the question of typicality but resolves to a large extent the question about the emergence of the fundamental postulate of a-priori probabilities purely from quantum mechanical laws. To do this we first need to specify what does it mean for some property to hold for most quantum states, or even more so, what it means to sample a state at random and seeing that some property for it then holds.

\subsection{Random states and the Haar measure}\label{sec: Random states and Haar}
To discuss the issue of sampling of quantum states, we should first discuss the geometry of the set of quantum states and how to place a probability measure on it.

As mentioned in Section~\ref{sec: quantum states and measurements}, we are concerned here only with studying finite quantum systems which can be associated with the complex Hilbert space $\mscr{H}\cong\mbb{C}^d$. A pure state is defined by a vector $|\psi\rangle\in\mscr{H}$ such that $\langle\psi|\psi\rangle=1$. Hence it follows that any other vector $|\varphi\rangle$ can be obtained by a unitary transformation $U|\psi\rangle$ where $U$ is a unitary matrix of dimension $d$, i.e. such that $UU^\dg=U^\dg{U}=\mbb1$. That is, we can define $|\varphi\rangle=U|\psi\rangle$, which is such that $\langle\varphi|\varphi\rangle=1$, and, similarly, we can see that for the state $\Phi=|\varphi\rangle\!\langle\varphi|$, Hermiticity and positivity,
 are satisfied, making $\Phi$ a valid pure quantum state. This can be understood intuitively as well if we consider that $\mbb{C}^d$ is isomorphic to $\mbb{R}^{2d}$, i.e. any vector $|\psi\rangle\in\mbb{C}^d$ can be thought of as a real vector in $\mbb{R}^{2d}$. Furthermore, all of pure states form a $2d-1$ dimensional sphere\footnote{ More precisely, a set of equivalence classes of points of the sphere, as the global phase means that there is a circle of points on the sphere for each distinct pure state.} $\mbb{S}^{2d-1}=\{v\in\mbb{R}^{2d}:\|v\|_2=1\}$, and so we can think of unitary operators acting on these as rotations. A similar discussion can be made for mixed states e.g. by writing any $\rho\in\$(\mscr{H})$ as a convex combination of pure states, although the geometrical analogy is not quite the same. Details can be found e.g. in Ref.~\cite{bengtsson2006geometry}. For our purposes we only need to have clear that quantum states can always be related by a unitary transformation.

Sampling quantum states \emph{at random} then can be seen to be induced by the sampling of unitaries at random. To do this we need a probability measure on the $d$ dimensional unitary group,\footnote{ A group is a set together with a binary operation satisfying axioms of closure, associativity, identity and invertibility. We take this as a standard concept but present all relevant detail in Appendix~\ref{appendix - Haar measure}.}
\begin{equation}
    \mbb{U}(d)=\{U\in\mbb{C}^{d\times{d}}:UU^\dg=U^{\dg}U=\mbb{1}\},    
\end{equation}
where $\mbb{C}^{d\times{d}}$ denotes complex square matrices of dimension $d$. Now, a natural way to place a probability measure on the unitary group is to demand that it is an invariant measure in the sense that it should not change under fixed unitary shifts, i.e. it should assign the same measure to all of $\mbb{U}(d)$.

To begin with, consider the reals: we can measure the size of an interval $\mbb{I}=[a,b]\subset\mbb{R}$ of the real line by its length $\ell(\mbb{I})=b-a$. Invariance in this case means that $\ell(\mbb{I}+c)=\ell(\mbb{I})$ for any $c\in\mbb{R}$. Now, the length of the interval is simply $\ell(\mbb{I}):=\int_a^bdx$, so more generally, for any subset $\mbb{L}\subset\mbb{R}$ we can measure its size by $\mu_\ell(\mbb{L}):=\int_\mbb{L}dx$, and we also have
\begin{equation}
    \mu_\ell(\mbb{L}+c)=\int_{\mbb{L}+c}\,dx=\int_\mbb{L}\,d(x+c)=\int_\mbb{L}\,dx=\mu_\ell(\mbb{L}),
\end{equation}
and it can be said that $\mu_\ell$ is a translation-invariant measure of the additive group $(\mbb{R},+)$.

This property can then be extended to the unitary group $\mbb{U}(d)$ so that for a given subset $\mbb{W}\subseteq\mbb{U}(d)$ and any $V\in\mbb{U}(d)$ we define
\begin{equation}
    \muhaar(\mbb{W})=\int_\mbb{W}d\muhaar(U),
\end{equation}
as the Haar measure,\footnote{ The Haar measure is named after Alfr\'{e}d Haar, who introduced such invariant measure in 1932 more generally over locally compact groups, allowing an analogue of Lebesgue integrals, such as in the cited example over the additive real group.} $\muhaar$, of the set $\mbb{W}$, which satisfies the left-right invariance property
\begin{align}
   \muhaar(\mbb{W})&=\muhaar(V\mbb{W})=\int_\mbb{W}d\muhaar(VU)=\int_\mbb{W}d\muhaar(UV)=\muhaar(\mbb{W}V).
\end{align}

Now what exactly does this mean? Remember the analogy of rotations on real space, then the Haar measure implies that if we measure a ``cap'' of the sphere and rotate it we get exactly the same measure. Similarly, it means that all unitaries have the same measure, so that it does not matter if we rotate to any other subset of the unitary group, we still get the same answer.

How are we to practically use the Haar measure to actually \emph{measure} the size of a subset of unitaries? It turns out that this can be done relatively easy on a computer~\cite{Mezzadri}, as we detail in Appendix~\ref{appendix - numerics Almost} (this was done for the results of Ref.~\cite{FigueroaRomero2019almostmarkovian}, which we discuss in Chapter~\ref{sec:typicality}). However, it is an elaborate task in the sense that we need to parametrize the unitary matrix entering $d\muhaar$, and there are $d^2$ independent real parameters in a unitary matrix,\footnote{ A complex square matrix of dimension $d$ has $2d^2$ real parameters, and unitarity $UU^\dg=\sum{U}_{ij}U^*_{\ell{j}}|i\rangle\!\langle\ell|=\mbb1$ imposes $\sum_j {U}_{ij}U^*_{\ell{j}}=\delta_{i\ell}$ which gives $d^2$ real constraints.} so that
\begin{align}
    d\muhaar(U)=f(\theta_1,\ldots,\theta_{d^2})\,d\theta_1\,\cdots\,d\theta_{d^2},
\end{align}
with $\{\theta_1,\ldots,\theta_{d^2}\}$ the parameters of $U$ acting as a set of local real coordinates on the manifold\footnote{ A manifold is a topological space for which every point has a neighborhood which is Euclidean.} described by $\mbb{U}(d)$ embedded in $\mbb{R}^{2d^2}$ and with a given \emph{probability density function}\footnote{ We define precisely a probability density function in Chapter~\ref{sec:equilibration}. See Ref.~\cite{Zyczkowski_1994} for an explicit construction.} $f$.

Finally, the Haar measure can be seen to be unique up to a multiplicative constant~\cite{halmos2013measure}, i.e. for any two Haar measures $\mu_\haar$ and $\mu_\haar^\prime$, we can relate these as $\mu_\haar=\alpha\mu_\haar^\prime$ for some constant $\alpha\in\mbb{R}^+$. The Haar measure is finite and positive, so we can make this a legitimate probability measure by normalizing it as $\muhaar(\mbb{U})=1$. In general, this property holds for groups known as \emph{compact} (closed and bounded). However, here we are only concerned with the unitary group, which satisfies such property. This means then that we can say that we \emph{sample a unitary matrix} uniformly at random whenever we take a unitary matrix that is distributed according to the Haar measure, and we denote this by $U\sim\muhaar$.

\subsection{Twirling and the Schur-Weyl duality}
We can now wrap up the previous discussion and connect it with sampling a pure quantum state at random simply by sampling a unitary from the Haar measure. Ultimately what we care about are the statistical properties of random quantum states, or more general quantities depending on unitaries. Specifically, we will rely on the moments of the unitary group to compute statistical properties of quantities relying on compositions of unitaries, such as the purity of a reduced state.

While we will later be concerned with higher order moments of the unitary group with respect to the Haar measure, it is instructive to consider here the calculation of the first and second moments of the unitary group over the Haar measure. For this we can rely on the following.

\begin{theorem}[Schur-Weyl duality~\cite{Roberts2017}]
\label{Thm: Schur-Weyl}
Any operator $\mc{O}$ acting on $\mscr{H}^{\otimes{n}}:=\bigotimes_{i=1}^n\mscr{H}_n$ commutes with all operators $V^{\otimes{n}}$, where $V\in\mbb{U}(d)$, if and only if $\mc{O}$ is a linear combination of permutation operators:
\begin{equation}
    [\mc{O},V^{\otimes{n}}]=0,\,\forall\,{V}\in\mbb{U}(d)\qquad
    \Longleftrightarrow\qquad
    \mc{O}=\sum_{\sigma\in\mathfrak{G}_n}c_\sigma\wp_\sigma,
\end{equation}
with $\mathfrak{G}_n$ denoting the symmetric group over $\{1,2,\ldots,n\}$ (i.e. the set of permutation operations that can be performed on $n$ symbols), and where the permutation operator is defined as
\begin{equation}
    \wp_\sigma|v_1,\ldots,v_n\rangle=|v_{\sigma(1)},\ldots,v_{\sigma(n)}\rangle,
\end{equation}
for any $|v_1,\ldots,v_n\rangle\in\mscr{H}^{\otimes{n}}$.
\end{theorem}

For example, for $n=3$, there are $3!=6$ permutations and $\sigma$ denotes the assignment of each, e.g. the permutation $1\to2\to1$, $3\to3$ corresponds to $\sigma(1)=2$, $\sigma(2)=1$, $\sigma(3)=3$. This can also be denoted in so-called cycle notation simply as $(1,2)(3)$.

The Schur-Weyl duality is a result from representation theory and it carries deep consequences that go far beyond the results of this thesis; for our purposes, however, it allows us to compute integrals in a so-called \emph{twirled} map, defined as
\begin{gather}
    \Xi^{(n)}(X):=\int_{\mbb{U}(d)}{U}^{\otimes{n}}X({U}^{\otimes{n}})^\dg\,d\muhaar(U),
    \label{eq: SW def}
\end{gather}
 where $X\in\mscr{H}^{\otimes{n}}$. For now let us simply take this definition for the twirl as a map $\Xi:\mscr{H}^{\otimes{n}}\to\mscr{H}^{\otimes{n}}$. We will revisit this briefly in Section~\ref{sec: quantum channels}.
 
 It is straightforward to see that $\Xi^{(n)}(X)$ commutes with all $V^{\otimes{n}}$ by the invariant property of the Haar measure,
 \begin{align}
     V^{\otimes{n}}\Xi^{(n)}(X)&=\int_{\mbb{U}(d)}V^{\otimes{n}}U^{\otimes{n}}X(V^{\otimes{n}})^\dg\,d\muhaar(U)\nonumber\\
     &=\int_{\mbb{U}(d)}W^{\otimes{n}}X(W^{\otimes{n}})^\dg{V}^{\otimes{n}}\,d\muhaar(V^\dg{W})\nonumber\\
     &=\int_{\mbb{U}(d)}W^{\otimes{n}}X(W^{\otimes{n}})^\dg{V}^{\otimes{n}}\,d\muhaar(W)\nonumber\\
     &=\Xi^{(n)}(X)V^{\otimes{n}},
 \end{align}
where we defined $W=VU$ in the second line. Thus, the Schur-Weyl duality ensures that
\begin{align}
    \Xi^{(n)}(X)=\sum_{\sigma\in\mathfrak{G}_n}c_\sigma(X)\,\wp_\sigma,
    \label{eq: SW twirl}
\end{align}
where $c_\sigma$ is now a linear function of $X$.

\subsection{Average states and average purity}
The decomposition of the twirl map in Eq.~\eqref{eq: SW twirl} is quite relevant for our purposes because it lets us swiftly compute the first and second moment with $n=1,2$. By the $n\textsuperscript{th}$ statistical moment we are referring to the expectation of products of components of Haar-distributed unitaries, i.e.
\begin{equation}
    \mbb{E}_\haar[U_{i_1j_1}\cdots{U}_{i_nj_n}U_{i_1^\prime j_1^\prime}^*\cdots{U}_{i_n^\prime j_n^\prime}^*]=\int\limits_{\mbb{U}(d)}U_{i_1j_1}\cdots{U}_{i_nj_n}U_{i_1^\prime j_1^\prime}^*\cdots{U}_{i_n^\prime j_n^\prime}\,d\muhaar(U),
    \label{eq: n-moments of the unitary group}
\end{equation}
where $U_{ab}$ are entries of the unitary $U=\sum{U}_{ab}|a\rangle\!\langle{b}|$ in a given basis. Here the notation $\mbb{E}_\haar[\cdot]=\int_{\mbb{U}(d)}[\cdot]\,d\muhaar$ explicitly means the expectation with respect to the Haar measure with $X\sim\muhaar$. Thus computing the $n$-moments of the unitary group can be rendered equivalent to computing an $n$-twirl, $\mbb{E}_\haar[U^{\otimes{n}}X(U^{\otimes{n}})^\dg]=\Xi^{(n)}(X)$.

For the average we have by definition
\begin{equation}
    \mbb{E}_\haar[U\,\rho\,U^\dg]=\int_{\mbb{U}(d)}U\,\rho\,U^\dg\,d\muhaar(U),
    \label{eq: average state Haar}
\end{equation}
so this corresponds to the 1-twirl, $\Xi^{(1)}(\rho)$, which by Eq.~\eqref{eq: SW twirl} must imply
\begin{equation}
    \mbb{E}_\haar[U\,\rho\,U^\dg]=c(\rho)\,\mbb{1},
    \label{eq: SW average state Haar}
\end{equation}
because the only possible permutation of one object is to itself and thus the identity is the only possible value for $\wp$. Now, we can exploit both the cyclicity and linearity of the trace, together with the normalization of the Haar measure, to see that
\begin{align}
    \tr\{\mbb{E}_\haar[U\,\rho\,U^\dg]\}&=\muhaar(\mbb{U})\tr(\rho)=1\nonumber\\
    &=c(\rho)\,\tr(\mbb{1})=d\,c(\rho),
\end{align}
where in the first line we took the trace of Eq.~\eqref{eq: average state Haar}, whilst in the second we took that of Eq.~\eqref{eq: SW average state Haar}.

Thus, it follows that $c(\rho)=\tr(\rho)=1$, and hence the average quantum state from the uniform probability measure is the maximally mixed state,
\begin{equation}
    \mbb{E}_\haar\left[U\,\rho\,U^\dg\right]=\f{\mbb{1}}{d},
    \label{eq: average quantum state}
\end{equation}
and it similarly follows, in general, that for any random $d\times{d}$ matrix $X$ that is Haar-distributed, $\mbb{E}_\haar[X]=\tr(X)\mbb{1}/d$. This result is intuitive since if we take into account that the maximally mixed state is the maximal ignorance state, so for any quantum state drawn uniformly at random, naturally we would expect to get the equiprobable state.

For the second moment we can encounter different relevant quantities, but particularly we will care about the purity $\tr\left[\rho_\mathsf{A}^2\right]$ of a reduced state $\rho_\mathsf{A}=\tr_\mathsf{B}[\rho_\mathsf{AB}]$, where $\rho_\mathsf{AB}\in\$(\mscr{H}_\mathsf{A}\otimes\mathscr{H}_\mathsf{B})$, i.e.
\begin{equation}
    \mbb{E}_\haar\left\{\tr\left[\left(\rho_\mathsf{A}^\prime\right)^2\right]\right\}=\int_{\mbb{U}(d)}\tr\left\{\tr_\mathsf{B}\left[U\,\rho_\mathsf{AB}\,U^\dg\right]\tr_\mathsf{B}\left[U\,\rho_\mathsf{AB}\,U^\dg\right]\right\}\,d\muhaar(U),
\end{equation}
where here $\rho_\mathsf{A}^\prime=\tr_\mathsf{B}[U\,\rho_\mathsf{AB}\,U^\dg]$. This means we can now use the 2-twirl, $\Xi^{(2)}$; towards this, we now have via Schur-Weyl duality
\begin{equation}
    \Xi^{(2)}[X]=\alpha(X)\,\mbb{1}+\beta(X)\,\swap,
    \label{eq: Schur-Weyl 2-twirl}
\end{equation}
where here $\mbb1\in\mscr{H}^{\otimes2}$ is a $d^2$ dimensional identity operator and $\swap:=\sum|ij\rangle\!\langle{ji}|$ is known as a swap operator, which, as the name suggests, swaps the respective states, i.e. $\swap|uv\rangle=|vu\rangle$, and is such that $\swap=\swap^\dg$ and $\swap^2=\mbb1$. This is because the only possible permutations on two elements are precisely the identity and interchanging the elements. Now we need two equations to determine $\alpha$ and $\beta$, but this can be done in a similar manner to the 1-twirl: first we have
\begin{equation}
    \tr\{\Xi^{(2)}[X]\}=d^2\alpha(X)+d\,\beta(X)=\tr(X),
\end{equation}
where the first equality follows from tracing the Schur-Weyl duality expression, while the second follows from the definition of the twirl. Now, similarly,
\begin{equation}
    \tr\{\swap\Xi^{(2)}[X]\}=d\,\alpha(X)+d^2\,\beta(X)=\tr(\swap X),
\end{equation}
as $\swap$ commutes with $U^{\otimes2}$, so we can solve for $\alpha$ and $\beta$,
\begin{equation}
    \alpha(X)=\f{\tr(X)}{d^2-1}-\f{\tr(\swap X)}{d(d^2-1)},\qquad
    \beta(X)=\f{\tr(\swap X)}{d^2-1}-\f{\tr(X)}{d(d^2-1)},
    \label{eq: Schur-Weyl constants}
\end{equation}
and consequently Eq.~\eqref{eq: Schur-Weyl 2-twirl} is now solved.

Going back to the motivation of the purity of a reduced state, first notice that
\begin{equation}
    \tr[\swap(\rho_\mathsf{A}\otimes\rho_\mathsf{A})]=\sum_{i,j=1}^{d_\mathsf{A}}\langle{j}|\rho_\mathsf{A}|i\rangle\!\langle{i}|\rho_\mathsf{A}|j\rangle=\tr\left[\rho_\mathsf{A}^2\right],
    \label{eq: purity swap}
\end{equation}
so that we may readily use the 2-twirl if we take $\mscr{H}\cong\mscr{H}_\mathsf{A}\otimes\mscr{H}_\mathsf{B}$ and $X=\rho_\mathsf{AB}^{\otimes2}$, then partial trace over the subspace $\mathsf{B}$ for each copy,
\begin{align}
    \tr_{\mathsf{B}\mathsf{B}}\left\{\Xi^{(2)}\left[\rho_\mathsf{AB}^{\otimes2}\right]\right\}&=\alpha\left(\rho_\mathsf{AB}^{\otimes2}\right)\,d_\mathsf{B}^2\,\mbb1+\beta\left(\rho_\mathsf{AB}^{\otimes2}\right)\,d_\mathsf{B}\swap,
   \label{eq: partial trace 2-twirl}
\end{align}
where implicitly both the identity and the swap act on $\mscr{H}_\mathsf{A}^{\otimes2}$. Taking the trace then, we get the expected purity
\begin{align}
    \mbb{E}_\haar\left\{\tr\left[\left(\rho_\mathsf{A}^\prime\right)^2\right]\right\}&=\tr\left[\swap\left(\tr_{\mathsf{B}\mathsf{B}}\left\{\Xi^{(2)}\left[\rho_\mathsf{AB}^{\otimes2}\right]\right\}\right)\right]\nonumber\\[0.1in]
    &=\f{d_\mathsf{A}d_\mathsf{B}^2}{d_\mathsf{AB}^2-1}-\f{d_\mathsf{B}\tr(\rho_\mathsf{AB}^2)}{d_\mathsf{AB}^2-1}+\f{d_\mathsf{A}^2d_\mathsf{B}\tr(\rho_\mathsf{AB}^2)}{d_\mathsf{AB}^2-1}-\f{d_\mathsf{A}}{d_\mathsf{AB}^2-1}\nonumber\\[0.1in]
    &=\f{d_\mathsf{A}(d_\mathsf{B}^2-1)+d_\mathsf{B}(d_\mathsf{A}^2-1)\tr(\rho_\mathsf{AB}^2)}{d_\mathsf{AB}^2-1},
\end{align}
and in particular if the original state is pure, $\tr(\rho_\mathsf{AB}^2)=1$, this yields
\begin{equation}
    \mbb{E}_\haar\left\{\tr\left[\left(\rho_\mathsf{A}^\prime\right)^2\right]\right\}=\f{d_\mathsf{A}+d_\mathsf{B}}{d_\mathsf{AB}+1}.
    \label{eq: average purity reduced state}
\end{equation}

Moreover notice that
\begin{equation}
    \lim_{d_\mathsf{B}\to\infty}\mbb{E}_\haar\left\{\tr\left[\left(\rho_\mathsf{A}^\prime\right)^2\right]\right\}=\f{1}{d_\mathsf{A}},
\end{equation}
i.e. the expected purity of the reduced state $\rho_\mathsf{A}$ will approach that of the maximally mixed state whenever $d_\mathsf{B}\gg{d}_\mathsf{A}$. This implies that, on average, we would have maximum ignorance about the small subsystem despite knowing everything about the state of the full composite. As it turns out, this is the case whenever $\mathsf{A}$ and $\mathsf{B}$ have the maximum amount of \emph{entanglement} possible. We will introduce precisely what this means below. Bipartite maximally entangled states happen to be typical, in the sense that if we sample a bipartite pure state at random, it will be highly entangled with very high probability. To formalize this notion of typicality let us introduce a concept known as \emph{concentration of measure}, which will let us to go back to the main issue we set out to investigate, which is the emergence of the fundamental postulate of statistical mechanics.

\subsection{Concentration of measure}
While Eq.~\eqref{eq: average purity reduced state} is an interesting result, it does not tell us anything about the distribution of bipartite states and how much their reduced states will differ from the maximally mixed state. A powerful concept, however, that allows us to dig into the structure of such distribution without having to  compute all of the moments of the unitary group is one called concentration of measure.

As the name suggests, concentration of measure refers to a given measure being concentrated in a region of a metric space~\cite{Ledoux}. More precisely, let $\mscr{M}$ be a metric space with metric (i.e. distance) $\tilde{\Delta}_{\mscr{M}}$ and probability measure $\mu_\mathsf{M}$. Then we say that a function $f:\mscr{M}\to\mbb{R}$ satisfies a concentration of measure around its mean if, for any point $x\in\mscr{M}$ and any $\delta>0$,
\begin{equation}
    \mbb{P}_\mathsf{M}[f(x)\geq\mbb{E}_\mathsf{M}(f)+\delta]\leq\alpha_\mathsf{M}(\delta/\mscr{L}),
    \label{eq: def concentration of measure}
\end{equation}
where here, as in the case of the Haar measure, $\mbb{P}_\mathsf{M}$ and $\mbb{E}_\mathsf{M}$ explicitly refer to the probability and expectation with $x\sim\mu_\mathsf{M}$, and where $\mscr{L}$ is the so-called Lipschitz constant of $f$, and the function $\alpha_\mathsf{M}$ is known as the concentration function or concentration rate, which must vanish in increasing $\delta$ in order to have concentration of measure.

The Lipschitz constant has a special role because, intuitively, it tells us how fast the function $f$ changes, which in turn will have an impact on the concentration rate and thus imply concentration of measure or absence thereof. Here we are concerned with functions mapping to the reals; however, in general, we say that a function $f:\mbb{X}\to\mbb{Y}$ between metric spaces $(\mbb{X},\tilde{\Delta}_{_\mbb{X}})$ and $(\mbb{Y},\tilde{\Delta}_{_\mbb{Y}})$ is $\mscr{L}$-Lipschitz if there is a real constant $\mscr{L}\geq0$ such that
\begin{gather}
    \tilde{\Delta}_{_\mbb{Y}}\left(f(x_1),f(x_2)\right)\leq\mscr{L}\,\tilde{\Delta}_{_\mbb{X}}(x_1,x_2),
\end{gather}
for any $x_1,x_2\in\mbb{X}$.

This means that the function in Eq.~\eqref{eq: def concentration of measure} will satisfy a concentration of measure around its mean whenever its rate of change is low, i.e. whenever it has a small Lipschitz constant, given that the concentration rate $\alpha_\mathsf{M}$ has to be vanishing in increasing $\delta$.

Whilst concentration of measure can be defined as in Eq.~\eqref{eq: def concentration of measure}, it is a mathematical concept with a much larger depth~\cite{Ledoux}. However, as we discussed at the beginning of this section, pure quantum states can be represented as equivalence classes of points in the hypersphere $\mbb{S}^{2d-1}$, so for our purposes the relevant quantities are the hypersphere as our metric space together with its concentration rate. This result is known as Levy's lemma.\footnote{ Also known as an isoperimetric inequality, it is named after Paul L\'{e}vy who derived it in 1919 albeit in a slightly different form.}

\begin{theorem}[Levy's lemma~\cite{Popescu2006}]
\label{def: Levy's Lemma}
Let $f:\mbb{S}^{d-1}\to\mbb{R}$ be an $\mscr{L}$-Lipschitz continuous function on the $d-1$ dimensional hypersphere $\mbb{S}^{d-1}$, then for any $x\in\mbb{S}^{d-1}$ and any $\delta>0$,
\begin{equation}
    \mbb{P}_\haar[|f(x)-\mbb{E}_\haar(f)|\geq\delta]\leq2\exp\left(-\f{d\,\delta^2}{9\pi^3\mscr{L}^2}\right).
    \label{eq: Levy's lemma}
\end{equation}
\end{theorem}
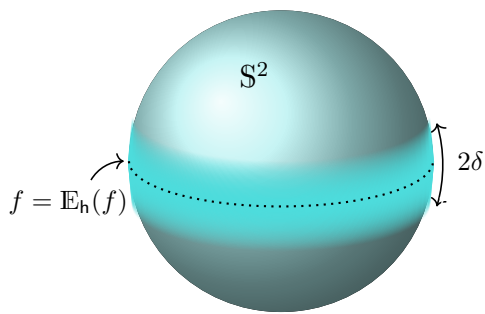
\begin{figure}[t]
\centering
    \begin{tikzpicture}
    \fill[ball color=C3!30, opacity=1] (0,0) circle (2cm);
    \foreach \x in {0,2,...,100}{
    \draw[thick, C3!70, opacity=1-\x*0.01](-2+\x*0.00052,0-\x*0.006) arc (180:360:2-\x*0.00052 and 0.6);
    \draw[thick, C3!70, opacity=1-\x*0.01](-2+\x*0.00052,0+\x*0.006) arc (180:360:2-\x*0.00052 and 0.6);
    }
    \node[right] at (2.2,0) {$2\delta$};
    \draw[semithick, <->] (2.05,-0.6) arc(-16:16:2) ;
    \draw[dotted, thick] (-2,0) arc (180:360:2 and 0.6);
    \node at (-0.35,1.15) {\Large$\mbb{S}^2$};
    \draw[semithick, ->] (-2.5,-0.25) arc(150:90:0.5);
    \node[left,below] at (-2.8,-0.25) {$f=\mbb{E}_\haar(f)$};
    \node[right,below] at (2.8,-0.25) {$\textcolor{white}{f=\mbb{E}_\haar(f)}$};
    \end{tikzpicture}
    \caption[Levy's lemma: concentration of measure on the sphere]{\textbf{Levy's lemma for the hypersphere} $\mbb{S}^d$ implies that the probability of picking a point at random outside a band of width $2\delta$ along the equator converges exponentially to zero in $d\delta^2$. For the unit sphere with $d=2$, concentration is rather weak, however for high-dimensional hyperspheres most random points will lie along the equator with high probability.}
    \label{fig: Levy's lemma sphere}
\end{figure}
That is, for all functions on a high-dimensional hypersphere that do not change too rapidly, i.e. with a small Lipschitz constant $\mscr{L}$, the function evaluated on a point picked uniformly at random will be close to its expectation with very high probability. To make this concrete, consider $x=(x_1,\ldots,x_d)$ and the function defined by $f(x)=x_1$, then the expectation is the equator, $x_1=0$, and thus Levy's lemma in Eq.~\eqref{eq: Levy's lemma} tells us that the probability of finding a random point outside a band of width $2\delta$ along the equator converges exponentially to zero in $d\,\delta^2$. Thus \emph{large-dimensional hyperspheres are fat along every equator}, as illustrated in Fig.~\ref{fig: Levy's lemma sphere}.

Let us now go back to the physical scenario. Levy's lemma allows for a neat application of the previous result, letting us show that pure bipartite maximally entangled states are typical, in the sense that most bipartite states are concentrated around them. In other words, if we sample pure bipartite states at random according to the Haar measure, these will be highly entangled with very high probability. Furthermore, as we will see, it is a crucial ingredient for the main result of this section on the emergence of the fundamental postulate of statistical mechanics.

\subsection{Typicality of entangled states}
An average reduced state under the Haar measure will be maximally mixed, as can be seen from Eq.~\eqref{eq: average state Haar}, as we can always take the partial trace after averaging. However, we can now ask, what is the average distinguishability between a generic reduced state and the maximally mixed state,
\begin{equation}
    \mbb{E}_\haar\left[{D}\left(\rho^\prime_\mathsf{A},\f{\mbb1_\mathsf{A}}{d}\right)\right]=\f{1}{2}\mbb{E}_\haar\|\rho^\prime_\mathsf{A}-\f{\mbb1_\mathsf{A}}{d_\mathsf{A}}\|_1,
\end{equation}
where here $\rho^\prime_\mathsf{A}=\tr_\mathsf{B}[U\rho_\mathsf{AB}U^\dg]$ with $U\in\mbb{U}(d_\mathsf{AB})$ uniformly distributed, $U\sim\muhaar$. The choice of the trace distance here is motivated as in Sections~\ref{sec: trace distance qm101} and~\ref{sec: tr dist equilibration}.

Now we may use both the relation $\|X\|_1\leq\sqrt{\dim(X)}\|X\|_2$ between Schatten norms and Jensen's inequality for the square root to get
\begin{align}
    \mbb{E}_\haar\left[{D}\left(\rho^\prime_\mathsf{A},\f{\mbb1_\mathsf{A}}{d_\mathsf{A}}\right)\right]&\leq\f{1}{2}\sqrt{d_\mathsf{A}}\sqrt{\mbb{E}_\haar\left[\tr(\rho^{\prime\,2}_\mathsf{A})\right]-\f{1}{d_\mathsf{A}}}=\f{1}{2}\sqrt{\f{d_\mathsf{A}^2-1}{d_\mathsf{AB}+1}},
\end{align}
where in the second line we used the average purity in Eq.~\eqref{eq: average purity reduced state}, assuming that we care about the situation in which the global state $\rho_\mathsf{AB}$ is pure, $\tr(\rho_\mathsf{AB}^2)=1$. Now, in particular when $\mathsf{B}$ is much greater than $\mathsf{A}$, we get
\begin{equation}
    \mbb{E}_\haar\left[{D}\left(\rho^\prime_\mathsf{A},\f{\mbb1_\mathsf{A}}{d_\mathsf{A}}\right)\right]\lesssim\f{1}{2}\sqrt{\f{d_\mathsf{A}}{d_\mathsf{B}}},
\end{equation} which itself converges to zero in the $d_\mathsf{B}\gg{d}_\mathsf{A}$ limit.

Now we can regard this trace distance to the maximally mixed state as a function ${D}(\tr_\mathsf{B}[\cdot],\mbb{1}_\mathsf{A}/d_\mathsf{A}):\mbb{S}^{2d_\mathsf{AB}-1}\to\mbb{R}^+$ from pure states on the sphere to the reals, and apply Levy's lemma, provided it is Lipschitz continuous. We have, for any two $\rho_\mathsf{AB},\sigma_\mathsf{AB}\in\$(\mscr{H}_\mathsf{A}\otimes\mscr{H}_\mathsf{B})$,
\begin{align}
    \left|{D}\left(\rho_\mathsf{A},\f{\mbb1_\mathsf{A}}{d_\mathsf{A}}\right)-{D}\left(\sigma_\mathsf{A},\f{\mbb1_\mathsf{A}}{d_\mathsf{A}}\right)\right|&\leq{D}(\rho_\mathsf{A},\sigma_\mathsf{A})\leq{D}(\rho_\mathsf{AB},\sigma_\mathsf{AB}),
    \label{eq: Lipschitz tr dist}
\end{align}
where in the first inequality we used the triangle inequality and in the second we used the fact that the partial trace does not increase the trace distance.

This last property can be seen by noticing that the difference of density matrices is Hermitian and thus it can be diagonalized with real eigenvalues, i.e. we have $\rho-\sigma=UDU^\dg$ and we can further split this as $\rho-\sigma=Q-S$ where $Q$, $S$ are positive semi-definite with orthogonal eigenspaces. Now, because $\rho$, $\sigma$ are states, it follows that $\tr(Q)=\tr(S)$, and thus ${D}(\rho,\sigma)=\tr(Q)$ given that $|\rho-\sigma|=Q+S$. Now, finally $\tr(Q_\mathsf{A})\geq\tr(\Pi\,Q_\mathsf{A})$ for any projector $\Pi$, so taking the trace distance definition maximizing over projectors, ${D}(\rho,\sigma)=\max_\Pi\tr[\Pi(\rho-\sigma)]$ it follows that ${D}(\rho,\sigma)\geq\tr(\Pi\,Q_\mathsf{B})\geq\tr[\Pi(\rho_\mathsf{A}-\sigma_\mathsf{A})]\geq{D}(\rho_\mathsf{A},\sigma_\mathsf{A})$. This property holds in general for any trace preserving map, but this will be discussed in more detail in Section~\ref{sec: quantum channels}.

Thus Eq.~\eqref{eq: Lipschitz tr dist} means that the Lipschitz constant of ${D}(\tr_\mathsf{B}[\cdot],\mbb1_\mathsf{A}/d_\mathsf{A})$ can simply be taken to $\mscr{L}=1$, as for pure states ${D}(|\phi\rangle\!\langle\phi|,|\psi\rangle\!\langle\psi|)^2=1-|\langle\psi|\phi\rangle|^2$~\cite{nielsen2000quantum}, so that if $\rho_\mathsf{AB}=|\phi\rangle\!\langle\phi|$ and $\sigma_\mathsf{AB}=|\psi\rangle\!\langle\psi|$,
\begin{align}
    {D}(\rho_\mathsf{AB},\sigma_\mathsf{AB})&\leq\sqrt{(1-|\langle\psi|\phi\rangle|)(1+|\langle\psi|\phi\rangle|)}\nonumber\\
    &\leq\sqrt{2(1-\mathrm{Re}[\langle\psi|\phi\rangle])}\nonumber\\
    &\leq\left||\phi\rangle-|\psi\rangle\right|,
    \label{eq: Lipschits tr dist 2}
\end{align}
and we now can readily apply Levy's lemma.

\begin{theorem}[Almost all pure quantum states are almost maximally entangled~\cite{glassle2013almost}]
\label{Thm: entanglement concentration}
For any pure state $\rho_\mathsf{AB}\in\$(\mscr{H}_\mathsf{A}\otimes\mscr{H}_\mathsf{B})$ of dimension $d_\mathsf{AB}=d_\mathsf{A}d_\mathsf{B}$ with $d_\mathsf{B}\gg{d}_\mathsf{A}$, drawn uniformly at random, and for any $\delta>0$,
\begin{equation}
    \mbb{P}_\haar\left[{D}\left(\rho_\mathsf{A}^\prime,\f{\mbb1_\mathsf{A}}{d_\mathsf{A}}\right)\geq\f{1}{2}\sqrt{\f{d_\mathsf{A}}{d_\mathsf{B}}}+\delta\right]\leq2\exp\left(-\f{2d_\mathsf{AB}\,\delta^2}{9\pi^3}\right),
    \label{eq: typicality entangled}
\end{equation}
so when $\mathsf{B}$ is much bigger than $\mathsf{A}$, the probability of a random reduced state of being distinguishable from the maximally mixed state is vanishing. This implies that almost all pure states will be almost maximally entangled.
\end{theorem}

A similar result for the average entanglement entropy was derived in Ref.~\cite{Page_purity}, as well as a concentration of measure result in Ref.~\cite{Hayden2006}, which can be argued for by using a relation between the entanglement entropy and the trace distance, known as Fannes inequality~\cite{nielsen2000quantum}.

Theorem~\ref{Thm: entanglement concentration} already hints at implications for the emergence of the fundamental postulate of statistical mechanics; despite the universe being in a pure state, we may find small systems in almost maximal ignorance states, with the reason behind this being precisely entanglement.

\subsection{Entanglement as a canonical principle}
While much progress in topics regarding typicality has been made since von Neumann's results, a turning point in the topic can arguably be attributed to the work of Popescu, Short and Winter in Ref.~\cite{Popescu2006}, where the authors explicitly drew a connection between the emergence of the fundamental postulate of statistical mechanics and the typicality of entanglement.

Consider an \gls{syst-env} system-environment composite given by the space $\mscr{H}\cong\mscr{H}_\mathsf{S}\otimes\mscr{H}_\mathsf{E}$ as a closed system corresponding to the universe. A global constraint on these, which thermodynamically would correspond to the total energy of the universe, can be modelled in general by restricting the allowed global states to a subspace $\mscr{H}_\mathsf{R}\subseteq\mscr{H}$ of dimension $d_\mathsf{R}$. Now the fundamental postulate of statistical mechanics would correspond to a-priori assigning all pure states on $\mathsf{R}$ the same probability, i.e. having the equilibrium thermodynamics of the universe under $\mathsf{R}$ completely described by the maximally mixed state $\mbb1_\mathsf{R}/d_\mathsf{R}$.

Now the canonical state of the system is defined as the equal a-priori probability state of $\mathsf{R}$ with the degrees of the \gls{env} traced out,
\begin{equation}
    \Omega_\mathsf{S}=\tr_\mathsf{E}\left(\f{\mbb1_\mathsf{R}}{d_\mathsf{R}}\right).
\end{equation}

Now the question can be posed in similar terms to that of the typicality of entangled states. The crucial insight in Ref.~\cite{Popescu2006} is that indeed the \gls{syst-env} universe might be in a pure state, but despite this, the reduced states in subsystem \gls{syst} will typically be close to the canonical state. Now we consider the average trace distance between an \gls{syst} state $\rho_\mathsf{S}$ and the canonical state $\Omega_\mathsf{S}$, which we can bound as
\begin{align}
    \mbb{E}_\haar[{D}\left(\rho^\prime_\mathsf{S},\Omega_\mathsf{S}\right)]&\leq\f{1}{2}\mbb{E}_\haar[\|\rho_\mathsf{S}^\prime-\Omega_\mathsf{S}\|_2]\leq\f{1}{2}\sqrt{d_\mathsf{S}}\sqrt{\mbb{E}[\tr(\rho_\mathsf{S}^{\prime\,2})]-\tr(\Omega_\mathsf{S}^2)},
    \label{eq: Popescu tr dist 1}
\end{align}
where here $\rho^\prime_\mathsf{S}=\tr_\mathsf{E}[U\rho_\mathsf{R}U^\dg]$ with $U\in\mbb{U}(d_\mathsf{R})$, and $\rho_\mathsf{R}\in\$(\mscr{H}_\mathsf{R})$ is a pure state $\tr(\rho_\mathsf{R})=1$ of the universe; in the first line we used $\|X\|_1\leq\sqrt{\dim(X)}\|X\|_2$, while in the second one we used Jensen's inequality together with $\mbb{E}_\haar[\rho_S^\prime]=\Omega_\mathsf{S}$.

We can bound the expectation of the purity of $\rho_\mathsf{S}$ by writing, similar to Eq.~\eqref{eq: purity swap},
\begin{equation}
    \tr[\rho_\mathsf{S}^2]=\tr[\swap(\rho_\mathsf{S}\otimes\rho_\mathsf{S})]=\tr[(\mbb1_\mathsf{EE}\otimes\swap_\mathsf{SS})\rho_\mathsf{R}\otimes\rho_\mathsf{R}],
\end{equation}
where $\swap$ acts on $\mathsf{SS}$ and the full trace in the second equality is over $\mathsf{RR}$. Thus we need the 2-twirl
\begin{equation}
    \Xi^{(2)}[\rho_\mathsf{R}\otimes\rho_\mathsf{R}]=\f{1}{d_\mathsf{R}(d_\mathsf{R}+1)}\left(\mbb1+\swap\right),
\end{equation}
where we used the Schur-Weyl duality as in Eq.~\eqref{eq: Schur-Weyl 2-twirl} with the constants determined by means of Eq.~\eqref{eq: Schur-Weyl constants}. The identity and the $\swap$ are in space $\mscr{H}_\mathsf{R}\otimes\mscr{H}_\mathsf{R}$. This 2-twirl is equivalent to $\mbb{E}[\rho_\mathsf{R}^\prime\otimes\rho_\mathsf{R}^\prime]$, where $\rho_\mathsf{R}^\prime=U\rho_\mathsf{R}U^\dg$. Thus we have
\begin{align}
    \mbb{E}[\tr(\rho_\mathsf{S}^{\prime\,2})]&=\tr[(\mbb1_\mathsf{EE}\otimes\swap_\mathsf{SS})\mbb{E}_\haar[\rho_\mathsf{R}\otimes\rho_\mathsf{R}]]\nonumber\\
    &=\f{1}{d_\mathsf{R}(d_\mathsf{R}+1)}\tr[(\mbb1_\mathsf{EE}\otimes\swap_\mathsf{SS})(\mbb1_\mathsf{RR}+\swap_\mathsf{RR})],
\end{align}
now we notice that $\swap_\mathsf{RR}=\mbb{1}_\mathsf{RR}(\swap_\mathsf{EE}\otimes\swap_\mathsf{SS})$, where the identity $\mbb{1}_\mathsf{RR}$ here simply means that the \gls{syst-env} swaps are restricted to $\mathsf{R}$, thus
\begin{align}
    \mbb{E}[\tr(\rho_\mathsf{S}^{\prime\,2})]&=\f{\tr[\mbb1_\mathsf{RR}(\mbb1_\mathsf{EE}\otimes\swap_\mathsf{SS})]}{d_\mathsf{R}(d_\mathsf{R}+1)}+\f{\tr[\mbb1_\mathsf{RR}(\swap_\mathsf{EE}\otimes\mbb1_\mathsf{SS})]}{d_\mathsf{R}(d_\mathsf{R}+1)},\nonumber\\[0.1in]
    &\leq\tr\left[\left(\f{\mbb1_\mathsf{R}}{d_\mathsf{R}}\right)^{\otimes2}(\mbb1_\mathsf{EE}\otimes\swap_\mathsf{SS})\right]+\tr\left[\left(\f{\mbb1_\mathsf{R}}{d_\mathsf{R}}\right)^{\otimes2}(\swap_\mathsf{EE}\otimes\mbb1_\mathsf{SS})\right]\nonumber\\[0.1in]
    &=\tr[(\Omega_\mathsf{S}\otimes\Omega_\mathsf{S})\swap]+\tr[(\Omega_\mathsf{E}\otimes\Omega_\mathsf{E})\swap]\nonumber\\[0.1in]
    &=\tr(\Omega_\mathsf{S}^2)+\tr(\Omega_\mathsf{E}^2),
\end{align}
which neatly renders Eq.~\eqref{eq: Popescu tr dist 1} as
\begin{equation}
    \mbb{E}_\haar[{D}\left(\rho^\prime_\mathsf{S},\Omega_\mathsf{S}\right)]\leq\f{1}{2}\sqrt{d_\mathsf{S}\tr(\Omega_\mathsf{E}^2)},
\end{equation}
where $\tr(\Omega_\mathsf{E}^2)=d_\mathsf{R}^{-2}\tr[(\tr_\mathsf{S}(\mbb1_\mathsf{R}))^2]$ is called an (inverse) effective dimension of the environment, as it measures the dimension of the space in which the environment is most likely to be found. This is also in analogy with the effective dimension introduced in Section~\ref{sec: equilibration on average} by means of Eq.~\eqref{eq: deff inverse} (there for a probability of occupation of energy eigenstates). This can further be bounded as
\begin{equation}
    \tr(\Omega_\mathsf{E}^2)\leq{d}_\mathsf{S}/d_\mathsf{R},
\end{equation} by writing this effective dimension in terms of the eigenvalues $\lambda_i$ of $\Omega_\mathsf{E}$, and pulling out the largest one $\Lambda$ as $\tr(\Omega_\mathsf{E}^2)=\sum\lambda_i^2\leq\Lambda\sum\lambda_i=\Lambda\leq{d}_\mathsf{S}/d_\mathsf{R}$, given that $\tr_\mathsf{S}(\mbb{1}_\mathsf{R})=d_\mathsf{S}$.

Now, the Lipschitz constant of the trace distance ${D}(\tr_\mathsf{E}[\cdot],\Omega_\mathsf{S})$ can be seen to be equal to one, as was done in Eq.~\eqref{eq: Lipschitz tr dist}. Thus Levy's lemma can be readily applied to get\footnote{ Notice that in Ref.~\cite{Popescu2006} Levy's lemma is applied to the trace norm as opposed to the trace distance.}
\begin{theorem}[General Canonical Principle~\cite{Popescu2006}]
\label{theo: Canonical Principle}
For any pure state $\rho_\mathsf{R}\in\$(\mscr{H}_\mathsf{R})$ selected uniformly at random, where $\mscr{H}_\mathsf{R}\subseteq\mscr{H}_\mathsf{E}\otimes\mscr{H}_\mathsf{S}$ of dimension $d_\mathsf{R}$ is a subspace of an \gls{syst-env} composite of dimension $d_\mathsf{ES}=d_\mathsf{E}d_\mathsf{S}$, and any $\delta>0$,
\begin{equation}
    \mbb{P}_\haar\left[{D}(\rho_\mathsf{S},\Omega_\mathsf{S})\geq\f{1}{2}\sqrt{d_\mathsf{S}\tr(\Omega_\mathsf{E}^2)}+\delta\right]\leq2\exp\left(-\f{2d_\mathsf{R}\,\delta^2}{9\,\pi^3}\right),
\end{equation}
where $\Omega_\mathsf{S}=\tr_\mathsf{E}(\mbb1_\mathsf{R})/d_\mathsf{R}$ is called the canonical state of the system and with
\begin{equation}
    \tr(\Omega_\mathsf{E}^2)\leq\f{d_\mathsf{S}}{d_\mathsf{R}},
\end{equation}
satisfied for the inverse dimension of the environment $\tr(\Omega_\mathsf{E}^2)$.
\end{theorem}

This result implies that whenever $d_\mathsf{S}\ll1/\tr(\Omega_\mathsf{E}^2)$ and $\delta\ll1\ll d_\mathsf{R}\delta^2$, most quantum states will be almost canonical with very high probability. The second condition in particular reduces to $d_\mathsf{R}\gg1$ for $\delta=d_\mathsf{R}^{-1/3}$. Now, as discussed previously on equilibration in Section~\ref{sec: fluctuations infinite}, the systems of interest in realistic scenarios are typically much smaller that their environments and in particular we expect $d_\mathsf{R}\gg{d}_\mathsf{S}$ as well. This has the implication then that we do not need to invoke the a-priori equal probabilities postulate, but that we rather have typicality of canonical states emerging from their high entanglement with the effective environment. As seen in Ref.~\cite{Popescu2006}, when the restriction to the accessible space imposed by $\mathsf{R}$ corresponds to the total energy, the canonical state can be seen to correspond to the canonical distribution, or Gibbs so-called state, $\Omega_\mathsf{S}\simeq\exp(-\beta\,H_\mathsf{S})/\tr[\exp(-\beta\,H_\mathsf{S})]$ with an inverse temperature $\beta$ and system Hamiltonian $H_\mathsf{S}$, which is further exemplified through the model of a spin chain.

\begin{remark}We notice that this argument can be similarly applied to the expectation value of a given observable $A\in\mscr{B}(\mscr{H}_\mathsf{S})$. Notice that $\mbb{E}_\haar[\tr(A\rho^\prime_\mathsf{S})]=\tr[A\Omega_\mathsf{S}]$, so we really only need the Lipschitz constant of $\langle{A}\rangle_{\tr_\mathsf{E}(\cdot)}=\tr[A\tr_\mathsf{E}(\cdot)]$, which can be obtained similarly as in Eq.~\eqref{eq: Lipschits tr dist 2} by letting two $\rho_\mathsf{R}=|\phi\rangle\!\langle\phi|$ and $\rho_\mathsf{R}=|\phi\rangle\!\langle\phi|$ pure states so that
\begin{align}
    |\langle{A}\rangle_{\rho_\mathsf{R}-\sigma_\mathsf{R}}|&=\f{1}{2}\left|(\langle\phi|+\langle\psi|)A(|\phi\rangle-|\psi\rangle)+
    (\langle\phi|-\langle\psi|)A(\langle\phi|+\langle\psi|)\right|\nonumber\\
    &\leq\|A\|\left||\phi\rangle-|\psi\rangle\left|\,\right||\phi\rangle+|\psi\rangle\right|\nonumber\\[0.1in]
    &\leq2\|A\|\,\left||\psi\rangle-|\phi\rangle\right|,
\end{align}
and if $A$ only acts in subsystem $\mathsf{A}$ we get the same answer given that $\langle{A}\rangle_{\rho_\mathsf{S}}=\langle{A\otimes\mbb1}\rangle_{\rho_\mathsf{R}}$ and $\|A\otimes\mbb1\|=\|A\|$, so we conclude that $\mscr{L}=2$, and thus
\begin{equation}
    \mbb{P}_\haar\left[|\langle{A}\rangle_{\rho_\mathsf{R}-\sigma_\mathsf{R}}|\geq\delta\right]\leq2\exp\left(-\f{d_\mathsf{R}\,\delta^2}{18\pi^3\|A\|^2}\right),
\end{equation}
and here similarly $\delta$ can be chosen to a suitable value, e.g. $\delta=d_\mathsf{R}^{-1/3}$ so that for a large accessible space, the expectation value of any observable in a random state of the system will be almost the one on the canonical state with high probability.
\end{remark}

The approach of typicality is rather complementary to that of dynamical equilibration but it nevertheless provides a further understanding of the quantum foundations of statistical mechanics. One of its features, which can be seen as a drawback, is that it gives a kinematic argument for equilibration (or thermalization) rather than a dynamical one: it just speaks about almost all quantum states looking almost canonical but it does not say anything about how they get there. Another feature of the approach presented here, that will prove challenging in Part II of this thesis, is the drawing of state vectors from the Haar measure: Ref.~\cite{Gogolin_2011} in Section 6.2 discusses several different approaches with different measures that constrain typicality to a more physically meaningful notion. Finally, it is worth mentioning the approach to equilibration by Ref.~\cite{Masanes} employing so-called \emph{unitary designs}, which are distributions of unitary operators reproducing a finite number of moments of the Haar measure. We will formally introduce and employ this concept in Chapter~\ref{sec: Markovianization by design}.

    \chapter{Quantum Processes}
\label{sec:processes}
\setlength{\epigraphwidth}{0.6825\textwidth}
\epigraph{\emph{A philosopher once said, ``It is necessary for the very existence of science that the same conditions always produce the same results.''\\Well, they don't!}}{-- Richard P. Feynman (\emph{Character of Physical Law}).}

In the previous chapter, we already encountered more general transformations of quantum states that are not given simply by the action of an operator. In fact, the whole dynamics, either of the state of a whole Hilbert space (closed system), or of a subspace of a larger composite (open system), can be described as a the action of one of these \emph{superoperators} or \emph{quantum maps}. As we will see, however, when we try and do this quantum map description for an open quantum system composed of a subsystem \gls{syst} and an environment \gls{env}, we are faced with having to give up an important property known as Complete Positivity, unless we assume that the initial state has no correlations between \gls{syst} and \gls{env}.

Even if we assume that the initial state has no correlations, the unitary evolution of the whole \gls{syst-env} system will lead to a correlated state at a later time $t$ when the system is observed, so that if an experimenter were to intervene in this system at a later time $\tau>t$, they would be faced with the same conundrum as before. A resolution to this can be given by changing the approach to one which considers a dynamical map taking the operations that the experimenter can control, which mathematically are just quantum channels (or more generally, \gls{CP} maps), to output quantum states in what is known as the \emph{superchannel}~\cite{Modi2012}. As we will see, the superchannel naturally satisfies the complete positivity property and can account for two-time correlations between the initial preparation of the system and a final measurement.

However, as we mentioned, if the experimenter were to intervene on the system a third time, they would require a \emph{bigger} superchannel, i.e. a generalization taking as input more than a couple of interventions and that hence can account for multitime correlations. This is achieved by an object known as the \emph{process tensor}, which has a tensor-like structure in the sense of being a map from multiple \gls{CP} maps to a quantum state, and similarly satisfies the relevant mathematical properties such as being \gls{CP} itself.

The process tensor framework then naturally leads to a generalization of the classical notion of stochastic processes as a collection of random variables in time~\cite{Accardi_1976, Accardi_1982}. In particular, this gives a generalization of the concept of Markovianity from classical stochastic processes as that of a dependence in the past to make predictions. While there have been many attempts to do this~\cite{Breuer_2016}, these have proved inconsistent or insufficient to account for temporal correlations~\cite{Milz_PRL2019}. We will see that the process tensor gives a clear operational Markov condition generalizing the classical one, as well as providing a non-ambiguous measure of non-Markovianity.

Still, assuming no initial correlations and a weak coupling between \gls{syst} and \gls{env} along evolution is widely applicable and it has been fertile ground for research for many years~\cite{breuer2002theory}. Together with Chapter~\ref{sec:statmech}, this will bring forward questions that seem almost parallel to those in the foundations of statistical mechanics, namely, why are Markovian processes are so prevalent, when the theory of open quantum systems tells us---as we will see below---that temporal correlations should be the norm, and how does equilibration hold when temporal correlations come into play in multitime processes.

\section{Quantum channels}
\label{sec: quantum channels}

The transformations of quantum states that are considered (deterministically) physically realizable are known as quantum channels, and their usage date back from the 1960's with the work of George Sudarshan and collaborators~\cite{Sudarshan_1961, Sudarshan_2_1961} and some years after with Karl Kraus~\cite{Kraus}.

Quantum channels are linear maps $\Phi:\mscr{B}(\mscr{H}_\mathsf{in})\to\mscr{B}(\mscr{H}_\mathsf{out})$ for any two choice of spaces $\mscr{H}_\mathsf{in}$ and $\mscr{H}_\mathsf{out}$, with the additional conditions of being \gls{CP} and \gls{TP}, so as to preserve the properties of quantum states. The channel $\Phi$ can thus interchangeably be called a \gls{CPTP} map. In general the input and output spaces can differ; for simplicity and here we will usually assume $\mscr{H}_\mathsf{in}\cong\mscr{H}_\mathsf{out}$, unless stated otherwise, where we distinguish output from input spaces with a prime, e.g. $\mscr{H}_\mathsf{A}$ refers to an input space and $\mscr{H}_{\mathsf{A}^\prime}$ to an output space. Note also that whenever we refer to a map, we are axiomatically implying throughout this thesis a \emph{linear} map unless stated otherwise.\footnote{ A discussion can be seen in Ref.~\cite{Milz_operational}; in particular the requirement of linearity does not follow from the linearity of quantum mechanics but rather from linearity of mixing in a statistical theory.}

A map $\Phi$ being \gls{CP} means not only that its action on a positive operator is positive, $\Phi(X)\geq0$ for any $X\geq0$, but also that if its domain is only a subspace of a larger space, it will remain positive. That is, let $\mscr{H}\cong\mscr{H}_\ell\otimes\mscr{H}_\mathsf{in}$ where $\mscr{H}_\ell$ is an $\ell$-dimensional Hilbert space, and denote by $\mc{I}_\ell$ the identity map on $\mscr{H}_\ell$, i.e. a map acting trivially as $\mc{I}_\ell(X)=X$ for any $X\in\mscr{B}(\mscr{H})$; then the \gls{CP} property means
\begin{equation}
    (\mc{I}_\ell\otimes\Phi)(Y)\geq0,
\end{equation}
for all positive $Y\in\mc{B}(\mscr{H}_\ell\otimes\mscr{H}_{\mathsf{in}})$ and all $\ell$. In particular if positivity holds only for $\ell\leq\mathrm{L}$, then the map is said to be $\mathrm{L}$-positive. One such example is the map $\Phi_\mathrm{L}(X) = (\mathrm{L}-1)\tr(X)\mbb1_\mathrm{L}-X$~\cite{choi_1972}. The \gls{CP} property is motivated physically to ensure that these map states to states, even when these are correlated with another space. Nevertheless, historically, the \gls{CP} property hasn't been without contention, in particular for dynamical systems with initial correlations~\cite{Pechukas_1994, Alicki_1995, Pechukas_1995}, where it was argued that in such case either linearity or complete positivity would need to be given up. However as is explained in the following sections, when correctly accounting for correlations in composite systems it is not necessary to give up\footnote{ While we restrict ourselves to working with \gls{CP} maps, non-\gls{CP} maps have also been studied and are not without applications, see e.g. Ref.~\cite{Wood_thesis}} complete positivity~\cite{Modi2012, Milz_operational, milz_2019, Taranto_2020}.

Secondly, a map $\Phi$ is \gls{TP} whenever
$\tr[\Phi(X)]=\tr(X)$ for any $X\in\mscr{B}(\mscr{H}_\mathsf{in})$, and in particular this means that probabilities are conserved after the action of a \gls{TP} map on quantum states. As opposed to the \gls{CP} property, if a quantum map in question fails to be trace preserving but is \gls{TNI}, it remains a physical map, albeit with the interpretation that it can only ever be successfully realised with some probability. There is always a chance that some other transformation could take place.~\cite{Cappellini_2007}.

An example of a quantum channel that is ubiquitous in quantum theory is the partial trace, $\tr_{\mathsf{B}^\prime}:\mscr{B}(\mscr{H}_\mathsf{A}\otimes\mscr{H}_\mathsf{B})\to\mscr{B}(\mscr{H}_{\mathsf{A}^\prime})$, already introduced in Section~\ref{sec: quantum states and measurements} through Eq.~\eqref{eq: partial trace}. Other standard examples of quantum channels include the identity channel $\mc{I}$, with
\begin{equation}
        \mc{I}(\rho)=\rho,
\end{equation}

the unitary channel $\mc{U}$, with
\begin{equation}
    \mc{U}(\rho)=U\rho\,U^\dg,\quad\text{where}\quad\,UU^\dg=U^\dg{U}=\mbb1,
\end{equation}

and the depolarizing channel $\Lambda_q$, with
\begin{equation}
    \Lambda_q(\rho)=q\rho+(1-q)\f{\mbb1_\mc{H}}{d},\quad\,\text{with}\quad\,0\leq{q}\leq1.
    \label{eq: depolarizing}
\end{equation}

Similarly, we have also dealt already with the \emph{dephasing} channel, which as the name suggests, gets rid of the phases of a quantum state, or nullifies the off-diagonal terms of a state with respect to a given basis. Specifically, if we have a system with initial state $\rho\in\$(\mscr{H})$ and with evolution $U=\exp(-iHt)$ for a Hamiltonian $H=\sum_{n=1}^\mathfrak{D}{E}_nP_n$, with $P_n$ the projector to the eigenspace of energy $E_n$, then we can write
\begin{equation}
   \omega=\mc{D}(\rho),\quad\text{where}\,\quad\mc{D}(\cdot):=\sum_{n=1}^\mathfrak{D}P_n(\cdot)P_n,
   \label{eq: dephasing omega}
\end{equation}
for the dephasing map $\mc{D}$ under the Hamiltonian $H$. The twirling in Eq.~\eqref{eq: SW def} as well is another example of a quantum channel.

\section{Three representations of quantum channels}
The definition of a quantum channel can be somewhat abstract and described purely by its inputs and outputs. There are, however, ways to work practically with any quantum channel through their different representations, some of which can further be extended to any \gls{CPTP} or \gls{CPTNI} map or their generalizations. In the following we introduce three of the main representations for quantum channels.

\subsection{Operator sum representation}
The first representation we mention is a decomposition into operators which also gives a condition for the \gls{CP} property. This representation was first introduced in physics by Karl Kraus in 1971 in Ref.~\cite{Kraus} based on work by W. Forrest Stinespring in 1955~\cite{Stinespring}, and it is thus known as the operator sum representation, the Kraus representation or the Stinespring representation.

\begin{theorem}[Kraus representation~\cite{bengtsson2006geometry}]
\label{thm: Kraus representation}
A map $\Phi:\mscr{B}(\mscr{H}_\mathsf{in})\to\mscr{B}(\mscr{H}_\mathsf{out})$ is \gls{CP} if and only if its action has the form
\begin{equation}
    \Phi(\rho)=\sum_{i}K_i\rho{K}_i^\dg,
\end{equation}
where the operators $K_i$ are known as Kraus operators~\cite{kraus1983states}.
\end{theorem}

In particular, if $\Phi$ is \gls{TP}, then \begin{equation}
    \sum_iK_i^\dg{K}_i=\mbb1_\mathsf{in},
    \label{eq: Kraus complete}
\end{equation}
whereas if it is such that $\Phi(\mbb1_\mathsf{in})=\mbb1_\mathsf{out}$, i.e. so-called \emph{unital}, then
\begin{equation}
    \sum_iK_i{K}_i^\dg=\mbb1_\mathsf{out}.
\end{equation}

The Kraus representation is clearly not unique. The minimal number of Kraus operators is known as the Kraus rank, and satisfies $\kappa\leq{d}_\mathsf{in}d_\mathsf{out}$; in particular, there is always a representation with $\kappa$ orthogonal Kraus operators, i.e. such that $\tr[K_iK_j]=\delta_{ij}$ called the canonical Kraus form.

\begin{example}As a simple example, consider a two-level, qubit system, where any density matrix can be written as $\rho=\f{1}{2}(\mbb1+\undervec{r}\cdot\undervec{\sigma})$, where $\undervec{r}\in\mbb{R}^3$ is such that $|\undervec{r}|\leq1$, and $\undervec\sigma$ is a vector of the Pauli matrices
\begin{equation}
    \sigma_x=\begin{pmatrix}0&1\\1&0\end{pmatrix},\quad
    \sigma_y=\begin{pmatrix}0&-i\\i&0\end{pmatrix}
    ,\quad
    \sigma_z=\begin{pmatrix}1&0\\0&-1\end{pmatrix}.
\end{equation}

Then the Kraus operators
\begin{equation}
    K_1=\f{\sqrt{1+3q}}{2}\mbb1
    ,\quad{K}_2=\f{\sqrt{1-q}}{2}\sigma_x,
    \quad{K}_3=\f{\sqrt{1-q}}{2}\sigma_y,
    \quad{K}_4=\f{\sqrt{1-q}}{2}\sigma_z
\end{equation}
describe the action of a qubit depolarizing channel $\Lambda_q$ in Eq~\eqref{eq: depolarizing}. This map can equivalently be described as one with an action $\Lambda_Q(\rho)=Q\rho+(1-Q)\sum_i\sigma_i\rho\,\sigma_i$ through a relation between $Q$ and $q$~\cite{nielsen2000quantum}.
\end{example}

\subsection{Dilations}
A so-called purification was introduced in Section~\ref{sec: quantum states and measurements} as the fact that every mixed quantum state $\rho\in\mscr{B}(\mscr{H}_\mathsf{A})$ can be represented by means of a pure quantum state in a larger Hilbert space, $|\psi\rangle\!\langle\psi|\in\mscr{B}(\mscr{H}_\mathsf{A}\otimes\mscr{H}_\mathsf{B})$, as the reduced state $\rho=\tr_\mathsf{B}(|\psi\rangle\!\langle\psi|)$. Quantum states, however, can actually be thought of as being a particular case of a quantum channel $\rho:\mbb{C}\to\mscr{B}(\mscr{H})$, so really the idea of purification more generally extends to any quantum channel.

The analogous purification of quantum channels is more precisely referred to as a \emph{Stinespring dilation} and follows from a more general mathematical result known as the Stinespring dilation theorem~\cite{Stinespring}. In particular, it implies that there exists an ancillary system, say $\mathsf{B}$, such that the action of a \gls{CPTP} map $\Phi:\mscr{B}(\mscr{H}_{\mathsf{A}_\mathsf{in}})\to\mscr{B}(\mscr{H}_{\mathsf{A}_\mathsf{out}})$ can be written as 
\begin{equation}
    \Phi(\rho)=\tr_{\mathsf{B}_\mathsf{out}}[\,\mc{U}(\rho\otimes\beta)],
    \label{eq: dilation}
\end{equation}
where here $\mc{U}:\mscr{B}(\mscr{H}_{\mathsf{A}_\mathsf{in}}\otimes\mscr{H}_{\mathsf{B}_\mathsf{in}})\to\mscr{B}(\mscr{H}_{\mathsf{A}_\mathsf{out}}\otimes\mscr{H}_{\mathsf{B}_\mathsf{out}})$ is a unitary map, and with $\beta\in\mscr{B}(\mscr{H}_{\mathsf{B}_\mathsf{in}})$ a given quantum state in the ancillary space. Notice $d_{\mathsf{AB}_\mathsf{in}}=d_{\mathsf{AB}_\mathsf{out}}$, although the individual dimensions may differ. Similarly, the more general Stinespring theorem considers isometries\footnote{ An isometry on a Hilbert space is a linear operator that preserves distances; this is equivalent to any linear operator $W$ such that $\|Wv\|=\|v\|$ for all $v$ in such space, with $\|\cdot\|$ a corresponding norm. A unitary operator is a particular case of an isometry.} that might not actually be unitary and thus applies in general to \gls{CP} maps; here we only deal with dilations of the form of Eq.~\eqref{eq: dilation}.

The dilation representation is a powerful one that serves as a cornerstone in open quantum dynamics, as it can be interpreted as all quantum channels arising from a unitary interaction with an environment, as we detail in section~\ref{sec: open quantum dynamics}. Notice this implies \gls{CPTP} maps are precisely the crucial type of maps when it comes to the physical picture; this will be dealt with in depth in the following sections. More generally, the dilation representation is useful when the properties of the global unitary channel can be used to deduce properties of the channel in question, as it can be done with more general maps as in section~\ref{sec: process tensor}; Stinespring theorem assures that a dilation exists for all \gls{CP} maps, even if it is not unique (more precisely, they can be said to be unique up to a unitary transformation).

\begin{figure}[t]
    \centering
    \begin{tikzpicture}
    \begin{scope}
    \shade[outer color=C3!50!white, inner color=white, draw=black, rounded corners, thick] (-0.75,-0.75) rectangle (0.75,0.75);
    \draw[thick,-] (-1.75,0) -- (-0.75,0);
    \draw[thick, -] (0.75,0) -- (1.75,0);
    \node[left] at (-1.75,0) {$\rho$};
    \node[above] at (-1.25,0) {$\mathsf{A}_\mathsf{in}$};
    \node at (0,0) {\LARGE$\Phi$};
    \node[above] at (1.25,0) {$\mathsf{A}_\mathsf{out}$};
    \node[right] at (1.75,0) {$\rho_\mathsf{out}$};
    \end{scope}
    \begin{scope}[shift={(3.5,0)}]
    \node at (0,0) {$\Longleftrightarrow$};
    \end{scope}
    \begin{scope}[shift={(1.75,0)}]
    \node[left] at (3.5,0.5) {$\beta$};
    \node[above] at (4,0.5) {$\mathsf{B}_\mathsf{in}$};
    \node[left] at (3.5,-0.5) {$\rho$};
    \node[above] at (4,-0.5) {$\mathsf{A}_\mathsf{in}$};
    \draw[thick, -] (3.5,0.5) -- (4.5,0.5);
    \draw[thick, -] (3.5,-0.5) -- (4.5,-0.5);
    \shade[outer color=C2!60!white, inner color=white, draw=black, rounded corners, thick] (4.5,-1) rectangle (6,1);
    \node at (5.25,0) {\LARGE$\mc{U}$};
    \draw[thick, -] (6,0.5) -- (7,0.5);
    \draw[thick, -] (6,-0.5) -- (7,-0.5);
    \node[above] at (6.5,0.5) {$\mathsf{B}_\mathsf{out}$};
    \node[above] at (6.5,-0.5) {$\mathsf{A}_\mathsf{out}$};
    \node[right] at (7,-0.5) {$\rho_\mathsf{out}$};
    \node[right] at (7,0.5) {\trash};
    \end{scope}
    \end{tikzpicture}
    \caption[Stinespring dilation of a quantum channel]{\textbf{Stinespring dilation of a quantum channel:} The action of a quantum channel $\Phi:\mscr{B}(\mscr{H}_\mathsf{A})\to\mscr{B}(\mscr{H}_{\mathsf{A}\mathsf{out}})$ on a state $\rho$ can be represented as the reduced state of the outcome of a unitary map $\mc{U}:\mscr{B}(\mscr{H}_\mathsf{A}\mathsf{in}\otimes\mscr{H}_\mathsf{B}\mathsf{in})\to\mscr{B}(\mscr{H}_{\mathsf{A}_\mathsf{out}}\otimes\mscr{H}_{\mathsf{B}_\mathsf{out}})$ acting on a joint state $\rho\otimes\beta$. We denote quantum maps by boxes acting on inputs on the left and yielding outputs to the right. Lines denote Hilbert spaces, with no notion of time or direction attached to them; we normally assume $\mscr{H}_\mathsf{in}\cong\mscr{H}_\mathsf{out}$ and omit the output space labels when sufficiently clear.}
    \label{fig: Stinespring dilation of a quantum channel}
\end{figure}

A practical way to think intuitively about quantum operations is by means of graphical diagrams depicting the relations between maps and states. As we will see, this seemingly simple way of depicting abstract quantities is a handy but powerful tool that allows to deal with more general and complex situations; in essence it is used to represent any quantum computation as a circuit of inputs, quantum channels, which here are the equivalent of logical operations, and outputs. While we do not stick to all the conventions, a thorough description of the graphical calculus often employed in open quantum systems and quantum information can be seen in Ref.~\cite{coecke_kissinger_2017, Wood2015Tensor}.

In Fig.~\ref{fig: Stinespring dilation of a quantum channel} we depict the dilation of a \gls{CP} map $\Phi:\mscr{B}(\mscr{H}_{\mathsf{A}_\mathsf{in}})\to\mscr{B}(\mc{H}_{\mathsf{A}_\mathsf{out}})$ acting on an input state $\rho$ and rendering as outcome a state $\rho_\mathsf{out}=\Phi(\rho)$. Conventionally we depict inputs to the left and outputs to the right, with lines representing Hilbert spaces and boxes the quantum maps; we further represent pictorially the partial tracing operation with a trash can (\trash)\,symbol. Whenever there are parallel lines, these represent a tensor product, as in Fig.~\ref{fig: Stinespring dilation of a quantum channel} for $\mscr{H}_\mathsf{A}\otimes\mscr{H}_\mathsf{B}$; in this case the input state is a product state $\rho\otimes\beta$ as well, but in general this need not be the case.

It is instructive to see how the Stinespring dilation relates to the Kraus representation. Let us take Eq.~\eqref{eq: dilation} assuming that the ancillary space is prepared in a pure state $|\beta\rangle\!\langle\beta|$; then we can write explicitly the partial trace by introducing a basis $\{|b\rangle\}_{b=1}^{d_{\mathsf{B}_\mathsf{out}}}$ for $\mscr{H}_{\mathsf{B}_\mathsf{out}}$ as
\begin{align}
    \Phi(\rho)&=\sum_{b=1}^{d_{\mathsf{B}_\mathsf{out}}}(\mbb1_{\mathsf{A}_\mathsf{out}}\otimes\langle{b}|)\,U\,(\rho\otimes|\beta\rangle\!\langle\beta|)\,U^\dg(\mbb1_{\mathsf{A}_\mathsf{out}}\otimes|b\rangle)\nonumber\\
    &=\sum_{b=1}^{d_{\mathsf{B}_\mathsf{out}}}(\mbb1_{\mathsf{A}_\mathsf{out}}\otimes\langle{b}|)\,U\,(\mbb1_{\mathsf{A}_\mathsf{in}}\otimes|\beta\rangle)\,(\rho\otimes\mbb1_{\mathsf{B}_\mathsf{in}})\,(\mbb1_{\mathsf{A}_\mathsf{in}}\otimes\langle\beta|)\,U^\dg(\mbb1_{\mathsf{A}_\mathsf{out}}\otimes|b\rangle)\nonumber\\
    &:=\sum_{b=1}^{d_{\mathsf{B}_\mathsf{out}}}K_b\,\rho\,{K}_b^\dg,
    \label{eq: Stinespring CPTP}
\end{align}
where we identified $K_b:=(\mbb1_{\mathsf{A}_\mathsf{out}}\otimes\langle{b}|)\,U\,(\mbb1_{\mathsf{A}_\mathsf{in}}\otimes|\beta\rangle)$ with  the Kraus operators~\cite{kraus1983states}. These clearly satisfy the completeness property in Eq.~\eqref{eq: Kraus complete}. For a canonical Kraus representation this in turn implies a minimal $d_{\mathsf{B}_\mathsf{out}}\leq{d}_{\mathsf{A}_\mathsf{in}}d_{\mathsf{A}_\mathsf{out}}$. This can then be done similarly if the ancilla is prepared in a mixed state if written as a convex combination~\cite{bengtsson2006geometry}.

\subsection{Choi-Jamio\l{}kowski isomorphism}
The third main representation that we will deal with in this thesis is based on a (one-to-one) correspondence between quantum channels and quantum states. To begin with, consider a basis $\{|i\rangle\}$ for a $d$-dimensional Hilbert space $\mscr{H}$, and define the \emph{vectorization} map by
\begin{equation}
    \mathrm{vec}(|i\rangle\!\langle{j}|)=|ij\rangle.
\end{equation}

This gives an isomorphism $\mscr{B}(\mscr{H})\cong\mscr{H}\otimes\mscr{H}$ in the sense that it is a one to one correspondence between linear bounded operators and vectors. In general we can do this for any operator $A\in\mscr{B}(\mscr{H})$ in the basis above with components $\alpha_{ij}=\langle{i}|A|j\rangle$ as
\begin{equation}
    |\phi_A\rangle:=\mathrm{vec}(A)=\sum_{i,j=1}^d{\alpha_{ij}}|ij\rangle,
\end{equation}
or similarly for a vector $|\varphi\rangle\in\mscr{H}\otimes\mscr{H}$, we can write this in a given basis and turn one of the basis vectors in a covector. Now, how exactly do we get this correspondence? Notice that we may write
\begin{align}
    |\phi_A\rangle&=\sum_{i,j=1}^d\alpha_{ij}|i\rangle\otimes|j\rangle=\sum_{i,j,k=1}^d\alpha_{ij}|k\rangle\otimes|i\rangle\!\langle{j}|k\rangle=\sum_{k=1}^d|k\rangle\otimes{A}|k\rangle\nonumber\\
    &=(\mbb1\otimes{A})\sum_{i=1}^d|ii\rangle,
\end{align}
where the state $|\tilde{\psi}\rangle:=\sum|ii\rangle$ is precisely the maximally entangled state in Eq.~\eqref{eq: def maximally entangled state} up to a normalization factor. We will usually refer to this as an unnormalized maximally entangled state. It also does not matter if we act with $A$ on the first or the second Hilbert space. This means we can write $|\phi_A\rangle=(\mbb1\otimes{A})|\tilde{\psi}\rangle=(A\otimes\mbb1)|\tilde{\psi}\rangle$.

Let us label the pair of copies of the spaces with $\mathsf{1}$ and $\mathsf{2}$, so that $|\phi_A\rangle\in\mscr{H}_\mathsf{1}\otimes\mscr{H}_\mathsf{2}$. We can go the opposite way via
\begin{equation}
    A=\tr_\mathsf{2}\left(|\phi_{A^\mathrm{T}}\rangle\!\langle\tilde{\psi}|\right)=\tr_\mathsf{2}\left[\left(\mbb1_\mathsf{1}\otimes{A}^\mathrm{T}\right)|\tilde{\psi}\rangle\!\langle\tilde{\psi}|\right],
\end{equation}
where $\mathrm{T}$ denotes a transpose, i.e. $A^\mathrm{T}=\sum\alpha_{ji}|i\rangle\!\langle{j}|$.

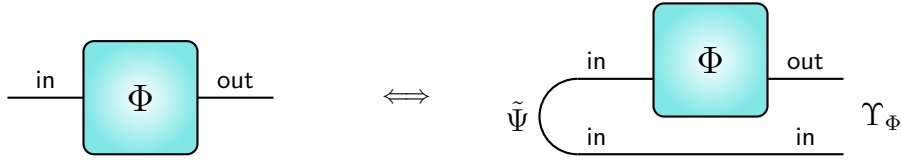
\begin{figure}[t]
    \centering
    \begin{tikzpicture}
    \begin{scope}
    \shade[outer color=C3!50!white, inner color=white, draw=black, rounded corners, thick] (-0.75,-0.75) rectangle (0.75,0.75);
    \draw[thick,-] (-1.75,0) -- (-0.75,0);
    \draw[thick, -] (0.75,0) -- (1.75,0);
    \node[above] at (-1.25,0) {$\mathsf{in}$};
    \node at (0,0) {\LARGE$\Phi$};
    \node[above] at (1.25,0) {$\mathsf{out}$};
    \end{scope}
    \begin{scope}[shift={(3.5,0)}]
    \node at (0,0) {$\Longleftrightarrow$};
    \end{scope}
    \begin{scope}[shift={(2.25,0.5)}]
    \node[left] at (3,-0.75) {\Large$\tilde{\Psi}$};
    \node[above] at (3.75,-0.25) {$\mathsf{in}$};
    \node[above] at (3.75,-1.25) {$\mathsf{in}$};
    \draw[thick, -] (3.5,-0.25) -- (4.5,-0.25);
    \draw[thick] (3.5,-0.25) arc (90:270:0.5);
    \draw[thick, -] (3.5,-1.25) -- (7,-1.25);
    \shade[outer color=C3!50!white, inner color=white, draw=black, rounded corners, thick] (4.5,-0.75) rectangle (6,0.75);
    \node at (5.25,0) {\LARGE$\Phi$};
    \draw[thick, -] (6,-0.25) -- (7,-0.25);
    \node[above] at (6.5,-0.25) {$\mathsf{out}$};
    \node[above] at (6.5,-1.25) {$\mathsf{in}$};
    \node[right] at (7.1,-0.75) {\Large$\Upsilon_\Phi$};
    \end{scope}
    \end{tikzpicture}
    \caption[Choi-Jamio\l{}kowski isomorphism]{\textbf{The Choi-Jamio\l{}kowski isomorphism} gives a one-to-one correspondence between quantum maps and states. A map $\Phi:\mscr{B}(\mscr{H}_\mathsf{in})\to\mscr{B}(\mscr{H}_{\mathsf{out}})$ can be represented as a quantum state $\Upsilon_\Phi\in\mscr{B}(\mscr{H}_{\mathsf{out}}\otimes\mscr{H}_{\mathsf{in}})$, by letting $\Phi$ act on half an unnormalized maximally entangled state $\tilde\Psi\in\mscr{B}(\mscr{H}_{\mathsf{in}}\otimes\mscr{H}_{\mathsf{in}})$, which is represented by an arch joining the entangled systems.}
    \label{fig: Choi-Jamiolkowski}
\end{figure}

So what we expect now is to have an analogous correspondence between quantum channels and density matrices; and indeed, if we let
\begin{equation}
    \tilde{\Psi}:=\sum_{i,j=1}^{d_\mathsf{A}}|ii\rangle\!\langle{jj}|\in\mscr{B}(\mscr{H}_\mathsf{in}\otimes\mscr{H}_\mathsf{in}),
\end{equation}
and $\Phi:\mscr{B}(\mscr{H}_\mathsf{in})\to\mscr{B}(\mscr{H}_\mathsf{out})$ a quantum channel, we have
\begin{equation}
    \Upsilon_\Phi=(\Phi\otimes\mc{I})\tilde{\Psi},
    \label{eq: definition CJI}
\end{equation}
and we can recover the action of the map $\Phi$ through
\begin{equation}
    \Phi(\rho)=\tr_\mathsf{in}[\Upsilon_\Phi(\mbb1_\mathsf{out}\otimes{\rho}^\mathrm{T})],
    \label{eq: map through Choi}
\end{equation}
for any $\rho\in\$(\mscr{H}_\mathsf{in})$, and where the partial trace is over the input space. The matrix in Eq.~\eqref{eq: definition CJI} defines what is referred to as the \gls{CJI}~\cite{Jamiolkowski, Choi_1975}, and we will refer to $\Upsilon_\Phi$ as the Choi state of $\Phi$. It can also be visualized diagrammatically as in Fig.~\ref{fig: Choi-Jamiolkowski}.

The \gls{CJI}, however, as opposed to the Kraus representation or the dilation representation, not only applies to quantum channels or \gls{CP} maps but generally to any linear map. Here we will consider primarily \gls{CP} maps, for which Eq.~\eqref{eq: definition CJI} implies $\Upsilon_\Phi\geq0$; indeed a positive Choi state can be decomposed as $\Upsilon_\Phi=\sum_{i=1}^{\mathsf{D}}{g_i}|\gamma_i\rangle\!\langle\gamma_i|$ with its $\mathsf{D}=d_\mathsf{in}d_\mathsf{out}$ nonnegative eigenvalues $g_i$ and eigenvectors $\{|\gamma_i\rangle\}$, then
\begin{align}
    \Phi(X)&=\sum_{i=1}^{\mathsf{D}}g_i\tr_\mathsf{in}[|\gamma_i\rangle\!\langle\gamma_i|(\mbb1_\mathsf{out}\otimes{X}^\mathrm{T})]\nonumber\\
    &=\sum_{i=1}^{\mathsf{D}}\sum_{e=1}^{d_\mathsf{in}}g_i(\mbb1_\mathsf{out}\otimes\langle{e}_\mathsf{in}|)|\gamma_i\rangle\!\langle\gamma_i|(\mbb1_\mathsf{out}\otimes{X}^\mathrm{T}|e_\mathsf{in}\rangle)\nonumber\\
    &=\sum_{i=1}^{\mathsf{D}}\sum_{e,\varepsilon=1}^{d_\mathsf{in}}g_i(\mbb1_\mathsf{out}\otimes\langle{e}_\mathsf{in}|)|\gamma_i\rangle\!\langle\gamma_i|(\mbb1_\mathsf{out}\otimes|\varepsilon_\mathsf{in}\rangle)\langle\varepsilon_\mathsf{in}|{X}^\mathrm{T}|e_\mathsf{in}\rangle\nonumber\\
    &=\sum_{i=1}^{\mathsf{D}}\sum_{e,\varepsilon=1}^{d_\mathsf{in}}g_i(\mbb1_\mathsf{out}\otimes\langle{e}_\mathsf{in}|)|\gamma_i\rangle\!\langle{e}_\mathsf{in}|{X}|\varepsilon_\mathsf{in}\rangle\!\langle\gamma_i|(\mbb1_\mathsf{out}\otimes|\varepsilon_\mathsf{in}\rangle)\nonumber\\
    &:=\sum_{i=1}^\mathsf{D}\mathfrak{G}_iX\mathfrak{G}_i^\dg,
\end{align}
which gives a Kraus representation with $\mathfrak{G}_i:=\sum_{e=1}^{d_\mathsf{in}}\sqrt{g_i}(\mbb1_\mathsf{out}\otimes\langle{e}_\mathsf{in}|)|\gamma_i\rangle\!\langle{e}_\mathsf{in}|$ each of the Kraus operators.

If additionally the map is \gls{TP} as well, we have
\begin{equation}
    \tr_\mathsf{out}[\Upsilon_\Phi]=\sum_{i,j=1}^{d_\mathsf{in}}\tr[\Phi(|i\rangle\!\langle{j}|)]|i\rangle\!\langle{j}|=\mbb1_\mathsf{in},
    \label{eq: Choi state TP condition}
\end{equation}
which makes evident as well that the Choi \emph{state} does not have unit trace by construction.

The punchline of the \gls{CJI} is that linear maps can be represented as matrices in a larger space, and in particular that quantum channels have corresponding positive matrices, also in a larger space. Two straightforward examples come for the identity map, which leads simply to $\Upsilon_\mc{I}=\tilde\Psi$ and the unitary map, which leads also to a maximally entangled state $\Upsilon_\mc{U}=(\mc{U}\otimes\mc{I})\tilde\Psi$ between the input and output states.

The Choi state in general will be a relevant tool in this thesis, albeit for more general maps, as we will see below.

\section{Open quantum dynamics}\label{sec: open quantum dynamics}
Let us now go back to the physical motivation and consider a generic scenario where an experimenter has access to the subpart \gls{syst} of a larger composite \gls{syst-env}, where an environment \gls{env} is out of access and control to the experimenter. The whole \gls{syst-env} universe is closed so in general it will evolve unitarily. The system \gls{syst}, however, from the perspective of the experimenter, will be described by a \gls{CP} map taking preparations and yielding outcome states, which suggests a Stinespring representation as in Eq.~\eqref{eq: dilation}.

The standard experimental procedures to reconstruct quantum states and quantum dynamical maps are known as quantum state tomography (\texttt{QST}) and quantum process tomography (\texttt{QPT}), respectively. \texttt{QST} relies on the measurement statistics, $\vec{p}=(\tr[\mathrm{M}_1\rho],\ldots,\tr[\mathrm{M}_\ell\rho])$, of a given measurement $\{\mathrm{M}_i\}$ (which has to form a basis on the space of the system) to reconstruct $\rho$. \texttt{QPT}, on the other hand, aims to reconstruct the dynamical \gls{CP} map, $\Phi$, by preparing a set of linearly independent input states $\{\rho_i\}$, sending them through the map, and then collecting the outputs as a linear combination of the inputs, which then allows to reconstruct a pair of Kraus operators of the channel. We will not deal in detail with either procedure but it lets us put into perspective what the real world situation is; more on either can be seen e.g. in Ref.~\cite{nielsen2000quantum}.

Let us consider then the following scenario: at some time, which we may set as $t=0$, the system \gls{syst} is uncorrelated from \gls{env} so that the experimenter is able to prepare inputs $\rho_\mathsf{S}(0)$ and reconstruct the final states by \texttt{QST} at some time $t=\delta{t}$. We will only consider unitary evolution of the whole under a time-independent Hamiltonian $H$, so that
\begin{equation}
    \mc{U}_{\delta{t}}(\cdot)=\exp(-iH\delta{t})(\cdot)\exp(iH\delta{t}),
    \label{eq: unitary map Hamiltonian}
\end{equation}
is the unitary map of the evolution during $\delta{t}$. Letting the environment state at $t=0$ be $\varepsilon$, we have
\begin{equation}
    \rho_\mathsf{S}(t)=\mscr{Z}_{\delta{t}}[\rho_\mathsf{S}(0)]=\tr_\mathsf{E}[\,\mc{U}_{\delta{t}}(\rho_\mathsf{S}(0)\otimes\varepsilon)],
\end{equation}
where here $\mscr{Z}_{\delta{t}}:\$(\mscr{H}_\mathsf{S})\to\$(\mscr{H}_\mathsf{S})$ is the dynamical map taking \gls{syst} states to \gls{syst} states.

We can allow time to vary for the dynamical map so that $\{\mscr{Z}_t:t\geq0\}$ constitutes a one-parameter family of dynamical maps with $\mscr{Z}_0=\mc{I}$; we can furthermore see that these are \gls{CPTP}, as we did for the Stinespring dilation in Eq.~\eqref{eq: Stinespring CPTP}. Now, experimentally, \texttt{QPT} can be used to reconstruct the dynamical map.

\subsection{The Born-Markov approximation}
The reconstruction of the dynamical maps $\mscr{Z}_t$ can be done for numerous physical cases and phenomenological models~\cite{breuer2002theory}. Despite this, even ignoring the fact that it is very difficult to prepare a system that is uncoupled to the environment~\cite{PhysRevA.76.042113}, computing the dynamical map is typically unmanageable without making some further simplifying assumptions.

Notice that even if the state of the whole composite, $\rho_\mathsf{SE}$, obeys the Schr\"{o}dinger equation, this does not immediately imply that an analogue differential equation, i.e. a time-local master equation, for $\rho_\mathsf{S}$ exists. In particular, if $\rho_\mathsf{S}$ obeys a differential equation, this means that the state of \gls{syst} is local in time~\cite{Preskill_QI_notes}, i.e. to determine $\rho_\mathsf{S}(t+dt)$ we would only need to know $\rho_\mathsf{S}(t)$. However, the interaction with the environment makes this generically impossible, as information will be exchanged between \gls{syst} and \gls{env} and we would need to know the state $\rho_\mathsf{S}$ at earlier times as well.

The first notion of locality in time is what is commonly understood in open quantum dynamics as \emph{Markovianity}. We will expand on this notion below, however, the general concept is the same: that of \emph{memorylessness}, as opposed to \emph{non-Markovianity}, in which we need to know the past states of a given system to determine its future. Crucially, notice once again that open quantum systems theory is non-Markovian \emph{by definition}, and that Markovian open quantum dynamics are in reality impossible.

Despite this, as is common when doing idealizations in all of physics, approximating open systems as Markovian is effective for a large class of physical scenarios and has a wide applicability~\cite{carmichael1993open, blanchard2000decoherence, schlosshauer2007decoherence, alicki2007quantum}. The simplest assumptions that can be made to render the dynamics Markovian are known as the Born-Markov approximation, which lead to an important class, albeit the simplest one, of open quantum dynamics and master equations.

\begin{definition}[Born-Markov approximation]
\label{def: Born-Markov approximation}
These are two assumptions:
\begin{compactenum}[\itshape i.]
    \item Weak coupling (Born): The coupling between \gls{syst} and \gls{env} is sufficiently weak and \gls{env} is reasonably large, so that $\varepsilon$ is practically unaffected and the whole \gls{syst-env} state remains approximately in product state at all times: $\rho_\mathsf{SE}(t)\approx\rho_\mathsf{S}(t)\otimes\varepsilon,\,\forall{t}\geq0$.
    \item Forgetful environment (Markov): The self-correlations within \gls{env} induced by the interaction with \gls{syst} decay rapidly compared to the timescale over which $\rho_\mathsf{S}$ changes noticeably.
\end{compactenum}
\end{definition}

Beginning from the Schr\"{o}dinger equation, Eq.~\eqref{eq: Schroedinger}, for a total Hamiltonian taking into account the \gls{syst-env} interaction, a master equation by the name of Lindblad-Franke-GKS equation\footnote{ GKS standing after Gorini, Kossakowski and Sudarshan. The equation is sometimes only referred to after Lindblad. The contributions are: Franke~\cite{Franke1976}, GKS~\cite{GKS} and Lindblad~\cite{Lindblad1976} (all from 1976).} can be derived by assuming the Born-Markov approximation~\cite{breuer2002theory}. While the Born-Markov conditions are the underlying physical conditions to render the Hamiltonian dynamics Markovian, the Lindblad-Franke-GKS equation in a general form can be obtained by positing that the dynamical map satisfies
\begin{equation}
   \mscr{Z}_{t_1+t_2}=\mscr{Z}_{t_1}\mscr{Z}_{t_2},
    \label{eq: divisibility dynamical maps} 
\end{equation}
where $\mscr{Z}_{t_1}\mscr{Z}_{t_2}:=\mscr{Z}_{t_1}\circ\mscr{Z}_{t_2}$ implicitly means a composition of maps. This is known as the divisibility, or semigroup, property, and it is usually employed in the definition of quantum Markovianity, as it resembles a property known by the same name in the theory of classical stochastic processes. We will come back to this issue below, as it will be important for the discussion in this thesis, the point about this property now is that it allows to write $\mscr{Z}_t=\exp(\mc{L}\,t)$, with $\mc{L}$ known as the generator of the semigroup, which then can be written in its most general form~\cite{breuer2002theory}, and which in particular can be derived explicitly using the Born-Markov approximation.

The assumptions in the Born-Markov approximation can be boiled down to positing that the environment is negligibly affected by its interaction with the system and that it is forgetful or dissipative. Both assumptions are widely applicable in physical settings, the Born assumption for example could be taken considered for macroscopic systems or on systems within very large environments and the Markov condition could be used if one considers an environment at a high enough temperature.

This already begs the question of how Markovian behaviour comes about with no a-priori assumptions. As expected from the discussion in Chapter~\ref{sec:statmech}, we can make a very clear analogy between this question and that of the foundations of statistical mechanics. To do so, let us continue analysing the problem of system-environment correlations in time and how to resolve them.

\subsection{The initial correlation problem}
Let us embrace now the inevitability of the coupling of system \gls{syst} to the environment \gls{env}, i.e. in a standard experimental setup, the experimenter's preparation of input states might affect the environment as well. The reconstruction of dynamical maps by \texttt{QPT} was suggested in the late nineties~\cite{Nielsen_QPT} and soon after experiments began to be carried out~\cite{Nielsen1998,Childs_2001,Mitchell_2003,PhysRevLett.90.193601,Weinstein_2004, OBrien_2004, Myrskog_2005, Howard_2006, Chow_2009}. However, non completely positive behaviour was being obtained, which was discarded as experimental noise and hence just not physically valid. On the other hand, despite not having access to the whole \gls{syst-env} composite, experiments were also devised to detect initial correlations~\cite{Li_2011, Gessner2014}.

As it turns out, in 1994 and 95, Philip Pechukas and Robert Alicki showed that if initial correlations are to be considered, then one either has to give up complete positivity or linearity on the dynamical map~\cite{Pechukas_1994, Alicki_1995, Pechukas_1995}. A simple argument based on the original by Pechukas and Alicki can be seen in Ref.~\cite{Milz_operational}, where a map labelled an \emph{assignment map}, $\zeta:\mscr{B}(\mscr{H}_\mathsf{S})\to\mscr{B}(\mscr{H}_\mathsf{E}\otimes\mscr{H}_\mathsf{S})$, such that it satisfies a \emph{consistency condition} $\tr_\mathsf{E}[\zeta(\rho_\mathsf{S})]=\tr_\mathsf{E}(\rho_\mathsf{SE})=\rho_\mathsf{S}$ is used to write 
\begin{equation}
    \mscr{Z}_t[\rho_\mathsf{S}]=\tr_\mathsf{E}[\,\mc{U}_t\,\zeta(\rho_\mathsf{S})]=\tr_\mathsf{E}[\,\mc{U}_t(\rho_\mathsf{SE})],
\end{equation}
where here $\mc{U}_t$ is a unitary map as defined in Eq.~\eqref{eq: unitary map Hamiltonian} and where $\rho_\mathsf{S}$ and $\rho_\mathsf{SE}$ stands for the initial state on either \gls{syst} or \gls{syst-env}. Then demanding the \gls{CP} property on $\mscr{Z}_t$ implies (as can be seen in detail in Ref.~\cite{Milz_operational}) that $\zeta(\rho_\mathsf{S})=\rho_\mathsf{S}\otimes\varepsilon$, i.e. the initial state is uncorrelated, or otherwise either the \gls{CP} property or linearity have to be given up. This conundrum is known as the initial correlation problem.

Now, while the \gls{CP} property is desirable, it seems tempting to try and go around it, as opposed to try and do the same with linearity, which as we mentioned at the beginning of this chapter is related to the mixing property of any statistical theory, and furthermore, it is a pillar assumption for \texttt{QPT}. In particular, simple examples of \gls{NCP} behaviour, i.e. maps that are positive only for some subset $\{\rho_\mathsf{S}:\zeta(\rho_\mathsf{S})\geq0\}$ of \emph{compatible} states, can be seen to arise whenever the system and the environment are entangled~\cite{PhysRevA.64.062106}, and more generally for other types of correlations~\cite{Rodriguez_Rosario_2008,Aspuru_2010}. This led to several arguments for embracing \gls{NCP} maps~\cite{Jordan_2004, Shaji_2005, Jordan_2006,PhysRevA.77.042113} at the beginning of the current century; the main problem, however, is operationally determining the set of compatible states, i.e. an experimenter does not necessarily know which is this set and determining it experimentally leads to ambiguity, since the very act of trying to identify a correlated system will disturb the environment~\cite{Modi_2011,Xu_2012,Majeed_2019}. A simple example with only two qubits can be seen in detail in Ref.~\cite{Milz_operational}.

\subsection{The superchannel}
A resolution to the initial correlation problem comes rather with a change of perspective. As Asher Peres puts it~\cite{Peres_2003}: ``\emph{The simple and obvious truth is that quantum phenomena do not occur in a Hilbert space. They occur in a laboratory. (\ldots) The experimenter controls the emission process and observes detection events. The theorist’s problem is to predict the probability of response of this or that detector, for a given emission procedure}''. The ambiguity in the \gls{NCP} scenario in a sense already gives us a hint: in any experiment the very first step is to prepare some unknown fiducial state into a known input. But if the fiducial state is correlated, the experimenter cannot do this without disturbing the environment, so the prepared states are not really inputs anymore. The crucial insight in Ref.~\cite{Modi2012} was thus to adopt a description that focuses on the objects that can be controlled operationally: preparations and measurements. Moreover, we will see that we do not need to give up linearity or positivity in any meaningful way.

\begin{figure}[t]
    \centering
    \begin{tikzpicture}
    \begin{scope}
    \fill[outer color=C4!60!white, inner color=white, draw=black, rounded corners, thick] (0.5,-1.25) -- (0.5,1.25) -- (5.5,1.25) -- (5.5,-1.25) -- (4,-1.25) -- (4,0.25) -- (2,0.25) -- (2,-1.25) -- cycle;
    \draw[thick, -] (2,-0.5) -- (2.5,-0.5);
    \draw[thick, -] (3.5,-0.5) -- (4,-0.5);
    \draw[thick, -] (5.5,-0.5) -- (6,-0.5);
    \node[below,right] at (2,-0.75) {$\mathsf{in}$};
    \node[below,left] at (4,-0.75) {$\mathsf{in^\prime}$};
    \node[below,right] at (5.5,-0.75) {$\mathsf{out}$};
    \node at (4.5,0.5) {\Large$\mc{M}$};
    \node[below] at (3,-1.5) {$\mathsf{(a)}$};
    \end{scope}
    \begin{scope}[shift={(6.6,0)}]
    \fill[inner color=white, outer color=C4!5!white, draw=C4, dashed, thick, rounded corners] (0.9,-1.25) -- (0.9,1.25) -- (7.65,1.25) -- (7.65,0) -- (6.25,0) -- (6.25,-1.25) -- (4.25,-1.25) -- (4.25,0.25) -- (2.75,0.25) -- (2.75,-1.25) -- cycle;
    \node at (8.15,0.75) {\Large$\mc{M}$};
    \node[above] at (2.4,0.5) {$\mathsf{E}$};
    \node[above] at (2.4,-0.5) {$\mathsf{S}$};
    \draw[thick, -] (2,0.5) -- (7,0.5);
    \draw[thick, -] (2,-0.5) -- (7,-0.5);
    \fill[outer color=C1!10!white, inner color=white, draw=black, thick] (1.6,0) ellipse (0.5cm and 1cm);
    \node[left] at (2,0) {$\rho_\mathsf{SE}$};
    \shade[outer color=C3!50!white, inner color=white, draw=black, rounded corners, thick] (3,-1) rectangle (4,0);
    \node at (3.5,-0.5) {\large$\mc{A}$};
    \shade[outer color=C2!60!white, inner color=white, draw=black, rounded corners, thick] (4.5,-1) rectangle (6,1);
    \node at (5.25,0) {\LARGE$\mc{U}$};
    \node[right] at (7,-0.5) {$\rho_\mathsf{S}^\prime$};
    \node[right] at (7,0.5) {\trash};
    \node[below] at (4,-1.5) {$\mathsf{(b)}$};
    \end{scope}
    \end{tikzpicture}
    \caption[Circuit diagram of the superchannel]{\textbf{Circuit diagram of the superchannel:} $\mathsf{(a)}$ The superchannel $\mc{M}$ can be thought of as a two-legged box taking as input a \gls{CPTNI} map $\mc{A}:\mscr{B}(\mscr{H}_\mathsf{in})\to\mscr{B}(\mscr{H}_{\mathsf{in^\prime}})$ and yielding an output state $\rho^\prime_\mathsf{S}\in\$(\mscr{H}_\mathsf{out})$. $\mathsf{(b)}$ An initial state $\rho_\mathsf{SE}$ on the whole \gls{syst-env} system, in general correlated, is acted on with a preparation described by $\mc{A}$, then the whole evolves unitarily with the unitary channel $\mc{U}$, yielding the state $\mc{M}[\mc{A}]=\rho_\mathsf{S}^\prime$. The dashed box displays the contents of $\mc{M}$ explicitly as all dynamical content out of the control of the experimenter.}
    \label{fig: Superchannel}
\end{figure}

In general, any experimental intervention can be described by any \gls{CPTNI} maps $\mc{A}_i$. This means that for a preparation $\mc{A}$, the reduced output state of a system, \gls{syst}, after a given unitary evolution $\mc{U}$ can be described as an object $\mc{M}$ taking as input $\mc{A}$ and giving as an output a quantum state,
\begin{equation}
    \rho_\mathsf{S}^\prime=\tr_\mathsf{E}[\,\mc{U}\mc{A}(\rho_\mathsf{SE})]:=\mc{M}[\mc{A}],
    \label{eq: superchannel def 1}
\end{equation}
where implicitly we denote $\mc{A}$ for $\mc{A}\otimes\mc{I}_\mathsf{E}$, as we will throughout this thesis unless the distinction need to be made explicitly. This can be visualized as in Fig.~\ref{fig: Superchannel}. Keep in mind that if the map $\mc{A}$ is \gls{TNI}, then in general the output can be subnormalized, $\tr(\rho_\mathsf{S}^\prime)\leq1$, so despite us referring to it as a quantum state, we are allowing for it to have some missing information whenever the operation $\mc{A}$ might fail to be implemented.

The map $\mc{M}$ is known as the \emph{superchannel}, as it generalizes the idea of a quantum channel as a map from superoperators to quantum states. The superchannel contains all the dynamical information inaccessible to the experimenter. It can be said to be a supersuperoperator $\mc{M}:\mscr{B}_(\mscr{H}_\mathsf{in})\otimes\mscr{B}(\mscr{H}_\mathsf{in'})\to\mscr{B}(\mscr{H}_\mathsf{out})$. Linearity can be seen by inspection of Eq.~\eqref{eq: superchannel def 1}, as
\begin{equation}
    \mc{M}[\sum\alpha_i\mc{A}_i]=\sum\alpha_i\mc{M}[\mc{A}_i],
\end{equation}
and similarly trace preservation can be seen in the sense that $\tr[\mc{M}[\mc{A}]]=\tr[\mc{A}(\rho_\mathsf{S})]$, which corresponds to unity when $\mc{A}$ is \gls{TP}, or otherwise to a probability of successfully implementing $\mc{A}$.

Complete positivity, however, now means that if we have a channel acting together on an ancillary space of dimension $\ell$, i.e. $\Theta:\mscr{B}(\mscr{H}_\mathsf{in}\otimes\mscr{H}_\ell)\to\mscr{B}(\mscr{H}_\mathsf{in^\prime}\otimes\mscr{H}_{\ell})$, then the output corresponds to a positive map $\Theta^\prime:\mscr{B}(\mscr{H}_\ell)\to\mscr{B}(\mscr{H}_{\mathsf{out}}\otimes\mscr{H}_{\ell})$ for all dimensions $\ell$, with
\begin{equation}
    (\mc{M}\otimes\mc{I}_\ell)\Theta=\Theta^\prime\geq0,\qquad\forall\ell\geq0,
\end{equation}
where $\mc{M}$ acts only on the ``$\mathsf{in}$'' part of $\Theta$. This is more easily visualized as in Fig.~\ref{fig: CP superchannel}.

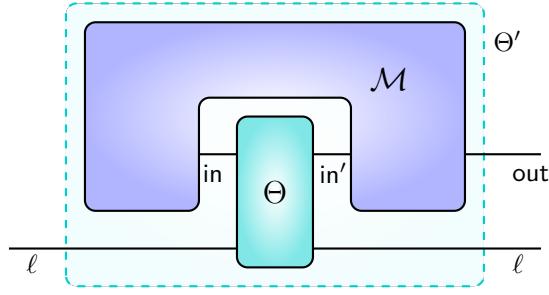
\begin{figure}[t]
    \centering
    \begin{tikzpicture}
    \begin{scope}
    \fill[outer color=C3!5!white, inner color=white, draw=C3, thick, dashed, rounded corners] (0.25,-2.25) -- (0.25,1.5) -- (5.75,1.5) -- (5.75,-2.25) -- (0.25,-2.25) -- cycle;
    \fill[outer color=C4!60!white, inner color=white, draw=black, rounded corners, thick] (0.5,-1.25) -- (0.5,1.25) -- (5.5,1.25) -- (5.5,-1.25) -- (4,-1.25) -- (4,0.25) -- (2,0.25) -- (2,-1.25) -- cycle;
    \draw[thick, -] (2,-0.5) -- (2.5,-0.5);
    \draw[thick, -] (3.5,-0.5) -- (4,-0.5);
    \draw[thick, -] (5.5,-0.5) -- (6.5,-0.5);
    \node[below,right] at (1.925,-0.75) {$\mathsf{in}$};
    \node[below,left] at (4.075,-0.75) {$\mathsf{in'}$};
    \node[below,right] at (6,-0.75) {$\mathsf{out}$};
    \node at (4.5,0.5) {\Large$\mc{M}$};
    \fill[outer color=C3!50!white, inner color=white, rounded corners, draw=black, thick] (2.5,-2) rectangle (3.5,0);
    \node at (3,-1) {\Large$\Theta$};
    \draw[thick, -] (-0.5,-1.75) -- (2.5,-1.75);
    \draw[thick, -] (3.5,-1.75) -- (6.5,-1.75);
    \node[below,left] at (0,-1.95) {$\ell$};
    \node[below,right] at (6,-1.95) {$\ell$};
    \node[below,right] at (5.75,1) {$\Theta^\prime$};
    \end{scope}
    \end{tikzpicture}
    \caption[Complete positivity of the superchannel]{\textbf{The \gls{CP} property for a superchannel} means $(\mc{M}\otimes\mc{I}_\ell)\Theta=\Theta^\prime\geq0$ for any \gls{CP} map $\Theta$ and any ancillary dimension $\ell\ge0$; here $\Theta^\prime$ can be thought of as a map from an input space of dimension $\ell$ to a product of such space together with the output space from the action of $\mc{M}$.}
    \label{fig: CP superchannel}
\end{figure}

We know that a channel being \gls{CP} implies a positive Choi state and vice versa. Now, as mentioned before, the superchannel is just a slightly more elaborate channel, so we may similarly obtain its Choi state. Remember that for a channel, the way to obtain the corresponding choi state is simply to act on the input space with half an unnormalized maximally entangled state $\tilde{\Psi}=\sum|ii\rangle\!\langle{jj}|\in\mscr{B}(\mscr{H}_\mathsf{in}\otimes\mscr{H}_\mathsf{in})$ and we obtained an unnormalized state in a larger space including the second copy of the input space. In this case, looking at Fig.~\ref{fig: Superchannel}$\mathsf{(a)}$, we see that we need to act on $\mscr{H}_\mathsf{in'}$ with half $\tilde{\Psi}\in\mscr{B}(\mscr{H}_\mathsf{in'}\otimes\mscr{H}_\mathsf{in'})$ and the remaining spaces should go out (almost) intact, giving an unnormalized state in a larger space composed of the input and output spaces.

To act with half the maximally entangled state we can simply swap half of it with the input spaces; let us label the new ancillary spaces by $\mathsf{A}$ and $\mathsf{B}$, which really just copies of \gls{syst} and have the same dimension, i.e. $d_\mathsf{A}=d_\mathsf{B}=d_\mathsf{S}$. Then
\begin{equation}
    \Upsilon_\mc{M}:=\tr_\mathsf{E}[\,\mc{U}\,\mc{S}\,(\rho_\mathsf{SE}\otimes\tilde{\Psi})],
\end{equation}
is the Choi state of the superchannel, where here $\mc{U}$ stands for $\mc{U}\otimes\mc{I}_\mathsf{AB}$, and
\begin{align}
    \mc{S}(\cdot)&:=\swap_\mathsf{SA}(\cdot)\swap_\mathsf{SA}\\
    \text{where here}\quad \swap_\mathsf{SA}&:=\mbb1_\mathsf{E}\sum_{i,j=1}^{d_\mathsf{S}}|ij\rangle\!\langle{ji}|\otimes\mbb1_\mathsf{B},
\end{align}
given that here the input and output space is simply the subsystem \gls{syst}. The swap of course can be done with either $\mathsf{A}$ or $\mathsf{B}$. This can also be seen clearly through a circuit diagram as in Fig.~\ref{fig: superchannel Choi state}.

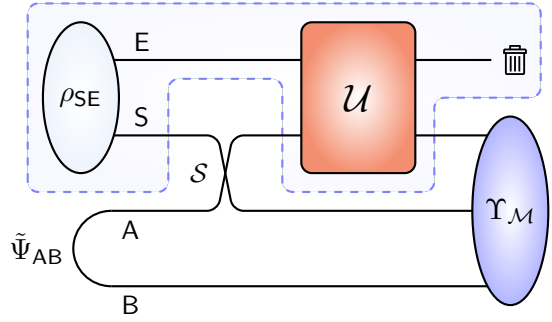
\begin{figure}[t]
    \centering
    \begin{tikzpicture}
    \fill[inner color=white, outer color=C4!5!white, draw=C4, dashed, thick, rounded corners] (0.9,-1.25) -- (0.9,1.25) -- (7.65,1.25) -- (7.65,0) -- (6.25,0) -- (6.25,-1.25) -- (4.25,-1.25) -- (4.25,0.25) -- (2.75,0.25) -- (2.75,-1.25) -- cycle;
    \node[above] at (2.4,0.5) {$\mathsf{E}$};
    \node[above] at (2.4,-0.5) {$\mathsf{S}$};
    \draw[thick, -] (2,0.5) -- (7,0.5);
    \draw[thick, -, rounded corners] (2,-0.5) -- (3.4,-0.5) -- (3.6,-1.5) --  (7,-1.5);
    \draw[thick, rounded corners] (2,-1.5) -- (3.4,-1.5) -- (3.6,-0.5) -- (7,-0.5);
    \node[left] at (3.4,-1) {$\mc{S}$};
    \draw[thick] (2,-1.5) arc (90:270:0.5);
    \draw[thick] (2,-2.5) -- (7,-2.5);
    \node[left] at (1.5,-2) {\large$\tilde{\Psi}_\mathsf{AB}$};
    \node[below] at (2.25,-1.5) {$\mathsf{A}$};
    \node[below] at (2.25,-2.5) {$\mathsf{B}$};
    \fill[outer color=C1!10!white, inner color=white, draw=black, thick] (1.6,0) ellipse (0.5cm and 1cm);
    \node[left] at (2,0) {\large$\rho_\mathsf{SE}$};
    \shade[outer color=C2!60!white, inner color=white, draw=black, rounded corners, thick] (4.5,-1) rectangle (6,1);
    \node at (5.25,0) {\LARGE$\mc{U}$};
    \node[right] at (7,0.5) {\trash};
    \fill[outer color=C4!60!white, inner color=white, draw=black, thick] (7.25,-1.5) ellipse (0.5cm and 1.25cm);
    \node at (7.25,-1.5) {\large$\Upsilon_\mc{M}$};
    \end{tikzpicture}
    \caption[The Choi state of the superchannel]{\textbf{The Choi state $\Upsilon_\mc{M}$ of the superchannel} can similarly be obtained by letting it act on half an unnormalized maximally entangled state $\tilde\Psi_\mathsf{AB}$; to do this, half of the ancillary space is swapped with the input on \gls{syst}. Here $\mathsf{A}$ and $\mathsf{B}$ are labels: both correspond to copies of \gls{syst}.}
    \label{fig: superchannel Choi state}
\end{figure}

Let us define $\FS_{\alpha\beta}:=\mbb1_\mathsf{E}\otimes|\alpha\rangle\!\langle\beta|$, so that we can expand the definition of the Choi state as
\begin{align}
    \Upsilon_\mc{M}=\sum_{\alpha,\ldots,\delta}\tr_\mathsf{E}[\,U\,\FS_{\alpha\beta}\rho_\mathsf{SE}\FS_{\delta\gamma}U^\dg]\otimes|\beta\alpha\rangle\!\langle\delta\gamma|,
    \label{eq: Choi superchannel expanded}
\end{align}
then to recover the action of the superchannel on a given map $\mc{A}$, we need to contract the Choi state of $\mc{M}$ with that of $\mc{A}$. First notice that we have to write $\Upsilon_\mc{A}=(\mbb1_\mathsf{A}\otimes\mc{A})\tilde\Psi$, i.e. $\mc{A}$ has to act on the $\mathsf{B}$ ancillary space (analogous to Eq.~\eqref{eq: map through Choi}, where we contract with the entangled half that was not acted on in $\Upsilon_\mc{M}$); this can also be understood easily through diagrams, as in Fig.~\ref{fig: Choi superchannel contraction}. This is, tracing over all inputs
\begin{align}
    \tr_\mathsf{in}[\Upsilon_\mc{M}(\mbb1_\mathsf{out}\otimes\Upsilon_\mc{A}^\mathrm{T})]&=\tr_\mathsf{in}\left\{\sum_{\alpha,\ldots,\delta}\tr_\mathsf{E}[U\FS_{\alpha\beta}\rho_\mathsf{SE}\FS_{\delta\gamma}U^\dg]\otimes|\beta\alpha\rangle\!\langle\delta\gamma|\Upsilon_\mc{A}^\mathrm{T}\right\}\nonumber\\
    &=\sum_{\alpha,\ldots,\delta,i,j}\tr_\mathsf{E}[U\FS_{\alpha\beta}\rho_\mathsf{SE}\FS_{\delta\gamma}U^\dg]\langle\delta|j\rangle\!\langle{i}|\beta\rangle\!\langle\gamma|\mc{A}^\mathrm{T}(|i\rangle\!\langle{j}|)|\alpha\rangle\nonumber\\
    &=\sum_{\alpha,\ldots,\delta}\tr_\mathsf{E}[U(\mbb1_\mathsf{E}\otimes|\alpha\rangle\!\langle\beta|)\rho_\mathsf{SE}(\mbb1_\mathsf{E}\otimes|\delta\rangle\!\langle\gamma|)U^\dg]\langle\alpha|\mc{A}(|\beta\rangle\!\langle\delta|)|\gamma\rangle\nonumber\\
    &=\sum_{\alpha,\ldots,\delta}\sum_\mu\tr_\mathsf{E}[U(\mbb1_\mathsf{E}\otimes|\alpha\rangle\!\langle\beta|)\rho_\mathsf{SE}(\mbb1_\mathsf{E}\otimes|\delta\rangle\!\langle\gamma|)U^\dg]\langle\alpha|A_\mu|\beta\rangle\!\langle\delta|A_\mu^\dg|\gamma\rangle\nonumber\\
    &=\sum_\mu\tr_\mathsf{E}[U(\mbb1_\mathsf{E}\otimes{A}_\mu)\rho_\mathsf{SE}(\mbb1_\mathsf{E}\otimes{A}_\mu^\dg)U^\dg]\nonumber\\
    &=\tr_\mathsf{E}[\,\mc{U}\,\mc{A}\,(\rho_\mathsf{SE})],
\end{align}
where in the fifth line we used the fact that $\mc{A}$ is \gls{CP} so that $\mc{A}(\cdot)=\sum{A}_\mu(\cdot)A_\mu^\dg$.

The Choi state is manifestly positive, given that everything in its definition is positive; this in turn implies that $\mc{M}$ is \gls{CP}, as expected. We may now think of this operationally as a situation in which the experimenter brings along an ancillary system of dimension $\ell$ and performs an entangling operation $\Theta_{\mathsf{S}\ell}$, with the subsequent \gls{syst-env} dynamics occurring; then we are assured that the output $\Theta^\prime_{\mathsf{S}\ell}$ will be positive. The Kraus operators of $\mc{M}$ can be obtained as well, as done in Ref.~\cite{Modi_2011, Milz_operational}. This shows conclusively that the dynamics of initially correlated systems are indeed completely positive, we were only looking in the wrong place. Notice as well that the \gls{CPTP} property for the superchannel was a natural consequence rather than an a-priori condition.

\begin{figure}[t]
    \centering
    \begin{tikzpicture}
    \begin{scope}
    \fill[outer color=C4!60!white, inner color=white, draw=black, rounded corners, thick] (0.5,-1.25) -- (0.5,1.25) -- (5.5,1.25) -- (5.5,-1.25) -- (4,-1.25) -- (4,0.25) -- (2,0.25) -- (2,-1.25) -- cycle;
    \fill[outer color=C3!5!white, inner color=white, draw=C3, dashed, rounded corners, thick] (4.85,-1.35) -- (7,-1.35) -- (7,-3.15) -- (4.85,-3.15) -- cycle;
    \draw[thick, -, rounded corners] (2,-0.5) -- (2.9,-0.5) -- (3.1,-1.5) -- (6,-1.5);
    \draw[thick, -, rounded corners] (0.75,-1.5) -- (2.9,-1.5) -- (3.1,-0.5) -- (4,-0.5);
    \draw[thick] (0.75,-1.5) arc (90:270:0.5);
    \node[left] at (0.25,-2) {$\tilde{\Psi}$};
    \draw[thick] (0.75,-2.5) -- (5,-2.5);
    \fill[outer color=C3!50!white, inner color=white, draw=black, rounded corners, thick] (5,-2) rectangle (6,-3);
    \node at (5.5,-2.5) {\large$\mc{A}^\mathrm{T}$};
    \draw[thick] (6,-1.5) arc (90:-90:0.5);
    \node[right] at (6.5,-2) {$\tilde{\Psi}$};
    \draw[thick, -] (5.5,-0.5) -- (6,-0.5);
    \node[below,right] at (2,-0.75) {$\mathsf{in}$};
    \node[below,left] at (4,-0.75) {$\mathsf{in^\prime}$};
    \node[below,right] at (5.5,-0.75) {$\mathsf{out}$};
    \node at (4.5,0.5) {\Large$\mc{M}$};
    \end{scope}
    \begin{scope}[shift={(8,-0.5)}]
    \node at (-0.6,-0.25) {\Huge{\textbf{=}}};
    \fill[outer color=C4!60!white, inner color=white, draw=black, rounded corners, thick] (0.5,-1.25) -- (0.5,1.25) -- (5.5,1.25) -- (5.5,-1.25) -- (4,-1.25) -- (4,0.25) -- (2,0.25) -- (2,-1.25) -- cycle;
    \draw[thick, -] (2,-0.5) -- (2.5,-0.5);
    \draw[thick, -] (3.5,-0.5) -- (4,-0.5);
    \draw[thick, -] (5.5,-0.5) -- (6,-0.5);
    \node[below,right] at (2,-0.75) {$\mathsf{in}$};
    \node[below,left] at (4.05,-0.75) {$\mathsf{in'}$};
    \node[below,right] at (5.5,-0.75) {$\mathsf{out}$};
    \node at (4.5,0.5) {\Large$\mc{M}$};
     \fill[outer color=C3!50!white, inner color=white, draw=black, rounded corners, thick] (2.5,-1) rectangle (3.5,0);
    \node at (3,-0.5) {\large$\mc{A}$};
    \end{scope}
    \end{tikzpicture}
    \caption[The action of the superchannel in terms of Choi states]{\textbf{The action of a superchannel} $\mc{M}$ on a \gls{CP} map $\mc{A}$ can be given as the contraction of the respective Choi states as $\mc{M}[\mc{A}]=\tr_\mathsf{in^{(\prime)}}[\Upsilon_\mc{M}(\mbb1_\mathsf{out}\otimes\Upsilon_\mc{A}^\mathrm{T})]$, where $\Upsilon_\mc{A}^\mathrm{T}=(\mbb1_\mathsf{in}\otimes\mc{A}^\mathrm{T})\tilde\Psi$.}
    \label{fig: Choi superchannel contraction}
\end{figure}

Now, the superchannel, just as in the case of quantum channels, can be reconstructed experimentally through \texttt{QPT} without any a-priori knowledge about the initial state or the \gls{syst-env} dynamics. In a similar fashion to standard \texttt{QPT}, now the experimenter requires $d_\mathsf{S}^4$ linearly independent \gls{CP} maps $\{\mc{A}_i\}$ and measuring the corresponding outputs $\rho_\mathsf{S}^{(i)}$, with which e.g. the Choi state can be reconstructed. This procedure, just as any \texttt{QPT} or \texttt{QST} procedure requires a relatively great amount of resources, however it was reported already in Ref.~\cite{Ringbauer_2015}.

Going back to the discussion of \gls{NCP} dynamical maps, notice that the case $\mc{M}[\mc{I}]$ gives the usual dynamical map scenario, which however is operationally void without determining the superchannel. Notice that the initial uncorrelated state case with a dynamical map is contained as a particular one under the superchannel, i.e. for $\rho_{\mathsf{SE}}=\rho_\mathsf{S}\otimes\varepsilon$, we have
\begin{align}
    \mc{M}[\mc{A}]=\tr_\mathsf{E}[\,\mc{U}_t\,(\mc{A}(\rho_\mathsf{S})\otimes\varepsilon)]=\mscr{Z}_t[\sigma_\mathsf{S}],
\end{align}
where $\sigma_\mathsf{S}=\mc{A}(\rho_\mathsf{S})$. In such case the dynamical map description is sufficient, however, in general the superchannel can be employed to quantify initial correlations e.g. by letting $\rho_\mathsf{SE}=\rho_\mathsf{S}\otimes\varepsilon+\chi_\mathsf{SE}$~\cite{Modi_2011}. As we will see below, this also turns out to be a special case of the general notion of a Markovian process.

\section{Multiple time-steps}\label{sec: process tensor}
One of the crucial features of the superchannel is that it allows to diagnose the presence of initial correlations playing a role in the reduced subsystem dynamics. This is already some form of \emph{memory} or non-Markovian effect. This, however, only accounts for correlations between the initial preparation and the final measurement; if we were to fully describe non-Markovianity, we would look for an extension to any dependence on any given time in the past, analogous to the classical definition for non-Markovianity.

This is crucial because if we are to deal consistently with non-Markovianity, we first require a non-ambiguous characterization that genuinely deals with temporal correlations across an arbitrary number of points. Far from being just an interesting theoretical aspect, this is a tremendously practical question, given the ever-increasing technological capability and interest to deal with noisy systems over several time-steps~\cite{Blok2014, Shrapnel_2018,Giarmatzi2018, Winick2019, Morris2019, Taranto_2019, Taranto_2019_2, Guo2020, shapiro2012quantum, Altafini}.

Let us then first revisit the classical notion of a stochastic process and of Markovianity.

\subsection{Classical stochastic processes}\label{sec: classical stochastic processes}
A classical stochastic process is a probabilistic process in time, i.e. it can be described mathematically as a collection of random variables indexed by time. As these are mainly used to model situations in which only partial knowledge or certainty is available, it is no surprise that classical stochastic processes are ubiquitous in our attempt to describe anything from coin tossing to financial markets or disease spread. We have already appealed to some of the notions needed to discuss stochastic processes, however, we can formalize these a bit more to make these concepts concrete. Most of these concepts can be consulted in full detail in standard textbooks such as~\cite{feller1968} or in Ref.~\cite{breuer2002theory} related to open systems or furthermore in particular in the context related to this thesis in Ref.~\cite{milz_2019,Taranto_2020}.

It is worth mentioning that some of the notation we use here is only for local purposes i.e. it does not necessarily have any relation with the symbols on previous sections unless explicitly stated.

First we need to introduce a \emph{probability space}, which sets the scene to model probabilistically a given class of situations: this is a triple $(\Omega,\Sigma,\mu)$ where
\begin{compactenum}[\itshape i.]
    \item $\Omega$ is called a \emph{sample space}; mathematically, this is an arbitrary nonempty set. It stands for the set of all (discrete or continuous) outcomes, i.e. single realisations of the model.
    \item $\Sigma$ is called an \emph{event space} and is a collection of subsets of the sample space called events. The event space $\Sigma$ is a $\sigma$-algebra: this means that, denoting by $\pi(\Omega)$ the power set (set of all subsets) of the sample space, $\Sigma\subseteq\pi(\Omega)$ satisfies\footnote{ We use standard set notation with $\setminus$ meaning set subtraction, $\cup$ union of sets and $\cap$ intersection of sets.}
    \begin{compactitem}[--]
        \item $\Omega\in\Sigma$.
        \item If $\alpha\in\Sigma$, then also $(\Omega\setminus\alpha)\in\Sigma$.
        \item If $\alpha_i\in\Sigma$ for $i=1,2,\ldots$, then also $\cup_i\alpha_i\in\Sigma$.
        \item If $\alpha_i\in\Sigma$ for $i=1,2,\ldots$, then also $\cap_i\alpha_i\in\Sigma$.
    \end{compactitem}
    Together, the pair $(\Omega,\Sigma)$ is called a \emph{measurable space}.

    \item $\mu$ is called a \emph{probability measure}. It is a function $\mu:\Sigma\to[0,1]$ such that
    \begin{compactitem}[--]
    \item $\mu(\Omega)=1$
    \item $\mu(\cup_i\alpha_i)=\sum_i\mu(\alpha_i)$ for any countable collection of pairwise disjoint\footnote{ That is, $\alpha_i\cap\alpha_j=\emptyset,\forall i\neq{j}$.} subsets $\{\alpha_i\}\in\Sigma$.
    \end{compactitem}
\end{compactenum}

\begin{example}Perhaps the simplest example is a coin toss: if a fair coin is tossed three times, then $\Omega=\{\texttt{HHH,HHT,HTT,HTH,THT,THH,TTH,TTT}\}$ is the sample space, where \texttt{H} and \texttt{T} stand for `heads' and `tails', respectively, and the order in each element stands for the first, second and third toss. The power set $\pi(\Omega)$ has $2^{|\Omega|}=2^8$ events. If the experiment concerns the event ``at least two heads occur'', i.e. $\omega=\{\texttt{HHH,HHT,HTH,THH}\}$, then the event space is $\Sigma=\{\emptyset,\varkappa,\omega,\Omega\}$, where $\varkappa=\Omega\setminus\omega$, so the elements correspond to the experiment not being performed, not more than one head occurred, at least two head occurred and the experiment being performed. The probability measure gives $\mu(\emptyset)=0$, $\mu(\omega)=\mu(\varkappa)=0.5$, $\mu(\Omega)=1$.
\end{example}

Now, we may define a \emph{random variable} $X$ on the probability space $(\Omega,\Sigma,\mu)$ as a function $X:\Omega\to\mbb{R}$; in general these can correspond to functions from the sample space to any other measurable space, e.g. in Section~\ref{sec: Random states and Haar} we considered random unitary matrices as random variables. Intuitively, random variables are real-valued quantities that can be measured from outcomes of random trials. This point of view is useful simply because an experimenter often cares about some function of the outcomes of the experiment rather than the outcomes themselves. This allows to ask about the probability that the random variable takes a value within a subset $s\in\mathtt{S}$, where $\mathtt{S}\subseteq\mbb{R}$ is a collection of open subsets\footnote{ Formally these have to be so-called Borel subsets~\cite{feller1968}.} of the real numbers, i.e. in a slight abuse of notation we may define
\begin{equation}
    \mbb{P}[X\in{s}]:=\mbb{P}[\{\omega\in\Omega:X(\omega)\in{s}\}],
\end{equation}
where $\mbb{P}$ is a probability distribution $\mbb{P}:\mathtt{S}\to[0,1]$, related to the probability measure $\mu$ via $\mbb{P}[s]=\mu[X^{-1}(s)]$, with $X^{-1}(s)\in\Sigma$. This is so because $\mathtt{S}$ is a $\sigma$-algebra and $(\mbb{R},\mathtt{S})$ a measurable space. We notice that given a function $f:\mbb{R}\to\mbb{R}$ such that $f^{-1}(s)\in\mathtt{S}$, then any other $Y$ defined by $Y(\omega)=f(X(\omega))$ for an event $\omega$ is also a random variable. In general we will refer to the value $X(\omega)$ as a \emph{realisation}.

A \emph{stochastic process} then concerns $\mathtt{S}$-valued random variables $X$ in a parameter $t\geq0$, usually standing for time. That is, a stochastic process is a random variable $X:\Omega\times\mbb{R}^+_0\to\mbb{R}$, where $\mbb{R}^+_0$ stands for the set of non-negative real numbers. We can then think of a stochastic process in two ways. For a fixed event $\omega\in\Omega$, the function $X:\mbb{R}^+_0\to\mbb{R}$ defines a so-called trajectory of the stochastic process. Conversely, $X:\Omega\to\mbb{R}$ for a fixed $t\geq0$ denotes a collection of random variables at a given time.

In practice we usually care about the probability distribution of a given trajectory and a sensible way to describe it is through a discrete number of $k$ times
\begin{equation}
    \mscr{T}_k=\{t_0,t_1,,\ldots,t_{k-1}\},
    \label{eq: set of times}
\end{equation}
and whenever we explicitly assume that times are ordered as $t_0<t_1<\ldots<t_n$, we will denote this by $n:0$. We can then construct a vector of random variables $\mathbf{X}_{\mscr{T}_k}=(X_0(t_0),X_1(t_1),\ldots,X_{k-1}(t_{k-1}))$ for each time in $\mscr{T}_k$ with each $X_i(\cdot,t_i):\Omega\to\mbb{R}$. This then defines the joint probability distribution
\begin{equation}
    \mbb{P}(\vec{x}_{_{\mscr{T}_k}})=\mu\left[\mathbf{X}_{\mscr{T}_k}^{-1}(\vec{x}_{_{\mscr{T}_k}})\right],
    \label{eq: classical joint probability}
\end{equation}
for realisations $X_i(\omega,t_i)$ of an event $\omega\in\Omega$ to lie within $x_i$ at time $t_i$, where we also defined $\vec{x}_{_{\mscr{T}_k}}=(x_0,x_1,\ldots,x_{k-1})$, i.e. the expression in Eq.~\eqref{eq: classical joint probability} explicitly means, for example, for a two time-step process,
\begin{equation}
    \mbb{P}(x_0,x_1)=\mbb{P}[\{\omega\in\Omega:X_0(\omega,t_0)\in{x}_0\,\,\text{and}\,\,{X}_1(\omega,t_1)\in{x}_1\}].
\end{equation}

This is enough to statistically characterize the stochastic process with finite time-steps as the probability for a given trajectory to fall within sets $x_i$ at each time $t_i$, together with $k$-point correlations between the random variables at each time-step.

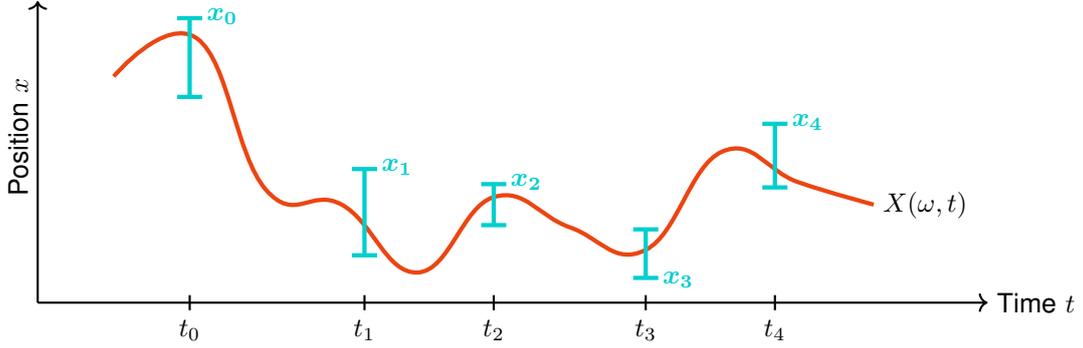
\begin{figure}[t]
    \centering
    \begin{tikzpicture}
    \draw[thick,->] (-1,-1) -- (-1,3);
    \draw[thick, ->] (-1,-1) -- (11.5,-1);
    \draw [ultra thick, C2] plot [smooth, tension=0.75] coordinates { (0,2) (1.1,2.5) (2,0.5) (3,0.3) (4,-0.6) (5,0.4) (6,0) (7,-0.3) (8,1) (9,0.6) (10,0.3)};
    \node[rotate=90, above, shift={(0.2,0)}] at (-1,1) {\textsf{Position} $x$};
    \node[right] at (10,0.3) {$X(\omega,t)$};
    \node[right] at (11.5,-1) {\textsf{Time} $t$};
    \draw[thick] (1,-1.1) -- (1,-0.9);
    \draw[|-|, ultra thick, C3] (1,2.8) -- (1,1.7);
    \node[C3, right] at (1.1,2.8) {$\boldsymbol{x_0}$};
    \node[below] at (1,-1.1) {$t_0$};
    \draw[thick] (3.3,-1.1) -- (3.3,-0.9);
    \draw[|-|, ultra thick, C3] (3.3,-0.4) -- (3.3,0.8);
    \node[C3, below, right] at (3.4,0.8) {$\boldsymbol{x_1}$};
    \node[below] at (3.3,-1.1) {$t_1$};
    \draw[thick] (5,-1.1) -- (5,-0.9);
    \draw[|-|, ultra thick, C3] (5,0) -- (5,0.6);
    \node[C3, right] at (5.1,0.6) {$\boldsymbol{x_2}$};
    \node[below] at (5,-1.1) {$t_2$};
    \draw[thick] (7,-1.1) -- (7,-0.9);
    \draw[|-|, ultra thick, C3] (7,-0.7) -- (7,0);
    \node[C3, right] at (7.1,-0.7) {$\boldsymbol{x_3}$};
    \node[below] at (7,-1.1) {$t_3$};
    \draw[thick] (8.7,-1.1) -- (8.7,-0.9);
    \draw[|-|, ultra thick, C3] (8.7,0.5) -- (8.7,1.4);
    \node[C3, right] at (8.8,1.4) {$\boldsymbol{x_4}$};
    \node[below] at (8.7,-1.1) {$t_4$};
    \end{tikzpicture}
    \caption[Trajectory in classical stochastic processes]{\textbf{An example of a stochastic trajectory} for a particle in a one-dimensional random walk: for given times $\mscr{T}_5=\{t_0,t_1,\ldots,t_4\}$, a classical stochastic process is determined by the joint probability distribution $\mbb{P}_{\mscr{T}_5}(x_0,x_1,\ldots,x_4)$ to find the particle in the region $x_i$ at time $t_i$.}
    \label{fig: Joint probability distribution}
\end{figure}

The choice of $\mscr{T}_k$ is arbitrary, so this already tell us that if obtain realisations for some but not all such times, say $\mscr{T}_\ell$, with $|\mscr{T}_\ell|<|\mscr{T}_k|$, then the joint probability distribution $\mbb{P}(\vec{x}_{_{\mscr{T}_\ell}})$ should be contained in $\mbb{P}_{\mscr{T}_k}(\vec{x}_{_{\mscr{T}_k}})$. Indeed, we can marginalise, i.e. sum out the probabilities for the extra random variables in the larger set of times to recover the joint distribution in the smaller one,
\begin{equation}
    \mbb{P}(\vec{x}_{_{\mscr{T}_\ell}})=\sum_{\mscr{T}_k\setminus\mscr{T}_\ell}\mbb{P}(\vec{x}_{_{\mscr{T}_k}}),
\end{equation}
where the sum runs over realisations for the excessive times in $\mscr{T}_k$ but not in $\mscr{T}_\ell$.

This containment property generalizes to the fact that there exists an infinite joint probability distribution that contains all the finite ones: this is reconciled through the so-called Kolmogorov extension theorem~\cite{breuer2002theory,tao2011, kolmogorov2018}, which binds the realistic scenario of a discrete number of observations with the firm mathematical footing of a general stochastic process. More specifically, it gives the consistency conditions for a family of joint probabilities to guarantee the existence of an underlying continuous stochastic process.

\subsection{The classical Markov condition}
Within any classical stochastic process we can consider different particular cases for the joint distribution in Eq.~\eqref{fig: Joint probability distribution}. In particular, a Markov process as introduced for open quantum systems as being \emph{memoryless} implies that, relative to a given time-step, the probabilities for future steps should only depend on such time-step\footnote{ Strictly speaking, outright \emph{memoryless} should refer to no dependence even in the present state.} and not on the past ones.

\begin{definition}[Classical Markov condition]
\label{classical Markov}
A classical stochastic process is called Markovian if for any ordered times $t_0<t_1<\cdots<t_k$ we have
\begin{equation}
    \mbb{P}(x_k|\vec{x}_{k-1:0})=\mbb{P}(x_k|x_{k-1}),
    \label{eq: classical Markov}
\end{equation}
where here $\vec{x}_{j:0}=(x_0,\ldots,x_j)$ and $\mbb{P}(A|B)$ refers to the \emph{conditional probability} of event $A$ given event $B$, defined by $\mbb{P}(A|B)=\mbb{P}(A,B)/\mbb{P}(B)$, i.e.
\begin{equation}
    \f{\mbb{P}(\vec{x}_{k:0})}{\mbb{P}(\vec{x}_{k-1:0})}=\mbb{P}(x_k|\vec{x}_{k-1:0})\,\,\stackrel{\textsf{Markov}}{=}\,\,\mbb{P}(x_k|x_{k-1})=\f{\mbb{P}(x_k,x_{k-1})}{\mbb{P}(x_{k-1})}.
    \label{eq: def conditional probability}
\end{equation}
Conversely, any classical process that does not satisfy Eq.~\eqref{eq: classical Markov} is called non-Markovian.
\end{definition}

Physically, a conditional probability $\mbb{P}(x_j|x_i)$ can be understood as a probability for a system to go from a state $x_i$ to state $x_j$ and is often referred to as a transition probability or a propagator. The Markov condition is a very significant simplification, since normally we would need all $(k+1)$-point correlations to compute the following joint probability distribution. Notice that in general, by repeatedly applying the definition of the conditional probability,
\begin{equation}
    \mbb{P}(\vec{x}_{k:0})=\mbb{P}(x_k|\vec{x}_{k-1:0})\mbb{P}(x_{k-1}|\vec{x}_{k-2:0})\cdots\mbb{P}(x_1|x_0)\mbb{P}(x_0),
\end{equation}
so that to determine the full joint probability distribution we require an increasing number of transition probabilities. However, if the the Markov condition holds, a Markovian joint distribution satisfies
\begin{equation}
    \mbb{P}(\vec{x}_{k:0})=\mbb{P}(x_k|x_{k-1})\mbb{P}(x_{k-1}|x_{k-2})\cdots\mbb{P}(x_1|x_0)\mbb{P}(x_0),
\end{equation}
requiring only 2-point correlations to determine the future joint distribution.

Specifically, suppose we have a classical system with $n$ possible outcomes (as in the case of a coin flip, $n=2$, for example). Then to determine the distribution $\mbb{P}(x_0)$ we require $n-1$ probabilities, to determine $\mbb{P}(x_1|x_0)\mbb{P}(x_0)$ we require $n-1+n(n-1)=n^2-1$ probabilities, and so on until we reach $n^{k+1}-1$ to determine $\mbb{P}(\vec{x}_{k:0})$, i.e. we get an exponential growth in the number of time-steps $k$. However, if the process is Markovian, we require $(n-1)(1+kn)$ probabilities, so just a linear growing number of terms in the number of time-steps $k$.

Notice that the Markov condition is a statement regarding multiple time-steps, or equivalently a statement regarding the observation of multiple events. That is, to have a general notion of Markovianity or memorylessness we must consider correlations across several events in time.

Now, continuing with the case of $n$ possible outcomes, notice that we can store the initial probability distribution $\mbb{P}(x_0)$ as an $n$-dimensional vector $\mbb{P}_0$ with entries corresponding to the probability of each outcome so that
\begin{equation}
    \mbb{P}_k=\mathfrak{P}_{k:0}\mbb{P}_0,
\end{equation}
is the vector corresponding to $\mbb{P}(x_k)$, where $\mathfrak{P}_{k:0}$ is an $n\times{n}$ matrix containing the propagator for each outcome from time-step $t_0$ to time-step $t_k$. Explicitly, this means a propagator matrix $\mathfrak{P}_{j:i}$ from time-step $t_i$ to step $t_j>t_i$, with random variables taking given values $X_i=u_a$ and $X_j=v_b$, where $1\leq{a,b}\leq{n}$, will have the probability entries
\begin{equation}
    \left(\mathfrak{P}_{j:i}\right)_{ab}=\mbb{P}[X_i=u_a|X_j=v_b],
\end{equation}
which are such that
\begin{equation}
    \sum_{b=1}^n\left(\mathfrak{P}_{j:i}\right)_{ab}=1,
\end{equation}
and are called \emph{stochastic matrices}. This property can be thought as justifying the label of these matrices as propagators, as they would give an evolution of a probability vector $\mbb{P}_i$ at time $t_i$ to a legitimate probability vector $\mbb{P}_j$ at time $t_j$.

A Markovian process then implies that the full propagator can be broken up in the intermediate steps as
\begin{equation}
    \mathfrak{P}_{k:0}=\mathfrak{P}_{k:k-1}\mathfrak{P}_{k-1:k-2}\cdots\mathfrak{P}_{1:0},
\end{equation}
and we can further marginalise to break it up in any two propagator matrices as
\begin{equation}
    \mathfrak{P}_{k:0}=\mathfrak{P}_{k:j}\mathfrak{P}_{j:0},\qquad\forall{t_k}>t_j>t_0,
    \label{eq: Classical divisibility}
\end{equation}
which is known as a divisibility property, and which is the motivation for the analogous property on quantum dynamical maps to define Markovianity as in Eq.~\eqref{eq: divisibility dynamical maps}.

The divisibility property on two contiguous time-steps can equivalently be expressed as
\begin{equation}
    \mbb{P}(x_{\ell+1}|x_{\ell-1})=\sum_{t_\ell}\mbb{P}(x_{\ell+1}|x_\ell)\mbb{P}(x_\ell|x_{\ell-1}),
    \label{eq: Chapman-Kolmogorov}
\end{equation}
where the time-steps are ordered $t_{\ell-1}<t_\ell<t_{\ell+1}$ for any $\ell>0$, which usually in this form is known as the Chapman-Kolmogorov equation~\cite{Breuer_2016}. The intuitive explanation of this equation is that the probability of transitioning from a state $x_{\ell-1}$ to another $x_{\ell+1}$ can be obtained by multiplying the transition probabilities to and from an intermediate state $x_\ell$ and summing over all the possible of these.

Both divisibility in Eq.~\eqref{eq: Classical divisibility} and the Chapman-Kolmogorov equation in Eq.~\eqref{eq: Chapman-Kolmogorov}, however, do not say anything about higher point correlations, i.e. divisibility is insufficient to characterize Markovianity and it is possible to find non-Markovian processes which can satisfy divisibility~\cite{Hanggi1977,Hanggi1978,Vacchini_2011}, where however, the propagator would not correspond to a conditional transition probability matrix~\cite{Hanggi_1982}.

Clearly the Markov condition is a massive mathematical simplification: because of this it has been largely studied and has a wide reach and applicability~\cite{Dynkin_2006, stroock2013}. If we turn to physics, however, not only in the more general quantum case is it an idealization but also in the classical case as soon as we start considering fully realistic scenarios we realise that Markovian processes are rather exceptional~\cite{VanKampen_1998}.

A simple example contrasting both kinds of processes is that of drawing a marble from a bag full of marbles with a handful of colors: if we were to replace or put back each marble as we draw it from the bag, the process will be Markovian and we do not need to remember anything to update our predictions of what color the next marble will have; if, however, we discard each marble once we take it out, by knowing the full amount of marbles we have to keep track of all the colors as they come out to update our predictions accordingly. Other examples for realistic scenarios can be seen in Ref.~\cite{VanKampen_1998}, in particular for cases where Markovianity becomes a reasonable approximation within certain limits in which the initial configuration is forgotten.

In a sense, with the previous Chapter~\ref{sec:statmech} we have seen that equilibration on average is indeed a process where the initial state is effectively forgotten. Let us thus go back to approaching the case of non-Markovianity and stochastic processes in quantum mechanics.

\subsection{The process tensor}\label{sec: the process tensor}
To consider the generalization of a stochastic process in open quantum systems we first need to consider multiple interventions. Given the operational scenario of the superchannel, we see that what we need is to extend this picture to not only a preparation and a measurement but to an arbitrary number of interventions. Crucially, notice that in the discussion of classical stochastic processes, there was an implicit assumption that realisations do not alter in any way the subsequent states, and thus such assumption will not hold in the quantum extension anymore.

The physical scenario is now the following: an experimenter prepares a fiducial quantum state $\rho_\mathsf{SE}$ on system \gls{syst} of a joint \gls{syst-env} composite through an operation $\mc{A}_0$, which in general is an \gls{CPTNI} map; subsequently the whole composite evolves unitarily through a unitary map $\mc{U}_1$, after which an operation $\mc{A}_1$ is performed on \gls{syst}, then the whole evolves unitarily under a unitary map $\mc{U}_2$, and so on, until an intervention $\mc{A}_{k-1}$, followed finally by a unitary map $\mc{U}_k$. This means the final state in system \gls{syst} will be given by
\begin{equation}
    \rho_\mathsf{S}^{(k)}=\tr_\mathsf{E}\left[\,\mc{U}_k\,\mc{A}_{k-1}\,\mc{U}_{k-1}\,\mc{A}_{k-1}\cdots\mc{U}_1\mc{A}_0\,(\rho_\mathsf{SE})\right],
    \label{eq: process tensor open}
\end{equation}
where here we also implicitly write $\mc{A}_\ell$ for $\mc{A}_\ell\otimes\mc{I}_\mathsf{E}$. As for the case of the superchannel, keep in mind that $\rho_\mathsf{S}^{(k)}$ can be subnormalized (even if we refer to it as a quantum state) if the operations are in general \gls{TNI} but not necessarily \gls{TP}.

Now we see that the generalization of the superchannel to $k$ time-steps must be a map $\mc{T}_{k:0}:\mscr{B}(\mscr{H}_\mathsf{S})^{\otimes{2k}}\to\mscr{B}(\mscr{H}_\mathsf{S})$, taking $k$ \gls{CPTNI} maps as arguments and giving a quantum state as output at time-step $k$, i.e.
\begin{equation}
    \mc{T}_{k:0}[\vec{\mc{A}}_{k-1:0}]=\rho_\mathsf{S}^{(k)},
\end{equation}
where $\vec{\mc{A}}_{k-1:0}=(\mc{A}_0,\mc{A}_1,\ldots,\mc{A}_{k-1})$. This can depicted diagramatically as in Fig.~\ref{fig: Process tensor}.

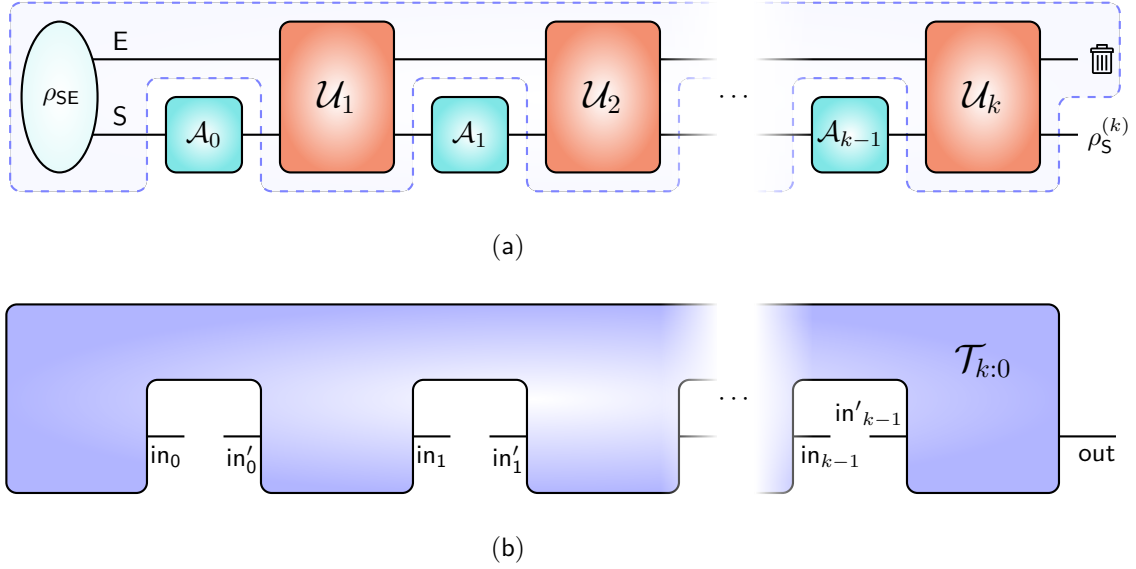
\begin{figure}[t]
    \centering
    \begin{tikzpicture}
    \begin{scope}
    \fill[inner color=white, outer color=C4!5!white, draw=C4, dashed, thick, rounded corners] (0.95,-1.25) -- (0.95,1.25) -- (15.55,1.25) -- (15.55,0) -- (14.75,0) -- (14.75,-1.25) -- (12.75,-1.25) -- (12.75,0.25) -- (11.25,0.25) -- (11.25,-1.25) -- (10.5,-1.25) -- (10.5,0.25) -- (9.75,0.25) -- (9.75,-1.25) -- (7.75,-1.25) -- (7.75,0.25) -- (6.25,0.25) -- (6.25,-1.25) -- (4.25,-1.25) -- (4.25,0.25) -- (2.75,0.25) -- (2.75,-1.25) -- cycle;
    \node[above] at (2.4,0.5) {$\mathsf{E}$};
    \node[above] at (2.4,-0.5) {$\mathsf{S}$};
    \draw[thick, -] (2,0.5) -- (15,0.5);
    \draw[thick, -] (2,-0.5) -- (15,-0.5);
    \fill[outer color=C3!10!white, inner color=white, draw=black, thick] (1.6,0) ellipse (0.5cm and 1cm);
    \node[left] at (2,0) {$\rho_\mathsf{SE}$};
    \shade[outer color=C3!50!white, inner color=white, draw=black, rounded corners, thick] (3,-1) rectangle (4,0);
    \node at (3.5,-0.5) {\large$\mc{A}_0$};
    \shade[outer color=C2!60!white, inner color=white, draw=black, rounded corners, thick] (4.5,-1) rectangle (6,1);
    \node at (5.25,0) {\LARGE$\mc{U}_1$};
    \shade[outer color=C3!50!white, inner color=white, draw=black, rounded corners, thick] (6.5,-1) rectangle (7.5,0);
    \node at (7,-0.5) {\large$\mc{A}_1$};
    \shade[outer color=C2!60!white, inner color=white, draw=black, rounded corners, thick] (8,-1) rectangle (9.5,1);
    \node at (8.75,0) {\LARGE$\mc{U}_2$};
    \draw[white, fill=white, path fading= west] (9.5,-1.5) -- (9.5,1.5) -- (10.25,1.5) -- (10.25,-1.5) ;
    \draw[white, fill=white] (10.25,1.5) -- (10.75,1.5) -- (10.75,-1.5) -- (10.25,-1.5);
    \draw[white, fill=white, path fading= east] (10.75,1.5) -- (11.5,1.5) -- (11.5,-1.5) -- (10.75,-1.5);
    \node at (10.5,0) {$\cdots$};
    \shade[outer color=C3!50!white, inner color=white, draw=black, rounded corners, thick] (11.5,-1) rectangle (12.5,0);
    \node at (12,-0.5) {\large$\mc{A}_{k-1}$};
    \shade[outer color=C2!60!white, inner color=white, draw=black, rounded corners, thick] (13,-1) rectangle (14.5,1);
    \node at (13.75,0) {\LARGE$\mc{U}_k$};
    \node[right] at (15,-0.5) {$\rho_\mathsf{S}^{(k)}$};
    \node[right] at (15,0.5) {\trash};
    \node at (7.5,-2) {$\mathsf{(a)}$};
    \end{scope}
    \begin{scope}[shift = {(0,-4)}]
    \draw[thick, -] (2,-0.5) -- (15.5,-0.5);
    \fill[inner color=white, outer color=C4!60!white, draw=black, thick, rounded corners] (0.9,-1.25) -- (0.9,1.25) -- (14.75,1.25) -- (14.75,-1.25) -- (12.75,-1.25) -- (12.75,0.25) -- (11.25,0.25) -- (11.25,-1.25) -- (10.5,-1.25) -- (10.5,0.25) -- (9.75,0.25) -- (9.75,-1.25) -- (7.75,-1.25) -- (7.75,0.25) -- (6.25,0.25) -- (6.25,-1.25) -- (4.25,-1.25) -- (4.25,0.25) -- (2.75,0.25) -- (2.75,-1.25) -- cycle;
    \draw[white, fill=white, path fading= west] (9.5,-1.5) -- (9.5,1.5) -- (10.25,1.5) -- (10.25,-1.5) ;
    \draw[white, fill=white] (10.25,1.5) -- (10.75,1.5) -- (10.75,-1.5) -- (10.25,-1.5);
    \draw[white, fill=white, path fading= east] (10.75,1.5) -- (11.5,1.5) -- (11.5,-1.5) -- (10.75,-1.5);
    \node at (10.5,0) {$\cdots$};
    \node at (7.5,-2) {$\mathsf{(b)}$};
    \node at (13.75,0.5) {\LARGE$\mc{T}_{k:0}$};
    \shade[outer color=white, inner color=white, draw=white, rounded corners] (3.25,-1) rectangle (3.75,0);
    \node[below] at (3,-0.5) {$\mathsf{in_0}$};
    \node[below] at (4,-0.5) {$\mathsf{in'_0}$};
    \shade[outer color=white, inner color=white, draw=white, rounded corners] (6.75,-1) rectangle (7.25,0);
    \node[below] at (6.5,-0.5) {$\mathsf{in_1}$};
    \node[below] at (7.5,-0.5) {$\mathsf{in'_1}$};
    \shade[outer color=white, inner color=white, draw=white, rounded corners] (11.75,-1) rectangle (12.25,0);
    \node[below] at (11.75,-0.5) {$\mathsf{in}_{k-1}$};
    \node[above] at (12.25,-0.5) {$\mathsf{in'}_{k-1}$};
    \node[below] at (15.25,-0.5) {$\mathsf{out}$};
    \end{scope}
    \end{tikzpicture}
    \caption[A $k$-step process tensor]{\textbf{A $k$-step process tensor is a map from $k$ operations (\gls{CPTNI} maps) to a quantum state:} $\mathsf{(a)}$ A joint system-environment in a fiducial state $\rho_\mathsf{SE}$ is prepared with an operation $\mc{A}_0$ and subsequently undergoes a joint unitary evolution $\mc{U}_1$, until an intervention with an operation $\mc{A}_1$ is made with subsequent evolution $\mc{U}_1$, and so on until a time-step $k$, where the state of the system \gls{syst} is given by $\rho_\mathsf{S}^{(k)}$. $\mathsf{(b)}$ The process tensor $\mc{T}_{k:0}$ is all of the content out of direct control to the experimenter, and contains all the dynamical information about the system, including \gls{syst-env} correlations across all points in time.}
    \label{fig: Process tensor}
\end{figure}

Such a map, $\mc{T}_{k:0}$ is called a \emph{process tensor}~\cite{Pollock_2018, Pollock_2018_Markov} and we will refer to it as a $k$-step process tensor, or often simply as a $k$-step process. In this sense, a $0$-step process is a joint fiducial state, which might have undergone a unitary evolution, and yielding an \gls{syst} quantum state, e.g. as is considered for equilibration on average. A superchannel is the particular case $\mc{M}=\mc{T}_{1:0}$, i.e. a $1$-step process, taking a preparation and yielding a reduced \gls{syst} quantum state. The name of $\mc{T}_{k:0}$ makes reference to the multilinear structure of such map, which can easily be seen by insertion through Eq.~\eqref{eq: process tensor open}, as $\mc{T}_{k:0}[\alpha\vec{\mc{A}}_{k:0}]+\beta\vec{\mc{B}}_{k:0}]=\alpha\mc{T}_{k:0}[\vec{\mc{A}}_{k:0}+\beta\mc{T}_{k:0}[\vec{\mc{B}}_{k:0}]$ for any two sets of operations $\{\mc{A}_i\}$ and $\{\mc{B}_i\}$ and scalars $\alpha$, $\beta$. This implies that the process tensor can be reconstructed through a relevant tomographic scheme~\cite{PhysRevA.98.012108}. Henceforth, whenever we consider a general $k$-step process tensor, we will simply denote it by $\mc{T}$ unless necessary otherwise.

Now, we know that the superchannel is \gls{CP}, however, we can similarly check that this will be the case for an arbitrary number of time-steps also by constructing its Choi state and checking that is is positive. Knowing the procedure for the superchannel, depicted in Fig.~\ref{fig: superchannel Choi state}, we can see that the generalization to an arbitrary number of $k$ steps follows by introducing $k$ maximally entangled states $\tilde{\Psi}_{\mathsf{A}_i\mathsf{B}_i}\in\mscr{B}(\mscr{H}_{\mathsf{A}_i}\otimes\mscr{H}_{\mathsf{B}_i})$, where $\mscr{H}_{\mathsf{A}_i}\cong\mscr{H}_\mathsf{S}$ and similarly for $\mathsf{B}$, and letting half of each act as an input at every step by swapping the input spaces with the corresponding ancilla. This is more clearly illustrated in Fig.~\ref{fig: Choi process tensor}.

\begin{figure}[t]
    \centering
    \begin{tikzpicture}
    \begin{scope}[xscale=1, yscale=0.8]
    \fill[inner color=white, outer color=C4!5!white, draw=C4, dashed, thick, rounded corners] (0.95,-1.25) -- (0.95,1.25) -- (14.5,1.25) -- (14.5,0) -- (13.25,0) -- (13.25,-1.25) -- (11.25,-1.25) -- (11.25,0.25) -- (10.25,0.25) -- (10.25,-1.25) -- (9.5,-1.25) -- (9.5,0.25) -- (8.75,0.25) -- (8.75,-1.25) -- (6.75,-1.25) -- (6.75,0.25) -- (5.75,0.25) -- (5.75,-1.25) -- (3.75,-1.25) -- (3.75,0.25) -- (2.75,0.25) -- (2.75,-1.25) -- cycle;
    \node[above] at (2.4,0.5) {$\mathsf{E}$};
    \node[above] at (2.4,-0.5) {$\mathsf{S}$};
    \draw[thick] (2,0.5) -- (13.75,0.5);
    \draw[thick] (2,-0.5) -- (14,-0.5) ;
    \draw[thick] (2.25,-1.5) arc (90:270:0.5);
    \draw[thick] (2.25,-1.5) -- (14,-1.5) ;
    \draw[thick] (2.25,-2.5) -- (14,-2.5);
    \node[left] at (1.75,-2) {$\tilde{\Psi}_{\mathsf{A}_1\mathsf{B}_1}$};
    \draw[thick] (2.25,-3) arc (90:270:0.5);
    \draw[ultra thick, C3] (3.25,-0.5) -- (3.25,-1.5);
    \node[C3] at (3.25,-0.5) {$\boldsymbol{\mathsf{X}}$};
    \node[C3] at (3.25,-1.5) {$\boldsymbol{\mathsf{X}}$};
    \node[below right] at (3.25,-1.5) {$\mc{S}_1$};
    \draw[thick] (2.25,-3) -- (14,-3);
    \draw[thick] (2.25,-4) -- (14,-4);
    \node[left] at (1.75,-3.5) {$\tilde{\Psi}_{\mathsf{A}_2\mathsf{B}_2}$};
    \draw[ultra thick, C3] (6.25,-0.5) -- (6.25,-3);
    \node[C3] at (6.25,-0.5) {$\boldsymbol{\mathsf{X}}$};
    \node[C3] at (6.25,-3) {$\boldsymbol{\mathsf{X}}$};
    \node[below right] at (6.25,-3) {$\mc{S}_2$};
    \node[left] at (1.75,-4.5) {$\vdots$};
    \draw[thick] (2.25,-5) arc (90:270:0.5);
    \draw[thick] (2.25,-5) -- (14,-5);
    \draw[thick] (2.25,-6) -- (14,-6);
    \node[left] at (1.75,-5.5) {$\tilde{\Psi}_{\mathsf{A}_k\mathsf{B}_k}$};
    \draw[ultra thick, C3] (10.75,-0.5) -- (10.75,-5);
    \node[C3] at (10.75,-0.5) {$\boldsymbol{\mathsf{X}}$};
    \node[C3] at (10.75,-5) {$\boldsymbol{\mathsf{X}}$};
    \node[below right] at (10.75,-5) {$\mc{S}_k$};
    \fill[outer color=C3!10!white, inner color=white, draw=black, thick] (1.6,0) ellipse (0.5cm and 1cm);
    \node[left] at (2,0) {$\rho_\mathsf{SE}$};
    \shade[outer color=C2!60!white, inner color=white, draw=black, rounded corners, thick] (4,-1) rectangle (5.5,1);
    \node at (4.75,0) {\LARGE$\mc{U}_1$};
    \shade[outer color=C2!60!white, inner color=white, draw=black, rounded corners, thick] (7,-1) rectangle (8.5,1);
    \node at (7.75,0) {\LARGE$\mc{U}_2$};
    \draw[white, fill=white, path fading= west] (8.5,-6.25) -- (8.5,1.5) -- (9.25,1.5) -- (9.25,-6.25) ;
    \draw[white, fill=white] (9.25,1.5) -- (9.75,1.5) -- (9.75,-6.25) -- (9.25,-6.25);
    \draw[white, fill=white, path fading= east] (9.75,1.5) -- (10.5,1.5) -- (10.5,-6.25) -- (9.75,-6.25);
    \node at (9.5,0) {$\cdots$};
    \node at (9.5,-2) {$\cdots$};
    \node at (9.5,-3.5) {$\cdots$};
    \node at (9.5,-4.4) {$\vdots$};
    \node at (9.5,-5.5) {$\cdots$};
    \shade[outer color=C2!60!white, inner color=white, draw=black, rounded corners, thick] (11.5,-1) rectangle (13,1);
    \node at (12.25,0) {\LARGE$\mc{U}_k$};
    \node[right] at (13.75,0.5) {\trash};
    \fill[outer color=C4!60!white, inner color=white, draw=black, thick] (13.5,-0.25) rectangle (14.5,-6.25);
    \draw[white, fill=white, path fading= north] (13.25,-4) -- (13.25,-4.5) -- (14.75,-4.5) --(14.75,-4) ;
    \draw[white, fill=white] (13.25,-4.4) -- (14.75,-4.4) -- (14.75,-4.6) -- (13.25,-4.6);
    \draw[white, fill=white, path fading= south] (13.25,-4.5) -- (14.75,-4.5) -- (14.75,-5) -- (13.25,-5);
    \node at (14,-4.4) {$\vdots$};
    \node at (14,-3.25) {\large$\Upsilon_{k:0}$};
    \end{scope}
    \end{tikzpicture}
    \caption[Choi state representation of the process tensor]{\textbf{The Choi state representation of a $k$-step process tensor}, denoted $\Upsilon_{k:0}$, can be obtained by swapping out the system, $\mc{S}_i$ with half a maximally entangled state, $\Psi_{\mathsf{A}_i\mathsf{B}_i}$, at each step $i$. The final state is an unnormalized many-body state acting on a $d_\mathsf{S}^{2k+1}$ dimensional system.}
    \label{fig: Choi process tensor}
\end{figure}
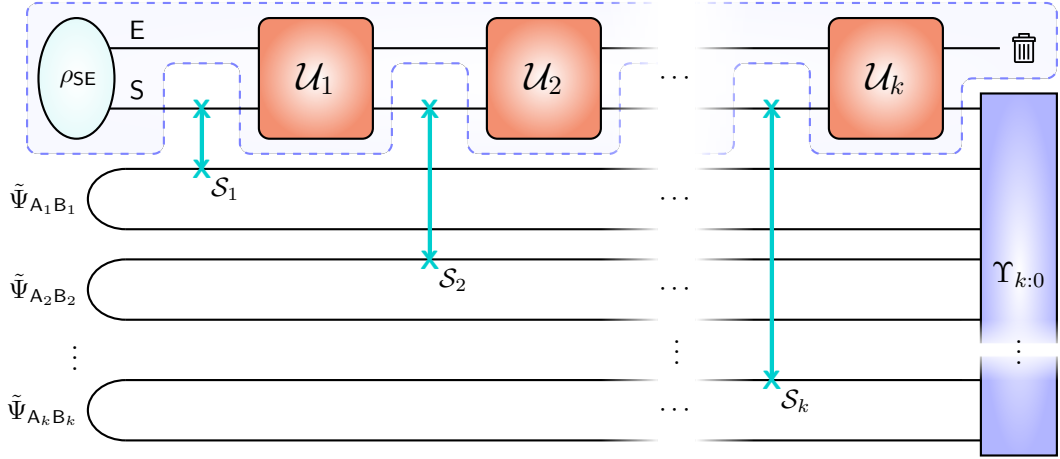

Specifically, the Choi state of the process tensor takes the form
\begin{equation}
    \Upsilon_{k:0}=\tr_\mathsf{E}\left[\,\mc{U}_k\,\mc{S}_k\,\mc{U}_{k-1}\,\mc{U}_{k-1}\cdots\,\mc{U}_1\,\mc{S}_1\,(\rho_\mathsf{SE}\otimes\tilde{\Psi}^{\otimes{k}})\right],
    \label{eq: process tensor Choi state}
\end{equation}
where here we are implicitly writing $\mc{U}_i$ for $\mc{U}_i\otimes\mc{I}_{\mathsf{A}_1\mathsf{B}_1\cdots\mathsf{A}_k\mathsf{B}_k}$ and $\tilde{\Psi}^{\otimes{k}}=\tilde{\Psi}_{\mathsf{A}_1\mathsf{B}_1}\otimes\cdots\otimes\tilde{\Psi}_{\mathsf{A}_k\mathsf{B}_k}$. The generalized swap between system \gls{syst} and an ancilla $\mathsf{A}_i$ at time-step $i$ is defined by 
\begin{equation}
    \mc{S}_i(\cdot):=\swap_{\mathsf{S}\mathsf{A}_i}(\cdot)\swap_{\mathsf{S}\mathsf{A}_i}
\end{equation}
where here
\begin{equation}
    \swap_{\mathsf{S}\mathsf{A}_i}:=\sum_{i,j=1}^{d_\mathsf{S}}\mc{I}_\mathsf{E}\otimes|i\rangle\!\langle{j}|\otimes\mc{I}_{\mathsf{A}_1\mathsf{B}_1\cdots\mathsf{A}_{i-1}\mathsf{B}_{i-1}}\otimes|j\rangle\!\langle{i}|\otimes\mc{I}_{\mathsf{B}_i\mathsf{A}_{i+1}\mathsf{B}_{i+1}\cdots\mathsf{A}_k\mathsf{B}_k},
\end{equation}
and as in the case of the superchannel, the Choi state of the process tensor is manifestly positive by definition.

Similar to the case of the process tensor, we will commonly denote simply by $\Upsilon$ a $k$-step process tensor unless it is relevant to denote explicitly the number of time-steps.

To see that this reproduces correctly the action of the process tensor, we can contract with the Choi state of the collection of time-ordered operations $\mc{A}_0,\mc{A}_1,\ldots,\mc{A}_{k-1}$: these are uncorrelated, so the full Choi state is simply a tensor product of Choi states.

We can illustrate this for $k=2$ and the generalization follows trivially. First let us expand the swaps and the maximally entangled states as we did in Eq.~\eqref{eq: Choi superchannel expanded} for the superchannel, i.e.
\begin{align}
    \Upsilon_{2:0}=\sum_{\alpha,\ldots,\delta}\tr_\mathsf{E}\left[U_2\FS_{\alpha_2\beta_2}U_1\FS_{\alpha_1\beta_1}\rho_\mathsf{SE}\FS_{\delta_1\gamma_1}U_1^\dg\FS_{\delta_2\gamma_2}U_2^\dg\right]\otimes|\beta_1\alpha_1\beta_2\alpha_2\rangle\!\langle\delta_1\gamma_1\delta_2\gamma_2|,
\end{align}
so that, contracting with $\Upsilon_{\vec{\mc{A}}_{1:0}}^\mathrm{T}:=\Upsilon_{\mc{A}_0}^\mathrm{T}\otimes\Upsilon_{\mc{A}_1}^\mathrm{T}$, where here similarly the single operation Choi states are defined by $\Upsilon_{\mc{A}_i}=(\mbb1_{\mathsf{A}_i}\otimes\mc{A}_i)\tilde{\Psi}$, then
\begin{align}
    &\tr_\mathsf{in}[\Upsilon_{2:0}(\mbb{1}_\mathsf{out}\otimes\Upsilon_{\vec{\mc{A}}_{1:0}}^\mathrm{T})]\nonumber\\
    &=\sum_{\alpha,\ldots,\delta}\tr_\mathsf{E}\left[U_2\FS_{\alpha_2\beta_2}U_1\FS_{\alpha_1\beta_1}\rho_\mathsf{SE}\FS_{\delta_1\gamma_1}U_1^\dg\FS_{\delta_2\gamma_2}U_2^\dg\right]\langle\gamma_1|\mc{A}_0^\mathrm{T}(|\beta_1\rangle\!\langle\delta_1|)|\alpha_1\rangle\!\langle\gamma_2|\mc{A}_1^\mathrm{T}(|\beta_2\rangle\!\langle\delta_2|)|\alpha_2\rangle\nonumber\\
    &=\sum_{\alpha,\ldots,\delta}\tr_\mathsf{E}\left[U_2\FS_{\alpha_2\beta_2}U_1\FS_{\alpha_1\beta_1}\rho_\mathsf{SE}\FS_{\delta_1\gamma_1}U_1^\dg\FS_{\delta_2\gamma_2}U_2^\dg\right]\langle\alpha_1|\mc{A}_0(|\beta_1\rangle\!\langle\delta_1|)|\gamma_1\rangle\!\langle\alpha_2|\mc{A}_1(|\beta_2\rangle\!\langle\delta_2|)|\gamma_2\rangle\nonumber\\
    &=\sum_\mu\tr_\mathsf{E}\left[U_2A_{\mu_1}U_1A_{\mu_0}\rho_\mathsf{SE}A_{\mu_0}^{\dg}U_1^{\dg}A_{\mu_1}^{\dg}U_2^\dg\right]\nonumber\\
    &=\tr_\mathsf{E}\left[\,\mc{U}_2\,\mc{A}_1\,\mc{U}_1\,\mc{A}_0(\rho_\mathsf{SE})\right]\nonumber\\
    &=\mc{T}_{2:0}[\vec{\mc{A}}_{1:0}],
\end{align}
where the trace is over all input spaces, as labelled in Fig.~\ref{fig: Process tensor}, and where we used the fact that the operations $\mc{A}_i$ are \gls{CP} by decomposing $\mc{A}_i(\cdot)=\sum_{\mu_i}A_{\mu_i}(\cdot)A_{\mu_i}^\dg$. This can be checked similarly for any number of time-steps $k$, so that
\begin{equation}
    \mc{T}[\vec{\mc{A}}_{k-1:0}]=\tr_\mathsf{in}\left[\Upsilon\left(\mbb{1}_\mathsf{out}\otimes\Upsilon_{\vec{\mc{A}}_{k-1:0}}^\mathrm{T}\right)\right],
\end{equation}
for any $k$-step process $\mc{T}$ with Choi state $\Upsilon$.

Now, with the Choi state at hand we can also readily verify that the process tensor is \gls{TP}, i.e.
\begin{align}
    \tr_\mathsf{out}[\Upsilon]&=\sum_{\alpha,\ldots,\delta}\tr\left[U_k\FS_{\alpha_k\beta_k}\cdots{U}_1\FS_{\alpha_1\beta_1}\rho_\mathsf{SE}\FS_{\delta_1\gamma_1}U_1^\dg\cdots\FS_{\delta_k\gamma_k}U_k^\dg\right]|\beta_1\alpha_1\cdots\beta_k\alpha_k\rangle\!\langle\delta_1\gamma_1\cdots\delta_k\gamma_k|\nonumber\\
    &=\sum_{\alpha,\beta}|\beta_1\alpha_1\cdots\beta_k\alpha_k\rangle\!\langle\beta_1\alpha_1\cdots\beta_k\alpha_k|\nonumber\\
    &=\mbb1_{\mathsf{A}_1\mathsf{B}_1\cdots\mathsf{A}_k\mathsf{B}_k}\nonumber\\
    &=\mbb1_\mathsf{in},
\end{align}
where the trace is over the output in system \gls{syst}, thus whenever the operations $\mc{A}_i$ are \gls{TP}, the action of the process tensor on these will yield a properly normalized quantum state.

The final ingredient we require from the process tensor is a containment property (which amounts to a \emph{causality} property~\cite{milz_2019}) and more generally a consistency condition. Now, in principle this is akin to doing a marginalization, however, when we did this for classical processes we implicitly made the assumption that there was a single way to probe the system in question and that this probing had no influence at all in the process. So, even if classically summing over all events amounts to doing nothing, we cannot do the same in the quantum case. So, for example, the sum over outcomes of a measurement $\mc{B}$ in a basis $\{|\beta\rangle\}$ is given by $\mc{B}(\rho)=\sum_\beta|\beta\rangle\!\langle\beta|\rho|\beta\rangle\!\langle\beta|$, thus a state being classical corresponds to it being diagonal in the basis $\{|\beta\rangle\}$ and $\mc{B}(\rho)=\rho$, however when $\rho$ contains non-diagonal non-zero elements this stops being true.

The containment property clearly holds for the process tensor by replacing the relevant operations with identities, however, it was recently proved that consistency in general holds not only for the process tensor as well, but for any stochastic theory~\cite{Milz_2020}. This result serves to define quantum stochastic processes by means of the process tensor in an unambiguous way. For our purposes, we should notice that indeed smaller process tensors are contained in bigger ones in the sense that
\begin{equation}
    \mc{T}_{\mscr{T}_\ell}\left[\vec{A}_{\mscr{T}_\ell}\right]=\mc{T}_{\mscr{T}_k}\left[\vec{A}_{\mscr{T}_\ell}\cup\mc{I}_{\mscr{T}_k\setminus\mscr{T}_\ell}\right],\qquad\forall\mscr{T}_\ell,\mscr{T}_k:\,|\mscr{T}_\ell|<|\mscr{T}_\ell|,
\end{equation}
where $\mscr{T}_n$ is a discrete set of times as defined in Eq.~\eqref{eq: set of times}. This can also be visualized easily through Fig.~\ref{fig: Process tensor containment}.

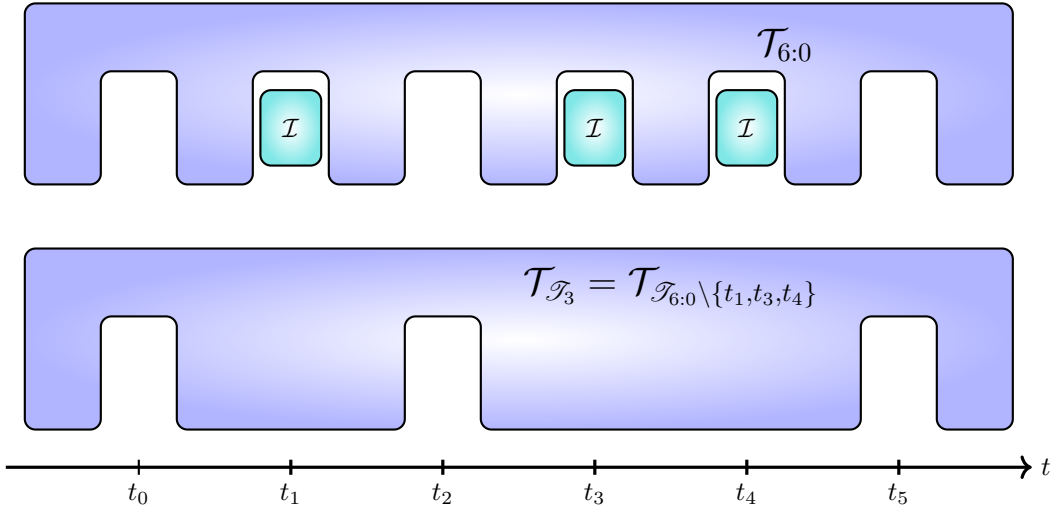
\begin{figure}[t]
    \centering
    \begin{tikzpicture}
    \begin{scope}
    \fill[inner color=white, outer color=C4!60!white, draw=black, thick, rounded corners] (2,1.15) -- (15,1.15) -- (15,-1.25) -- (14,-1.25) -- (14,0.25) -- (13,0.25) -- (13,-1.25) -- (12,-1.25) -- (12,0.25) -- (11,0.25) -- (11,-1.25) -- (10,-1.25) -- (10,0.25) -- (9,0.25) -- (9,-1.25) -- (8,-1.25) -- (8,0.25) -- (7,0.25) -- (7,-1.25) -- (6,-1.25) -- (6,0.25) -- (5,0.25) --  (5,-1.25) -- (4,-1.25) -- (4,0.25) -- (3,0.25) -- (3,-1.25) -- (2,-1.25)  -- cycle;
    \node at (12,0.6) {\LARGE$\mc{T}_{6:0}$};
    \shade[outer color=C3!50!white, inner color=white, draw=black, rounded corners, thick] (5.1,-1) rectangle (5.9,0);
    \node at (5.5,-0.5) {$\mc{I}$};
    \shade[outer color=C3!50!white, inner color=white, draw=black, rounded corners, thick] (9.1,-1) rectangle (9.9,0);
    \node at (9.5,-0.5) {$\mc{I}$};
    \shade[outer color=C3!50!white, inner color=white, draw=black, rounded corners, thick] (11.1,-1) rectangle (11.9,0);
    \node at (11.5,-0.5) {$\mc{I}$};
    \end{scope}
    \begin{scope}[shift={(0,-3.25)}]
    \fill[inner color=white, outer color=C4!60!white, draw=black, thick, rounded corners] (2,1.15) -- (15,1.15) -- (15,-1.25) -- (14,-1.25) -- (14,0.25) -- (13,0.25) -- (13,-1.25) -- (12,-1.25) -- (11,-1.25) -- (10,-1.25) -- (9,-1.25) -- (8,-1.25) -- (8,0.25) -- (7,0.25) -- (7,-1.25) -- (6,-1.25) --  (5,-1.25) -- (4,-1.25) -- (4,0.25) -- (3,0.25) -- (3,-1.25) -- (2,-1.25)  -- cycle;
    \node at (10.5,0.625) {\LARGE$\mc{T}_{\mscr{T}_3}=\mc{T}_{\mscr{T}_{6:0}\setminus\{t_1,t_3,t_4\}}$};
    \end{scope}
    \begin{scope}[shift={(0,-5)}]
    \draw[very thick, ->] (1.75,0) -- (15.25,0);
    \node[right] at (15.25,0) {$t$};
    \draw[thick, -] (3.5,-0.1) -- (3.5,0.1);
    \node[below] at (3.5,-0.1) {$t_0$};
    \draw[very thick] (5.5,-0.1) -- (5.5,0.1);
    \node[below] at (5.5,-0.1) {$t_1$};
    \draw[very thick] (7.5,-0.1) -- (7.5,0.1);
    \node[below] at (7.5,-0.1) {$t_2$};
    \draw[very thick] (9.5,-0.1) -- (9.5,0.1);
    \node[below] at (9.5,-0.1) {$t_3$};
    \draw[very thick] (11.5,-0.1) -- (11.5,0.1);
    \node[below] at (11.5,-0.1) {$t_4$};
    \draw[very thick] (13.5,-0.1) -- (13.5,0.1);
    \node[below] at (13.5,-0.1) {$t_5$};
    \end{scope}
    \end{tikzpicture}
    \caption[Containment property of the process tensor]{\textbf{Process tensor containment property (example):} A 6-step process tensor $\mscr{T}_{6:0}$ contains the 3-step process tensor $\mc{T}_{\mscr{T}_3}$ where $\mscr{T}_3=\{t_0,t_2,t_5\}$; this can be seen by applying identity operations $\mc{I}$ at timesteps $t_1,t_3,t_4$. The overarching idea is that there exists a unique maximal description containing any process on a reduced number of time-steps~\cite{Milz_2020}.}
    \label{fig: Process tensor containment}
\end{figure}

Finally, let us draw a more precise picture in relation with the discussion of the classical case. A quantum event $x_i$ at the $i$\textsuperscript{th} time-step corresponds to an outcome, i.e. a random variable, of the corresponding intervention, which is a \gls{CPTNI} map $\mc{A}^{(x_i)} (\cdot):= \sum_\nu {A}_\nu^{(x_i)}(\cdot)A_\nu^{(x_i)\,\dg}$ with Kraus operators $\{A_\nu^{(x_i)}\}$ satisfying $\sum_\nu{A}_\nu^{(x_i)}{A}_\nu^{(x_i)\,\dg}\leq\mbb1$. More generally, an intervention corresponds to the action of an \emph{instrument} $\mathcal{J} = \left\{\mc{A}^{(x_i)}\right\}$, which overall yields a \gls{CPTP} map when summed over possible outcomes $\sum_{x_i} \mc{A}^{(x_i)}:=\mc{A}^{\mc{J}}$~\cite{Lindblad1979}, naturally generalizing the concept of a \gls{POVM}.

Thus, if an experimenter applies the sequence of \gls{CP} maps $\mc{A}_0^{(x_0)},\mc{A}_1^{(x_1)},\ldots,\mc{A}_{k-1}^{(x_{k-1})}$, each being an element of a corresponding instrument $\mc{J}_0,\mc{J}_1,\ldots,\mc{J}_{k-1}$, and measures the final state $\rho_\mathsf{S}^{(k)}$ with a \gls{POVM} through an instrument $\mc{J}_k=\left\{\mc{A}_k^{(x_k)}=\mathrm{M}_k^{(x_k)}\right\}$, the probability to observe a sequence of quantum events is given by
\begin{align}
    \mbb{P}\left(\vec{x}_{k:0}|\vec{\mathcal{J}}_{k:0}\right)&=\tr\left[\mc{A}_k^{(x_k)} \mathcal{U}_{k} \mc{A}_{k-1}^{(x_{k-1})}\,\mc{U}_{k-1}\ \!\dots\mc{A}_1^{x_1}\mc{U}_1\,\mathcal{A}_0^{x_0} (\rho_\mathsf{SE})\,\right]\nonumber\\
    &=\tr\left\{\,\mc{A}_k^{(x_k)}\,\mc{T}[\vec{\mc{A}}_{k-1:0}^{\,(x_{k-1:0})}]\right\},
    \label{eq: Born rule 1}
\end{align}
where here $\vec{\mathcal{J}}_{k:0}=\{\mc{J}_0,\mc{J}_1,\ldots,\mc{J}_k\}$, so that the process tensor contains the probability distributions for all possible measurements, naturally generalising the classical setting described above.

This can be neatly rewritten by means of the corresponding Choi states, clearly separating the influence of the environment from that of the interventions, in a multitime generalization of the Born rule~\cite{Oreshkov2012, Costa_2016, Shrapnel_2018}:
\begin{gather}
    \mbb{P}\left(\vec{x}_{k:0}|\vec{\mathcal{J}}_{k:0}\right)=\tr\left[\,\Upsilon \Lambda^\mathrm{T}\,\right],
    \label{eq: Born rule 2}
\end{gather}
where we have now defined $\Lambda$ as the Choi state of the interventions $\left\{\mc{A}_i^{(x_i)}\right\}_{i=0}^k$. Notice that $\Lambda$ itself can be thought of having a process-tensor-like structure, here with a set of uncorrelated operations between time-steps. So in fact, we can generalize this to a situation where the experimental interventions are correlated with one another by means of an ancillary space $\mathsf{\Gamma}$, equivalent to an environment, albeit one that the experiment has access to. We can visualize this as a comb contracting with a process tensor as in Fig.~\ref{fig: tester}.

Finally, we can more generally consider the operations $\Lambda$ as an element of a so-called \emph{tester}~\cite{Chiribella_2009}, which is a set of operation combs that overall lead to a deterministic operation. That is, explicitly denoting $\Lambda_{k:0}^{(x_{k:0})}=\Upsilon_{\vec{\mc{A}}_{k:0}^{\,(x_{k:0})}}$, then a tester is given by the set $\left\{\Lambda_{k:0}^{(x_{k:0})}\right\}$, where each element is positive and where the sum, $\Lambda_{k:0}^{\vec{\mc{J}}_{k:0}}:=\sum_{x_i}\Lambda_{k:0}^{(x_{k:0})}$, adds up to a \gls{CPTP} map, $\tr_\mathsf{in}[\Lambda_{k:0}^{\vec{\mc{J}}_{k:0}}]=\mbb1_\mathsf{out}$. In a nutshell, this also generalizes the concept of instrument, i.e. what a tester element is to a tester, a \gls{CP} map is to an instrument and a \gls{POVM} element is to a \gls{POVM}.

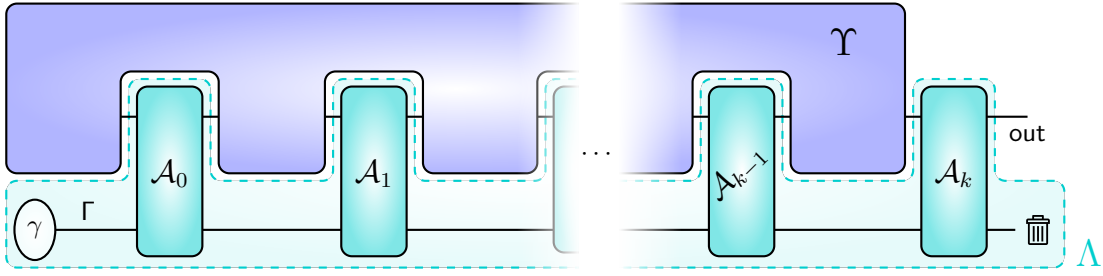
\begin{figure}[t]
    \centering
    \begin{tikzpicture}
    \begin{scope}[xscale=1.075]
    \fill[inner color=white, outer color=C3!10!white, draw=C3, dashed, thick, rounded corners] (3.25,-1.35) -- (3.25,-2.5) -- (16.2,-2.5) -- (16.2,-1.35) -- (15.35,-1.35) --  (15.35,0) -- (14.35,0) -- (14.35,-1.35) -- (12.75,-1.35) -- (12.75,0) -- (11.75,0) -- (11.75,-1.35) -- (10.5,-1.35) -- (10.5,0) -- (9.85,0) -- (9.85,-1.35) -- (8.25,-1.35) -- (8.25,0) -- (7.25,0) -- (7.25,-1.35) -- (5.75,-1.35) -- (5.75,0) -- (4.75,0) -- (4.75,-1.35) -- cycle;
    \node[below] at (15.75,-0.5) {$\mathsf{out}$};
    \draw[thick, -] (4,-0.5) -- (15.75,-0.5);
    \fill[inner color=white, outer color=C4!60!white, draw=black, thick, rounded corners] (3.25,-1.25) -- (3.25,1) -- (14.25,1) -- (14.25,-1.25) -- (12.85,-1.25) -- (12.85,0.1) -- (11.65,0.1) -- (11.65,-1.25) -- (10.5,-1.25) -- (10.5,0.1) -- (9.75,0.1) -- (9.75,-1.25) -- (8.35,-1.25) -- (8.35,0.1) -- (7.15,0.1) -- (7.15,-1.25) -- (5.85,-1.25) -- (5.85,0.1) -- (4.65,0.1) -- (4.65,-1.25) -- cycle;
    \node at (10.5,0.5) {$\cdots$};
    \node at (13.5,0.5) {\LARGE$\Upsilon$};
    \draw[thick, -] (3.6,-2) -- (15.6,-2);
    \node[right] at (15.58,-2) {\trash};
    \fill[outer color=C3!5!white, inner color=white, draw=black, thick] (3.6,-2) ellipse (0.25 and 0.4);
    \node at (3.6,-2) {$\gamma$};
    \node[above] at (4.25,-2) {$\mathsf{\Gamma}$};
    \shade[outer color=C3!50!white, inner color=white, draw=black, rounded corners, thick] (4.85,-2.35) rectangle (5.65,-0.1);
    \node at (5.25,-1.25) {\Large$\mc{A}_0$};
    \shade[outer color=C3!50!white, inner color=white, draw=black, rounded corners, thick] (7.35,-2.35) rectangle (8.15,-0.1);
    \node at (7.75,-1.25) {\Large$\mc{A}_1$};
    \shade[outer color=C3!50!white, inner color=white, draw=black, rounded corners, thick] (9.95,-2.3) rectangle (10.5,-0.1);
    \shade[outer color=C3!50!white, inner color=white, draw=black, rounded corners, thick] (11.85,-2.35) rectangle (12.65,-0.1);
    \node[rotate=45] at (12.2,-1.25) {\Large$\mc{A}_{k-1}$};
    \shade[outer color=C3!50!white, inner color=white, draw=black, rounded corners, thick] (14.45,-2.35) rectangle (15.25,-0.1);
    \node at (14.85,-1.25) {\Large$\mc{A}_k$};
    \draw[white, fill=white, path fading= west] (9.5,-2.75) -- (9.5,1.5) -- (10.25,1.5) -- (10.25,-2.75) ;
    \draw[white, fill=white] (10.25,1.5) -- (10.75,1.5) -- (10.75,-2.75) -- (10.25,-2.75);
    \draw[white, fill=white, path fading= east] (10.75,1.5) -- (11.5,1.5) -- (11.5,-2.75) -- (10.75,-2.75);
    \node at (10.5,-1) {$\cdots$};
    \node[below, C3] at (16.5,-2) {\LARGE$\Lambda$};
    \end{scope}
    \end{tikzpicture}
    \caption[Correlated operations on a process tensor]{\textbf{A set of correlated operations} $\left\{\mc{A}_0,\mc{A}_1,\ldots,\mc{A}_{k-1}\right\}$ followed by a \gls{POVM} element $\mc{A}_k=\mathrm{M}_k$ acting on a joint $\mathsf{S\Gamma}$ system, where $\mathsf{\Gamma}$ is an ancillary system initialized in state $\gamma$, can be described as a process-tensor like object that is \gls{CP} with Choi state $\Lambda$. Here we equivalently denote the process tensor as $\Upsilon$.}
    \label{fig: tester}
\end{figure}

The topic of quantum stochastic processes is rapidly developing and by no means are we making an exhaustive presentation in this thesis. While many questions date back to several decades, with noticeable approaches by Accardi~\cite{Accardi_1976, Accardi_1982} and Lindblad~\cite{Lindblad1979}, similar to the case of the foundations of statistical mechanics, the topic has experienced a renewed interest from different perspectives (and names), with a great deal of advancement in recent years~\cite{Hardy_2012, Hardy_2016, Werner_2005,Caruso_2014, Chiribella_2008, Chiribella_2009, Chiribella_2013, Oreshkov2012, Costa_2016, Oreshkov_2016_2,Portmann_2017, Milz_2018, Strasberg_2019, Strasberg_2019_2}. As we will see, in a sense, with this thesis we have contributed to this program from a front that we can directly relate with equilibration and typicality. To draw this connection let us then discuss how the process tensor naturally generalizes the concept of Markovianity and accounts for a non-ambiguous measure of non-Markovianity. 

\subsection{The quantum Markov condition}
As we have now argued, the standard approach to open quantum dynamics by means of dynamical maps offers a necessary but not sufficient condition for Markovianity. Attempts to define what a quantum Markovian process is have been made~\cite{Rivas_2014, Breuer_2016} and several \emph{measures} of non-Markovianity have been proposed e.g. based on divisibility of dynamical maps~\cite{Rivas_2010, Hou_2011}, positivity~\cite{Wolf_2008, Usha_2011, Devi_2012, Manisalco_2014}, or on trace distance, positing that this distinguishability for two initial quantum states must be monotonically decreasing over time, essentially implying that if one can distinguish better these two states it is because some information is flowing back from the environment~\cite{Breuer_2009}. All of these are valid witnesses, however, they have been seen to be inconsistent with each other~\cite{Rivas_2011}, disagreeing on whether a process is Markovian and/or on the degree of non-Markovianity.

As we have stressed before, at its core a quantum Markov condition has to consider multitime correlations and approaches based on dynamical maps or generally quantum channels capture by definition only two-point correlations.

Let us begin by writing the conditional probability of getting an outcome $x_k$ given outcomes $\vec{x}_{k-1:0}$, all with respective instruments $\vec{\mc{J}}_{k:0}$ analogous to Eq.~\eqref{eq: def conditional probability}, i.e.
\begin{equation}
    \mbb{P}(x_k|\vec{x}_{k-1:0};\vec{\mc{J}}_{k:0})=\f{\mbb{P}(\vec{x}_{k:0}|\vec{\mc{J}}_{k:0})}{\mbb{P}(\vec{x}_{k-1:0}|\vec{\mc{J}}_{k-1:0})},
    \label{eq: conditional quantum}
\end{equation}
where the numerator is given by Eq.~\eqref{eq: Born rule 1} and similarly for the denominator, $\mbb{P}(\vec{x}_{k-1:0}|\vec{\mc{J}}_{k-1:0})=\tr\left\{\mc{A}_{k-1}^{\,(x_{k-1})}\mc{T}_{k-1:0}\left[\vec{\mc{A}}_{k-2:0}^{\,(\vec{x}_{k-2:0})}\right]\right\}$. The main issue for defining a Markov condition lies in the invasiveness the operations $\mc{A}_i^{(x_i)}$. Notice that the denominator of Eq.~\eqref{eq: conditional quantum} won't necessarily capture all of the information of the output state after interrogation at time-step $k-1$, while this will affect the numerator at time-step $k$.

The main idea in Ref.~\cite{Pollock_2018_Markov} to resolve this is called a \emph{causal break}. Essentially, the history dependence of a process at time-step $k$ can be checked by fixing its state at such time-step and analyzing its future. This can be described as follows: an experimenter performs $k-1$ operations, and at time-step $k$ performs a measurement, given by a \gls{POVM} $\left\{\mathrm{M}^{(m_k)}\right\}$, and independently reprepares the system in an independent known state, $\sigma_\mathsf{S}^{(s_k)}$. Here the labels $m_k$ and $s_k$ just distinguish the measured outcome and the input for the fresh state. What this achieves is to break the information flow on subsystem \gls{syst}, while leaving it in a known state, thus allowing to condition future statistics on this state. Put differently, after a causal break the \gls{syst-env} state will be in a product form, and independent of the previous state of \gls{syst}, however, the \gls{env} part an explicitly depend on the state $\rho_\mathsf{SE}^{(k)}$ before the causal break, and as such on all operations $\vec{\mc{A}}_{k-1:0}^{(\vec{x}_{k-1:0})}$ on \gls{syst} before the causal break.

Specifically, in terms of an instrument, this means a causal break at time-step $k$ is given by
\begin{equation}
    \mc{J}_k^{\mc{B}}:=\left\{\mc{B}_k^{(x_k)}:=\mathrm{M}^{(m_k)}\otimes\sigma_\mathsf{S}^{(s_k)}\right\},
\end{equation}
which can be depicted as in Fig.~\ref{fig: causal break}. In general any operation with an input independent of its output will constitute a causal break, with the respective Choi states being uncorrelated. After a causal break the joint \gls{syst-env} state is rendered in a product form, with the state \gls{syst} only depending on the label $s_k$ but not on the previous $\rho_{\mathsf{SE}}^{(k)}$. That is, the causal break acts on the \gls{syst-env} state as
\begin{equation}
    \left[\mc{B}_k^{(x_k)}\otimes\mc{I}_\mathsf{E}\right]\rho_\mathsf{SE}^{(k)}=\sigma_\mathsf{S}^{(s_k)}\otimes\tr_\mathsf{S}\left[\rho_\mathsf{SE}^{(k)}\left(\mathrm{M}^{(m_k)}\otimes\mbb1_\mathsf{E}\right)^\mathrm{T}\right],
\end{equation}
so the environment, however, can still depend on $\rho_{\mathsf{SE}}^{(k)}$ and the full operation process $\Lambda_{k-1:0}$, thus if we are able to distinguish any two distinct operation tensors, $\Lambda_{k-1:0}\neq\Lambda^\prime_{k-1:0}$ and two different \gls{POVM} outcomes $m_k\neq{m}^\prime_k$, the environment must have carried the memory allowing us to do so.

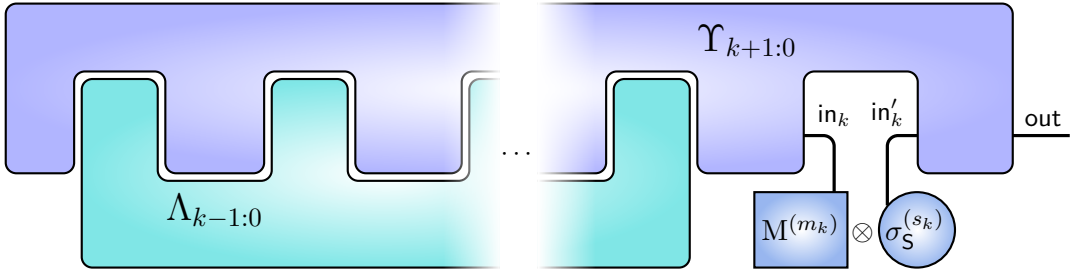
\begin{figure}[t]
    \centering
    \begin{tikzpicture}
    \begin{scope}
    \fill[inner color=white, outer color=C3!50!white, draw=black, thick, rounded corners]  (12.75,-2.5) -- (12.75,0) -- (11.75,0) -- (11.75,-1.35) -- (10.5,-1.35) -- (10.5,0) -- (9.85,0) -- (9.85,-1.35) -- (8.25,-1.35) -- (8.25,0) -- (7.25,0) -- (7.25,-1.35) -- (5.75,-1.35) -- (5.75,0) -- (4.75,0) -- (4.75,-2.5) -- cycle;
    \draw[very thick] (17,-0.75) -- (17.75,-0.75);
    \node[above] at (17.4,-0.75) {$\mathsf{out}$};
    \fill[inner color=white, outer color=C4!60!white, draw=black, thick, rounded corners] (3.75,-1.25) -- (3.75,1) -- (17,1) -- (17,-1.25) -- (15.75,-1.25) -- (15.75,0.1) -- (14.25,0.1) -- (14.25,-1.25) --  (12.85,-1.25) -- (12.85,0.1) -- (11.65,0.1) -- (11.65,-1.25) -- (10.5,-1.25) -- (10.5,0.1) -- (9.75,0.1) -- (9.75,-1.25) -- (8.35,-1.25) -- (8.35,0.1) -- (7.15,0.1) -- (7.15,-1.25) -- (5.85,-1.25) -- (5.85,0.1) -- (4.65,0.1) -- (4.65,-1.25) -- cycle;
    \draw[very thick, rounded corners] (14.25,-0.75) -- (14.65,-0.75) -- (14.65,-1.6);
    \node[above] at(14.65,-0.75) {$\mathsf{in}_k$} ;
    \draw[very thick, rounded corners] (15.75,-0.75) -- (15.35,-0.75) -- (15.35,-1.675);
    \node[above] at(15.35,-0.75) {$\mathsf{in}_k^\prime$} ;
    \fill[inner color=white, outer color=C1!60!white, draw=black,thick] (13.6,-1.5) rectangle (14.825,-2.5);
    \fill[inner color=white, outer color=C1!60!white, draw=black,thick] (15.74,-2) ellipse (0.5 and 0.5);
    \node[below] at (14.9,-1.6) {\large$\mathrm{M}^{(m_k)}\,\otimes\,\sigma_\mathsf{S}^{(s_k)}$};
    \node at (10.5,0.5) {$\cdots$};
    \node at (13.5,0.5) {\LARGE$\Upsilon_{k+1:0}$};
    \node at (6.5,-1.75) {\LARGE$\Lambda_{k-1:0}$};
    \draw[white, fill=white, path fading= west] (9.5,-2.75) -- (9.5,1.5) -- (10.25,1.5) -- (10.25,-2.75) ;
    \draw[white, fill=white] (10.25,1.5) -- (10.75,1.5) -- (10.75,-2.75) -- (10.25,-2.75);
    \draw[white, fill=white, path fading= east] (10.75,1.5) -- (11.5,1.5) -- (11.5,-2.75) -- (10.75,-2.75);
    \node at (10.5,-1) {$\cdots$};
    \end{scope}
    \end{tikzpicture}
    \caption[Causal break]{\textbf{Causal break} by means of an instrument $\mc{J}_k^\mc{B}=\left\{\mathrm{M}_k^{(m_k)}\otimes\sigma_\mathsf{S}^{(s_k)}\right\}$ with $\mathrm{M}_k^{(m_k)}$ a \gls{POVM} element with outcome $m_k$ and $\sigma_\mathsf{S}^{(s_k)}$ an independent input state labelled by $s_k$, acting at time-step $k$ on a process $\Upsilon_{k+1:0}$ which previously acts on an operation process $\Lambda_{k-1:0}$. Quantum Markovianity refers to an independence  at any future time-step $k+1$ from the previous operations $\Lambda_{k-1:0}$ and the outcome $m_k$; any dependence on these will point to information being carried through the environment on $\Upsilon_{k+1:0}$ and in such case the process is non-Markovian.}
    \label{fig: causal break}
\end{figure}

We can then introduce the following:

\begin{definition}[Quantum Markov condition~\cite{Pollock_2018_Markov, Costa_2016}]
\label{def: quantum Markov} A $(k+1)$-step quantum process $\Upsilon_{k+1:0}$ is Markovian if the statistics with respect to any measurement $\mc{J}_{k+1}$ after a causal break $\mc{J}_k^{\mc{B}}:=\left\{\mathrm{M}^{(m_k)}\otimes\sigma_\mathsf{S}^{(s_k)}\right\}$ are independent of outcome $m_k$ and all historic outcomes $\vec{x}_{k-1:0}$ with all possible historic instruments $\vec{\mc{J}}_{k-1:0}$,
\begin{equation}
    \mbb{P}\left(x_{k+1}|m_k,s_k,\vec{x}_{k-1:0};\mc{J}_{k+1},\mc{J}_k^\mc{B},\vec{\mc{J}}_{k-1:0}\right)=\mbb{P}\left(x_{k+1}|s_k;\mc{J}_{k+1},\mc{J}_k^\mc{B}\,\right),
    \label{eq: quantum Markov}
\end{equation}
for any given state $\sigma_\mathsf{S}^{(s_k)}$ prepared in the causal break.
\end{definition}

Now analogous to the classical case, if Eq.~\eqref{eq: quantum Markov} is not satisfied, the quantum process is called non-Markovian. This provides an unambiguous method to witness memory effects. Importantly, the \emph{history} of a process, given by the operation process $\Lambda_{k-1:0}$ with instruments $\vec{\mc{J}}_{k-1:0}$ and the measurement associated to $\mathrm{M}^{(m_k)}$ is the generalization of a \emph{trajectory} for a quantum process~\cite{Sakuldee_2018}.

Now, if we fix the choice of instruments at each time-step and these consist of causal breaks, $\mc{J}^\mc{B}_i=\left\{\mc{B}_i^{(x_i)}\right\}$ for $i=0,1,\ldots,k$, then Eq.~\eqref{eq: quantum Markov} implies the classical Markov condition $\mbb{P}\left(x_{k+1}|\vec{x}_{k:0}\right)=\mbb{P}(x_{k+1}|x_k)$ as in Eq.~\eqref{eq: classical Markov}.  Two other consequences that are shown in Ref.~\cite{Pollock_2018_Markov} are that Markovian dynamics are divisible, within the time-steps where they are defined, as in Eq.~\eqref{eq: divisibility dynamical maps} (with the converse, as discussed for the classical case, not being true), and that whenever the definitions e.g. in Ref.~\cite{Rivas_2010, Hou_2011, Wolf_2008, Usha_2011, Devi_2012, Manisalco_2014,Breuer_2009,Rivas_2011} predict non-Markovianity, Eq.~\eqref{eq: quantum Markov} will also predict non-Markovianity, while the converse is not true. As stressed before, many of the shortcomings in other approaches boil down to having necessary but not sufficient conditions for Markovianity.

Finally, the concept of a causal break simply gives us an operationally non-ambiguous way to characterize a quantum process as Markovian but whether this is the case should not depend on such concept; precisely, the process tensor was conceived around the idea of separating experimentally accessible quantities from those intrinsic to the dynamical process. As it turns out, the process tensor gives a straightforward way to classify Markovian processes: given that these have no temporal correlations and are divisible, this implies that their Choi state is simply a tensor product of dynamical maps connecting adjacent time-steps. More specifically, notice that pairs of subsystems of the Choi state $\Upsilon$ in Fig.~\ref{fig: Choi process tensor} correspond to the different $\mathsf{in}$ and $\mathsf{in'}$ spaces at each time-steps of $\mc{T}$, so that the Choi state encodes temporal correlations as spatial correlations.

\begin{figure}[t]
    \centering
    \begin{tikzpicture}
    \begin{scope}[xscale=1, yscale=0.8]
    \fill[inner color=white, outer color=C1!5!white, draw=C1, dashed, thick, rounded corners] (1.75,0.1) -- (1.75,1.5) -- (6.2,1.5) -- (6.2,0) -- (5.7,0) -- (5.7,-1.2) -- (3.8,-1.2) -- (3.8,0.1) -- cycle;
    \fill[inner color=white, outer color=C1!5!white, draw=C1, dashed, thick, rounded corners] (6.4,0.1) -- (6.4,1.5) -- (9.2,1.5) -- (9.2,0) -- (8.7,0) -- (8.7,-1.2) -- (6.8,-1.2) -- (6.8,0.1) -- cycle;
    \node[above] at (2.9,0.5) {$\mathsf{E}$};
    \node[above] at (2.9,-0.5) {$\mathsf{S}$};
    \draw[thick, rounded corners] (2,0.5) -- (5.9,0.5) -- (5.9,0.7);
    \node[above] at (5.9,0.7) {\trash};
    \fill[outer color=C3!10!white, inner color=white, draw=black, thick] (6.7,1) ellipse (0.2cm and 0.3cm);
    \node at (6.7,1) {$\varepsilon_2$};
    \draw[thick, rounded corners] (7,0.5) -- (6.7,0.5) -- (6.7,0.7);
    \draw[thick, rounded corners] (8.5,0.5) -- (8.8,0.5) -- (8.8,0.7);
    \node[above] at (8.8,0.7) {\trash};
     \begin{scope}[shift={(4.5,0)}]
    \fill[inner color=white, outer color=C1!5!white, draw=C1, dashed, thick, rounded corners] (6.4,0.1) -- (6.4,1.5) -- (9.2,1.5) -- (9.2,0) -- (8.7,0) -- (8.7,-1.2) -- (6.8,-1.2) -- (6.8,0.1) -- cycle;
    \fill[outer color=C3!10!white, inner color=white, draw=black, thick] (6.7,1) ellipse (0.2cm and 0.31cm);
    \node at (6.7,1) {$\varepsilon_k$};
    \draw[thick, rounded corners] (7,0.5) -- (6.7,0.5) -- (6.7,0.7);
    \draw[thick, rounded corners] (8.5,0.5) -- (8.8,0.5) -- (8.8,0.7);
    \node[above] at (8.8,0.7) {\trash};
    \end{scope}
    \draw[thick] (2,-0.5) -- (14,-0.5) ;
    \draw[thick] (2.25,-1.5) arc (90:270:0.5);
    \draw[thick] (2.25,-1.5) -- (14,-1.5) ;
    \draw[thick] (2.25,-2.5) -- (14,-2.5);
    \node[right] at (1.75,-2) {$\tilde{\Psi}_{\mathsf{A}_1\mathsf{B}_1}$};
    \draw[thick] (2.25,-3) arc (90:270:0.5);
    \draw[ultra thick, C3] (3.25,-0.5) -- (3.25,-1.5);
    \node[C3] at (3.25,-0.5) {$\boldsymbol{\mathsf{X}}$};
    \node[C3] at (3.25,-1.5) {$\boldsymbol{\mathsf{X}}$};
    \node[below right] at (3.25,-1.5) {$\mc{S}_1$};
    \draw[thick] (2.25,-3) -- (14,-3);
    \draw[thick] (2.25,-4) -- (14,-4);
    \node[right] at (1.75,-3.5) {$\tilde{\Psi}_{\mathsf{A}_2\mathsf{B}_2}$};
    \draw[ultra thick, C3] (6.25,-0.5) -- (6.25,-3);
    \node[C3] at (6.25,-0.5) {$\boldsymbol{\mathsf{X}}$};
    \node[C3] at (6.25,-3) {$\boldsymbol{\mathsf{X}}$};
    \node[below right] at (6.25,-3) {$\mc{S}_2$};
    \node[right] at (1.75,-4.4) {$\vdots$};
    \draw[thick] (2.25,-5) arc (90:270:0.5);
    \draw[thick] (2.25,-5) -- (14,-5);
    \draw[thick] (2.25,-6) -- (14,-6);
    \node[right] at (1.75,-5.5) {$\tilde{\Psi}_{\mathsf{A}_k\mathsf{B}_k}$};
    \draw[ultra thick, C3] (10.75,-0.5) -- (10.75,-5);
    \node[C3] at (10.75,-0.5) {$\boldsymbol{\mathsf{X}}$};
    \node[C3] at (10.75,-5) {$\boldsymbol{\mathsf{X}}$};
    \node[below right] at (10.75,-5) {$\mc{S}_k$};
    \fill[outer color=C3!10!white, inner color=white, draw=black, thick] (2.18,0.6) ellipse (0.3cm and 0.4cm);
    \fill[outer color=C3!10!white, inner color=white, draw=black, thick] (2.15,-0.5) ellipse (0.35cm and 0.45cm);
    \node[left] at (2.5,0.6) {$\varepsilon_1$};
    \node[left] at (2.575,-0.5) {$\rho_\mathsf{S}^{(0)}$};
    \shade[outer color=C2!60!white, inner color=white, draw=black, rounded corners, thick] (4,-1) rectangle (5.5,1);
    \node at (4.75,0) {\LARGE$\mc{U}_1$};
    \shade[outer color=C2!60!white, inner color=white, draw=black, rounded corners, thick] (7,-1) rectangle (8.5,1);
    \node at (7.75,0) {\LARGE$\mc{U}_2$};
    \draw[white, fill=white, path fading= west] (8.5,-6.25) -- (8.5,1.5) -- (9.25,1.5) -- (9.25,-6.25);
    \draw[white, fill=white] (9.25,1.5) -- (9.75,1.5) -- (9.75,-6.25) -- (9.25,-6.25);
    \draw[white, fill=white, path fading= east] (9.75,1.5) -- (10.5,1.5) -- (10.5,-6.25) -- (9.75,-6.25);
    \node at (9.5,0) {$\cdots$};
    \node at (9.5,-2) {$\cdots$};
    \node at (9.5,-3.5) {$\cdots$};
    \node at (9.5,-4.5) {$\vdots$};
    \node at (9.5,-5.5) {$\cdots$};
    \shade[outer color=C2!60!white, inner color=white, draw=black, rounded corners, thick] (11.5,-1) rectangle (13,1);
    \node at (12.25,0) {\LARGE$\mc{U}_k$};
    \fill[bottom color=C4!60!white, top color=white, thick, draw=black] (13.5,-0.25) -- (13.5,-0.9) -- (14.5,-0.9) -- (14.5,-0.25);
    \draw[very thick, dashed] (13.5,-0.25) -- (14.5,-0.25);
    \fill[outer color=C4!60!white, inner color=white, thick, draw=black] (13.5,-1.1) rectangle (14.5,-1.9);
    \fill[outer color=C4!60!white, inner color=white, thick, draw=black] (13.5,-2.1) rectangle (14.5,-3.4);
    \fill[outer color=C4!60!white, inner color=white, thick, draw=black] (13.5,-4.4) -- (13.5,-3.6) -- (14.5,-3.6) -- (14.5,-4.4);
    \node[right] at (14.5,-3.75) {\small$\Upsilon^{(\mscr{Z})}_{2:1}$};
    \fill[outer color=C4!60!white, inner color=white, thick, draw=black] (13.5,-4.6) -- (13.5,-5.4) -- (14.5,-5.4) -- (14.5,-4.6);
    \draw[white, fill=white, path fading= north] (13.25,-4) -- (13.25,-4.5) -- (14.7,-4.5) --(14.7,-4) ;
    \draw[white, fill=white] (13.25,-4.35) -- (14.7,-4.35) -- (14.7,-4.65) -- (13.25,-4.65);
    \draw[white, fill=white, path fading= south] (13.25,-4.5) -- (14.7,-4.5) -- (14.7,-5) -- (13.25,-5);
    \fill[top color=C4!60!white, bottom color=white, thick, draw=black] (13.5,-6.25) -- (13.5,-5.6) -- (14.5,-5.6) -- (14.5,-6.25);
    \draw[very thick, dashed] (13.5,-6.25) -- (14.5,-6.25);
    \node at (14,-4.4) {$\vdots$};
    \node[right] at (14.75,-4.35) {\small$\vdots$};
    \node[right] at (14.5,-1.5) {\small$\rho_\mathsf{S}^{(0)}$};
    \node[right] at (14.5,-2.75) {\small$\Upsilon^{(\mscr{Z})}_{1:0}$};
    \node[right] at (14.5,-5.2) {\small$\Upsilon^{(\mscr{Z})}_{k-1:k-2}$};
    \node[right] at (14.5,-0.5) {\small$\Upsilon^{(\mscr{Z})}_{k:k-1}$};
    \end{scope}
    \end{tikzpicture}
    \caption[Markovian process with open quantum evolution]{\textbf{The process tensor for a Markov process} has no temporal correlations and its Choi state takes a tensor product form of maps connecting adjacent time-steps. For an open quantum evolution through a dynamical map $\mscr{Z}_{j:i}(\cdot)=\tr_\mathsf{E}[\,\mc{U}_{j:i}(\cdot\otimes\varepsilon_j)]$ from time-step $i$ to $j$, it takes a form equivalent to $\Upsilon_{k:0}=\Upsilon^{(\mscr{Z})}_{k:{k-1}}\otimes\cdots\otimes\Upsilon^{(\mscr{Z})}_{1:0}\otimes\rho_\mathsf{S}^{(0)}$, where $\Upsilon^{(\mscr{Z})}_{j:i}$ is the Choi state of $\mscr{Z}_{j:i}$. The dashed lines represent correlation between the pair of spaces where $\Upsilon_{k:k-1}^{(\mscr{Z})}$ acts on.}
    \label{fig: Markovian processes}
\end{figure}
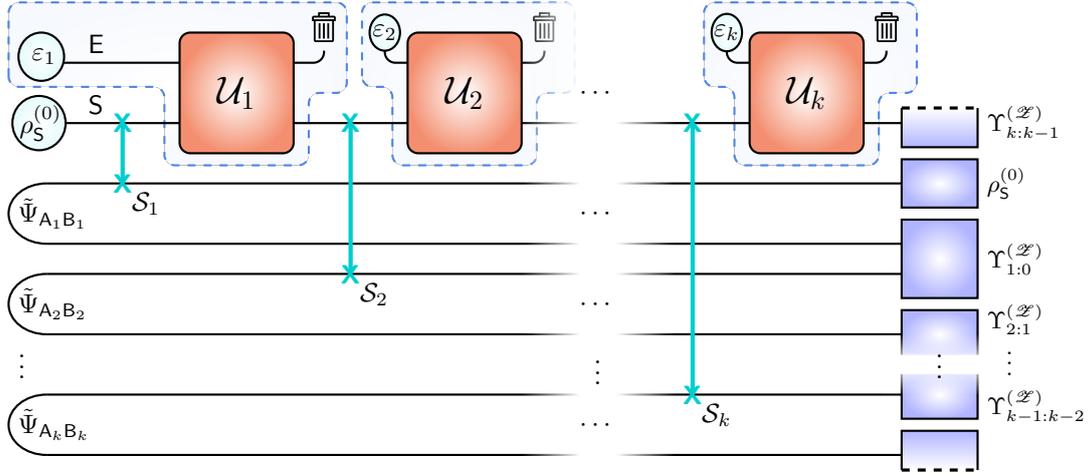

This means that all Markovian processes have the form~\cite{Pollock_2018_Markov, Milz_operational}
\begin{equation}
    \Upsilon^\markov=\mc{Z}_{k:k-1}\otimes\mc{Z}_{k-1:k-2}\otimes\cdots\mc{Z}_{1:0}\otimes\rho_\mathsf{S}^{(0)},
    \label{eq: Markov process}
\end{equation}
for dynamical \gls{CPTP} map Choi states $\mc{Z}_{j:i}$ connecting adjacent time-steps $i$ to $j$. A formal proof, showing that both Eq.~\eqref{eq: Markov process} implies Eq.~\eqref{eq: quantum Markov} and vice-versa, can be seen in Ref.~\cite{Taranto_2020}.

In particular for dynamical maps $\mscr{Z}_{j:i}(\cdot)=\tr_\mathsf{E}[\,\mc{U}_{j:i}(\cdot\otimes\varepsilon_j)]$, with unitary maps $\mc{U}_{j:i}$ dictating evolution from time $t_i$ to time $t_j$, and fiducial \gls{env} states $\varepsilon_j$, this implies that strictly Markovian dynamics in open quantum systems are those in which the environment absolutely forgets at every step, i.e. when the environment is completely discarded, or traced out, and the joint dynamics at each state has no memory whatsoever of previous \gls{env} states. We can visualize this as in Fig.~\ref{fig: Markovian processes}.

\subsection{An unambiguous measure for non-Markovianity}\label{sec: measure of non-Markovianity}
The previous discussion makes it clear that strictly Markovian processes cannot be physically realised by continuous open quantum system dynamics, or, in other words, all open quantum evolutions generated by a time-independent \gls{syst-env} Hamiltonians are  non-Markovian according to the criterion in Eq.~\eqref{eq: quantum Markov}.

Far from this feature rendering the process tensor framework useless, it provides a clear and non-ambiguous way to measure non-Markovianity. Consider any relevant distinguishability measure $\Delta$ satisfying Def.~\ref{def: distance measure}. Then we can quantify the non-Markovianity of any process $\Upsilon$ by measuring \emph{how far} it is from the closest Markovian process,
\begin{equation}
    \mc{N}_\Delta:=\min_{\Upsilon^\markov}\Delta\left(\Upsilon,\,\Upsilon^\markov\right),
    \label{eq: non-Markov measure}
\end{equation}
where the minimum is taken to represent \emph{the nearest} Markovian process.

The choice of the distance measure $\Delta$ will depend on the problem at hand, the operational meaning that one is looking for, or simply computational convenience which e.g. can then be used to place bounds on a different measure of interest. Furthermore, this choice can be relaxed to be any \emph{pseudo-distance}~\cite{Pollock_2018_Markov} respecting only condition \emph{i.} of Def.~\ref{def: distance measure} and being \gls{CP}-contractive, which means $\Delta(\Phi(\rho),\Phi(\sigma))\leq\Delta(\rho,\sigma)$ for any \gls{CP} map.

\begin{example}To exemplify, one such pseudo-distance is given by the relative entropy,
\begin{equation}
    S(\rho\,\|\,\sigma)=\tr[\rho\,(\log\rho-\log\sigma)],
\end{equation}
which is read as \emph{the relative entropy of $\rho$ with respect to $\sigma$}, which satisfies positivity but not the symmetry nor triangle inequality properties, but it is nevertheless a valid measure of the separation of a pair of states.

The relative entropy has several useful properties and has been widely studied, having a close connection as a generalization from classical entropy measures~\cite{nielsen2000quantum, bengtsson2006geometry}; in particular, the classical relative entropy of two probability distributions is related to the probability of distinguishing these after a large (finite) number of independent samples. This gives a large deviations bound known as Sanov’s theorem~\cite{Hoeffding_1994, cover2012elements}, which can be generalized to the quantum case~\cite{bengtsson2006geometry, Hiai_1991, Vedral_1997} as the probability of confusing two quantum states $\rho$, $\sigma$ after a large number of realizations $n$ as
\begin{equation}
    \mbb{P}_\text{confusion}=\exp\left[-n\,S(\rho\,\|\,\sigma)\right].
\end{equation}

This means that with the relative entropy, the measure $\mc{N}_S=\min_{\Upsilon^\markov}S(\Upsilon\|\Upsilon^\markov)$ gives an operational meaning through the probability of confusing a quantum process for a Markovian one decreasing exponentially in the number of realizations of the process, $\mbb{P}_\text{confusion}^\markov=\exp[-n\,\mc{N}_S]$.

Furthermore, the Markovian process minimizer of the relative entropy is given simply by the marginals (i.e. the reduced components) of the process tensor. We can see this as follows. First, given that we can write $\mc{N}_S=\min_{\Upsilon^\markov}\left\{-\tr[\,\Upsilon\,\log\,\Upsilon^\markov]-S(\Upsilon)\right\}$ with $S(\Upsilon)$ the von Neumann entropy of $\Upsilon$, we have
\begin{align}
    -\tr[\,\Upsilon\,\log\,\Upsilon^\markov]&=-\tr\left\{\,\Upsilon\,\log\left[\bigotimes_{\ell=1}^k\mc{Z}_{\ell:\ell-1}\otimes\rho_\mathsf{S}^{(0)}\right]\right\}\nonumber\\
    &=-\sum_{i=1}^k\tr\left\{\Upsilon\left(\log[\mc{Z}_{i:i-1}\otimes\mbb1]\right)\right\}-\tr\left\{\Upsilon\left(\log[\mbb1\otimes\rho_\mathsf{S}^{(0)}]\right)\right\}\nonumber\\
    &=-\sum_{i=1}^k\tr\left(\Upsilon_{i:i-1}\log\mc{Z}_{i:i-1}\right)-\tr\left(\Upsilon_{0}\log\rho_\mathsf{S}^{(0)}\right),
\end{align}
where in the second line the identities are in all remaining spaces and in the third line $\Upsilon_{\ell:\ell-1}=\tr_{\overline{\ell:\ell-1}}(\Upsilon)$ is the reduced state of $\Upsilon$ obtained by tracing all subsystems except the ones corresponding to dynamics from time-step $\ell-1$ to time-step $\ell$, and similarly for $\Upsilon_{\overline{0}}=\tr_{\overline{0}}(\Upsilon)$ just over the zero\textsuperscript{th} time-step; these are what we refer to as the \emph{marginals} of $\Upsilon$.

So the minimization in $\mc{N}_S$ amounts to minimizing the individual relative entropies, as $-\tr[\rho\log\sigma]=S(\rho\|\sigma)+S(\rho)$, with the minimum occurring for $\rho=\sigma$; thus we can conclude
\begin{equation}
    \Upsilon^\markov_S=\Upsilon_{k:k-1}\otimes\Upsilon_{k-1:k-2}\otimes\cdots\otimes\Upsilon_{1:0}\otimes\rho_\mathsf{S}^{(0)},
\end{equation}
is the closest Markovian process under relative entropy.
\end{example}

For other measures the minimization can potentially be non-trivial given that the Choi state is a $d_\mathsf{S}^{2k+1}$ dimensional many-body state. Despite this, bounds and different relations are usually accessible between different distinguishability measures. In the following we will see an extensive use and application of the measure of non-Markovianity in Eq.~\eqref{eq: non-Markov measure} to derive the main results of this PhD. Determining the degree of Markovianity of general quantum processes is not only a foundational task, with an ever-increasing interest and relevance in determining the breakdown of the Markovian approximation in modern experiments~\cite{Gessner2014, Ringbauer_2015, Morris2019, Winick2019}.

We now proceed to presenting these results.
    \part{Main Results}
    \chapter{Equilibration in Quantum Processes across Multiple Points in Time}
\label{sec:equilibration}

\setlength{\epigraphwidth}{0.825\textwidth}
\epigraph{\emph{We all behave like Maxwell’s demon \textup{[\,\dots]} We disturb the tendency toward equilibrium.}}{-- James Gleick (\emph{The Information: A History, a Theory, a Flood})}

As we described in Chapter~\ref{sec:statmech}, the seemingly paradoxical nature of the unitary evolution of quantum systems towards equilibrium can be resolved by means of the concept of equilibration on average. We saw that in essence, the mechanism behind equilibration is that of dephasing, as illustrated in Fig.~\ref{fig: omega dephased}. In particular, if we consider an open system \gls{syst}, we have seen in Chapter~\ref{sec:processes} that it might be the case that correlations in time or memory, manifesting in the statistics of sequential observations, maintain information about the initial perturbation. That is, put simply, it is unclear whether quantum stochastic processes equilibrate in an analogous way, i.e. whether they are most often found close to some average value.

The first problem we face when trying to draw a connection with equilibration on average on multitime processes is trying to pin down what exactly we mean by a process equilibrating, or introducing a notion of an \emph{equilibrium process} and a proper multitime \emph{timescale} on which a process would equilibrate on average within finite time-windows. This is not an easy task but here we elaborate on the notion we gave in Ref.~\cite{FigueroaRomero2020equilibration}. In our work, we derived sufficient conditions for general multitime observations to relax close to their corresponding equilibrium values at each time-step when the corresponding operations are implemented with an imperfect, \emph{fuzzy} clock, or equivalently, on a system with uniformly fluctuating energies.

While our approach to equilibration is entirely general, ultimately the idea is that there might be a similar connection between equilibration in quantum processes and the dynamical emergence of Markovianity, just as there is between equilibration on average of quantum systems and thermalization. In a sense it is expected that a Markovianization is a stronger condition (i.e. more constrained) than multitime process equilibration, just as thermalization is to system equilibration. Making this connection explicit (or for that matter making a different connection) will have to be done in the future, but as we stated before, it is important both for foundational and practical reasons knowing how and when does quantum processes dynamically Markovianize and it is mainly in this spirit that we explore how they equilibrate.

\section{Equilibration due to finite temporal resolution}
The approach of the results studied in Section~\ref{sec: fluctuations finite time} focuses on the temporal fluctuations of the expectation values of $A$ around equilibrium, in essence with a small variance relating to the expectation value of $A$ concentrating around its mean. As we saw, this gives a statistically meaningful characterization of how equilibration is dynamically achieved.

To explore equilibration on multitime processes, we took a similar approach but one focusing on procedures with a clear operational meaning. Precisely, one can picture a situation where an experimenter can implement operations with some finite temporal resolution only, i.e. where they have some uncertainty in the readings of their clock. In the words of the previous chapter, this \emph{fuzziness} can be associated to the instruments describing a set of operations at a given time in a process, with the fuzziness described by a given probability distribution. Thus we are asking how different an evolving quantum state appears from equilibrium when measured at a time that can vary in each realisation, being randomly drawn from some distribution that quantifies the fuzziness associated with finite temporal resolution.

To precisely define what we mean fuzziness, let us first define a \emph{probability density function} (\glsunset{pdf}\gls{pdf}), which will let us generalize the idea of finite temporal resolution. We now treat the time variables $t_i$, i.e. the waiting time between the $i\textsuperscript{th}$ and the $(i+1)\textsuperscript{th}$ interventions, as a non-negative real random variable, which now, however, are continuous rather than discrete (as opposed to the presentation in Section~\ref{sec: classical stochastic processes}). Let us focus on a single time $t$, analogous to the standard equilibration case, for now. This means that our sample space is the whole non-negative real line $\mbb{R}_0^+$, with events corresponding to intervals within it. In the discrete setting we employed random variables and their distributions, however the question now is how precisely do we define the distribution of a continuous random variable? A way to do this is to specify probabilities that the random variable will be within a given interval, rather than that of it taking a possible value. We can achieve this with the following.

\begin{definition}[Probability Density Function]
A probability density function (\gls{pdf}) for a real random variable $X$ is a function $\mscr{P}:\mathbb{R}\to\mathbb{R}_0^+$ such that
\begin{equation}
    \mbb{P}(a\leq{X}\leq{b}\,)=\int_a^b\mscr{P}(x)\,dx,
\end{equation}
for all $a\leq{b}$, and
\begin{equation}
    \mbb{P}(X\in\mbb{R})=\int_{-\infty}^\infty\mscr{P}(x)\,dx=1.
\end{equation}
\end{definition}

This definition of course can be generalized; in the case of finite distributions, for example, the corresponding quantity is the probability mass function.

Thus specifically by a finite temporal resolution observation we mean an observable $A$ (either on subsystem \gls{syst} or acting coarsely on \gls{syst-env}) measured after a time $t>0$ sampled from a one-parameter family of probability distributions with density function $\mscr{P}_{T}$, i.e. which is such that $\int_0^\infty\,dt\,\mscr{P}_T(t)=1$. The parameter $T$ represents the uncertainty or fuzziness of the distribution; for example, it could be associated with the variance of the distribution. With this definition, we may generalize the time-average over a time-window $T$ of a given time-dependent quantity $f$ by 
\begin{gather}
\overline{f}^{\mscr{P}_T} := \int_{0}^{\infty}dt\,\mscr{P}_T(t)\,f(t),
\end{gather}
so that the uniform average considered in Ref.~\ref{sec: fluctuations finite time} corresponds to the case $\mscr{P}_T=T^{-1}$.

Let us then consider a dynamical setup as in standard equilibration, focusing on the dynamics of a $d_\mathsf{S}$-dimensional subsystem \gls{syst} of a $d_\mathsf{S}d_\mathsf{E}$-dimensional composite system \gls{syst-env}, and refering to subsystem equilibration as the relaxation of \gls{syst} towards some steady state, while the whole \gls{syst-env} evolves unitarily via $U=\exp(-iHt)$, with a general time-independent Hamiltonian $H=\sum_{n=1}^\mathfrak{D}E_n{P}_n$. For simplicity, we denote the full \gls{syst-env} initial state as $\rho$.

Then we can define the time average of the initial state over a finite-interval of width $T$ as the fuzzy average
\begin{align}
     \overline{\rho}^{\mscr{P}_T} &:= \int_{0}^{\infty}dt\,\mscr{P}_T(t)\,\rho(t)=\int_0^\infty dt\,\mscr{P}_T(t)\,\ex^{-it(E_n-E_m)}\,P_{n}\,\rho\,P_m,
     \label{eq: fuzzy finite time state}
\end{align}
so we now need to make sense of the integral of the exponential factor. To do this, we know that in the infinite-time limit the average state \emph{should}, by definition, correspond to the uniform time-averaged state $\omega$. That is, we require
\begin{equation}
    \omega:=\lim_{T\to\infty}\overline{\rho}^{\mscr{P}_T}=\mc{D}(\rho),
\end{equation}
where $\mc{D}$ is the dephasing map with respect to $H$,
\begin{equation}
    \mc{D}(\rho)=\sum_{n=1}^\mathfrak{D}P_n\,\rho\,P_n,
    \label{eq: dephasing infinite}
\end{equation}
which we defined briefly before in Eq.~\eqref{eq: dephasing omega}. This means $\omega$ is independent of the choice of the \gls{pdf} $\mscr{P}_T$. We can further generalize this dephasing map to a finite-time one as
\begin{equation}
    \mc{G}_T(\rho):=G_{nm}^{(T)}P_n\,\rho\,P_m,\quad\text{where}\quad G_{nm}^{(T)}:=\overline{\ex^{-it(E_n-E_m)}}^{\mscr{P}_T},
    \label{eq: dephasing finite}
\end{equation}
which we similarly used briefly before only to simplify notation in Fig.~\ref{fig: omega dephased}. This means that we can make sense of the integral in Eq.~\eqref{eq: fuzzy finite time state} by requiring $\mscr{P}_T$ to be such that
\begin{equation}
    \lim_{T\to\infty}G_{nm}^{(T)}=\delta_{nm},
\end{equation}
which otherwise is entirely general.

We consider then the average distinguishability by means of an observable $A$ between the equilibrium state $\omega$, and the non-equilibrium, fuzzy state $\overline{\rho}^{\mscr{P}_T}$, which can be quantified as the difference of expectation values between these states as $\left|\left<A\right>_{\overline{\rho}^{\mscr{P}_T}-\omega}\right|=|\tr[A(\overline{\rho}^{\mscr{P}_T}-\omega)]|$. Given the finite time and fully dephasing operators, we can bound this as
\begin{equation}
    \left|\langle{A}\rangle_{\overline{\rho}^{\mscr{P}_T}-\omega}\right|=|\tr[A(\mc{G}_T-\mc{D})(\rho)]|\leq\|A\|\left\| \left(\mc{G}_T - \mc{D}\right)(\rho)\right\|_2,
\end{equation}
where we used $|\tr[X\sigma]\leq\|X\|\|\sigma\|_2$, with the Schatten norms defined as in Eq.~\eqref{eq: Schatten p-norm}, which can be seen to follow from H\"{o}lder's inequality in Eq.~\eqref{eq: Holders inequality} and the hierarchy of the Schatten norms in Eq.~\eqref{eq: Schatten hierarchy}. Now we can bound the right-hand-side as
\begin{align}
    \left\| \left(\mc{G}_T - \mc{D}\right)(\rho)\right\|_2^2 =&\tr\left|\sum_{n \neq m} {G}_{nm} P_n\rho P_m\right|^2
    =\sum_{n \neq m } |G_{nm}^{(T)}|^2\tr\left[ P_n \rho P_m\rho\right]\nonumber\\
    \leq&\max_{n\neq{m}}|G_{nm}^{(T)}|^2\sum_{n \neq m}\tr[P_n\rho P_m \rho]
    =\max_{n\neq{m}}|G_{nm}^{(T)}|^2\tr(\rho^2-\omega^2)\nonumber\\
    =&\|\rho-\omega\|_2^2 \ \max_{n\neq{m}}|G_{nm}^{(T)}|^2, \label{eq:singlestep multitime equilibration}
\end{align}
where in the second equality we implicitly used $P_nP_{n'}=\delta_{nn'}$, and in the last line we used $\tr(\rho^2-\omega^2) = \|\rho-\omega\|_2^2$, given that $\tr(\rho\,\omega) = \tr(\omega^2)$.

Thus it follows that
\begin{equation}
    \left|\left<A\right>_{\overline{\rho}^{\mscr{P}_T}-\omega}\right| \leq \mscr{S}_T\, \|A\|\|\rho-\omega\|_2,\quad\text{where}\quad\mscr{S}_T:=\max_{n\neq{m}}|G_{nm}^{(T)}|,
    \label{eq: main standard case multitime}
\end{equation}
with the rate of convergence to zero determined by $\mscr{S}_T$, which essentially tell us the off-diagonal term that will die off at the slowest rate on average with respect to the fuzzy clock.

In particular, when the fuzziness $T$ corresponds to that of the uniform distribution over an interval of width $T$ as we described in Section~\ref{sec: fluctuations finite time}, the \gls{pdf} is $\mscr{P}_T=T^{-1}$ and we get $|G_{nm}^{(T)}|=|\mathrm{sin}(T\mc{E}_{nm})/T\mc{E}_{nm}|$ with $\mc{E}_{nm} := (E_n-E_m)/2$. The bound in Eq.~\eqref{eq: main standard case multitime} then tells us that the evolved state $\rho(t)$ will differ from the equilibrium $\omega$ when measured at a given time with a temporal-resolution $T$ at most with proportion $|T\mc{E}_{nm}|^{-1}$ for the smallest energy gap $\mc{E}_{nm}$, with a scale set by the size of the observable $A$ and how different the initial state $\rho$ is from the equilibrium $\omega$. Notice as well that in general $\|\omega\|_2^2\leq{d}_\text{eff}^{-1}(\rho)$, where the inverse effective dimension is defined in Eq.~\eqref{eq: deff inverse}, with equality both for pure $\rho$ or when the Hamiltonian is non-degenerate; both quantities relate to how spread the initial state $\rho$ is in the energy eigenbasis.

This initial fuzziness can be interpreted as the observer not knowing exactly when the process actually began. However, one question we can ask is whether we are able to overcome the fuzziness of the initial interval by making a sequence of measurements. As we described in Section~\ref{sec:processes}, these operations can correspond to any possible experimental intervention, which can furthermore be correlated with each other through an ancillary system. In this case temporal correlations within the dynamics itself can propagate through the environment and similarly the disturbance introduced by the experimental operations might become relevant. Moreover, any fuzziness in the subsequent measurements also has to be accounted for.

\section{Multitime equilibration due to finite temporal resolution}
Of course one might not stop at a single observation but continue gathering data to assess how close the system remains to equilibrium with respect to a set of possible operations, $\{\mc{A}_i\}$, given by weighted \gls{CPTNI} maps $\mathcal{A}_\ell(\cdot) := \sum_\mu a_{\ell_\mu} K_{\ell_\mu}(\cdot)K_{\ell_\mu}^\dg$, with $\sum_\mu K_{\ell_\mu}^\dg\,K_{\ell_\mu}\leq \mbb{1}$ and $a_{\ell_\mu}\in\mbb{R}$ being the corresponding outcome weights. More specifically, the scenario is the following: an initial state $\rho$ in the full \gls{syst-env} composite evolves unitarily through a time-independent Hamiltonian dynamics until, at time $t_0$, an operation $\mc{A}_0$ is performed jointly on \gls{syst} along with an ancilla $\mathsf{\Gamma}$, which is initially uncorrelated in state $\gamma$. We now denote the full initial state by
\begin{equation}
    \varrho:=\rho\otimes\gamma,
\end{equation}
and after the first operation, \gls{syst-env} evolve unitarily again for a time $t_1$ until another operation $\mc{A}_1$ is made on \gls{syst-ancilla}, and so on for $k$ time-steps. The unitary evolution at each step is given by the map \begin{equation}
    \mc{U}_\ell(\cdot)=\exp\{-iH_\ell{t}_\ell\}(\cdot)\exp\{iH_\ell{t}_\ell\},
\end{equation}
acting on \gls{syst-env}. The time-independent Hamiltonians $H_\ell$ can in general be different at each step. The ancillary space $\mathsf{\Gamma}$ can be interpreted as a quantum memory device, and might carry information about previous interactions with the system. As done in Chapter~\ref{sec: process tensor}, we denote the Choi state of the operations $\{\mc{A}_i\}$ by $\Lambda$ and the underlying dynamical process by $\Upsilon$ (with the number of time-steps implicit).

The joint multitime expectation of these set of operations is thus given by
\begin{equation}
    \langle\Lambda\rangle_\Upsilon=\tr[\mc{A}_k\,\mc{U}_k\cdots\mc{A}_0\,\mc{U}_0(\varrho)],
\end{equation}
where implicitly $\mc{A}_i$'s act only on \gls{syst-ancilla}, while the unitaries $\mc{U}_i$ act on \gls{syst-env}. For simplicity, we consider a fixed set of projectors $\{P_n\}$ for all Hamiltonians such that $H_\ell=\sum P_nE_{n_\ell}$ at each step $\ell$, with $P_n$ projecting onto the energy eigenspaces of $H_i$ with energy $E_{n_i}$. Also, we denote simply by $\cdot$ the composition of superoperators when clear by context.

\begin{figure}[t]
\centering
    \begin{tikzpicture}[scale=1.85]
    \begin{scope}
    \end{scope}
    \begin{scope}[local bounding box=scope1]
        \draw[|-, very thick] (0.75,0)  -- (1.5,0);
        \draw[decoration={brace,mirror,raise=0.45em}, decorate] (0.77,0)  -- (1.48,0);
        \draw[|-, very thick] (1.5,0) -- (3,0);
        \draw[decoration={brace,mirror,raise=0.45em}, decorate] (1.52,0)  -- (2.98,0);
        \draw[|-, very thick] (3,0) -- (4,0);
        \draw[-, very thick] (4.5,0) -- (5.5,0);
        \draw[decoration={brace,mirror,raise=0.45em}, decorate] (4.35,0)  -- (5.48,0);
        \draw[|-, very thick] (5.5,0) -- (7.5,0);
        \draw[decoration={brace,mirror,raise=0.45em}, decorate] (5.52,0)  -- (7.48,0);
        \draw[|->, very thick] (7.5,0) -- (8.5,0);
        \node at (4.25,0) {\ldots};
        \node[C1!80!black] at (4.25,0.7) {\ldots};
        \node[C2] at (4.25,0.35) {\ldots};
         \node[below] at (1.125,-0.2) {\Large$t_0$};
        \draw[dashed, C2!90!white, very thick] (1,0) -- (1,1);
        \draw[dashed, C2!90!white, very thick] (2,0) -- (2,0.425);
        \node[below] at (2.25,-0.2) {\Large$t_1$};
        \draw[dashed, C2!90!white, very thick] (2.5,0) -- (2.5,0.55);
        \draw[dashed, C2!90!white, very thick] (3.5,0) -- (3.5,0.55);
        \node[below] at (5,-0.2) {\Large$t_{k-1}$};
        \draw[dashed, C2!90!white, very thick] (4.75,0) -- (4.75,0.975);
        \draw[dashed, C2!90!white, very thick] (6.25,0) -- (6.25,0.435);
        \node[below] at (6.5,-0.2) {\Large$t_{k}$};
        \draw[dashed, C2!90!white, very thick] (7.2,0) -- (7.2,0.6);
        \draw[dashed, C2!90!white, very thick] (7.8,0) -- (7.8,0.6);
        \draw[C2, thick,|<->|] (1,0.4) -- node[fill=white]{\large$T_0$} (2,0.4);
        \draw[C2, thick,|<->|] (2.5,0.4) -- node[fill=white]{\large$T_1$} (3.5,0.4);
        \draw[C2, thick,|<->|] (4.75,0.4) -- node[fill=white]{\large$T_{k-1}$} (6.25,0.4);
        \draw[C2, thick, |<->|] (7.2,0.4) -- node[fill=white]{\large$T_k$} (7.8,0.4);
        \node[right] at (8.5,0) {$\mathbf{t}$};
        \node[right] at (1,1.1) {\color{C1!80!black}\Large$\boldsymbol{\mathscr{P}}_{T_0}(t_0)$};
        \node[right] at (3,1.1) {\color{C1!80!black}\Large$\boldsymbol{\mathscr{P}}_{T_1}(t_1)$};
        \node[right] at (5,1.12) {\color{C1!80!black}\Large$\boldsymbol{\mathscr{P}}_{T_{k-1}}(t_{k-1})$};
        \node[right] at (7.6,1.1) {\color{C1!80!black}\Large$\boldsymbol{\mathscr{P}}_{T_k}(t_k)$};
    \end{scope}
    \begin{scope}[xscale = 0.225,
        yscale=0.14, shift={($(scope1.west)+(-0.15,0)$)},
        declare function={gamma(\z)=
    (2.506628274631*sqrt(1/\z) + 0.20888568*(1/\z)^(1.5) + 0.00870357*(1/\z)^(2.5) - (174.2106599*(1/\z)^(3.5))/25920 - (715.6423511*(1/\z)^(4.5))/1244160)*exp((-ln(1/\z)-1)*\z);},
    declare function={gammapdf(\x,\k,\theta) = \x^(\k-1)*exp(-\x/\theta) / (\theta^\k*gamma(\k));}]
        \begin{axis}[every axis plot/.append style={
        mark=none,samples=50,smooth,line width=7pt},
        axis line style={draw=none},
        ticks = none]
        \addplot[domain=0:5, C1!80!black] {gammapdf(x,1,2)};
        \end{axis}
    \end{scope}
    \begin{scope}[xscale = 0.3,
        yscale=0.125, shift={($(scope1.west)+(4.7,0)$)},
        declare function={gamma(\z)=
    (2.506628274631*sqrt(1/\z) + 0.20888568*(1/\z)^(1.5) + 0.00870357*(1/\z)^(2.5) - (174.2106599*(1/\z)^(3.5))/25920 - (715.6423511*(1/\z)^(4.5))/1244160)*exp((-ln(1/\z)-1)*\z);},
    declare function={gammapdf(\x,\k,\theta) = \x^(\k-1)*exp(-\x/\theta) / (\theta^\k*gamma(\k));}]
        \begin{axis}[every axis plot post/.append style={
        mark=none,samples=50,smooth,line width=5pt},
        axis line style={draw=none},
        ticks = none]
        \addplot[domain=1:10, C1!80!black] {gammapdf(x,9,0.5)};
        \end{axis}
    \end{scope}
    \begin{scope}[xscale = 0.3,
        yscale=0.14, shift={($(scope1.west)+(12.5,0)$)},
        declare function={gamma(\z)=
    (2.506628274631*sqrt(1/\z) + 0.20888568*(1/\z)^(1.5) + 0.00870357*(1/\z)^(2.5) - (174.2106599*(1/\z)^(3.5))/25920 - (715.6423511*(1/\z)^(4.5))/1244160)*exp((-ln(1/\z)-1)*\z);},
    declare function={gammapdf(\x,\k,\theta) = \x^(\k-1)*exp(-\x/\theta) / (\theta^\k*gamma(\k));}]
        \begin{axis}[every axis plot post/.append style={
        mark=none,samples=50,smooth,line width=5pt},
        axis line style={draw=none},
        ticks = none]
        \addplot[domain=0:15, C1!80!black] {gammapdf(x,2,2)};
        \end{axis}
    \end{scope}
    \begin{scope}[xscale = 0.3,
        yscale=0.14, shift={($(scope1.west)+(19.1,0)$)}]
        \begin{axis}[every axis plot post/.append style={
        mark=none,domain=-5:5,samples=50,smooth,line width=5pt},
        axis line style={draw=none},
        ticks = none]
        \addplot[C1!80!black] {gauss(0,1)};
        \end{axis}
    \end{scope}
    \end{tikzpicture}
    \caption[Multitime fuzzy clocks]{\textbf{Finite temporal resolution in a quantum process:} between interventions, each Hamiltonian evolution is time-averaged over the waiting times between interventions $t_0,\,t_1,\cdots,\,t_k$ with corresponding average waiting times $\tau_0,\,\tau_1,\cdots,\,\tau_k$, over a probability distribution with \gls{pdf} given by $\mscr{P}_{T_i}$, with $T_i$ having a suitable uncertainty parameter role.}
    \label{fig: fuzziness}
\end{figure}
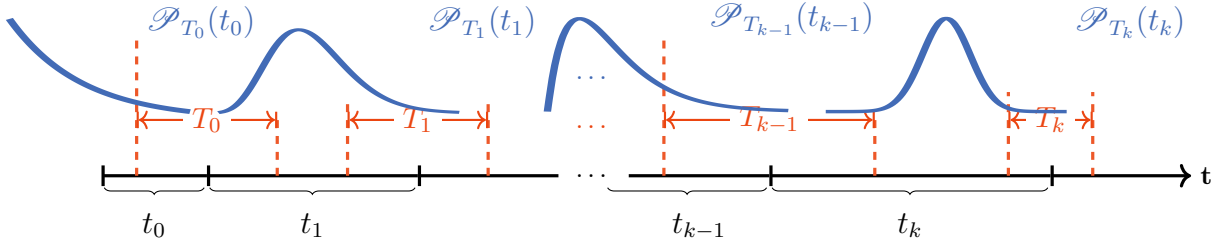

We denote the time intervals of the free evolutions as $t_j$, which is preceded by the $j\textsuperscript{th}$ measurement and followed by $(j+1)\textsuperscript{th}$ measurement. In other words, $t_j$ is the waiting time between $j\textsuperscript{th}$ and $(j+1)\textsuperscript{th}$ measurements. To avoid cluttering of notation we will denote the multitime probability distribution as $\mscr{P}_\mathbf{T}(\mathbf{t})=\prod_{i=0}^k\mscr{P}_{T_i}(t_i)$, where now we use the labels
\begin{gather}
\mathbf{t} := (t_0, t_1, \dots, t_k) \quad \mbox{and} \quad
\mathbf{T} = (T_0,T_1,\ldots,T_k)
\end{gather}
for waiting times and the fuzziness parameters for each time interval, respectively. We can now define a finite temporal resolution process as
\begin{align}
    \overline{\Upsilon}^{\mscr{P}_\mathbf{T}}:= \int_0^\infty dt_k \cdots \int_0^\infty dt_1 \int_0^\infty dt_0 \,\mscr{P}_\mathbf{T}(\mathbf{t})\,\Upsilon,
\end{align}
and we will denote the average waiting time between each pair of measurements by
\begin{equation}
    \tau_i := \int_{0}^{\infty} dt_i \ t_i \ \mscr{P}_{T_i}(t_i),
\end{equation}
and we can visualize these concepts pictorially as in Fig.~\ref{fig: fuzziness}.

We are interested in quantifying how different this out-of-equilibrium process, where time intervals are fuzzy, looks from an equilibrium process. To define the equilibrium process we follow the lead of earlier results, i.e. the initial state relaxes to the equilibrium state
\begin{equation}
    \varpi:=\varpi_0=\lim_{T_0\to\infty}\int_0^\infty\mc{U}_0(\varrho)\,\mscr{P}_{T_0}(t_0)\,dt_0=\mc{D}(\varrho)= \omega\otimes\gamma,
    \label{eq: def varpi}
\end{equation}
until an operation $\mc{A}_0$ is performed, and subsequently the system relaxes again to an equilibrium state $\varpi_1=\mc{D}\mc{A}_0(\varpi)$ until an operation $\mc{A}_1$ is made, and so on for $k$-time-steps. Then we can define
\begin{align}
    \varpi_i&:=\lim_{T_i\to\infty}\int_0^\infty\mc{U}_i\mc{A}_{i-1}(\varpi_{i-1})\,\mscr{P}_{T_i}(t_i)\,dt_i=\mc{D}\mc{A}_{i-1}\cdots\mc{A}_0\mc{D}(\varrho),
    \label{eq: varpi definition}
\end{align}
for any $i=0,\cdots,k$ as the equilibrium states after each intervention up to $\mc{A}_{i-1}$.

\begin{figure}[t]
\centering
\begin{tikzpicture}
  \begin{scope}[scale=0.43]
   \fill[-, dashed, thick, outer color=C4!10!white, inner color=white, rounded corners, draw=C4] (-5.5,0) -- (-5.5,2.5) -- (25.25,2.5) -- (25.25,0.5) -- (22.5,0.5) -- (22.5,-2) -- (19.5,-2) -- (19.5,0.5) -- (16.5,0.5) -- (16.5,-2) -- (10.5,-2) -- (10.5,0.5) -- (7.5,0.5) -- (7.5,-2) -- (4.5,-2) -- (4.5,0.5) -- (1.5,0.5) -- (1.5,-2.5) -- (-5.5,-2.5) -- (-5.5,0);
    \fill[-, dashed, thick, outer color=C3!10!white, inner color=white, rounded corners, draw=C3] (1.75,0) -- (1.75,-4.5) -- (27.5,-4.5)  -- (27.5,-2) -- (25.25,-2) -- (25.25,0.25) -- (22.75,0.25) -- (22.75,-2.25) -- (19.25,-2.25) -- (19.25,0.25) -- (16.75,0.25) -- (16.75,-2.25) -- (10.25,-2.25) -- (10.25,0.25) -- (7.75,0.25) -- (7.75,-2.25) -- (4.25,-2.25) -- (4.25,0.25) -- (1.75,0.25) -- (1.75,0) ;
    \node[C4!60!black] at (26.75,1.7) {\Large$\overline{\Upsilon}^{\mathbf{T}}$};
    \node[C3!70!black,above] at (26.5,-2) {\Large$\Lambda$};
    \draw[thick, -] (-3,1.25) -- (23.5,1.25);
    \draw[thick,-] (-3,-1) -- (23.5,-1);
    \shade[inner color=white, outer color=C3!10!white, draw=black, thick] (-4,0) ellipse (1.2cm and 2.2cm);
    \node at (-4,0) {\Large$\rho$};
    \node[above] at (-2.3,1.25) {$\mathsf{E}$};
    \node[above] at (-2.3,-1.1) {$\mathsf{S}$};
    \path[outer color=C4!5!white, inner color=C2!40!white, draw=black, thick, rounded corners] (-1.5,-1.55) rectangle (1,2);
    \node at (-0.25,0.3) {\Large$\mc{G}_{T_0}$};
    \shade[outer color=C3!60!white, inner color=white, draw=black, rounded corners, thick] (2,-4) rectangle (4,0);
    \node at (3,-2) {\Large$\mc{A}_0$};
    \draw[-, thick] (4,-3) -- (26,-3);
    \node[below] at (4.5,-3) {$\mathsf{\Gamma}$};
    \path[outer color=C4!5!white, inner color=C2!40!white, draw=black, thick, rounded corners] (4.75,-1.55) rectangle (7.25,2);
    \node at (6,0.3) {\Large$\mc{G}_{T_1}$};
    \shade[outer color=C3!60!white, inner color=white, draw=black, rounded corners, thick] (8,-4) rectangle (10,0);
    \node at (9,-2) {\Large$\mc{A}_1$};
    \path[outer color=C4!5!white, inner color=C2!40!white, draw=black, thick, rounded corners] (10.75,-1.55) rectangle (13.25,2);
    \node at (12,0.3) {\Large$\mc{G}_{T_2}$};
    \shade[outer color=C3!60!white, inner color=white, draw=black, rounded corners, thick] (17,-4) rectangle (19,0);
    \node[rotate=65] at (17.9,-2) {\Large$\mc{A}_{k-1}$};
    \path[outer color=C4!5!white, inner color=C2!40!white, draw=black, thick, rounded corners] (19.75,-1.55) rectangle (22.25,2);
    \node at (21,0.3) {\Large$\mc{G}_{T_k}$};
    \node[right] at (23.5,1.25) {\trash};
    \shade[outer color=C3!60!white, inner color=white, draw=black, rounded corners, thick] (23,-4) rectangle (25,0);
    \draw[white, fill=white, path fading= west] (13.5,2.75) -- (13.5,-4.75) -- (14,-4.75) -- (14,2.75) ;
    \draw[white, fill=white] (14,2.75) -- (16,2.75) -- (16,-4.75) -- (14,-4.75);
    \draw[white, fill=white, path fading= east] (16.5,2.75) -- (16,2.75) -- (16,-4.75) -- (16.5,-4.75);
    \node at (15,0) {$\cdots$};
    \node at (24,-2) {\Large$\mc{A}_k$};
    \node[right] at (26,-3) {\trash};
    \end{scope}
    
    \vspace{0.1in}
    
    \begin{scope}[shift = {(0,-4)}, scale=0.43]
    
   \fill[-, dashed, thick, outer color=C4!10!white, inner color=white, rounded corners, draw=C4] (-5.5,-4) -- (-5.5,2.5) -- (25.25,2.5) -- (25.25,0.5) -- (22.5,0.5) -- (22.5,-2) -- (19.5,-2) -- (19.5,0.5) -- (16.5,0.5) -- (16.5,-2) -- (10.5,-2) -- (10.5,0.5) -- (7.5,0.5) -- (7.5,-2) -- (4.5,-2) -- (4.5,0.5) -- (1.5,0.5) -- (1.5,-4) -- cycle;
    \fill[-, dashed, thick, outer color=C3!10!white, inner color=white, rounded corners, draw=C3] (1.75,0) -- (1.75,-4.5) -- (27.5,-4.5)  -- (27.5,-2) -- (25.25,-2) -- (25.25,0.25) -- (22.75,0.25) -- (22.75,-2.25) -- (19.25,-2.25) -- (19.25,0.25) -- (16.75,0.25) -- (16.75,-2.25) -- (10.25,-2.25) -- (10.25,0.25) -- (7.75,0.25) -- (7.75,-2.25) -- (4.25,-2.25) -- (4.25,0.25) -- (1.75,0.25) -- (1.75,0) ;
    \node[C4!60!black] at (26.5,1.7) {\Large$\Omega$};
    \node[C3!70!black,above] at (26.5,-2) {\Large$\Lambda$};
    \draw[thick, -] (-3,1.25) -- (23.5,1.25);
    \draw[thick,-] (-3,-1) -- (23.5,-1);
    \shade[inner color=white, outer color=C3!10!white, draw=black, thick] (-4,0) ellipse (1.2cm and 2.2cm);
    \node at (-4,0) {\Large$\rho$};
    \node[above] at (-2.3,1.25) {$\mathsf{E}$};
    \node[above] at (-2.3,-1.1) {$\mathsf{S}$};
    \fill[outer color=C2!60!white, inner color=white, draw=black, thick, rounded corners] (-1.5,-1.55) rectangle (1,2);
    \node at (-0.25,0.3) {\Large$\mc{D}$};
    \shade[outer color=C3!60!white, inner color=white, draw=black, rounded corners, thick] (2,-4) rectangle (4,0);
    \node at (3,-2) {\Large$\mc{A}_0$};
    \draw[-, thick] (4,-3) -- (26,-3);
    \node[below] at (4.5,-3) {$\mathsf{\Gamma}$};
    \fill[outer color=C2!60!white, inner color=white, draw=black, thick, rounded corners] (4.75,-1.55) rectangle (7.25,2);
    \node at (6,0.3) {\Large$\mc{D}$};
    \shade[outer color=C3!60!white, inner color=white, draw=black, rounded corners, thick] (8,-4) rectangle (10,0);
    \node at (9,-2) {\Large$\mc{A}_1$};
    \fill[outer color=C2!60!white, inner color=white, draw=black, thick, rounded corners] (10.75,-1.55) rectangle (13.25,2);
    \node at (12,0.3) {\Large$\mc{D}$};
    \shade[outer color=C3!60!white, inner color=white, draw=black, rounded corners, thick] (17,-4) rectangle (19,0);
    \node[rotate=65] at (17.9,-2) {\Large$\mc{A}_{k-1}$};
   \fill[outer color=C2!60!white, inner color=white, draw=black, thick, rounded corners] (19.75,-1.55) rectangle (22.25,2);
    \node at (21,0.3) {\Large$\mc{D}$};
    \node[right] at (23.5,1.25) {\trash};
    \shade[outer color=C3!60!white, inner color=white, draw=black, rounded corners, thick] (23,-4) rectangle (25,0);
    \node at (24,-2) {\Large$\mc{A}_k$};
    \node[right] at (26,-3) {\trash};
    \draw[white, fill=white, path fading= west] (13.5,2.75) -- (13.5,-4.75) -- (14,-4.75) -- (14,2.75) ;
    \draw[white, fill=white] (14,2.75) -- (16,2.75) -- (16,-4.75) -- (14,-4.75);
    \draw[white, fill=white, path fading= east] (16.5,2.75) -- (16,2.75) -- (16,-4.75) -- (16.5,-4.75);
    \node at (15,0) {$\cdots$};
    \node[rotate=-90] at (12,3.75) {\Large$\to$};
    \node[right] at (12.5,3.75) {$(\mathbf{T}\to\infty)$};
    \draw[thick, decoration={brace,mirror}, decorate] (-5,-2.5)  -- (1,-2.5);
    \node[below] at (-2,-2.8) {\Large$\omega$};
    \end{scope}
    \end{tikzpicture}
\caption[Equilibrium process and fuzzy multitime equilibration]{\textbf{Equilibration of quantum processes by finite temporal resolution:} This refers to a $k$-step process $\overline{\Upsilon}^{\mathbf{T}}$ time-averaged over each Hamiltonian evolution $(\mc{G})$ within time-windows of width $\mathbf{T}=(T_0,T_1,\ldots,T_k)$ remaining close to an equilibrium process $\Omega$. We define an equilibrium process as one dephased $(\mc{D})$ with respect to the corresponding Hamiltonian at each time step. Equilibration is determined according to a set of operations $\{\mc{A}_i\}$ on a subsystem $\mathsf{S}$, which can be correlated in time through an ancillary space $\mathsf{\Gamma}$, and is represented by a single tensor $\Lambda$. By definition, equality is attained in the limit of all $\mathbf{T}\to\infty$.}
\label{Fig: multitime equilibration}
\end{figure}
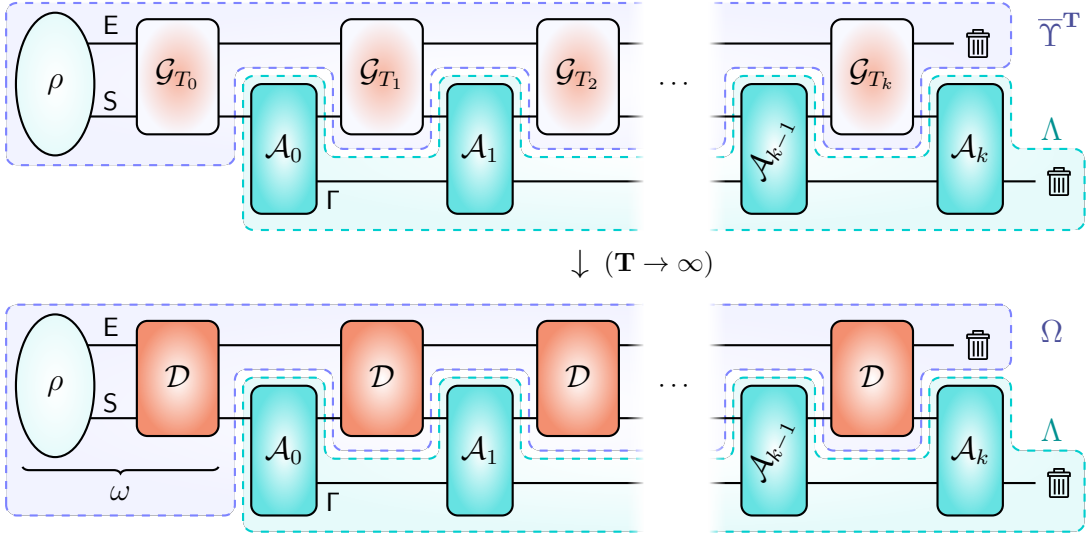

This is a sensible definition for the intermediate equilibrium states, which, however, is dependent on each operation $\mc{A}_j$. We can, however, define the equilibrium quantum process independently of the operations as
\begin{gather}
    \Omega:=\lim_{\mathbf{T}\to\infty} \overline{\Upsilon}^{\mscr{P}_\mathbf{T}},
    \label{eq: def Omega}
\end{gather}
which is depicted in Fig.~\ref{Fig: multitime equilibration} as a set of dephasing maps $\mc{D}$ at each timestep. This means then that we can write
\begin{align}
    \langle\Lambda\rangle_\Omega&=\tr[\Lambda\Omega]=\tr[\mc{A}_k\varpi_k]=\tr[\mc{A}_k\mc{D}\mc{A}_{k-1}\mc{D}\cdots\mc{A}_0\mc{D}(\varrho)],
    \label{eq: equilibrium expectation}
\end{align}
for the expectation of a sequence of operations $\{\mc{A}_i\}$ on the equilibrium process $\Omega$.

Since we can also express each finite averaging in the energy eigenbasis using the partial dephasing maps $\mc{G}_{T_i}$, defined in Eq.~\eqref{eq: dephasing finite}, we can similarly write $\langle\Lambda\rangle_{\overline{\Upsilon}^{\mscr{P}_\mathbf{T}}}=\tr[\mc{A}_k\varrho_k]$, where we now define
\begin{equation}
    \varrho_i := \mc{G}_{T_i}\mc{A}_{i-1} \cdots \mc{A}_0 \mc{G}_{T_0}(\varrho),\quad\text{for}\quad i=0,1,\cdots,k,
    \label{eq: varrho definition}
\end{equation}
as intermediate finite-time-averaged states after each intervention up to $\mc{A}_{i-1}$. As by definition $\lim_{T_i\to\infty}\mc{G}_{T_i}=\mc{D}$, the infinite-time limits $\mathbf{T}\to\infty$ make $\Upsilon$ indistinguishable from $\Omega$. We also depict this in Fig.~\ref{Fig: multitime equilibration}.

We may now generalize the left hand side of Eq.~\eqref{eq: main standard case multitime} with 
$|\langle\Lambda\rangle_{\overline{\Upsilon}^{\mathbf{T}}-\Omega}|$, asking how different the statistics of a set of operations $\{\mc{A}_i\}$ can be on a fuzzy clock process, $\overline{\Upsilon}^{\mathbf{T}}$, as opposed to those in the equilibrium one $\Omega$. For clarity, let us present the case $H_i=H=\sum E_nP_n$, i.e. with a fixed Hamiltonian for all time-steps, and with $T_i=T_j=T$, i.e. with a clock with the same amount of fuzziness at all steps, as in Ref.~\cite{FigueroaRomero2020equilibration}; we will then show how the general case reduces to this particular one.

\begin{theorem}[Multitime equilibration due to finite temporal resolution~\cite{FigueroaRomero2020equilibration}]\label{Thm: Multitime equilibration}
Given an environment-system-ancilla $\mathsf{(SE\Gamma)}$ composite with initial state $\varrho=\rho\otimes\gamma$ and initial equilibrium state $\varpi=\omega\otimes\gamma$, for any $k$-step process $\Upsilon$ with an evolution generated by a time-independent Hamiltonian on \gls{syst-env} at each step, and for any fuzzy multitime observable $\Lambda$ corresponding to a sequence of temporally local operations $\{\mc{A}_i\}_{i=0}^k$, each with fuzziness $T$ acting on the joint \gls{syst-ancilla} system,
\begin{align}
\left|\langle\Lambda\rangle_{\overline{\Upsilon}^{\mscr{P}_\mathbf{T}}-\Omega}\right|\leq \mbb{A}_k + \sum_{\ell=0}^{k-1}\|\mc{A}_{k:\ell+1}\|\left(\mbb{B}_\ell + \mbb{C}_\ell\right)
\quad\text{with}\quad
\mbb{A}_k:=\mscr{S}_T^{k+1}\|\mc{A}_{k:0}\|\,\|\varrho-\varpi\|_2,
    \label{eq: result main multi equilibration}
\end{align}
where here $\mscr{S}_T:=\max_{n\neq{m}}|G_{nm}^{(T)}|$ and $\mc{A}_{j:i}:=\mc{A}_j\cdots\mc{A}_i$ is a composition of operations; the norm $\|\cdot\|$ here stands for the norm on superoperators induced by the 2-norm, $\|\mc{X}\|=\sup_{\|\sigma\|_2=1}\|\mc{X}(\sigma)\|_2$; the first term is a single-time equilibration contribution, whereas the second term contains $k$ multitime contributions where
\begin{equation}
    \mbb{B}_\ell:=\|[\mc{G}_T^{k-\ell}-\mc{D},\mc{A}_\ell]\varrho_\ell\|_2, \qquad \mbb{C}_\ell:=\|[\mc{D},\mc{A}_\ell](\varrho_\ell-\varpi_\ell)\|_2,
    \label{eq: multitime bound corrections}
\end{equation}
with $\varrho_i$ and $\varpi_i$ intermediate finite-time averaged and equilibrium states at step $i$ defined in Eq.~\eqref{eq: varrho definition} and Eq.~\eqref{eq: varpi definition}, and where $[\cdot,\cdot]$ denotes a commutator of superoperators. 
\end{theorem}

\begin{proof}
The main idea to bound the difference $\left|\langle\Lambda\rangle_{\overline{\Upsilon}^{\mscr{P}_\mathbf{T}}-\Omega}\right|$ is to rearrange terms by commutation to obtain a single-time contribution with added correction terms. Let us slightly simplify notation by writing $\mc{G}_i=\mc{G}_{T_i}$ and only labelling $\mc{D}$ whenever it is relevant to know at which step this map is acting on. Let us first get this idea through with the $k=1$ case,
\begin{align}
    &\left|\langle\Lambda\rangle_{\overline{\Upsilon}^{\mscr{P}_\mathbf{T}}-\Omega}\right|=\left|\tr\left[\left(\bigcircop_{j=0}^1 \mathcal{A}_j  \mc{G}_j - \bigcircop_{j=0}^1 \mathcal{A}_j \mc{D}_j\right)(\varrho)\right]\right|\nonumber \\
    &= \Bigg|\tr\left[\mc{A}_1\mc{A}_0 \left(\mc{G}_1\mc{G}_0 - \mc{D}_1\mc{D}\right)(\varrho)\right] +\tr\Big[\mathcal{A}_1\Big([\mc{G}_1,\mc{A}_0]\mc{G}_0 - [\mc{D}_{1}, \mc{A}_0]\mc{D}\Big)(\varrho)\Big]\Bigg|\nonumber\\
    &\leq\left|\tr\left[\mc{A}_{1:0} \left(\mc{G}_{1:0} - \mc{D}\right)(\varrho)\right]\right| + \left|\tr\left\{\mathcal{A}_{1} [\mc{G}_{1} - \mc{D}, \mc{A}_0] \mc{G}_0(\varrho)\right\}\right| + \left|\tr\left\{\mathcal{A}_{1} [\mc{D},\mc{A}_0]  (\mc{G}_0-\mc{D})(\varrho)\right\}\right|,
\end{align}
where the third line follows by the triangle inequality ($|a-c|\leq|a-b|+|b-c|$, here with $b=\tr\{\mc{A}_1[\mc{D},\mc{A}_0]\mc{G}_0(\varrho)\}$). We now adopt the notation $\mc{A}_{j:i}:=\mc{A}_i\circ\cdots\circ\mc{A}_j$ and from now on ommit the $\circ$ symbol to denote composition of contiguous superoperators.

Then similarly we can do this for an arbitrary number of time-steps $k$,
\begin{align}
    &\left|\langle\Lambda\rangle_{\overline{\Upsilon}^{\mscr{P}_\mathbf{T}}-\Omega}\right|=\left|\tr\left[\left(\bigcircop_{j=0}^{k} \mathcal{A}_j \mc{G}_j - \bigcircop_{j=0}^{k} \mathcal{A}_j \mc{D}_j\right)(\varrho)\right]\right|\nonumber \\
    &=\Bigg|\tr\left[\mc{A}_{k:0} \left(\mc{G}_{k:0} - \mc{D}\right)(\varrho)\right] + \sum_{\ell=0}^{k-1}\tr\Big[\mathcal{A}_{k:\ell+1} \Big([\mc{G}_{k:\ell+1}, \mc{A}_\ell]  \mc{G}_\ell \bigcircop_{j=0}^{\ell-1} \mathcal{A}_j \mc{G}_j- [\mc{D}, \mc{A}_\ell]\mc{D}_\ell \bigcircop_{j=0}^{\ell-1} \mathcal{A}_j \mc{D}_j \Big)(\varrho)\Big]\Bigg|\nonumber\\
    &\leq \left|\tr\left[\mc{A}_{k:0} \left(\mc{G}_{k:0}\!- \mc{D}\right)(\varrho)\right]\right| + \!\sum_{\ell=0}^{k-1}\left|\tr\left[\mathcal{A}_{k:\ell+1} [\mc{G}_{k:\ell+1}\!- \mc{D}, \mc{A}_\ell] (\varrho_\ell)\right]\right| + \!\sum_{\ell=0}^{k-1}\left|\tr\left[\mathcal{A}_{k:\ell+1} [\mc{D}, \mc{A}_\ell]  (\varrho_\ell-\varpi_\ell)\right]\right|,
\end{align}
where $\varrho_\ell:=\mc{G}_\ell\bigcircop_{j=0}^{\ell-1}\mc{A}_j \mc{G}_j(\varrho)$. Using $|\tr \mc{X}(\varrho)|\leq \|\mc{X}\|\|\varrho\|_2$, where for simplicity, here $\|\cdot\|$ stands for the induced 2-norm defined as $\|\mc{X}\|:= \sup_{\|\sigma\|_2=1} \|\mc{X}(\sigma)\|_2$, which is a possible generalization of the operator norm for superoperators.\footnote{ We can write the action of any bounded linear map $\mc{X}$ (not necessarily \gls{CP}) as $\mc{X}(\cdot)=\sum_\mu{L}_\mu(\cdot)R_\mu^\dg$ as a generalization of the Kraus representation in Section~\ref{sec: quantum states and measurements}~\cite{Milz_operational}. The inequality can then be seen to follow with H\"{o}lder's inequality and the hierarchy of Schatten-norms. We use 2-norm for convenience in calculation.} Then we further can bound
\begin{align}
&\left|\langle\Lambda\rangle_{\overline{\Upsilon}^{\mscr{P}_\mathbf{T}}-\Omega}\right|\leq \left\|\mc{A}_{k:0}\right\|\, \left\| \left(\mc{G}_{k:0} - \mc{D}\right)(\varrho)\right\|_2 + \sum_{\ell=0}^{k-1} \left\|\mathcal{A}_{k:\ell+1} \right\|\, \left\| [\mc{G}_{k:\ell+1} - \mc{D}, \mc{A}_\ell] ( \varrho_\ell)\right\|_2 \nonumber \\
    & \qquad \qquad + \sum_{\ell=0}^{k-1}\left\|\mathcal{A}_{k:\ell+1} \right\|\, \left\|[\mc{D}, \mc{A}_\ell] ( \varrho_\ell-\varpi_\ell)\right\|_2,
    \label{appendix eq: 2normbound}
\end{align}
which contains the terms $\mbb{B}$ and $\mbb{C}$ of Eq.~\eqref{eq: multitime bound corrections} in the second and third terms of the inequality, and where the first term generalizes Eq.~\eqref{eq:singlestep multitime equilibration} as
\begin{align}
    \left\| \left(\mc{G}_{k:0} - \mc{D}\right)(\varrho)\right\|_2^2&=\tr\left|\sum_{n \neq m} G_{n_k{m}_k}^{(k)}\cdots{G}_{n_0m_0}^{(0)} P_n(\varrho)\,P_m\right|^2\nonumber\\
    &=\sum_{\substack{n \neq m \\ n^\prime \neq m^\prime}}\prod_{j=0}^k G_{n_jm_j}^{(j)}G_{m_j^\prime{n}_j^\prime}^{(j)}\tr\left[ P_n\varrho\,P_mP_{m'}\varrho\,P_{n'}\right]\nonumber\\
    &=\sum_{n \neq m }\prod_{j=0}^k |G_{n_jm_j}^{(j)}|^2\tr\left[ P_n \varrho P_m\varrho\right]\nonumber\\
    &\leq\prod_{j=0}^k\max_{n\neq{m}}|G_{n_jm_j}^{(j)}|^2\left\{\sum_{n , m}\tr[P_n\varrho P_m \varrho]-\sum_n\tr[P_n\varrho P_n \varrho]\right\}\nonumber\\
    &=\prod_{j=0}^k\max_{n\neq{m}}|G_{n_jm_j}^{(j)}|^2\tr(\varrho^2-\varpi^2)\nonumber\\
    &=\|\varrho-\varpi\|_2^2\prod_{j=0}^k\max_{n\neq{m}}|G_{n_jm_j}^{(j)}|^2.
\end{align}
where similarly in the last line $\tr(\varrho^2-\varpi^2)=\|\varrho-\varpi\|_2^2$, because $\tr(\varrho\,\varpi)=\tr(\varpi^2)$.
\end{proof}

\begin{remark}
In general, by definition the term $\mscr{S}_T$, which depends on the waiting time distribution $\mscr{P}_T$, converges to zero in increasing $T$, with the rate of convergence depending on the specific distribution. In particular, as we exemplified for the single-time case, for the uniform distribution on all time-steps as standard equilibration in Section~\ref{sec: fluctuations finite time}, we average over a time-window of width $T$ around each $\tau_i$ for all time-steps, with $\mscr{P}_T=T^{-1}$ in the interval $[\tau_i-T/2, \tau_i+T/2]$, and $\mscr{P}_T=0$ outside it. This yields
\begin{equation}
    \left|G_{mn}^{(T)}\right|=|\mathrm{sin}(T\mc{E}_{mn})/T\mc{E}_{mn}|,
\end{equation}
where the term $\mscr{S}_T$ then picks the smallest non-zero energy gap in the Hamiltonian.

Similarly, if the fuzziness corresponds to that of a half-normal distribution with variance $T$, then overall $\mscr{S}_T$ decays exponentially with
\begin{gather}
\left|G_{mn}^{(T)}\right|\sim\exp(-T\mc{E}_{mn}^2)\left|1-\mathrm{erf}(i\sqrt{T}\mc{E}_{mn})\right|,    
\end{gather}
where $\mathrm{erf}$ is the error function and $E_m-E_n=2\mc{E}_{mn}$. For both cases, if $T$ is small, $\mscr{S}_T$ will also be vanishingly small whenever the energy gap $\mc{E}_{nm}$ that maximizes $|G_{nm}^{(T)}|$ is large enough, i.e. $\mc{E}_{nm}\gg{T}$. This property holds in general, since distributions $\mscr{P}_T$ can be approximated as uniform for small $T$ or because the gaps $\mc{E}_{nm}$ can be seen as a rescaling factor on $T$ in the definition of $G_{nm}^{(T)}$.
\end{remark}

The term $\mbb{A}_k$ in Eq.~\eqref{eq: result main multi equilibration} neglects temporal correlations and the operations $\{\mc{A}_i\}$ are all composed as a single operation $\mc{A}_{k:0}=\mc{A}_k\cdots\mc{A}_0$. This is essence can be interpreted as a single-time contribution to equilibration. The two-norm distance satisfies $\|\varrho-\varpi\|_2^2\leq{1-(d_Ed_S)^{-1}}$ as the ancillary input $\gamma$ can be taken to be pure. As discussed above, this term is suppressed through the $\mscr{S}_T$ contributions when $i)$ the averaging window, or equivalently the fuzziness of the clock $T$ is large enough and $ii)$ for small $T$ whenever the energy gap maximizing the time-averaging $\left|G_{nm}^{(T)}\right|$ factor is large with respect to $T$.

Now, we can bound further the terms $\mbb{B}_\ell$ and $\mbb{C}_\ell$ in Theorem~\ref{Thm: Multitime equilibration}, which contain genuine multitime contributions relating to how well the intermediate states at step $\ell$ equilibrate. Continuing from Eq.~\eqref{appendix eq: 2normbound}, we have
\begin{align}
    &\mbb{B}_\ell=\left\| [\mc{G}_{k:\ell+1} - \mc{D}, \mc{A}_\ell] ( \varrho_\ell)\right\|_2\leq\left\| \mc{A}_\ell\right\|\, \left\|(\mc{G}_{k:\ell+1} - \mc{D}) ( \varrho_\ell)\right\|_2 + \left\| (\mc{G}_{k:\ell+1} - \mc{D})\mc{A}_\ell ( \varrho_\ell)\right\|_2,
\end{align}
where we used the triangle inequality on the commutator. Now from this inequality we have, similarly, for the first term,
\begin{align}
    \left\| \left(\mc{G}_{k:\ell+1} - \mc{D}\right)(\varrho_\ell)\right\|_2^2
    &\leq\prod_{j=\ell+1}^k\max_{n\neq{m}}|G_{n_jm_j}^{(j)}|^2\left\{\sum_{n , m}\tr[P_n\varrho_\ell P_m \varrho_\ell]-\sum_n\tr[P_n\varrho_\ell P_n \varrho_\ell]\right\}\nonumber\\
    &=\prod_{j=\ell+1}^k\max_{n\neq{m}}|G_{n_jm_j}^{(j)}|^2\tr[\varrho_\ell^2-\mc{D}(\varrho_\ell)\varrho_\ell]\nonumber\\
    &=\|\varrho_\ell-\mc{D}(\varrho_{\ell})\|_2^2\prod_{j=\ell+1}^k\max_{n\neq{m}}|G_{n_jm_j}^{(j)}|^2,
\end{align}
as $\tr[(\mc{D}(\varrho_\ell))^2]=\tr[\mc{D}(\varrho_\ell)\varrho_\ell]$, whilst for the second term, with $\varrho^\prime_\ell=\mc{A}_\ell(\varrho_\ell)$,
\begin{align}
    \left\| \left(\mc{G}_{k:\ell+1} - \mc{D}\right)(\varrho_\ell^\prime)\right\|_2^2
    &\leq\prod_{j=\ell+1}^k\max_{n\neq{m}}|G_{n_jm_j}^{(j)}|^2\left\{\sum_{n , m}\tr[P_n(\varrho_\ell^\prime) P_m (\varrho_\ell^\prime)]-\sum_n\tr[P_n(\varrho_\ell^\prime) P_n (\varrho_\ell^\prime)]\right\}\nonumber\\
    &=\prod_{j=\ell+1}^k\max_{n\neq{m}}|G_{n_jm_j}^{(j)}|^2\tr[\varrho_\ell^{\prime\,2}-\mc{D}(\varrho_\ell^\prime)\varrho_\ell^\prime]\nonumber\\
    &=\|\mc{A}_\ell(\varrho_\ell)-\mc{D}\mc{A}_\ell(\varrho_\ell)\|_2^2\prod_{j=\ell+1}^k\max_{n\neq{m}}|G_{n_jm_j}^{(j)}|^2,
\end{align}
so putting these together,
\begin{equation}
    \mbb{B}_\ell\leq\prod_{j=\ell+1}^k\max_{n\neq{m}}|G_{n_jm_j}^{(j)}|\left\{\|\mc{A}_\ell\|\|\varrho_\ell-\mc{D}(\varrho_{\ell})\|_2+\|\mc{A}_\ell(\varrho_\ell)-\mc{D}\mc{A}_\ell(\varrho_\ell)\|_2\right\}.
\end{equation}

This means that in the particular case of same evolution $H_i=H_j=H$ and same fuzziness, $T_i=T_j=T$, we have
\begin{equation}
    \mbb{B}_\ell\lesssim\mscr{S}_T^{k-\ell},
\end{equation}
so that, crucially, this term is suppressed overall in the width of the time-window $T$. 

Finally, for $\mbb{C}_\ell$, notice that we can further simplify the last term of Eq.~\eqref{eq: multitime final bound} as
\begin{align}
    \mbb{C}_\ell&=\|[\mc{D},\mc{A}_\ell](\varrho_\ell-\varpi_\ell)\|_2\leq\|\mc{D}(\varrho_{\ell+1})\|_2 + \|\varpi_{\ell+1}\|_2 + \|\mc{A}_\ell\| (\|\mc{D}(\varrho_\ell)\|_2 + \|\varpi_{\ell}\|_2)\nonumber\\
    &\leq\|\mc{D}(\varrho_{\ell+1})-\varpi_{\ell+1}\|_2+\|\mc{A}_\ell\|\|\mc{D}(\varrho_\ell)-\varpi_\ell\|_2,
\end{align}
and each term is the purity of a dephased state, which will decay as the inverse effective dimension of that state. This follows as in general, $\tr\left[(\mc{D}(\sigma))^2\right]\leq{d}_\text{eff}^{-1}(\sigma)$ for any state $\sigma$, with equality for either pure states or non-degenerate Hamiltonians. On the other hand, when the control operations from $0$ to $\ell$ succeed in driving $\varrho_\ell$ so that the action of the commutator does not dephase it significantly, the purity of $\varrho_\ell$ may be large and thus $\mbb{C}_\ell$ may become trivial (i.e. it approaches 1).

More concretely, the operations $\mc{A}_j$ interleaved within the intermediate states $\varrho_\ell$ and $\varpi_\ell$ will relate in the multitime correction terms in Eq.~\eqref{eq: multitime bound corrections} to how greatly they disturb either the finite-time averaged $\varrho_{j-1}$ or the equilibrated $\varpi_{j-1}$. This is most evident in the term $\mbb{C}_\ell$, which can be bounded as well as $\mbb{C}_\ell\leq \|[\mc{D}, \mc{A}_\ell]\| \|\varrho_\ell - \varpi_\ell\|_2$. The norm of the commutator can be written in terms of both the capacity of the operations $\mc{A}_\ell$ to generate coherences between different energy eigenspaces from equilibrium and the degree to which the operations can turn such coherences into populations. Environments in physical systems are typically much larger than the subsystems that can be probed, and, keeping in mind that the operations $\mc{A}_j$ act only on subsystem \gls{syst} and the ancilla $\mathsf{\Gamma}$, the ability to generate and detect energy coherences should be severely limited in many physically relevant cases.

A more general version of the bound in Theorem~\ref{Thm: Multitime equilibration} can thus be given as follows. Let us denote $\mf{S}_{b:a}:=\prod_{j=a}^b\max_{n\neq{m}}|G_{n_jm_j}^{(j)}|$, then,
\begin{align}
    \left|\langle\Lambda\rangle_{\overline{\Upsilon}^{\mscr{P}_\mathbf{T}}-\Omega}\right|&\leq \mf{S}_{k:0}\,\|\mc{A}_{k:0}\|\,\|\varrho-\varpi\|_2\nonumber\\
    &\quad +\sum_{\ell=0}^{k-1}\mf{S}_{k:\ell+1}\,\|\mc{A}_{k:\ell+1}\|\,\bigg\{\|\mc{A}_\ell\|\|\varrho_\ell-\mc{D}(\varrho_{\ell})\|_2+\|\mc{A}_\ell(\varrho_\ell)-\mc{D}\mc{A}_\ell(\varrho_\ell)\|_2\bigg\}\nonumber\\
    &\quad\quad +\sum_{\ell=0}^{k-1}\|\mc{A}_{k:\ell+1}\|\left\{\|\mc{D}(\varrho_{\ell+1})-\varpi_{\ell+1}\|_2+\|\mc{A}_\ell\|\|\mc{D}(\varrho_\ell)-\varpi_\ell\|_2\right\},
    \label{eq: multitime final bound}
\end{align}
so if we now take the Hamiltonian at each time-step to be a fixed $H=\sum{E}_nP_n$ and we fix the fuziness to be the same at each step $T_i=T$, we now have $\FS_{b:a}=\mscr{S}_T^{b-a}$ where $\mscr{S}_T:=\max_{n\neq{m}}\left|G_{nm}^{(T)}\right|$ as in Eq.~\eqref{eq: main standard case multitime}.

To summarize, our result in Theorem~\ref{Thm: Multitime equilibration} shows that either subsystems or global coarse properties of a closed time-independent Hamiltonian system will display equilibration for multiple sequential operations with a temporal uncertainty or fuzziness provided:
\begin{compactenum}[\itshape i.]
    \item Both the initial and intermediate states have a significant overlap with the energy eigenstates.
    \item The temporal fuzziness is large enough relative to the average measurement time or, equivalently, the energy gaps in the Hamiltonian are large enough with respect to the temporal fuzziness.
    \item The disturbance by the operations on intermediate states is small.
\end{compactenum}

\section{Genuine multitime equilibration}
In the previous section we have stressed the multitime nature of the bound in Theorem~\ref{Thm: Multitime equilibration}, however, how can we be sure that this is not simply an elaborated example of the results of Ref.~\cite{Short_finite} (described in Section~\ref{sec: fluctuations finite time})? Why would it not be possible to write the joint expectation $|\langle\Lambda\rangle_\Upsilon|$ as a single Heisenberg picture operator acting on the initial state, i.e. to simply group all of the time evolutions and measurements into a single Hermitian operator $\mathfrak{F}$ representing the measurement procedure acting on the initial state? This indeed suggests that a single-time equilibration bound, like the one in Eq.~\eqref{eq: Short finite time} would suffice to study equilibration in general quantum processes.

Let us then give a simple example that demonstrates that Theorem~\ref{Thm: Multitime equilibration} indeed captures genuine multitime phenomena.

\begin{example}
For simplicity, consider only two \gls{CPTNI} interventions acting on \gls{syst-ancilla} of the form
\begin{equation}
    \mc{A}(\cdot)=\sum_\mu{a}_\mu{A}_\mu(\cdot)A_\mu^\dg,\qquad
    \mc{B}(\cdot)=\sum_\mu{b}_\mu{B}_\mu(\cdot){B}_\mu^\dg,
\end{equation}
then the joint expectation for these operations interleaved with evolutions over time intervals $\delta{t}_0$ and $\delta{t}_1$ is
\begin{align}
    \langle\Lambda\rangle_\Upsilon&=\tr[\mc{B}\,\mc{U}_1\mc{A}\,\mc{U}_0(\varrho)]=\sum{a}_\mu b_\nu\tr[B_\nu U_1 A_\mu\,\mc{U}_0(\varrho)A_\mu^\dg U_1^\dg B_\nu^\dg],
\end{align}
where as above, $\mc{U}_\ell=U_\ell(\cdot)U_\ell^\dg$, with $U_\ell:=\exp(-iH_\ell{t}_\ell)$ with $H_\ell$ the Hamiltonian at time-step $\ell$ and $\varrho=\rho\otimes\gamma$ the full initial $\mathsf{SE\Gamma}$ state. Let us fix the basis for the Hamiltonians so that $H_\ell=\sum{E}_{n_\ell}{P}_n$ as above. We can now move terms around using cyclicity of trace to get
\begin{align}
    \langle\Lambda\rangle_\Upsilon
    &=\sum{a}_\mu b_\nu\tr[A_\mu^\dg U_1^\dg B_\nu^\dg{B}_\nu U_1 A_\mu\mc{U}_0(\varrho)]=\tr[\mc{A}\,\mc{U}_1^\star(\mathsf{B})\,\mc{U}_0(\varrho)],
\end{align}
where we define where $\mc{U}_\ell^\star(\cdot):=U_\ell^\dg(\cdot)U_\ell$ with $\mathsf{B}:=\sum\,b_\mu{B}_\mu^\dg{B}_\mu$. The argument we refer to is that we can write this as $\langle{\mathfrak{F}}\rangle_{\varrho}=\tr[\mathfrak{F}\varrho(\delta t_0)]$, where $\varrho(\delta t_0):=\mc{U}_0(\varrho)$ and
\begin{align}
    \mathfrak{F}:=\mc{A}\,\mc{U}_1^\star(\mathsf{B}),
\end{align}
and obtain that $|\langle\Lambda\rangle_{\overline{\Upsilon}^{\mathscr{P}^\mathbf{T}}-\Omega}|=|\langle{\mathfrak{F}}\rangle_{\overline{\varrho}^{\mathscr{P}^{T_0}}-\omega}|$.

This happens to be the case if the interval $\delta{t}_1$ is fixed, allowing fuzziness only in the first evolution time. In such case indeed we can simply apply the single-time result in  Eq.~\eqref{eq:singlestep multitime equilibration}, given that $\langle\Lambda\rangle_\Omega=\langle{\mathfrak{F}}\rangle_{\varpi_0}$ with $\varpi_0=\mc{D}(\varrho)$ as defined in Eq.~\eqref{eq: def varpi}.

However, when the fuzziness of the clock is present for each intervention, i.e. in both evolution times, we have
\begin{equation}
    \langle\Lambda\rangle_{\overline{\Upsilon}^{\mathscr{P}_\mathbf{T}}}=\tr[\mc{B}\,\mc{G}_1\mc{A}\,\mc{G}_0(\varrho)],
\end{equation}
where $\mc{G}_\ell=\mc{G}_{T_\ell}$, as we use in the proof of Theorem~\ref{Thm: Multitime equilibration} and defined in Eq.~\eqref{eq: dephasing finite}, is the time-evolution superoperator finite-time averaged with respect to the probability distributions $\mathscr{P}_{T_\ell}$ with a characteristic temporal fuzziness $T_\ell$. As defined in Eq.~\eqref{eq: equilibrium expectation}, the joint expectation with respect to the equilibrium process is
\begin{equation}
    \langle\Lambda\rangle_{\Omega}=\tr[\mc{B}\,\mc{D}\,\mc{A}\,\mc{D}(\varrho)].
\end{equation}

Now due to the double time average, it is impossible to write the difference of both quantities in terms of an operator expectation value, since
\begin{align}
    \left|\langle\Lambda\rangle_{\overline{\Upsilon}^{\mathscr{P}_T}-\Omega}\right|&=\tr[\mc{B}\,\mc{G}_1\,\mc{A}\,\mc{G}_0(\varrho)-\mc{B}\,\mc{D}\,\mc{A}\,\mc{D}(\varrho)]=\tr[\mc{B}\,\mc{G}_1\,\mc{A}\,\overline{\varrho}^{\mathscr{P}^{T_0}}-\mc{B}\,\mc{D}\mc{A}\,\varpi_0]\nonumber\\
    &\neq\left|\langle{\mathfrak{E}}\rangle_{\overline{\varrho}^{\mathscr{P}^{T_0}}-\varpi_0}\right|,
\end{align}
for any operator $\mathfrak{E}$. Exceptions occur when we take the $T_1\to\infty$ limit (so that $\mc{G}_1\to\mc{D}$), or if there is no fuzziness for a fixed $\delta{t}_1$ (so that $\mc{G}_1$ and $\mc{D}$ are replaced by a fixed $\mc{U}_{\delta{t}_1}$) as argued above. This means that in general the result in Theorem~\ref{Thm: Multitime equilibration} constitutes an equilibration result that cannot be reduced to a single-time one.
\end{example}

\section{Conclusions}
Let us finally note that our approach to describe the operations that can act on the process is general in the sense that these are \gls{CP} maps which can be correlated between time-steps and propagate information from their interactions with the subsystem \gls{syst} through the ancillary space $\mathsf{\Gamma}$. While these set a scale in all terms of the right-hand side of Eq.~\eqref{eq: main standard case multitime}, they can also contribute to loosening it, potentially allowing to distinguish the fuzzy process from the equilibrium one within a finite time. It is not entirely clear, however, if a departure from equilibration is more readily accessible with a larger ancillary space $\mathsf{\Gamma}$, and, for long time fuzziness $T$, the upper-bound in Eq.~\eqref{eq: main standard case multitime} should remain close to zero.

Similar to the single-time standard case described in Section~\ref{sec: fluctuations finite time}, equilibration over multiple observations in open systems is expected intuitively through decoherence arguments~\cite{Yukalov_2012}. The interplay with memory effects, through both the \gls{env} and $\mathsf{\Gamma}$ in the interventions is as yet not entirely clear, e.g., under which circumstances finite temporal resolution equilibration can occur without the dynamics being Markovian, i.e., memoryless, or if the temporal correlations among interventions through the ancillary space can display a departure from equilibration within a finite-time.

In the following Chapters we will see more clearly how the questions on the foundations of statistical mechanics can be posed in direct analogy in the context of quantum processes with respect to Markovianity. Whilst here we approached the question of equilibration in quantum processes somewhat pragmatically, we can conjecture a bridge between the characterization of the equilibration process and Markovianity akin to that between the generic, time-average equilibrium state and the Gibbs state for thermal equilibrium. The relationship between the two properties is as yet, however, not so transparent but we can certainly expect some progress in this direction in the near future.
    \chapter{Markovian Typicality}
\label{sec:typicality}

\setlength{\epigraphwidth}{0.4\textwidth}
\epigraph{\emph{Il n'y a de nouveau que ce qui est oublié.}\footnotemark}{-- Rose Bertin}
\footnotetext{\emph{There is nothing new except what has been forgotten}. The quote is sometimes attributed to Marie Antoinette, of whom Bertin was the dressmaker.}

As we saw in Chapter~\ref{sec:statmech}, the quest towards understanding how thermodynamics emerges purely from quantum mechanical laws has seen a great deal of progress in recent years. Most prominently, equilibration on average deals with the dynamical explanation of how reversible and recurrent Schr\"odinger dynamics lead to irreversible reduced dynamics which converge and revolve around an equilibrium state, specifically telling us that time-dependent quantum properties evolve towards a certain fixed equilibrium value and stay close to it for most times. Moreover, we have seen that this dynamical convergence towards and around equilibrium holds more generally for subparts of closed systems under general quantum stochastic processes.

Whenever equilibration on average holds, it implies that the dynamics erases the information contained in the initial state of the respective system, however, there may still be non-Markovian memory of the initial state encoded in the temporal correlations between observables. Moreover, as we highlighted in Chapter~\ref{sec:processes}, all open quantum evolutions generated by a time-independent \gls{syst-env} Hamiltonian are non-Markovian, i.e. we know that, as far as nature is concerned, non-Markovianity is the rule and Markovianity is the exception which at best is an idealization.

How can we explain then this apparent contradiction? Furthermore, the Born-Markov approximation has proven to be extremely fertile over the years, being applicable to a wide class of physical models and situations~\cite{carmichael1993open, blanchard2000decoherence, breuer2002theory, schlosshauer2007decoherence, alicki2007quantum, Dynkin_2006, stroock2013}.

There is now an evident parallel flow of ideas that we can draw from the emergence of statistical mechanics discussed in Chapter~\ref{sec:statmech}; quantum systems dynamically equilibrate despite non-equilibrium being generic, in turn fulfilling the second law of thermodynamics, whose emergence can be explained from first principles such as entanglement, rather than from the equal a-priori probabilities postulate. Now we know that quantum processes satisfy an analogous form of dynamical equilibration and we can similarly ask if the emergence of forgetful processes can arise fundamentally rather than from ad-hoc assumptions or approximations such as the Born-Markov condition.

In Ref.~\cite{FigueroaRomero2019almostmarkovian} without resorting to the Born-Markov assumption or any other approximation, we formally proved that quantum processes are close to Markovian ones, when the subsystem \gls{syst} is sufficiently small compared to the whole \gls{syst-env}, with a probability that tends to unity exponentially in the size of the latter. That is, we showed that Markovian processes are typical when these occur in small subsystems, with generic processes obeying a concentration of measure around Markovian ones. We also showed that, for a fixed global system size, it may not be possible to neglect non-Markovian effects when the process is allowed to continue for long enough, although detecting non-Markovianity for such processes would usually require non-trivial entangling resources. These results give birth to \textit{almost} Markovian processes from closed dynamics analogous to the way in which entanglement supersedes the fundamental postulate of statistical mechanics.

\section{Random quantum processes}
The main approach we take to formally prove that Markovian processes are typical is generally speaking the one described in Section~\ref{sec: state typicality}; here as well we want to explore the statistical properties of quantum processes and study how in this case Markovian processes turn out to be exceptional. As we will see, however, we will need some additional mathematical concepts when it comes to the moments of the unitary group.

Similar to the case of quantum states, to approach the question of what sampling a random quantum process means, we require a probability measure that assigns non-vanishing probabilities to mathematically generic unitary dynamics on the closed \gls{syst-env} composite. We similarly we achieved this by sampling the evolution from the unitarily invariant Haar measure, introduced in Section~\ref{sec: Random states and Haar}. As we saw, this has the additional advantage of allowing employ random matrix theory techniques~\cite{Weingarten,Collins_2003,GuMoments, Collins_2006, Puchala_2017, Guhr1998, mehta2004random} and leads to the relatively straightforward application of concentration of measure results~\cite{Ledoux, Milman, boucheron2013concentration}.

Consider then a $k$-step quantum process $\Upsilon$ on a \gls{syst-env} composite of dimension $d_\mathsf{SE}=d_\mathsf{S}d_\mathsf{E}$, with initial state $\rho$ and with unitary evolution given by unitary maps $\mc{U}_i=U_i(\cdot)\,U_i^\dg$ acting on the full \gls{syst-env} at the $i$\textsuperscript{th} timestep.  As we now are dealing with unitary evolution at several time-steps, we use the Haar measure to sample two distinct types of unitary \gls{syst-env} evolution.

\begin{definition}\label{def: random constant interaction}
We refer to these two ways of sampling as:
    \begin{compactenum}[\itshape i.]
    \item Random interaction: All $U_i$ independently chosen.
    \item Constant interaction: $U_i=U_j,\quad \forall \leq{i}\neq{j}\leq{k}$.
    \end{compactenum}
\end{definition}

We now depict this in Fig.~\ref{fig: Random Process Markov}$\mathsf{(a)}$. The entire set of unitaries enters into the process tensor as in the definition of the Choi state in Eq.~\eqref{eq: process tensor Choi state}. In the first case, the global system will quickly explore its entire (pure) state space for any initial state. The second case corresponds more closely to what one might expect for a truly closed system, where the Hamiltonian remains the same throughout the process. These correspond to two extremes; more generally, the dynamics from step to step may be related but not identical.

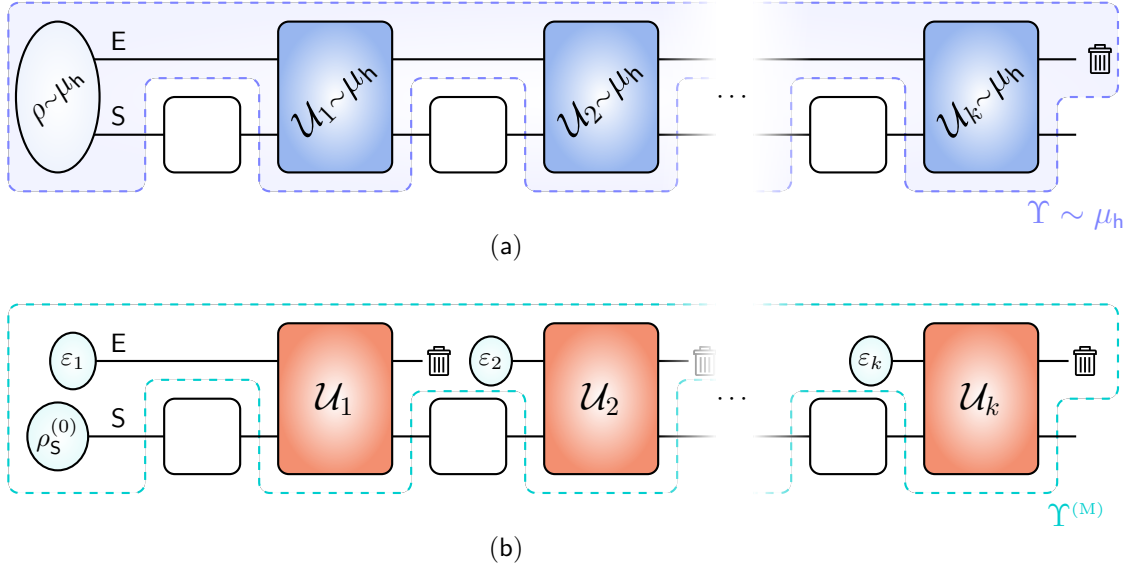
\begin{figure}[t]
    \centering
    \begin{tikzpicture}
    \begin{scope}
    \fill[inner color=white, outer color=C4!10!white, draw=C4, dashed, thick, rounded corners] (0.95,-1.25) -- (0.95,1.25) -- (15.55,1.25) -- (15.55,0) -- (14.75,0) -- (14.75,-1.25) -- (12.75,-1.25) -- (12.75,0.25) -- (11.25,0.25) -- (11.25,-1.25) -- (10.5,-1.25) -- (10.5,0.25) -- (9.75,0.25) -- (9.75,-1.25) -- (7.75,-1.25) -- (7.75,0.25) -- (6.25,0.25) -- (6.25,-1.25) -- (4.25,-1.25) -- (4.25,0.25) -- (2.75,0.25) -- (2.75,-1.25) -- cycle;
    \node[above] at (2.4,0.5) {$\mathsf{E}$};
    \node[above] at (2.4,-0.5) {$\mathsf{S}$};
    \node[below,C4] at (15,-1.25) {\Large$\Upsilon\sim\mu_\haar$};
    \draw[thick, -] (2,0.5) -- (15,0.5);
    \draw[thick, -] (2,-0.5) -- (15,-0.5);
    \fill[outer color=C1!10!white, inner color=white, draw=black, thick] (1.6,0) ellipse (0.55cm and 1cm);
    \node[rotate=45] at (1.6,0) {\Large$\rho\scriptstyle{\sim}\displaystyle\muhaar$};
    \shade[outer color=white, inner color=white, draw=black, rounded corners, thick] (3,-1) rectangle (4,0);
    \shade[outer color=C1!60!white, inner color=white, draw=black, rounded corners, thick] (4.5,-1) rectangle (6,1);
    \node[rotate=45] at (5.25,0) {\LARGE$\mc{U}_1\scriptstyle{\sim}\displaystyle\muhaar$};
    \shade[outer color=white, inner color=white, draw=black, rounded corners, thick] (6.5,-1) rectangle (7.5,0);
    \shade[outer color=C1!60!white, inner color=white, draw=black, rounded corners, thick] (8,-1) rectangle (9.5,1);
    \node[rotate=45] at (8.75,0) {\LARGE$\mc{U}_2\scriptstyle{\sim}\displaystyle\muhaar$};
    \draw[white, fill=white, path fading= west] (9.5,-1.5) -- (9.5,1.5) -- (10.25,1.5) -- (10.25,-1.5) ;
    \draw[white, fill=white] (10.25,1.5) -- (10.75,1.5) -- (10.75,-1.5) -- (10.25,-1.5);
    \draw[white, fill=white, path fading= east] (10.75,1.5) -- (11.5,1.5) -- (11.5,-1.5) -- (10.75,-1.5);
    \node at (10.5,0) {$\cdots$};
    \shade[outer color=white, inner color=white, draw=black, rounded corners, thick] (11.5,-1) rectangle (12.5,0);
    \shade[outer color=C1!60!white, inner color=white, draw=black, rounded corners, thick] (13,-1) rectangle (14.5,1);
    \node[rotate=45] at (13.75,0) {\LARGE$\mc{U}_k\scriptstyle{\sim}\displaystyle\muhaar$};
    \node[right] at (15,0.5) {\trash};
    \node at (7.5,-2) {$\mathsf{(a)}$};
    \end{scope}
    \begin{scope}[shift = {(0,-4)}]
    \fill[inner color=white, outer color=white, draw=C3, dashed, thick, rounded corners] (0.95,-1.25) -- (0.95,1.25) -- (15.55,1.25) -- (15.55,0) -- (14.75,0) -- (14.75,-1.25) -- (12.75,-1.25) -- (12.75,0.1) -- (11.25,0.1) -- (11.25,-1.25) -- (10.5,-1.25) -- (10.5,0.25) -- (9.75,0.25) -- (9.75,-1.25) -- (7.75,-1.25) -- (7.75,0.1) -- (6.25,0.1) -- (6.25,-1.25) -- (4.25,-1.25) -- (4.25,0.25) -- (2.75,0.25) -- (2.75,-1.25) -- cycle;
    \node[below,C3] at (15,-1.25) {\Large$\Upsilon^\markov$};
    \node[above] at (2.4,0.5) {$\mathsf{E}$};
    \node[above] at (2.4,-0.5) {$\mathsf{S}$};
    \draw[thick, -] (2,0.5) -- (6,0.5);
    \draw[thick, -] (6,0.5) -- (6.4,0.5);
    \node[right] at (6.3,0.5) {\trash};
    \draw[thick, -] (2,-0.5) -- (15,-0.5);
    \fill[outer color=C3!10!white, inner color=white, draw=black, thick] (1.8,0.5) ellipse (0.3 and 0.375);
    \fill[outer color=C3!10!white, inner color=white, draw=black, thick] (1.6,-0.5) ellipse (0.4 and 0.45);
    \node at (1.8,0.5) {$\varepsilon_1$};
    \node[left] at (2,-0.5) {$\rho_\mathsf{S}^{(0)}$};
    \shade[outer color=white, inner color=white, draw=black, rounded corners, thick] (3,-1) rectangle (4,0);
    \shade[outer color=C2!60!white, inner color=white, draw=black, rounded corners, thick] (4.5,-1) rectangle (6,1);
    \node at (5.25,0) {\LARGE$\mc{U}_1$};
    \shade[outer color=white, inner color=white, draw=black, rounded corners, thick] (6.5,-1) rectangle (7.5,0);
    \draw[thick] (7.3,0.5) -- (9.9,0.5);
    \node[right] at (9.8,0.5) {\trash};
    \shade[outer color=C2!60!white, inner color=white, draw=black, rounded corners, thick] (8,-1) rectangle (9.5,1);
    \node at (8.75,0) {\LARGE$\mc{U}_2$};
    \fill[outer color=C3!10!white, inner color=white, draw=black, thick] (7.3,0.5) ellipse (0.275 and 0.325);
    \node at (7.3,0.5) {$\varepsilon_2$};
    \draw[white, fill=white, path fading= west] (9.5,-1.5) -- (9.5,1.5) -- (10.25,1.5) -- (10.25,-1.5) ;
    \draw[white, fill=white] (10.25,1.5) -- (10.75,1.5) -- (10.75,-1.5) -- (10.25,-1.5);
    \draw[white, fill=white, path fading= east] (10.75,1.5) -- (11.5,1.5) -- (11.5,-1.5) -- (10.75,-1.5);
    \node at (10.5,0) {$\cdots$};
    \shade[outer color=white, inner color=white, draw=black, rounded corners, thick] (11.5,-1) rectangle (12.5,0);
    \draw[thick] (12.3,0.5) -- (14.9,0.5);
    \shade[outer color=C2!60!white, inner color=white, draw=black, rounded corners, thick] (13,-1) rectangle (14.5,1);
    \node at (13.75,0) {\LARGE$\mc{U}_k$};
    \fill[outer color=C3!10!white, inner color=white, draw=black, thick] (12.3,0.5) ellipse (0.275 and 0.325);
    \node at (12.3,0.5) {$\varepsilon_k$};
    \node[right] at (14.8,0.5) {\trash};
    \node at (7.5,-2) {$\mathsf{(b)}$};
    \end{scope}
    \end{tikzpicture}
    \caption[Sampling a quantum process at random and quantum Markovian processes]{\textbf{Sampling a quantum process at random and quantum Markovian processes:} $\mathsf{(a)}$ By the sampling of a random quantum process we mean a $k$-step process with a unitary evolution sampled from the Haar measure $\muhaar$ either according to a random interaction (independent $\mc{U}_i\neq\mc{U}_j$) or a constant interaction (all $\mc{U}_i=\mc{U}_j$) as per Definition~\ref{def: random constant interaction}. $(\mathsf{b})$ We consider how distinguishable is a generic process sampled at uniformly at random from a Markovian one, which is mathematically equivalent to one where \gls{syst-env} is initially uncorrelated, $\rho=\rho_\mathsf{S}^{(0)}\otimes\varepsilon_1$, with \gls{env} being subsequently discarded and replaced afresh at every time-step.}
    \label{fig: Random Process Markov}
\end{figure}

Two main features of this approach now stand out. The first is that we are directly sampling unitary dynamics from the Haar measure which are not generally given as a time-independent Hamiltonian dynamics; this is slightly different from sampling a pure quantum state at random, as done in Section~\ref{sec: state typicality}, as in this case the type of dynamics generated by the Haar measure will be relevant. The second is precisely that the interaction or information flow between all parts of the whole \gls{syst-env} composite will be relevant, i.e. no parts of the environment dimension are superfluous and in this sense we can think of having a strong interaction among all of the \gls{syst-env} composite degrees of freedom.

\section{The moments of the unitary group \& average processes}
As a $k$-step process tensor generically involves $k$ unitary maps $\mc{U}_1,\mc{U}_2,\ldots,\mc{U}_k$. The fiducial state $\rho$ can be taken to have undergone some evolution $\mc{U}_0$, which will let us interpret it as a random state. Back in Section~\ref{sec: Random states and Haar}, when we distributed \gls{syst-env} states according to the Haar measure, we only required the first and second moments of the unitary group to obtain a concentration of measure result. In this case, however, we are generically considering quantities involving $k+1$ unitary maps, either all the same or all different and independently sampled.

We showed in Eq.~\eqref{eq: average quantum state} that the average quantum state drawn from the Haar measure is the maximally mixed state. For a random process, with all $\mc{U}_i\neq\mc{U}_j$ independently chosen, we can simply apply the average independently over each unitary. First let us rewrite the definition of the Choi state representation of the process tensor,
\begin{align}
    &\Upsilon=\tr_\mathsf{E}\left[\,\mc{U}_k\,\mc{S}_k\,\mc{U}_{k-1}\mc{S}_{k-1}\cdots\mc{U}_1\mc{S}_1\left(\rho\otimes\tilde{\Psi}^{\otimes\,k}\right)\,\right]\nonumber\\
    &=\!\!\sum_{\alpha,\ldots,\delta}\!\!\tr_\mathsf{E}\left[U_k\FS_{\alpha_k\beta_k}\!\cdots{U}_1\FS_{\alpha_1\beta_1}\rho\FS_{\delta_1\gamma_1}U_1^\dg\cdots\FS_{\delta_k\gamma_k}U_k^\dg\right]\otimes|\beta_1\alpha_1\!\cdots\beta_k\alpha_k\rangle\!\langle\delta_1\gamma_1\!\cdots\delta_k\gamma_k|,
\end{align}
where the sum runs over all Greek indices from $1$ to $d_\mathsf{S}$, with $\tilde\Psi$ being an unnormalized maximally entangled state acting in the respective $d_\mathsf{S}$-dimensional ancillary spaces $\mathsf{A}_i\mathsf{B}_i$, where $\mc{S}_i$ are swaps between \gls{syst} and ancillary system $\mathsf{A}_i$ at time-step $i$, and where $\FS_{\alpha\beta}=\mbb{1}_\mathsf{E}\otimes|\alpha\rangle\!\langle\beta|$. Full detail can be revisited around the definition in Eq.~\eqref{eq: process tensor Choi state}. Then let us denote as
\begin{equation}
    \mbb{E}_\haarrand[f(U_0,U_1,\ldots,U_k)]=\int_{\mbb{U}(d)}f(U_0,U_1,\ldots,U_k)\,d\muhaar(U_0)\,d\muhaar(U_1)\cdots{d}\muhaar(U_k),
\end{equation}
the integration, or averaging, over the Haar measure independently over all time-steps, implicitly being over the different unitaries of the argument, so that we obtain
\begin{align}
    &\mbb{E}_\haarrand[\Upsilon]=\sum_{\alpha,\ldots,\gamma}\tr_\mathsf{E}\left\{\mbb{E}_\haarrand\left[U_k\FS_{\alpha_k\beta_k}\cdots{U}_1\FS_{\alpha_1\beta_1}U_0\rho\,U_0^\dg\FS_{\delta_1\gamma_1}U_1^\dg\cdots\FS_{\delta_k\gamma_k}U_k^\dg\right]\right\}\nonumber\\
    &\qquad\qquad\qquad\otimes|\beta_1\alpha_1\cdots\beta_k\alpha_k\rangle\!\langle\delta_1\gamma_1\cdots\delta_k\gamma_k|\nonumber\\
    &=\!\!\sum_{\alpha,\ldots,\gamma,\epsilon}\f{\mbb1_\mathsf{S}}{d_\mathsf{S}}\langle\epsilon\beta_k|\mbb{E}_\haarrand\left[U_{k-1}\FS_{\alpha_{k-1}\beta_{k-1}}\cdots{U}_1\FS_{\alpha_1\beta_1}U_0\rho\,U_0^\dg\FS_{\delta_1\gamma_1}U_1^\dg\cdots\FS_{\delta_{k-1}\gamma_{k-1}}U_{k-1}^\dg\right]|\epsilon\delta_k\rangle\nonumber\\
    &\qquad\qquad\otimes|\beta_1\alpha_1\cdots\beta_k\alpha_k\rangle\!\langle\delta_1\gamma_1\cdots\delta_k\alpha_k|\nonumber\\
    &=\!\!\sum_{\alpha,\ldots,\gamma,\epsilon}\f{\mbb1_\mathsf{S}}{d_\mathsf{S}^2}\langle\epsilon\beta_{k-1}|\mbb{E}_\haarrand\left[U_{k-2}\FS_{\alpha_{k-2}\beta_{k-2}}\!\!\cdots{U}_1\FS_{\alpha_1\beta_1}U_0\rho\,U_0^\dg\FS_{\delta_1\gamma_1}U_1^\dg\cdots\FS_{\delta_{k-2}\gamma_{k-2}}U_{k-2}^\dg\right]|\epsilon\delta_{k-1}\rangle\nonumber\\
    &\qquad\qquad\otimes|\beta_1\alpha_1\cdots\beta_{k-1}\alpha_{k-1}\rangle\!\langle\delta_1\gamma_1\cdots\delta_{k-1}\alpha_{k-1}|\otimes1_{\mathsf{A}_k\mathsf{B}_k}\nonumber\\[-0.05in]
    &\,\,\,\vdots\nonumber\\
    &=\f{\mbb{1}_{\mathsf{SA}_1\mathsf{B}_1\cdots\mathsf{A}_k\mathsf{B}_k}}{d_\mathsf{S}^{\,k+1}},
    \label{eq: average process rand}
\end{align}
where we introduced $\{|\epsilon\rangle\}$ as a basis for \gls{env} to perform each trace. This is, up to normalization of the maximally entangled states, a maximally mixed state in the full system-ancillary space, implying that the average process tensor is \emph{maximally noisy}.

\begin{example}
As a simple example consider $k=2$. Then the action of the average process tensor, say $\mc{T}_{2:0}^{\,\haarrand}$, on a pair of \gls{CPTNI} operations $\mc{A}_0$ and $\mc{A}_1$ is given by
\begin{align}
    \mc{T}_{2:0}^{\,\haarrand}[\{\mc{A}_i\}]&=\tr_\mathsf{in}\left\{\mbb{E}_\haarrand[\Upsilon_{2:0}]\left[(\mbb1_{\mathsf{A}_1}\otimes\mc{A}_0^\mathrm{T})\tilde{\Psi}\otimes(\mbb1_{\mathsf{A}_2}\otimes\mc{A}_1^\mathrm{T})\tilde{\Psi}\right]\right\}=\f{\tr[\mc{A}_0(\mbb1)]\tr[\mc{A}_1(\mbb1)]}{d_\mathsf{S}^3}\,\mbb1_\mathsf{S}\nonumber\\
    &\leq\f{\mbb1_\mathsf{S}}{d_\mathsf{S}},
\end{align}
with equality for \gls{TP} maps, i.e. with $\tr[\mc{A}_i(\mbb1)]=\tr(\mbb1)=d_\mathsf{S}$.
\end{example}

This generalizes similarly for any $k$, implying that the Haar average process tensor is maximally noisy, or analogous to quantum channel terms, completely depolarizing. Notice as well that this average process is Markovian, as the Choi state has a tensor product structure.

This fact relies upon the independent sampling of each unitary; however, if we consider a constant interaction, with all $\mc{U}_i=\mc{U}_j$, we require a single integral equivalent to the $(k+1)$-moment of the unitary group, as defined in Eq.~\eqref{eq: n-moments of the unitary group}. In Section~\ref{sec: Random states and Haar} we were able to compute the first and second moments by means of the twirl map and the Schur-Weyl duality. However, for higher-order moments this is not practical anymore, as per by the Schur-Weyl duality we need to consider all permutations and then relate back the twirl map with the moments of the unitary group. We can, however, capture in an abstract form the behaviour of the moments of the unitary group by means of an object called the \emph{Weingarten function} (which in a sense also relies on the Schur-Weyl duality); this will let us, in particular, understand the asymptotic behaviour in dimensions and time-steps for random processes, as will become clear below.

The Weingarten function can be defined in different ways; it is a fairly complicated function to evaluate explicitly~\cite{GuMoments, ZhangMInt} but in any case, for the $n$-moments of the unitary group, it only depends on a given permutation of $n$ and it gives a rational number in the dimension of the unitary in question. Tables with particular cases are often cited in the literature, which are helpful for computing lower moments, see e.g. Ref.~\cite{GuMoments, Roberts2017}. An alternative is to perform calculations numerically~\cite{Puchala_2017, Ginory_2019}). For our purposes, we keep the definition of the Weingaten function as follows.

\begin{definition}[Weingarten function~\cite{GuMoments}]
    Let $\mathfrak{G}_n$ be the symmetric group on $\{1,2,\ldots,n\}$, and let $n\leq{d}$ and $\sigma\in\mathfrak{G}_n$, then the (unitary) Weingarten function,\footnote{ $\Wg$ is named after Donald Weingarten~\cite{Weingarten}, who first studied asymptotic properties of the $n$-moments of $\mbb{U}(d)$. An explicit expression in terms of characters of symmetric groups and Schur functions was first derived by Beno\^{i}t Collins~\cite{Collins_2003} (later expanded in Ref.~\cite{Collins_2006}), who coined the term.} $\Wg$, is defined by
    \begin{equation}
        \Wg(\sigma,d)=\mbb{E}_\haar\left[\,\prod_{i=1}^nU_{ii}U_{i\sigma(i)}^*\right],
    \end{equation}
    where $U\sim\muhaar$ is a Haar distributed $d\times{d}$ unitary matrix.
\end{definition}

More explicit expressions of $\Wg$, such as the one derived in Ref.~\cite{Collins_2003}, are outside the scope of this thesis, however, it is important to point out that $\Wg$ depends not only on the particular permutation $\sigma$ but on its cycle structure, that is, for example, for $n=3$, the cyclic permutations $\sigma=(1,2,3)$ and $\sigma^\prime=(1,3,2)$, meaning $\sigma:1\to2\to3\to1$ and $\sigma^\prime:1\to3\to2\to1$, will evaluate to the same value on $\Wg$, i.e. $\Wg(\sigma,d)=\Wg(\sigma^\prime,d)$. We denote this cycle type with brackets as follows: given a set of positive integers $\{m_i\}$ such that $m_1+m_2+\ldots+m_\ell=n$, then $[m_1,m_2,\ldots,m_\ell]$ means cyclic permutations replacing $m_1$, $m_2$ and $m_\ell$ elements. For the previous example then the cycle type of $\sigma$ and $\sigma'$ is $[3]$, similarly for the permutations $(1,2)(3)$, $(1,3)(2)$ and $(2,3)(1)$ it will be $[2,1]$, whilst for $(1)(2)(3)$ we denote it as $\boldsymbol{1}^3:=[1,1,1]$.

The similarity with the expression for the $n$-moments of $\mbb{U}(d)$ in Eq.~\eqref{eq: n-moments of the unitary group} is now quite evident, and it implies the following.

\begin{theorem}[The $n$-moments of the unitary group~\cite{Collins_2003,GuMoments}]
Let $U\in\mbb{U}(d)$ with $U\sim\muhaar$ be a a $d\times{d}$ Haar random distributed unitary and $n\leq{d}$, then
\begin{equation}
    \mbb{E}_\haar\left[\,\prod_{\ell=1}^nU_{i_\ell{j}_\ell}U_{i^\prime_\ell{j}^\prime_\ell}^*\right]=\sum_{\sigma,\tau\in\mathfrak{G}_n}\prod_{\ell=1}^n\delta_{i_\ell{i}^\prime_{\sigma(\ell)}}\delta_{j_\ell{j}^\prime_{\tau(\ell)}}\,\mathrm{Wg}(\tau\sigma^{-1},d),\label{k moments of U}
\end{equation}
where $U_{ij}$ is the $ij$\textsuperscript{th} entry of $U$.
\end{theorem}

A property that follows from this result, together with the invariance of the Haar measure, is that $\mbb{E}_\haar\left[U_{i_1j_1}\cdots{U}_{i_nj_n}U_{i^\prime_1{j}^\prime_1}^*\cdots{U}_{i^\prime_{n'}{j}^\prime_{n'}}^*\right]=0$, i.e. there has to be the same amount of $U_{ij}^*$ as there are $U_{ij}$~\cite{GuMoments}. As we have already computed the first and second moments in Eq.~\eqref{eq: average quantum state} and in Eq.~\eqref{eq: Schur-Weyl 2-twirl} (together with Eq.~\eqref{eq: Schur-Weyl constants}), albeit indirectly through the twirl map, we can nevertheless compare with the values in the table of Ref.~\cite{GuMoments} for the corresponding Weingarten functions,
\begin{equation}
    \Wg([1],d)=\f{1}{d},\qquad\Wg([\boldsymbol{1}^2],d)=\f{1}{d^2-1},\qquad\Wg([2],d)=-\f{1}{d(d^2-1)}.
\end{equation}

For the first moment, the Weingarten function corresponds to the only coefficient of the 1-twirl, whilst for the second moment the Weingarten functions appear as coefficients in the $\alpha$, $\beta$ functions of the 2-twirl in Eq.~\eqref{eq: Schur-Weyl constants}, that is, in fact we can write the 1-twirl $\Xi^{(2)}$ and the 2-twirl $\Xi^{(2)}$ as
\begin{align}
    \Xi^{(1)}[(\cdot)]&=\Wg([1],d)\tr[(\cdot)],\\
    \Xi^{(2)}[(\cdot)]&=\Wg([\boldsymbol{1}^2],d)Z+\Wg([2],d)\swap Z,\quad\text{where}\quad{Z}=\tr[(\cdot)]\mbb1+\tr[\swap(\cdot)]\swap,
\end{align}
and so similarly, the $n$-twirl can be written in terms of the corresponding $n$ Weingarten functions.

As mentioned before, this is generally not an easy thing to do and is usually done only for moments of small $n$; however, the asymptotic behaviour is usually one of interest in random matrix theory (or more generally in any theory of non-commutative random variables) and in our case it will help us establish our results for Haar distributed quantum processes. Ultimately, the asymptotic behavior of the moments of the unitary group boils down to that of the $\Wg$ function.

\begin{remark}In Ref.~\cite{GuMoments} it is shown that
\begin{gather}
    \Wg(\sigma\in\mathfrak{G}_n,d)\sim\f{1}{d^{\,2n-\#\sigma}},\,\,\text{as}\,\,d\to\infty,
    \label{Wg asympt result}
\end{gather}
as a refinement of a result in Ref.~\cite{Collins_2006}, where $\#\sigma$ is the number of cycles of the permutation $\sigma$ counting also fixed points (assignments from an element to itself, $\sigma(x)=x$).
\end{remark}

Let us then consider the average process in the constant interaction case, $\mc{U}_i=\mc{U}_j=\mc{U},\,\forall{i,j\leq{k}}$. As it is clear that for us $d=d_\mathsf{SE}$, we will omit the dimension dependence in $\Wg$. As we originally derived\footnote{ In Ref.~\cite{FigueroaRomero2019almostmarkovian} we normalized each maximally entangled state entering the Choi state of the process tensor; the reason for this will be made clear in the following sections.} in Ref.~\cite{FigueroaRomero2019almostmarkovian} and we reproduce in Appendix~\ref{appendix - Average process constant}, the average $k$-step process tensor in the constant interaction case can be written by means of a set $\{|s_i^{(\prime)}\rangle\}_{s_i=1}^{d_\mathsf{S}}$ of \gls{syst} system bases for $i=0,1,\ldots,k$ as
\begin{equation}
    \mbb{E}_\haar[\Upsilon]=\sum_{\sigma,\tau\in\mathfrak{G}_{k+1}}\!\!\rho_{_{\tau(0);0}}\Wg(\tau\sigma^{-1})\Delta_{k,\sigma,\tau}^{(d_\mathsf{E})}|s_{\sigma(k)}\rangle\!\langle{s}_k|\bigotimes_{j=1}^k|s_{\sigma(j-1)}s^\prime_{\tau(j)}\rangle\!\langle{s}_{j-1}s^\prime_{j}|,
\label{average state Ui=Uj}
\end{equation}
with implicit sum over all repeated basis ($s^{(\prime)}_i$) indices, where here $\mathfrak{G}_{k+1}$ is the symmetric group on $\{0,1,\ldots,k\}$, and with the definitions
\begin{align}
    \rho_{_{\tau(0);0}}&=\langle{e}^\prime_{\tau(0)}s^\prime_{\tau(0)}|\rho|e^\prime_0s^\prime_0\rangle,
    \label{phi notation Ui=Uj}\\ \Delta_{k,\sigma,\tau}^{(d_\mathsf{E})}&=\delta_{e_{\sigma(k)}e_k}\prod_{\ell=1}^k\delta_{e_{\sigma(\ell-1)}e^\prime_{\tau(\ell)}}\delta_{e_{\ell-1}e^\prime_\ell},
    \label{DeltaEstate}
\end{align}
where $\{|e_i\rangle\}_{e_i=1}^{d_\mathsf{E}}$, $\{|e_i^\prime\rangle\}_{e_i^\prime=1}^{d_\mathsf{E}}$, with $i=0,1,\ldots,k$, also implicitly summed over all elements $e_i$ and $e_i^\prime$, is a set of \gls{env} bases and the $\Delta$ term is simply a monomial in $d_\mathsf{E}$, with degree determined by $\sigma$ and $\tau$.

The case $k=0$ recovers $\mbb{E}_\haar[\Upsilon_{0:0}]=\mbb{1}_S/d_\mathsf{S}$ as expected, as no process occurs; the $1/d_\mathsf{S}$ factor arises from the $\Wg([1],d_\mathsf{SE})$ function, with the unifying $d_\mathsf{E}$ factor coming from $\Delta_{0,\sigma,\tau}^{d_\mathsf{E}}$. In Appendix~\ref{appendix: average state Ui=Uj superchannel} we also write the case $k=1$ for a superchannel and its purity, which can be seen to be close to the maximally mixed one and coincide with it in the large \gls{env} limit, i.e. it coincides with the random interaction case.

As we saw in the whole Chapter~\ref{sec:statmech}, the small subsystem limit, $d_\mathsf{E}\gg{d}_\mathsf{S}$, is of particular interest. For the average process tensor, we saw that the random interaction case in Eq.~\eqref{eq: average process rand} is independent of $d_\mathsf{E}$, however for the constant interaction case we will get terms in inverse powers of $d_\mathsf{E}$ arising from the $\Wg$ functions. In particular, when looking at the limit $d_\mathsf{E}\to\infty$ of Eq.~\eqref{average state Ui=Uj}, the only term that does not vanish is the one with $\sigma,\tau=\boldsymbol{1}^{k+1}$, i.e. with both permutations being identities, as these generate the most numerator powers in $d_\mathsf{E}$ in the $\Delta_{k,\sigma,\tau}^{(d_\mathsf{E})}$ term in Eq.\eqref{DeltaEstate} when summed over all $e_i$'s. In other words, we get the contribution from $\Delta_{k,\sigma,\tau}^{(d_\mathsf{E})}$,
\begin{gather}
    \sum_{\substack{e_0,e_1,\ldots,e_k=1\\e^\prime_1,e^\prime_2,\ldots,e^\prime_k=1}}^{d_\mathsf{E}}\left[\delta_{e_k e_k}\prod_{\ell=1}^k\delta_{e_{\ell-1}e^\prime_{\ell}}\delta_{e_{\ell-1}e^\prime_\ell}\right]=d_\mathsf{E}^{k+1},
\end{gather}
and all other terms will vanish because of the $d_\mathsf{E}$ powers in the denominator generated by the $\Wg$ functions will dominate those from $\Delta_{k,\sigma,\tau}^{(d_\mathsf{E})}$. Given the asymptotic limit of $\Wg$ in Eq.~\eqref{Wg asympt result} we can see that indeed the least powers in $d$ produced by it are those when $\sigma\tau^{-1}=\mathbf{1}^{k+1}$ because $\#\boldsymbol{1}^n=n$, i.e. identity produces the greatest number of cycles, being the number of all possible fixed points. Finally, as $\sum_{\epsilon^\prime_0\varsigma^\prime_0}\langle\epsilon^\prime_0\varsigma^\prime_0|\rho|
\epsilon^\prime_0\varsigma^\prime_0\rangle=\tr\rho=1$, we get
\begin{align}
    \mbb{E}_\haar[\Upsilon]&\sim d_\mathsf{E}^{k+1}\mathrm{Wg}(\boldsymbol{1}^{k+1})\,\mbb{1}_{\mathsf{SA}_1\mathsf{B}_1\ldots\mathsf{A}_k\mathsf{B}_k}\nonumber\\
    &\sim\f{\mbb{1}_{\mathsf{SA}_1\mathsf{B}_1\ldots\mathsf{A}_k\mathsf{B}_k}}{d_\mathsf{S}^{k+1}},\quad\text{when}\quad d_\mathsf{E}\to\infty,
\end{align}
coinciding with the average over a random interaction process.

\section{A bound on non-Markovianity \& average noisiness}
As we saw previously in Section~\ref{sec: measure of non-Markovianity}, the process tensor leads to a well-defined Markov criterion from which it is possible to construct a family of operationally meaningful measures of non-Markovianity, many of which can be stated simply as \emph{distances} between a process tensor's Choi state $\Upsilon$ and the closest Markovian one $\Upsilon^\markov$. We saw as well that $\Upsilon^\markov$ must take the form of a tensor product of quantum maps $\mc{Z}_{i:i-1}$ connecting adjacent pairs of time steps, $\Upsilon^\markov = \bigotimes_{i=1}^k\mc{Z}_{i:i-1}\otimes\rho_\mathsf{S}^{(0)}$.

We are now interested in studying the non-Markovianity of a generic quantum process. In Ref.~\cite{FigueroaRomero2019almostmarkovian}, in analogy with the seminal studies on equilibration that we introduced in Chapter~\ref{sec:statmech}, we chose the measure of non-Markovianity defined in terms of the trace distance $D$ as
\begin{gather}
    \mc{N}_1:=\min_{\Upsilon^\markov}D\left(\Upsilon,\Upsilon^\markov\right)
\propto\min_{\Upsilon^\markov}\|\Upsilon-\Upsilon^\markov\|_1
\label{def nonMarkov trDistance}
\end{gather}
where $\|X\|_1:=\tr\sqrt{XX^\dg}$ is the trace norm (or Schatten 1-norm as defined by Eq.~\eqref{eq: Schatten p-norm}). In particular, this trace-distance measure is related to relative entropy, briefly discussed in Section~\ref{sec: measure of non-Markovianity}, through the so-called quantum Pinsker inequality, $\mc{R}(\Upsilon\|\Upsilon^\markov)\geq2\,D^2(\Upsilon,\Upsilon^\markov)$. Now, given that $\mc{N}_1$ is a trace-distance measure, we demand that it satisfies $0\leq\mc{N}_1\leq1$. This can be done imposing a normalization factor or directly normalizing the Choi states. We chose the latter,\footnote{ The reasons might be said to be somewhat historical; one reason that motivated this choice is that it is more intuitive to think of properly normalized maximally mixed states as maximally noisy processes, as opposed to having an identity with an incorrect normalization factor.} with the caveat that we employ it consistently only for the purposes of employing this measure of non-Markovianity.

\begin{notation}
From this point onward we normalize the Choi state of the process tensor to unity, which is equivalent to defining these via properly normalized maximally entangled states, i.e. from here on we redefine $\Upsilon$ as
\begin{equation}
    \Upsilon=\tr_\mathsf{E}[\,\mc{U}_k\mc{S}_k\,\mc{U}_{k-1}\mc{S}_{k-1}\cdots\mc{U}_1\mc{S}_1\mc{U}_0(\rho\otimes\Psi^{\otimes\,k})],\quad\text{where}\quad\mathsf{\Psi}=\f{1}{d_\mathsf{S}}\tilde{\Psi},
\end{equation}
where $\mathsf{\Psi}$ are now rightful maximally entangled states, so that now $\Upsilon$ satisfies $\tr[\Upsilon]=1$. This implies that now we explicitly have
\begin{equation}
    \mc{N}_1:=\f{1}{2}\min_{\Upsilon^\markov}\|\Upsilon-\Upsilon^\markov\|_1.
\end{equation}
\end{notation}

Now, we should point out that the choice of this measure is historical as well, as the trace distance is rather a distinguishability measure for quantum states rather than for quantum maps, and a distinguishability measure on Choi states would be directly a measure on these states rather than an operationally meaningful one on the corresponding processes. As we will see throughout the next chapter, this is not an issue in itself as we can tightly bound the measure $\mc{N}_1$ with an operationally relevant measure of non-Markovianity known as the diamond distance, and the result we obtained in Ref.~\cite{FigueroaRomero2019almostmarkovian} is only changed by a multiplicative constant in a minor way, ultimately not changing its consequences. In fact, we will see that we can very often relate families of non-Markovianity measures, such as those given by Schatten norms, by at most a multiplicative factor. However, the results in this chapter and the numerical calculations were originally obtained with $\mc{N}_1$. We will thus, for consistency and for simplicity's sake, present the results as originally derived with this measure in this thesis.\footnote{ See Eq.~\eqref{eq: non-Markov concentration Haar} for the main result in this Chapter in terms of diamond norm.}

To begin with, given the difficulty in computing and minimizing the Markovian Choi state, we may upper bound the distance $\mc{N}_1$ by a trace distance with respect to the maximally mixed state, which as we have seen would correspond to the noisiest Markovian process possible,
\begin{gather}
    \mc{N}_1\leq{D}\left(\Upsilon,\f{\mbb1}{d_\mathsf{S}^{2k+1}}\right),
    \label{TrDist 1st bound}
\end{gather}
where the identity acting on \gls{syst} together with the $k+1$ ancillas is implied. We may further bound this by considering the following cases separately.

\textbf{\textsf{1. Case}} $d_\mathsf{E}<d_\mathsf{S}^{2k+1}$: We notice that $\mathrm{rank}(\Upsilon)\leq{d_\mathsf{E}}$.\footnote{ One may see this by looking at the pure state ${|\Phi\rangle\!\langle\Phi|\equiv\mc{U}(\Theta\otimes\Psi^{\otimes{k}})\,\mc{U}^\dg}$ with Schmidt decomposition ${|\Phi\rangle=\sum_{i=1}^n\sqrt{\lambda_i}|e_is_i\rangle}$ where $n=\min\left(d_\mathsf{E},d_\mathsf{S}^{2k+1}\right)$.} Letting $\gamma$ be the diagonal matrix of up to $d_\mathsf{E}$ non-vanishing eigenvalues $\lambda_{\gamma_i}$ of the Choi state, we may write
\begin{align}
    \|\Upsilon-\f{\mbb1}{d_\mathsf{S}^{2k+1}}\|_1&=\sum_{i=1}^{d_\mathsf{E}}\left|\lambda_{\gamma_i}-\f{1}{d_\mathsf{S}^{2k+1}}\right|+\sum_{j=d_\mathsf{E}+1}^{d_\mathsf{S}^{2k+1}}\left|-\f{1}{d_\mathsf{S}^{2k+1}}\right|\nonumber\\
    &=\|\gamma-\f{\mbb1_\mathsf{E}}{d_\mathsf{S}^{2k+1}}\|_1+1-\f{d_\mathsf{E}}{d_\mathsf{S}^{2k+1}},
\end{align}
where $|\cdot|$ denotes the standard absolute value, so using the inequality ${\|X\|_1\leq\sqrt{\mathrm{dim}(X)}\|X\|_2}$ for a square matrix $X$, where $\|X\|_2=\sqrt{\tr(XX^\dg)}$ is the Schatten 2-norm,
\begin{align}
    \|\Upsilon-\f{\mbb1}{d_\mathsf{S}^{2k+1}}\|_1 & \leq \sqrt{d_\mathsf{E}}\|\gamma-\f{\mbb1_\mathsf{E}}{d_\mathsf{S}^{2k+1}}\|_2+1-\f{d_\mathsf{E}}{d_\mathsf{S}^{2k+1}} \nonumber \\
    &=\sqrt{d_\mathsf{E}\tr[\Upsilon^2]+\f{d^2_\mathsf{E}}{d_\mathsf{S}^{4k+2}}-\f{2 d_\mathsf{E}} {d_\mathsf{S}^{2k+1}}} +1-\f{d_\mathsf{E}}{d_\mathsf{S}^{2k+1}},
\end{align}
Furthermore, applying Jensen's inequality for the square-root, $\mbb{E}[\sqrt{X}]\leq\sqrt{\mbb{E}[X]}$, for any $\mbb{E}_\haar$ or $\mbb{E}_\haarrand$, we have
\begin{gather}
    \mbb{E}\left[\mc{N}_1\right]\leq\f{1}{2}\left(\sqrt{d_\mathsf{E}\mbb{E}[\tr(\Upsilon^2)]+\f{d_\mathsf{E}^2}{d_\mathsf{S}^{4k+2}}-\f{2d_\mathsf{E}}{d_\mathsf{S}^{2k+1}}}+1-\f{d_\mathsf{E}}{d_\mathsf{S}^{2k+1}}\right).
    \label{up bound2}
\end{gather}

\textbf{\textsf{2. Case}} $d_\mathsf{E}\geq{d}_\mathsf{S}^{2k+1}$: This case is a small subsystem limit for most $k$. Directly applying ${\|X\|_1\leq\sqrt{\mathrm{dim}(X)}\|X\|_2}$ as before,
\begin{gather}
    \|\Upsilon-\f{\mbb1}{d_\mathsf{S}^{2k+1}}\|_1\leq\sqrt{d_\mathsf{S}^{2k+1}}\|\Upsilon-\f{\mbb1}{d_\mathsf{S}^{2k+1}}\|_2 = \sqrt{d_\mathsf{S}^{2k+1}\tr[\Upsilon^2] -1},
\end{gather}
Similarly, taking the average over evolution, by means of Jensen's inequality,
\begin{equation}
    \mbb{E}\left[\mc{N}_1\right]\leq\f{1}{2}\sqrt{d_\mathsf{S}^{2k+1}\mbb{E}\left[\tr(\Upsilon^2)\right]-1}.
    \label{up bound1}
\end{equation}

This means that we can construct a piecewise function
\begin{gather}
    \mc{B}_k(d_\mathsf{E},d_\mathsf{S})\equiv
    \begin{cases}
    \displaystyle{\f{\sqrt{d_\mathsf{E}\,\mbb{E}[\tr(\Upsilon^2)]-x}+y}{2}}
    &\text{if}\quad d_\mathsf{E}<d_\mathsf{S}^{2k+1}
    \\[0.1in]
    \displaystyle{\f{\sqrt{d_\mathsf{S}^{2k+1}\mbb{E}[\tr(\Upsilon^2)]-1}}{2}}
    &\text{if}\quad d_\mathsf{E}\geq{d}_\mathsf{S}^{2k+1}
    \end{cases}\quad,
    \label{eq: bound on nM}
\end{gather}
with $x=d_\mathsf{E}d_\mathsf{S}^{-(2k+1)}\left(1+y\right)$ and $y=1-d_\mathsf{E}d_\mathsf{S}^{-(2k+1)}$, which provides an upper bound on the average non-Markovianity, i.e.
\begin{equation}
    \mbb{E}[\mc{N}_1]\leq\mc{B}_k.
\end{equation}

So similar to Section~\ref{sec: state typicality}, we now need to compute the average purity of $\Upsilon$.

The purity is a quantifier of the mixedness, or uniformity of eigenvalues of a positive operator, and as we saw as well in Section~\ref{sec: state typicality}, when computed on reduced states of bipartite systems it can serve as a quantifier of entanglement. In Appendix~\ref{appendix - average purity random} and in Appendix~\ref{appendix - average purity constant}, we reproduce the computation of the expected purity of the Choi state of a process tensor $\mathbb{E}[\Upsilon^2]$ in the random and the constant interaction pictures, respectively, which we derived originally in Ref.~\cite{FigueroaRomero2019almostmarkovian}. This average purity can be directly translated as a quantifier for \emph{noisiness} of the quantum process itself and can serve as well to measure the entanglement between system and environment. Bear in mind again that here the process tensors are normalized to unity. The average purities take the form
\begin{gather}
    \mbb{E}_\haarrand[\tr\left(\Upsilon^2\right)]=\f{d_\mathsf{E}^2-1}{d_\mathsf{E}(d_\mathsf{SE}+1)}\left(\f{d_\mathsf{E}^2-1}{d_\mathsf{SE}^2-1}\right)^k+\f{1}{d_\mathsf{E}},
    \label{average purity ergodic}
\end{gather}
for the random interaction picture, where we have assumed that the fiducial state $\rho$ is pure, and
\begin{gather}
    \mbb{E}_\haar[\tr\left(\Upsilon^2\right)]=d_\mathsf{S}^{-2k}\sum_{\sigma,\tau\in\mathfrak{G}_{2k+2}}\Wg(\tau\sigma^{-1})\,\rho_{_{\tau(0);k+1}}\rho_{_{\tau(k+1);0}}\,\Delta_{k,\sigma,\tau}^{(d_\mathsf{E},d_\mathsf{S})},
    \label{average time independent purity}
\end{gather}
in the constant interaction case, where $\rho_{\tau(\cdot);\ell}$ uses the same notation as in Eq.~\eqref{phi notation Ui=Uj} and is shown explicitly in Eq.~\eqref{appendix: purity const rho} on Appendix~\ref{appendix - average purity constant}, and $\Delta$ is a product, scaling with $k$, of monomials in $d_\mathsf{E}$ and $d_\mathsf{S}$ depending on permutations $\sigma$ and $\tau$, shown in full in Eq.~\eqref{TIpuritydef} also within Appendix~\ref{appendix - average purity constant}.

\subsection{Limiting cases}
In both the constant and random interaction cases, the bound $\mc{B}_k$ on the non-Markovianity $\mc{N}_1$ is a well-behaved rational function of $d_\mathsf{E}$, $d_\mathsf{S}$ and $k$. In the constant interaction case, Eq.~\eqref{average time independent purity} takes a non-trivial form mainly because of the $\Wg$ function (which is intrinsic to the Haar-unitary averaging in the constant interaction picture). However, due to the results in Ref.~\cite{GuMoments,Collins_2006}, we can still study analytically the behaviour of the bound $\mc{B}_k$ for both cases in the following limits.

\subsubsection{Small subsystem limit} For both interaction pictures,
\begin{gather}
    \lim_{d_\mathsf{E}\to\infty}\mbb{E}[\tr(\Upsilon^{2})]=\f{1}{d_\mathsf{S}^{2k+1}},
\end{gather}
which corresponds to the purity of the maximally mixed state. Note that the averaging occurs after computing the purity independently of which case is considered, i.e. this does not correspond to the purity of the average process but the average purity of a process.

This implies that on average, in the small subsystem limit, a process will be indistinguishable from the maximally noisy (and hence Markovian) one,
\begin{gather}
    \lim_{d_\mathsf{E}\to\infty}\mbb{E}[\mc{N}_1]=\lim_{d_\mathsf{E}\to\infty}\mc{B}_k(d_\mathsf{E},d_\mathsf{S})=0,
\end{gather}
and from Eq.~\eqref{average purity ergodic} we know it does so at a rate $\mc{O}(1/d_\mathsf{E})$ in the random interaction case.

\subsubsection{Long time limit}
The other interesting limiting case is the one where the \gls{syst-env} dimension is fixed, but the number of time steps is taken to be very large. The resulting process encodes all high order correlation functions between observables over a long period of time. Again for both cases, the expected purity in this limit goes as
\begin{gather}
    \lim_{k\to\infty}\mbb{E}[\tr(\Upsilon^{2})]=\f{1}{d_\mathsf{E}},
\end{gather}
which corresponds to the maximally mixed purity of the environment. This implies that the Choi state of the full \gls{syst-env} unitary process is maximally entangled between \gls{syst} and \gls{env}. In this limit, we get correspondingly
\begin{gather}
    \lim_{k\to\infty}\mbb{E}[\mc{N}_1]\leq\lim_{k\to\infty}\mc{B}_k(d_\mathsf{E},d_\mathsf{S})=1,
\end{gather}
meaning only that we cannot say much about non-Markovianity in this limit, i.e. a typical process in this limit could be highly non-Markovian. Indeed, we expect this to be the case, since the finite-dimensional \gls{syst-env} space will have a finite recurrence time.

\subsubsection{Average state purity limit}
As expected as well, our result generalizes the well-known result for quantum states, i.e., when $k=0$. In this case we recover the average purity
\begin{equation}
    \mbb{E}[\tr(\rho_S^{2})]=\f{d_\mathsf{E}+d_\mathsf{S}}{d_\mathsf{SE}+1},
\end{equation}
where $\rho_S\equiv\Upsilon_{0:0}=\tr_E(U\rho\,U^\dg)$, as we re-derived in Eq.~\eqref{eq: average purity reduced state}.

\section{Concentration around Markovian processes}
From the discussion in Section~\ref{sec: Random states and Haar} it seems clear that we have almost all the ingredients to turn the statements above into a concentration of measure result. We proved the following.

\begin{theorem}[Concentration of measure around Markovian processes~\cite{FigueroaRomero2019almostmarkovian}]
\label{Thm: typicality processes}
Let $\Upsilon\sim\muhaar$ be a $k$-step quantum process undergone by a $d_\mathsf{S}$-dimensional subsystem of a larger $d_\mathsf{SE}$-dimensional composite, sampled at random according to the Haar measure, then  for any $\delta>0$,
\begin{gather}
    \mbb{P}_\haar[\,\mc{N}_1\geq\mc{B}_k(d_\mathsf{E},d_\mathsf{S})+\delta]\leq\mathrm{e}^{-\mscr{C}(d_\mathsf{E},d_\mathsf{S})\delta^2},
    \label{eq: main result}
\end{gather}
where $\mscr{C}(d_\mathsf{E},d_\mathsf{S})=c \, d_\mathsf{SE}\left(\f{d_\mathsf{S}-1}{d_\mathsf{S}^{k+1}-1}\right)^2$ with $c=1/4$ for a constant interaction process and $c=(k+1)/4$ for a random interaction process. The function $\mc{B}_k$ is an upper bound on the expected non-Markovianity $\mbb{E}[\mc{N}_1]$, given in Eq.~\eqref{eq: bound on nM}, whose details depend on the way in which processes are sampled.
\end{theorem}

\begin{proof}We have previously obtained the function $\mc{B}_k$ completely, so we now derive the Lipschitz constants $\mscr{L}$ for both the constant and the random interaction cases, as well as the concentration rate.

\subsubsection{Lipschitz constant (constant interaction)}
Let $\tau:\mbb{U}(d_\mathsf{SE})\to\mathbb{R}$ defined by $\tau(U)=D\left(\Upsilon(U),\f{\mbb1}{d_\mathsf{S}^{2k+1}}\right)=\f{1}{2}\|\Upsilon(U)-\f{\mbb1}{d_\mathsf{S}^{2k+1}}\|_1$, where we explicitly mean $\Upsilon(U)=\tr_E[\mathsf{U}_{k:0}\Theta\mathsf{U}_{k:0}^\dg]$ where $\mathsf{U}_{k:0}=U_k\mc{S}_k\cdots{U}_1\mc{S}_1U_0$ and with $\Theta=\rho\otimes\mathsf{\Psi}^{\otimes{k}}$ a pure state, then for all $V\in\mbb{U}(d_\mathsf{SE})$,
\begin{align}
    |\tau({U})-\tau({V})|&=\left|D\left(\Upsilon(U),\f{\mbb1}{d_\mathsf{S}^{2k+1}}\right)-D\left(\Upsilon(V),\f{\mbb1}{d_\mathsf{S}^{2k+1}}\right)\right|\nonumber\\
    &\leq{D}(\Upsilon(U),\Upsilon(V))\nonumber\\
    &\leq{D}(\mathsf{U}_{k:0}\Theta\mathsf{U}_{k:0}^\dg,\mathsf{V}_{k:0}\Theta\mathsf{V}_{k:0}^\dg),
\end{align}
where here similarly $\mathsf{V}_{k:0}=V_k\mc{S}_k\cdots{V}_1\mc{S}_1V_0$.

Consider now different labelings for the unitary at each time step so that we can easily track each one. By the triangle inequality,
\begin{align}
    \|\mathsf{U}_{k:0}\Theta\mathsf{U}_{k:0}^\dg-\mathsf{V}_{k:0}\Theta\mathsf{V}_{k:0}^\dg\|_1&\leq\|\mathsf{U}_{k:0}\Theta\mathsf{U}_{k:0}^\dg-U_k\mc{S}_k\mathsf{V}_{k-1:0}\Theta\mathsf{V}_{k-1:0}^\dg\mc{S}_k^\dg{U}_k^\dg\|_1\nonumber\\
    &\qquad\qquad\qquad\qquad+\|U_k\mc{S}_k\mathsf{V}_{k-1:0}\Theta\mathsf{V}_{k-1:0}^\dg\mc{S}_k^\dg{U}_k^\dg-\mathsf{V}_{k:0}\Theta\mathsf{V}_{k:0}^\dg\|_1\nonumber\\
    &=\|\mathsf{U}_{k-1:0}\Theta\mathsf{U}_{k-1:0}^\dg-\mathsf{V}_{k-1:0}\Theta\mathsf{V}_{k-1:0}^\dg\|_1\nonumber\\
    &\qquad\qquad\qquad\qquad+\|U_k\mc{S}_k\mathsf{V}_{k-1:0}\Theta\mathsf{V}_{k-1:0}^\dg\mc{S}_k^\dg{U}_k^\dg-\mathsf{V}_{k:0}\Theta\mathsf{V}_{k:0}^\dg\|_1,
    \label{LipIneq2}
\end{align}
where it is clear by context that $U$ stands for $U\otimes\mbb1$.

Now we use Lemma 1 of Ref.~\cite{qspeedLipsch}, which states that
\begin{gather}\|A\sigma{A}^\dg-B\sigma{B}^\dg\|_1\leq2\|A-B\|_2\end{gather}
for two unitaries $A$, $B$ and any $\sigma$. For the simplest case with $k=1$, this gives
\begin{align}
    &\|U_1\mc{S}_1V_0\Theta{V}_0^\dg\mc{S}_1^\dg{U}_1^\dg-V_1\mc{S}_1V_0\Theta{V}_0^\dg\mc{S}_1^\dg{V}_1^\dg\|_1\nonumber\\
    &\qquad\qquad=\|\f{\sum}{d_\mathsf{S}}(U_1\FS_{\alpha\beta}V_0\rho{V}_0^\dg\FS_{\gamma\delta}U_1^\dg-{V}_1\FS_{\alpha\beta}V_0\rho{V}_0^\dg\FS_{\gamma\delta}V_1^\dg)\otimes|\beta\alpha\rangle\!\langle\gamma\delta|\|_1\nonumber\\
    &\leq\f{\sum_{\alpha\beta}}{d_\mathsf{S}}\tr\bigg[\sum_{\gamma\delta}(U_1\FS_{\alpha\beta}V_0\rho{V}_0^\dg\FS_{\gamma\delta}U_1^\dg-{V}_1\FS_{\alpha\beta}V_0\rho{V}_0^\dg\FS_{\gamma\delta}V_1^\dg)\nonumber\\
    &\qquad\qquad\qquad\qquad\qquad\qquad\qquad\qquad(U_1\FS_{\beta\alpha}V_0\rho{V}_0^\dg\FS_{\delta\gamma}U_1^\dg-{V}_1\FS_{\beta\alpha}V_0\rho{V}_0^\dg\FS_{\delta\gamma}V_1^\dg)\bigg]^{1/2}\nonumber\\
    &\leq\f{\sum_{\alpha\beta}}{d_\mathsf{S}}\|U_1\FS_{\alpha\beta}V_0\rho{V}_0^\dg\left(\sum\FS_{\gamma\delta}\right)U_1^\dg-{V}_1\FS_{\alpha\beta}V_0\rho{V}_0^\dg\left(\sum\FS_{\gamma^\prime\delta^\prime}\right)V_1^\dg\|_1\nonumber\\
    &\leq\f{2}{d_\mathsf{S}}\sum_{\alpha,\beta=1}^{d_\mathsf{S}}\|U_1-V_1\|_2=2d_\mathsf{S}\|U_1-V_1\|_2,
\end{align}
and doing similarly, for the $i$-th step,
\begin{align}
    \|U_i\mc{S}_i\mathsf{V}_{i-1:0}\Theta\mathsf{V}_{i-1:0}^\dg\mc{S}_i^\dg{U}_i^\dg-\mathsf{V}_{i:0}\Theta\mathsf{V}_{i:0}^\dg\|_1
    &\leq2d_\mathsf{S}^i\|U_i-V_i\|.
\end{align}

Thus bounding iteratively expression~\eqref{LipIneq2}, it follows that
\begin{align}
    \|\mathsf{U}_{k:0}\Theta\mathsf{U}_{k:0}^\dg-\mathsf{V}_{k:0}\Theta\mathsf{V}_{k:0}^\dg\|_1&\leq2\sum_{\ell=0}^kd_\mathsf{S}^\ell\|U_\ell-V_\ell\|_2,
\end{align}
finally giving
\begin{align}
    |\tau(U)-\tau(V)|&\leq\left(\f{d_\mathsf{S}^{k+1}-1}{d_\mathsf{S}-1}\right)\|U-V\|_2.
\end{align}

\subsubsection{Lipschitz constant (random interaction)}
On the other hand, for the ergodic case, let $\mbb{U}^{\times(k+1)}(d)=\underbrace{\mbb{U}(d)\times\cdots\times\mbb{U}(d)}_{k+1\,\text{times}}$ be the $k+1$ Cartesian product space of $d$-dimensional unitary groups, then we define $\zeta:\mbb{U}^{\times(k+1)}(d_\mathsf{SE})\to\mathbb{R}$ by $\zeta(\vec{U})=D\left(\Upsilon(\vec{U}),\f{\mbb1}{d_\mathsf{S}^{2k+1}}\right)$ where now $\vec{U}=(U_0,\cdots,U_k)$. Similarly as before, we now have
\begin{gather}
    \left|\zeta(\vec{U})-\zeta(\vec{V})\right|\leq\sum_{\ell=0}^kd_\mathsf{S}^\ell\|U_\ell-V_\ell\|_2,
\end{gather}
and we may let the metric on $\mbb{U}^{\times(k+1)}(d_\mathsf{SE})$ be the 2-product metric $\delta_\mathbb{U}$ defined~\cite{deza2009encyclopedia} by
\begin{gather}
    \delta_\mathbb{U}(\vec{x},\vec{y})=\sqrt{\sum_{\ell=0}^k\|x_\ell-y_\ell\|_2^2},
\end{gather}
which then satisfies
\begin{align}
    \sum_{\ell=0}^kd_\mathsf{S}^\ell\|U_\ell-V_\ell\|_2&\leq\sum_{\ell=0}^kd_\mathsf{S}^\ell\sqrt{\sum_{\ell^\prime=0}^k\|U_{\ell^\prime}-V_{\ell^\prime}\|_2^2}=\left(\f{d_\mathsf{S}^{k+1}-1}{d_\mathsf{S}-1}\right)\delta_\mathbb{U}(\vec{U},\vec{V}),
\end{align}
and thus we conclude that
\begin{gather}
    \left|\zeta(\vec{U})-\zeta(\vec{V})\right|\leq\left(\f{d_\mathsf{S}^{k+1}-1}{d_\mathsf{S}-1}\right)\delta_\mathbb{U}(\vec{U},\vec{V}),
\end{gather}
so the (bound on) Lipschitz constants coincide with 
\begin{equation}
    \mscr{L}\leq\f{d_\mathsf{S}^{k+1}-1}{d_\mathsf{S}-1},  
    \label{eq: N1 Lipschitz}
\end{equation}
which essentially behaves as $\mc{O}\left(d_\mathsf{S}^k\right)$.

\subsubsection{The concentration function (constant interaction)}
We now make use of a result related to the Gromov-Bishop inequality (see e.g. Theorem 7 in Ref.~\cite{Philthy}) stating that if $\mathrm{Ric}(M)\geq\mathrm{Ric}(\Sigma^n(R))=\f{n-1}{R^2}$ for an $n$-dimensional manifold $M$, where $\Sigma^n(R)$ is the $n$-dimensional sphere of radius $R$ and $\mathrm{Ric}(X)$ is the infimum of diagonal elements of the Ricci curvature tensor on $X$, then the respective concentration functions satisfy $\alpha_M(x)\leq\alpha_{\Sigma^n(R)}(x)\leq\exp\left[-\f{x^2(n-1)}{2R^2}\right]$~\cite{Ledoux}.

For the constant interaction case, the corresponding manifold is the group manifold $\mscr{U}$ of $\mbb{U}(d)$ (where here $d=d_\mathsf{SE}$), which is diffeomorphic to $\mbb{SU}(d)\times\Sigma^1(1)$~\cite{Ledoux}, where $\mbb{SU}(d)$ denotes the special unitary group (i.e. with added $\det(U)=1$ for any element) and thus has $\mathrm{Ric}(\mscr{U})=d/2$ and $\dim(\mscr{U})=d^2$. Then it follows that $\mathrm{Ric}(\mscr{U})\geq\mathrm{Ric}(\Sigma^{d^2}(r))$ if $r^2\geq2(d^2-1)/d$ and hence, taking the minimal case,
\begin{gather}
    \alpha_{_\mathscr{U}}(\delta/\mscr{L})\leq\alpha_{\Sigma^{d^2}(r)}(\delta/\mscr{L})\leq\exp\left(-\f{\delta^2d_\mathsf{SE}}{4\mscr{L}^2}\right).
\end{gather}

\subsubsection{The concentration function (random interaction)} For the random interaction case the corresponding manifold is the group manifold $\mathbf{U}$ of the $k+1$ Cartesian product space $\mbb{U}(d)\times\cdots\times\mbb{U}(d)$. In this case $\mathrm{Ric}(\mathbf{U})=(k+1)d/2$ and $\dim(\mathbf{U})=(k+1)d^2$. Then it follows that $\mathrm{Ric}(\mathbf{U})\geq\mathrm{Ric}(\Sigma^{(k+1)d^2}(\mc{R}))$ if $\mc{R}^2\geq2[(k+1)d^2-1]/[d(k+1)]$ and hence, taking the minimal case,
\begin{gather}
    \alpha_{_\mathbf{U}}(\delta/\mscr{L})\leq\alpha_{\Sigma^{(k+1)d_\mathsf{SE}^2}(\mc{R})}(\delta/\mscr{L})\leq\exp\left(-\f{\delta^2(k+1)d_\mathsf{SE}}{4\mscr{L}^2}\right).
\end{gather}
The factor $d_\mathsf{SE}\mscr{L}^{-2}$ behaves as $\mc{O}\left(d_\mathsf{E}d_\mathsf{S}^{-2k+1}\right)$, so both concentration functions will be small whenever $d_\mathsf{E}\gg{d}_\mathsf{S}^{2k}$.
\end{proof}

Theorem~\ref{Thm: typicality processes} assures that the probability for the non-Markovianity $\mc{N}_1$ to exceed a function of $k$, $d_\mathsf{S}$ and $d_\mathsf{E}$, that becomes very small in the large $d_\mathsf{E}$ limit, itself becomes small in that limit. Our result is meaningful when both $\mc{B}_k+\delta$ and $\exp\left(-\mscr{C}\delta^2\right)$ are small; the latter is fulfilled in the \emph{small subsystem} or \emph{large environment} limit, which in our setting means $d_\mathsf{E}\gg{d}_\mathsf{S}^{2k+1}$. We may also state the minimal value for $\delta$ that, assuming the large environment limit, renders both sides small,\footnote{ Similar to Ref.~\cite{Popescu2006} detailed in Section~\ref{sec: state typicality}, here we look for an $x>0$ such that $\delta=d_\mathsf{E}^{-x}$ and $\delta^2d_\mathsf{E}=d_\mathsf{E}^x$.} i.e., such that $\delta^2d_\mathsf{E}\gg1\gg\delta$; this is fulfilled for $\delta=d_\mathsf{E}^{-1/3}$. A geometrical cartoon to illustrate the result is presented in Fig.~\ref{main cartoon}.

Theorem~\ref{Thm: typicality processes} can be said to state that \emph{almost all quantum processes are almost Markovian}.

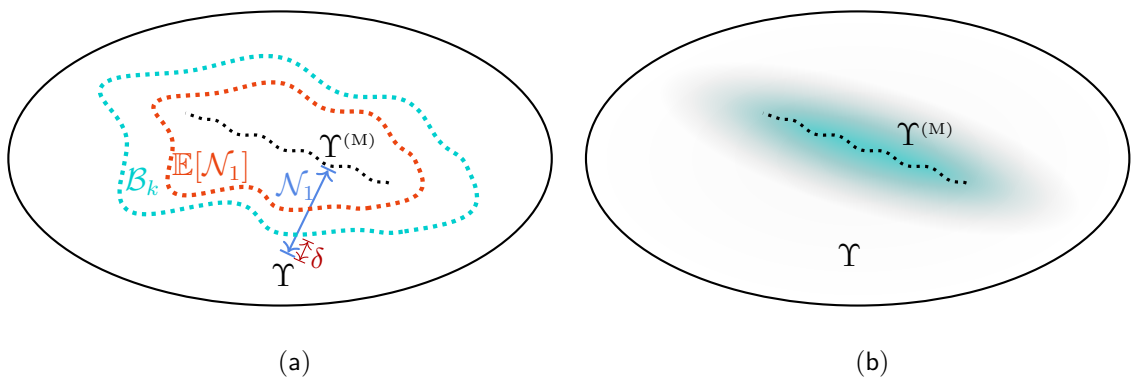
\begin{figure}[t]
\centering
\begin{tikzpicture}[xscale=2, yscale=1.2, every node/.style={scale=1}]
\begin{scope}
    \shade[thick, scale=0.65, inner color=white, outer color=white, draw=black]  (-2.225,0.4) ellipse (2.75cm and 2.5cm);
    \draw[-, black, very thick, dotted, decorate,decoration={snake, amplitude=.4mm, segment length=0.25in}, rotate around={150:(-0.75,0)}] (-0.7725,0) -- (0.7725,0);
    \node[right, above] at (-1,0.175) {\Large$\Upsilon^\markov$};
    \node at (-1.425,-1.025) {\Large$\Upsilon$};
    \draw[C1,|<->|, thick] (-1.4,-0.8) -- (-1.125,0.15);
    \node[C1, right] at (-1.53, -0.05) {\Large$\mc{N}_1$};
    \draw[red!70!black,|<->|] (-1.32,-0.875) -- (-1.26,-0.65);
    \node[right, below, red!70!black] at (-1.2,-0.625) {\large$\delta$};
    \node[C3] at (-2.35,0) {\Large$\mc{B}_k$};
    \node[C2] at (-1.9,0.16) {\Large$\mathbb{E}[\mc{N}_1]$};
    \draw[-, C2, ultra thick, dotted, rotate around={150:(-0.75,0)}] plot[smooth cycle, tension=.7] coordinates {(-1,0) (-0.8,0.3) (-0.6,0.4) (-0.3,0.6) (0,0.5) (0.4,0.8) (0.7,0.3) (1,0) (0.7,-0.3) (0.4,-0.6) (0,-0.5) (-0.3,-0.6) (-0.6,-0.4) (-0.8,-0.3)};
    \draw[-, C3, ultra thick, dotted, rotate around={150:(-0.75,0)},scale=1.4] plot[smooth cycle, tension=.7] coordinates {(-1,0) (-0.8,0.3) (-0.6,0.4) (-0.3,0.6) (0,0.5) (0.4,0.8) (0.7,0.3) (1,0) (0.7,-0.3) (0.4,-0.6) (0,-0.5) (-0.3,-0.6) (-0.6,-0.4) (-0.8,-0.3)};
    \node[scale=1] at (-1.35,-2) {$\mathsf{(a)}$};
\end{scope}
\begin{scope}[shift={(3.8,0)}]
    \shade[thick, scale=0.65, inner color=gray!5!white, outer color=white, draw=black]  (-2.225,0.4) ellipse (2.75cm and 2.5cm);
    \foreach\i in {0,0.01,...,0.8} {
        \fill[opacity=\i*0.03, C3, rotate around={60:(-1,-1)}] (0,0) ellipse ({0.8-\i} and {1.8-1.8*\i});         
      }
    \draw[-, black, very thick, dotted, decorate,decoration={snake, amplitude=.4mm, segment length=0.25in}, rotate around={150:(-0.75,0)}] (-0.7725,0) -- (0.7725,0);
    \node[right, above] at (-1,0.3) {\Large$\Upsilon^\markov$};
    \node at (-1.5,-0.8) {\Large$\Upsilon$};
    \node at (-1.35,-2) {$\mathsf{(b)}$};
\end{scope}
    \end{tikzpicture}
    \caption[Concentration around Markovian processes in large dimensional environments]{\textbf{Concentration around Markovian processes in large dimensional environments:} $\mathsf{(a)}$ A geometric cartoon of our main result in a space of process tensors: the probability of the non-Markovianity $\mc{N}_1$ of deviating from $\mc{B}_k$ by some $\delta>0$ decreases exponentially in $\delta^2$. In $\mathsf{(b)}$, quantum processes $\Upsilon$ on large dimensional environments (such that $d_E\gg{d}_S^{2k+1}$) concentrate around the Markovian ones $\Upsilon^\markov$.}\label{main cartoon}
\end{figure}

Our results imply that it is fundamentally hard to observe non-Markovianity in a typical process and thus go some way to explaining the overwhelming success of Markovian theories.

Specifically, Theorem~\ref{Thm: typicality processes} shows that even when interacting strongly with the wider composite system, a subsystem will typically undergo highly Markovian dynamics when the rest of the system has a sufficiently large dimension, and that the probability to be significantly non-Markovian vanishes with the latter. Our main result formalizes the notion that in the large environment limit a quantum process, taken uniformly at random, will be almost Markovian with very high probability. This corroborates the common understanding of the Born-Markov approximation, discusses in Section~\ref{sec: state typicality}, but, crucially, we make no assumptions about weak coupling between \gls{env} and \gls{syst}. Instead, in the Haar random interactions we consider, every part of the system typically interacts significantly with every part of \gls{env}. This is in contrast to many open systems models, even those with superficially infinite dimensional baths, where the effective dimension of the environment is relatively small~\cite{tamascelli2018}; it can always be bounded by a function of time scales in the system-environment Hamiltonian~\cite{Luchnikov2018}, which could be encoded in a bath spectral density. Our result is also more general than the scenario usually considered, since it accounts for interventions and thus the flow of information between \gls{syst} and \gls{env} across multiple times.

\subsection{Numerical sampling}
To support our results, we sampled process tensors $\Upsilon$ numerically in the random interaction case and computed their corresponding average non-Markovianity $\mbb{E}[\mc{N}_1]$ as a function of environment dimension $d_\mathsf{E}$ for a fixed system dimension $d_\mathsf{S}=2$, obtaining the behaviour shown in Fig.~\ref{numPlot}. The details on how this was done can be seen in Appendix~\ref{appendix - numerics Almost}. 

For constant interaction the numerical results are practically indistinguishable from those in the random case, but as mentioned, the analytical bound $\mc{B}_k$ is much harder to compute exactly. This suggests that either a simpler bound exists or that it might be possible to simplify the one we have obtained.\footnote{ Notice that despite expressions in Eq.~\eqref{average state Ui=Uj} and Eq.~\eqref{average time independent purity} being seemingly complicated, this complexity arises only because of the permutations and the Weingarten function.} As expected, our numerical results fall within the bound $\mc{B}_k(d_\mathsf{E},d_\mathsf{S}=2)$ and they behave similarly; we notice that the bound in general seems to be somewhat loose, and become loosest when $d_\mathsf{E}\simeq d_\mathsf{S}^{2k+1}$, implying that non-Markovianity might be hard to detect even when not strictly in the large environment limit. However, it does saturate rapidly as $d_\mathsf{E}$ increases.

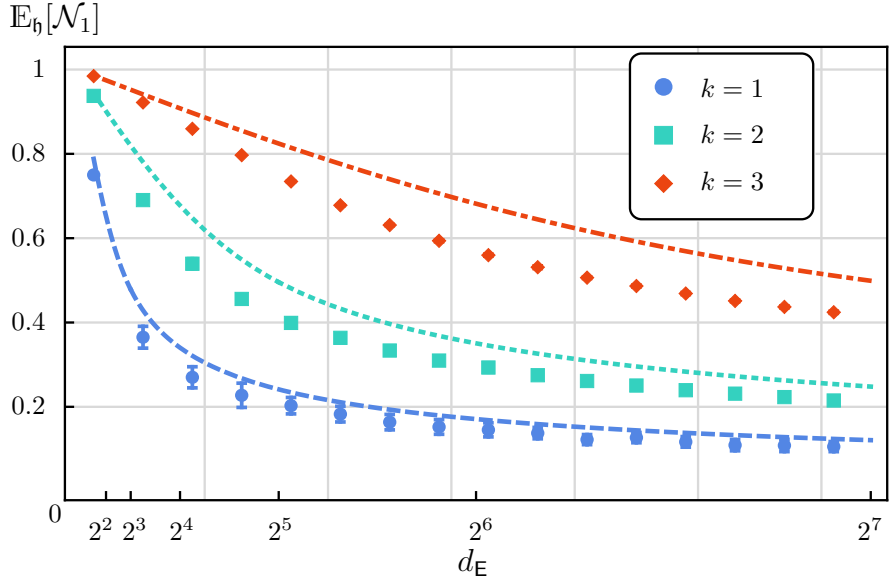
\begin{figure}[t]
\centering
\begin{tikzpicture}
\node[anchor=south west, inner sep=0] (image) at (0,0) {\includegraphics[width=0.75\textwidth]{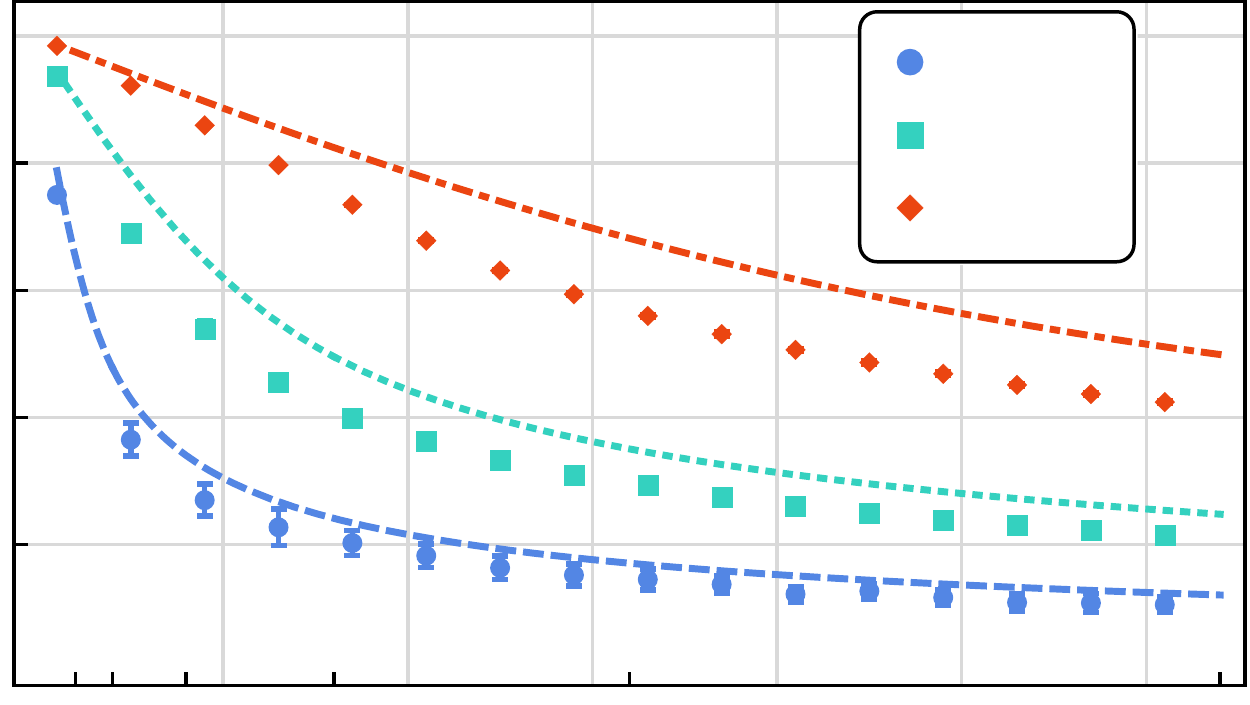}};
\begin{scope}[x={(image.south east)},y={(image.north west)}]
    \node at (0,0) {$0$};
    \node[left] at (0,0.23) {$0.2$};
    \node[left] at (0,0.41) {$0.4$};
    \node[left] at (0,0.59) {$0.6$};
    \node[left] at (0,0.77) {$0.8$};
    \node[left] at (0,0.95) {$1$};
    \node[above] at (0,1) {\Large$\mbb{E}_\haarrand[\mc{N}_1]$};
    \node at (0.81,0.91) {$k=1$};
    \node at (0.81,0.81) {$k=2$};
    \node at (0.81,0.71) {$k=3$};
    \node[below] at (0.5,-0.055) {\Large$d_\mathsf{E}$};
    \node[below] at (0.055,0.01) {$2^2$};
    \node[below] at (0.0955,0.01) {$2^3$};
    \node[below] at (0.152,0.01) {$2^4$};
    \node[below] at (0.27,0.01) {$2^5$};
    \node[below] at (0.51,0.01) {$2^6$};
    \node[below] at (0.975,0.01) {$2^7$};
\end{scope}
\end{tikzpicture}
     \caption[Average non-Markovianity on numerical sampling]{\textbf{Average non-Markovianity $\mbb{E}_\haarrand[\mc{N}_1]$ of a random interaction process for a qubit in the environment dimension $d_E$ at fixed time steps $k$:} Discrete values are shown for numerical averages over $\lfloor40/k\rfloor$ randomly generated process tensors $\Upsilon$ at time steps $k=1,2,3$ and with fixed $d_\mathsf{S}=2$; error bars denote the standard deviation due to sampling error. The lines above each set of points denote the upper bound $\mc{B}_k(d_E,2)$. Process tensors were generated by sampling Haar random unitaries according to Ref.~\cite{Mezzadri} and described in Appendix~\ref{appendix - numerics Almost}}\label{numPlot}
\end{figure}

So far, our results are valid for process tensors constructed with Haar random unitaries at $k$ evenly spaced steps; we are effectively considering a strong interaction between system and environment which rapidly scrambles quantum information in both~\cite{Roberts2017}. As we saw in Section~\ref{sec: equilibration on average}, the mechanism for equilibration is precisely that of dephasing, or effectively, the scrambling of information on the initial state of the system. This suggests that, even when timescales will differ with the type of evolution considered, most physical evolutions fall within our result, with e.g. a weaker behavior in $k$.

This is perhaps the most contentious part of our results from Ref.~\cite{FigueroaRomero2019almostmarkovian} presented in this chapter: as opposed to the typicality of quantum states where the Haar measure only plays a role in the sampling, in the case of quantum processes it has an implication on the class of dynamics that are being considered, and as we are well aware nature is far from random. We will discuss this point in more precise terms in the next chapter, where we directly approached this issue. Despite these features, we will now show that our results still hold at a coarse-grained level, where the intermediate dynamics corresponds to products of Haar random unitaries, which are not themselves Haar random.

\section{Observing non-Markovianity}
The choice of unitaries going into the process tensor in the previous sections (in our case drawn from the Haar measure), dictates a time scale for the system, up to a freely chosen energy scale. However, as we saw in Section~\ref{sec: the process tensor}, illustrated in Fig.~\ref{fig: Process tensor containment}, the process tensor satisfies a containment property, so that we can straightforwardly construct process tensors on a longer, coarse-grained time scale by simply allowing the system to evolve, or equivalently be acted with an identity operation, between some subset of time steps. And in fact, as we also mentioned before, process tensors at all time scales should be related in this way to an underlying process tensor with an infinite number of steps~\cite{Milz_2020}. We thus now refer to this construction through the containment property as \emph{coarse graining}.

To see that our main result directly applies to any coarse-grained process tensor, we again consider the definition of our non-Markovianity measure in Eq.~\eqref{def nonMarkov trDistance}. Consider the coarser grained process tensor
\begin{equation}
    \Upsilon_\textsf{coarse}=\Upsilon_{k:0\setminus\{i\in[0,k-1]\}},    
\end{equation}
where a subset of operations $\{\mc{A}_i\}_{i\in[0,k-1]}$ are replaced by identity operations, i.e. the system is simply left to evolve. Letting $\mc{N}_\textsf{coarse} \equiv \min_{\Upsilon_\textsf{coarse}^\markov} \mc{D}(\Upsilon_\textsf{coarse}, \Upsilon_\textsf{coarse}^\markov)$, we have
\begin{gather}
    \mc{N}_\textsf{coarse}\leq\mc{N}_1,
    \label{eq:coarsegraining}
\end{gather}
since the set of allowed $\Upsilon_\textsf{coarse}^\markov$ strictly contains the allowed $\Upsilon^\markov$ at the finer-grained level. This renders the new process less distinguishable from a Markovian one, i.e., coarse-graining can only make processes more Markovian.

The physical intuition behind this result is that the amount of information which can be encoded in a coarse grained process tensor is strictly less than that in its parent process tensor. In fact, this is a key feature of non-Markovian memory: the memory should decrease under coarse graining. On the other hand, due to the same reasoning, we cannot say anything about finer-grained dynamics once given a process up to a finite number of time-steps. One approach to tackling this issue would be to choose a different sampling procedure which explicitly takes scales into account and can deal with this problem directly incorporating this notion of graining.

Moreover, there is another important limitation for observing non-Markovianity. The operational interpretation of the trace distance, discussed in the previous section, implies that observing non-Markovianity requires applying a measurement that is an eigenprojector operator of $\Upsilon-\Upsilon^\markov$. The optimal measurement will, in general, be entangled across all time steps. In practice, this is hard to achieve and typically one considers a sequence of local measurements $m\in\mathbb{M}$. In general, for an any set of measurements $\mathbb{M}$ we can define a restricted measure of non-Markovianity detectable with that set: $\mc{D}_{\mathbb{M}} (\Upsilon,\Upsilon^\markov)\equiv\max_{m\in\mathbb{M}} \f{1}{2} |\tr[m(\Upsilon-\Upsilon^\markov)]|  \leq \mc{D} (\Upsilon,\Upsilon^\markov)$, which means that the detectable non-Markovianity will be smaller. This is akin to the \emph{eigenstate thermalization hypothesis}~\cite{PhysRevLett.120.150603}, where all eigenstates of a physical Hamiltonian look uniformly distributed with respect to most \emph{physically reasonable} observables. In our setting, this means that looking for non-Markovianity with observables that are local in time --i.e., physically reasonable-- we  find almost no temporal correlations.

The locality constraint, along with monotonicity of non-Markovianity under coarse graining, have further important consequences for a broad class of open systems studies where master equations are employed~\cite{Pollock2018tomographically}. Since master equations usually only account for two-point correlations with local measurements, they will be insensitive to most of the temporal correlations being accounted for by our measure, leading to an even greater likelihood for their descriptions to be Markovian. We will also discuss this point in the next chapter when going beyond the Haar measure.

\section{Conclusions}
The generic form of open quantum dynamics is non-Markovian, but, despite this, it is often very well approximated by simpler Markovian dynamics. How this \emph{memorylessness} emerges is not dissimilar to questions, regarding the emergence of thermodynamic behaviour, which have pervaded quantum mechanics since its conception. We now know that the fundamental postulate of equal a-priori probabilities of statistical mechanics can be traced back to the entanglement between subsystems and their environment~\cite{Popescu2006}. With Theorem~\ref{Thm: typicality processes}, obtained originally in Ref.~\cite{FigueroaRomero2019almostmarkovian}, we have shown that very similarly, if we sample a generic quantum process occurring in a large finite environment at random, it will be almost Markovian with very high probability.

Specifically, we have showed that, even when interacting strongly with the wider composite system, a subsystem will typically undergo highly Markovian dynamics when the rest of the system has a sufficiently large dimension, and that the probability to be significantly non-Markovian vanishes with the latter. Theorem~\ref{Thm: typicality processes} formalizes the notion that in the large environment limit a quantum process, taken uniformly at random, will be almost Markovian with very high probability, in turn also corroborating the common understanding of the Born-Markov approximation described in Section~\ref{sec: open quantum dynamics}, where crucially, however, we make no assumptions about weak coupling between \gls{env} and \gls{syst}. It is important that this is different to the case of many open systems models which consider superficially infinite dimensional baths, i.e. where the effective dimension of the environment is relatively small~\cite{tamascelli2018}, as these can always be bounded by a function of time scales in the system-environment Hamiltonian~\cite{Luchnikov2018} and which could be encoded in a bath spectral density. Most importantly, as have been stressed throughout Chapter~\ref{sec:processes}, our result naturally contains the scenario usually considered, since it accounts for multiple interventions and thus the flow of information between \gls{syst} and \gls{env} across multiple times.

It is also important to highlight that while it may still be possible to observe non-Markovian behaviour at a time scale that is smaller than the fundamental time scale set by the chosen unitaries, Eq.~\eqref{eq:coarsegraining} tells us that any coarse grained process will remain concentrated around the Markovian ones in the large environment limit. Otherwise, for larger and larger systems, one needs an ever increasing number of time steps, corresponding to higher order correlations, in order to increase the probability of witnessing non-Markovianity. However, even in this case, from the discussion in the previous section, we know that the measurement on this large number of times steps will be temporally entangled, which may also be difficult to achieve.

Finally, we highlight the interplay of the typicality of Markovian processes and dynamical equilibration on multiple time-steps. We have drawn an analogy before with the relation between equilibration, thermalization and the typicality of canonical states, and it now seems highly plausible that an analogous relationship holds between multitime equilibration, \emph{Markovianization} and the typicality of Markovian processes. While a fully dynamical characterization of a notion of Markovianization is still to be achieved, in the following chapter we will describe a step towards this notion as an emergence of Markovian processes by addressing the bothersome aspect in Theorem~\ref{Thm: typicality processes} of being somewhat far from physical, that is, we will show that Markovian processes still satisfy a large deviation bound even if we step away from the Haar measure and consider more physically motivated evolutions.
    \chapter{Markovianization by Design}
\label{sec: Markovianization by design}
\setlength{\epigraphwidth}{0.65\textwidth}
\epigraph{\emph{Matter is matter, neither noble nor vile, infinitely transformable, and its proximate origin is of no importance whatsoever.}}{-- Primo Levi (\emph{The Periodic Table})}

An important limitation of the typicality of Markovian processes in large environments by means of Theorem~\ref{Thm: typicality processes}, is that it encompasses too wide a class of \gls{syst-env} interactions, many of which can be deemed highly unphysical. Implementing a Haar random unitary requires an exponential number two-body interactions and random bits~\cite{Knill_1995}, so Haar random dynamics cannot be obtained efficiently in a physical setting. This seems to be at odds with the applicability of the Born-Markov assumption on a wide variety of physical models~\cite{carmichael1993open, blanchard2000decoherence, schlosshauer2007decoherence, alicki2007quantum}. Forgetfulness is indeed a common feature of the world around us, and one that is crucial for doing science: without forgetfulness, repeatability would be impossible. From a somewhat philosophical standpoint, consider that if any given atom remembered its own past, then it would be unique and there would be no sense in classifying atoms and molecules.

So can we say physically relevant models satisfy a concentration of measure with respect to Markovianity?

In addressing this problem, in Ref.~\cite{FigueroaRomero2020makovianization} we identified a class of isolated physical processes which approximately Markovianize. We show that, similar to the way in which quantum states thermalize, quantum processes can \emph{Markovianize} in the sense that they can converge to a class of typical processes, satisfying a meaningful large deviations principle whenever they are undergone within a large environment and under complex enough---but not necessarily fully random---dynamics. To accomplish this, we employ large deviation bounds\footnote{ I.e. bounds on probabilities for rare events.} for so-called \emph{approximate unitary designs} derived in Ref.~\cite{Low_2009}, and apply them to the process tensor formalism.

The concept of a unitary design provides a finite approximation to uniform Haar randomness and it refers to a set of unitaries that reproduce a finite number of moments of the unitary group with the Haar measure. In a physical scenario, unitary designs reproducing an ever-increasing number of moments of the Haar measure can be seen to arise naturally from seemingly simple situations as the evolution time increases~\cite{Roberts2017, Winter_HamDesign, 2designsXZ}. One can further allow some small error in such reproduction of the Haar moments, thus rendering the respective unitary designs approximate. To further establish these ideas together with the main claim in our result, we make use of an efficient construction of an approximate unitary design with an $n$-qubit quantum circuit using two-qubit interactions only, showing how seemingly simple, physically motivated systems, can speedily become forgetful.

Given the ever-increasing interest and relevance in determining the breakdown of the Markovian approximation in modern experiments~\cite{Gessner2014, Ringbauer_2015, Morris2019, Winick2019}, we discuss potential applications and extensions to our results, as well as their limitations and possible ways to overcome them.

\section{Approximate unitary designs}
The core of the issue with our previous approach for the typicality of Markovian processes lies in the unitaries entering a process being uniformly distributed. Physical unitary evolution has a very specific structure determined by a Hamiltonian, and in this sense it can be said to be far from random. In some circumstances, however, there are physical processes which can approximate some of the statistical features of the Haar measure~\cite{Guhr1998, DAlessio_2016, Gogolin_2016, mehta2004random}. For example, consider the toy model depicted in Fig.~\ref{Fig: Design box and process}$\mathsf{(a)}$, comprising a dilute gas of $n$ particles evolving autonomously in a closed box. The gas particles interact with each other in one of two ways as they randomly move inside the box. Following and intervening on an impurity particle, taken to be the system, this model can be well described by a circuit such as the one in Fig.~\ref{Fig: Design box and process}$\mathsf{(b)}$. The simplicity of this system suggests that it can only \emph{uniformly randomize} after a large number of random two-qubit interactions, progressively resembling genuine Haar random dynamics. One possible way to quantify this progressive resemblance of the Haar measure is given by the concept of unitary designs.

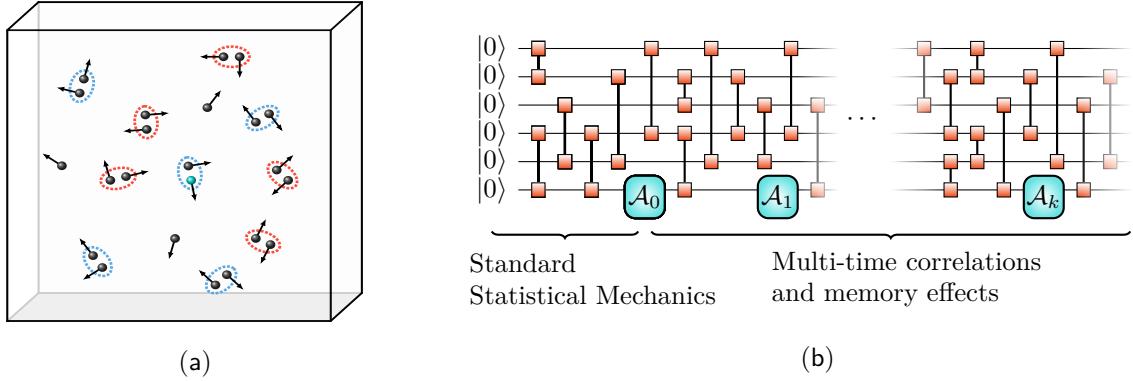
\begin{figure}[t]
\centering
\begin{minipage}{0.4\textwidth}
\begin{tikzpicture}
\node[anchor=south west, inner sep=0] (image) at (0,0) {\includegraphics[width=0.8\textwidth, height=0.725\textwidth]{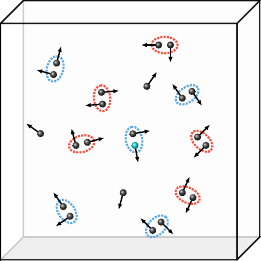}};
\node[below] at (0.425\textwidth,-0.25) {$\mathsf{(a)}$};
\end{tikzpicture}
  \end{minipage}
  \begin{minipage}{0.59\textwidth}
  \begin{tikzpicture}
  \begin{scope}[xscale=0.35, yscale=0.375,shift={(-0.2,0)}]
 \foreach \i in {1,...,6}{
 \node at (0,-\i) {$|0\rangle$};
 \draw[] (1,-\i) -- (24,-\i);
 
 \shade[bottom color=C2, top color=white, draw=black] (1.5,-0.75) rectangle (2,-1.25);
 \draw[thick] (1.75,-1.25) -- (1.75,-1.75);
 \shade[bottom color=C2, top color=white, draw=black] (1.5,-1.75) rectangle (2,-2.25);
 
 \shade[bottom color=C2, top color=white, draw=black] (1.5,-3.75) rectangle (2,-4.25);
 \draw[thick] (1.75,-4.25) -- (1.75,-5.75);
 \shade[bottom color=C2, top color=white, draw=black] (1.5,-5.75) rectangle (2,-6.25);
 
 \shade[bottom color=C2, top color=white, draw=black] (2.5,-2.75) rectangle (3,-3.25);
 \draw[thick] (2.75,-3.25) -- (2.75,-4.75);
 \shade[bottom color=C2, top color=white, draw=black] (2.5,-4.75) rectangle (3,-5.25);
 
 \shade[bottom color=C2, top color=white, draw=black] (3.5,-3.75) rectangle (4,-4.25);
 \draw[thick] (3.75,-4.25) -- (3.75,-5.75);
 \shade[bottom color=C2, top color=white, draw=black] (3.5,-5.75) rectangle (4,-6.25);
 
 \shade[bottom color=C2, top color=white, draw=black] (4.5,-1.75) rectangle (5,-2.25);
 \draw[thick] (4.75,-2.25) -- (4.75,-4.75);
 \shade[bottom color=C2, top color=white, draw=black] (4.5,-4.75) rectangle (5,-5.25);
 
  \shade[outer color=C3!60!white, inner color=white, draw=black, thick, rounded corners] (5,-5.5) rectangle (6.5,-7);
  \node at (5.75,-6.25) {$\mc{A}_0$};
 }
 
 \shade[bottom color=C2, top color=white, draw=black] (5.75,-0.75) rectangle (6.25,-1.25);
 \draw[thick] (6,-1.25) -- (6,-3.75);
 \shade[bottom color=C2, top color=white, draw=black] (5.75,-3.75) rectangle (6.25,-4.25);
 
 \shade[bottom color=C2, top color=white, draw=black] (7,-3.75) rectangle (7.5,-4.25);
 \draw[thick] (7.25,-4.25) -- (7.25,-5.75);
 \shade[bottom color=C2, top color=white, draw=black] (7,-5.75) rectangle (7.5,-6.25);
 
 \shade[bottom color=C2, top color=white, draw=black] (7,-1.75) rectangle (7.5,-2.25);
 \draw[thick] (7.25,-2.25) -- (7.25,-2.75);
 \shade[bottom color=C2, top color=white, draw=black] (7,-2.75) rectangle (7.5,-3.25);
 
  \shade[bottom color=C2, top color=white, draw=black] (8,-0.75) rectangle (8.5,-1.25);
 \draw[thick] (8.25,-1.25) -- (8.25,-4.75);
 \shade[bottom color=C2, top color=white, draw=black] (8,-4.75) rectangle (8.5,-5.25);
 
 \shade[bottom color=C2, top color=white, draw=black] (9,-1.75) rectangle (9.5,-2.25);
 \draw[thick] (9.25,-2.25) -- (9.25,-3.75);
 \shade[bottom color=C2, top color=white, draw=black] (9,-3.75) rectangle (9.5,-4.25);
 
 \shade[outer color=C3!60!white, inner color=white, draw=black, thick, draw=black, rounded corners] (10,-5.5) rectangle (11.5,-7);
  \node at (10.75,-6.25) {$\mc{A}_1$};
  
  \shade[bottom color=C2, top color=white, draw=black] (10,-2.75) rectangle (10.5,-3.25);
 \draw[thick] (10.25,-3.25) -- (10.25,-4.75);
 \shade[bottom color=C2, top color=white, draw=black] (10,-4.75) rectangle (10.5,-5.25);
 
 \shade[bottom color=C2, top color=white, draw=black] (11,-0.75) rectangle (11.5,-1.25);
 \draw[thick] (11.25,-1.25) -- (11.25,-3.75);
 \shade[bottom color=C2, top color=white, draw=black] (11,-3.75) rectangle (11.5,-4.25);
 
 \shade[bottom color=C2, top color=white, draw=black] (12,-2.75) rectangle (12.5,-3.25);
 \draw[thick] (12.25,-3.25) -- (12.25,-5.75);
 \shade[bottom color=C2, top color=white, draw=black] (12,-5.75) rectangle (12.5,-6.25);
 
 \shade[bottom color=C2, top color=white, draw=black] (16,-0.75) rectangle (16.5,-1.25);
 \draw[thick] (16.25,-1.25) -- (16.25,-2.75);
 \shade[bottom color=C2, top color=white, draw=black] (16,-2.75) rectangle (16.5,-3.25);
 
 \draw[white, fill=white, path fading= east] (15,0.5) -- (17,0.5) -- (17,-6.5) -- (15,-6.5);
 \draw[white, fill=white] (13,0.5) -- (15,0.5) -- (15,-6.5) -- (13,-6.5);
 \draw[white, fill=white, path fading= west] (11.5,-6.5) -- (11.5,0.5) -- (13,0.5) -- (13,-6.5) ;
 
 \node at (14,-3.5) {$\cdots$};
 
 \shade[bottom color=C2, top color=white, draw=black] (17,-1.75) rectangle (17.5,-2.25);
 \draw[thick] (17.25,-2.25) -- (17.25,-3.75);
 \shade[bottom color=C2, top color=white, draw=black] (17,-3.75) rectangle (17.5,-4.25);
 
 \shade[bottom color=C2, top color=white, draw=black] (17,-4.75) rectangle (17.5,-5.25);
 \draw[thick] (17.25,-5.25) -- (17.25,-5.75);
 \shade[bottom color=C2, top color=white, draw=black] (17,-5.75) rectangle (17.5,-6.25);
 
 \shade[bottom color=C2, top color=white, draw=black] (18,-0.75) rectangle (18.5,-1.25);
 \draw[thick] (18.25,-1.25) -- (18.25,-1.75);
 \shade[bottom color=C2, top color=white, draw=black] (18,-1.75) rectangle (18.5,-2.25);
 
 \shade[bottom color=C2, top color=white, draw=black] (18,-3.75) rectangle (18.5,-4.25);
 \draw[thick] (18.25,-4.25) -- (18.25,-4.75);
 \shade[bottom color=C2, top color=white, draw=black] (18,-4.75) rectangle (18.5,-5.25);
 
 \shade[bottom color=C2, top color=white, draw=black] (19,-2.75) rectangle (19.5,-3.25);
 \draw[thick] (19.25,-3.25) -- (19.25,-5.75);
 \shade[bottom color=C2, top color=white, draw=black] (19,-5.75) rectangle (19.5,-6.25);
 
 \shade[outer color=C3!60!white, inner color=white, draw=black,rounded corners, thick] (20,-5.5) rectangle (21.5,-7);
  \node at (20.75,-6.25) {$\mc{A}_k$};
  
  \shade[bottom color=C2, top color=white, draw=black] (20,-1.75) rectangle (20.5,-2.25);
 \draw[thick] (20.25,-2.25) -- (20.25,-3.75);
 \shade[bottom color=C2, top color=white, draw=black] (20,-3.75) rectangle (20.5,-4.25);
 
 \shade[bottom color=C2, top color=white, draw=black] (21,-0.75) rectangle (21.5,-1.25);
 \draw[thick] (21.25,-1.25) -- (21.25,-4.75);
 \shade[bottom color=C2, top color=white, draw=black] (21,-4.75) rectangle (21.5,-5.25);
  
  \shade[bottom color=C2, top color=white, draw=black] (22,-2.75) rectangle (22.5,-3.25);
 \draw[thick] (22.25,-3.25) -- (22.25,-5.75);
 \shade[bottom color=C2, top color=white, draw=black] (22,-5.75) rectangle (22.5,-6.25);
 
 \shade[bottom color=C2, top color=white, draw=black] (23,-1.75) rectangle (23.5,-2.25);
 \draw[thick] (23.25,-2.25) -- (23.25,-4.75);
 \shade[bottom color=C2, top color=white, draw=black] (23,-4.75) rectangle (23.5,-5.25);
  
  \draw[white, fill=white, path fading= west] (22,0.5) -- (24,0.5) -- (24,-6.5) -- (22,-6.5);
  
  \draw [thick,decoration={brace,mirror,raise=0.55cm},decorate] (6,-6) -- (24,-6); 
  
  \node[align=left] at (15.6,-9.15) {Multi-time correlations\\and memory effects};
  
  \draw [thick,decoration={brace,mirror,raise=0.55cm},decorate] (0,-6) -- (5.5,-6); 
    
  \node[align=left] at (3.8,-9.15) {Standard\\Statistical Mechanics};
  \node at (12.25,-12) {$\mathsf{(b)}$};
  \end{scope}
  \end{tikzpicture}
  \end{minipage}
\caption[Unitary designs, from simplicity to randomness]{\textbf{Unitary designs, from simplicity to randomness:} $\mathsf{(a)}$ A toy model with dynamics given by a unitary design is that of an impurity particle (teal) immersed in a gas of $n_\mathsf{E}$ particles within a closed box, where all particles interact in pairs in one of two ways at random. $\mathsf{(b)}$ An \gls{syst-env}-system with random two-qubit gate interactions only, and multiple interventions $\{\mc{A}_i\}$. While the standard approach towards typicality or equilibrium properties concerns the whole \gls{syst-env} dynamics and/or a single measurement on system \gls{syst}, we show that complex dynamics within large environments will be highly Markovian with high probability. On the other hand, if probed enough times, information about past correlations will eventually become non-negligible.}
\label{Fig: Design box and process}
\end{figure}

\begin{definition}[Exact unitary $\design$-design~\cite{Low_2009}]
An exact unitary $\design$-design is defined as a probability measure $\mu_{\,\design}$ on $\mbb{U}(d)$ such that for all positive $s\leq\design$, and all $d^s\times{d}^s$ complex matrices $X$,
\begin{gather}
    \mbb{E}_{\,\design}\left[\mc{V}^{\otimes{s}}(X)\right]=\mbb{E}_\mathsf{h}\left[\mc{U}^{\otimes{s}}(X)\right],\quad\forall{s\leq\design},
    \label{eq: def exact t design}
\end{gather}
where $\mc{U}(\cdot):=U(\cdot)U^\dg$ and $\mc{V}(\cdot):=V(\cdot)V^\dg$ are unitary maps with $U,V\in\mbb{U}(d)$.
\end{definition}

Here, as above, the notation $\mathbb{E}_\Omega$ indicates the expectation value with respect to a given probability measure $\mu_\Omega$, i.e. we have $V\sim\mu_{\,\design}$ and $U\sim\mu_\mathsf{h}$ in Eq.~\eqref{eq: def exact t design}. In the case we will be interested in, the unitary maps will correspond to \gls{syst-env} unitaries distributed according to the either the Haar measure or a unitary design. As per the definition in Eq.~\eqref{eq: def exact t design}, a unitary $\design$-design reproduces up to the $\design$-th moment over the uniform distribution given by the Haar measure. In particular, $\mu_{\,\design}$ can consist of a finite ensemble $\{V_i,p_i\}_{i=1}^N$ of unitaries $V_i$ and probabilities $p_i$, as is now common in applications such as so-called randomized benchmarking of error rates in quantum gates~\cite{Dankert_2009, Wallman_2014}.

Moreover, this definition of a unitary design can be relaxed by allowing a small error $\varepsilon$. In general an $\varepsilon$-approximate $\design$-design, which we denote $\mu_{\,\appdesign}$, can be defined through
\begin{gather}
    \left|\mbb{E}_{\,\appdesign}\left[\mc{V}^{\otimes{s}}(X)\right]-\mbb{E}_\mathsf{h}\left[\mc{U}^{\otimes{s}}(X)\right]\right|\leq\varepsilon,\quad\forall\,{s}\leq\design
    \label{eq: def approx t design generic}
\end{gather}
for a suitable metric $|\cdot|$, with all quantities as in Eq.~\eqref{eq: def exact t design}, now with $V \sim \mu_{\,\appdesign}$.

We specifically employed the definition of Ref.~\cite{Low_2009} for unitary designs, which uses the fact that the definition of an exact $\design$-design can be written in terms of a balanced monomial $\Theta$ of degree less or equal to $\design$ in the components of the unitaries $U$.

\begin{definition}[Balanced monomial~\cite{Low_2009}]
A (unitary) balanced monomial $\Theta_s:\mbb{U}(d)\to\mbb{R}$ of degree $s$ is a monomial in unitary elements with $s$ conjugated and $s$ unconjugated elements.
\end{definition}

This means, for example, that $\Theta_2(U)=U_{ab}U_{cd}U_{ef}^*U_{hg}^*$ is a balanced monomial of degree 2 for given components $a,b,\ldots,g$. In this language, we have seen as well that the $n$-moments of the unitary group are non-vanishing only for balanced monomials. Thus, writing Eq.~\eqref{eq: def exact t design} in terms of matrix elements, this can be seen to be equivalent to requiring $\mbb{E}_{\,\design}[\Theta_s(V)]=\mbb{E}_\haar[\Theta_s(U)]$ for all monomials $\Theta_s$ of degree $s\leq\design$. Similarly then, for an $\varepsilon$-approximate $\design$-design, we adopt the definition of Ref.~\cite{Low_2009} with Eq.~\eqref{eq: def approx t design generic}, which implies the following.

\begin{definition}[$\varepsilon$-approximate unitary $\design$-design]
The distribution $\mu_{\appdesign}$ is a $\varepsilon$-approximate unitary $\design$-design  if for all balanced monomials $\Theta_s$ of degree $s\leq\design$,
\begin{gather}
	\left|\mbb{E}_{\,\appdesign}\Theta_s(V)-\mbb{E}_\mathsf{h}\Theta_s(U)\right|\leq\f{\varepsilon}{d^{\,\design}},
	\label{eq: approximate design Low}
\end{gather}
with $U,V\in\mbb{U}(d)$.
\end{definition}

From now on we focus on the more general approximate designs. We will see below that the degree $\varepsilon$ to which the distribution of the unitary dynamics on $\mu_{\,\appdesign}$ differs from an exact design for given $\design$ depends on the complexity of the model.

Notice what this means for the toy model of Fig.~\ref{Fig: Design box and process}$\mathsf{(a)}$: as individual random two-body interactions of each kind accumulate, what we expect is for the dynamics to start \emph{scrambling} their information across the whole gas in the box, progressively becoming more complex and uniformly random~\cite{Roberts2017}. Unitary designs give us this finite quantification of the resemblance towards uniform Haar randomness and, in this case, it can give us a precise way to account for the progressive emergence of complexity from seemingly simple individual two-body interactions.

Unitary designs for $\design=2,3$ have been widely studied~\cite{Emerson_2005, Gross2007, Harrow2009, Dankert_2009, Nakata2013, Wallman_2014, Webb_2016, 2designsXZ, Zhu_2017} and efficient constructions are known for larger values of $\design$~\cite{Brandao2016, Winter_HamDesign, Harrow2009}. The latter are of particular relevance, precisely as designs for large $\design$, i.e., those with a higher complexity~\cite{Roberts2017}, are expected to satisfy tighter large deviation bounds. Indeed, a statement of the form $\mbb{P}_{\appdesign}[\mc{N}\geq\delta]$ for a non-Markovianity measure $\mc{N}$ is expected to satisfy a bound similar to that in Theorem~\ref{Thm: typicality processes}, approaching genuine concentration of measure as the level and quality of the design increases. Such large deviation bounds --which here simply refers to bounds on probabilities for \emph{rare} events-- over approximate unitary designs were derived in a general form in Ref.~\cite{Low_2009} for a polynomial function satisfying a concentration of measure bound. 

\begin{theorem}[Large deviation bounds for $\design$-designs~\cite{Low_2009}]
\label{thm: large dev Low}
Let $\mc{X}$ be a polynomial of degree $\mathsf{T}$. Let $f(U)=\sum_i\alpha_i\Theta_{s_i}(U)$ where $\Theta_{s_i}(U)$ are monomials and let $\alpha(f)=\sum_i|\alpha_i|$. Suppose that $f$ has probability concentration
\begin{equation}
    \mbb{P}_\haar[|f-\zeta|\geq\delta]\leq C\exp\left(-\mscr{C}\delta^2\right),
\end{equation}
and let $\mu_{\appdesign}$, be an $\varepsilon$-approximate unitary $\design$-design, then
\begin{equation}
    \mbb{P}_\appdesign[|f-\zeta|\geq\delta]\leq\f{1}{\delta^{2m}}\left(C\left(\f{m}{\mscr{C}}\right)^m+\f{\varepsilon}{d^{\,\design}}(\alpha+|\zeta|)^{2m}\right),
    \label{eq: Low main large dev}
\end{equation}
for any integer $m$ with $2m\mathsf{T}\leq\design$.
\end{theorem}

This is the most general expression, where $\zeta$ can be any quantity, in particular the expectation of $f$. The main idea from this result in Ref.~\cite{Low_2009} (similarly applied before in Ref.~\cite{Brandao2016}) is that given a $\mu_\appdesign$ distribution as an $\varepsilon$-approximate unitary $\design$-design and a concentration result for a polynomial $f$ of degree $\mathsf{T}$, then one can compute
\begin{gather}
	\mbb{E}_{\,\appdesign}\left[f^m\right]=\mbb{E}_\haar\left[f^m\right]+g(\varepsilon,\design,f),
\end{gather}
where $m\leq{\design/2\mathsf{T}}$. Using Markov's inequality we have
\begin{align}
	\mbb{P}_\appdesign(f\geq\delta)=\mbb{P}_\appdesign\left(f^m\geq\delta^m\right)&\leq\f{\mbb{E}_{\,\appdesign}\left[f^m\right]}{\delta^m}=\f{1}{\delta^m}\left[\mbb{E}_\haar\left[f^m\right]+g(\varepsilon,\mathsf{t},f)\right],
	\label{eq: large dev form}
\end{align}
which is the form of the main large deviations bound in Eq.~\eqref{eq: Low main large dev}. More precisely, the other two main results that come along with the proof of Theorem~\ref{thm: large dev Low} given in Ref.~\cite{Low_2009} and allowing to compute the right hand-side of Eq.~\eqref{eq: large dev form} are the following.

\begin{lemma}[3.4 of Ref.~\cite{Low_2009}]
\label{Lemma: Low 1}
Let $\mc{X}$ be a polynomial of degree $\mathsf{T}$ and $\zeta$ any constant. Let $f(U)=\sum_i\alpha_i\Theta_{s_i}(U)$ where $\Theta_{s_i}(U)$ are monomials and let $\alpha(f)=\sum_i|\alpha_i|$. Then for an integer $m$ such that $2m\mathsf{T}\leq\design$ and $\mu_\appdesign$ an $\varepsilon$-approximate unitary $\design$-design,
\begin{equation}
    \mbb{E}_{\,\appdesign}\left[|f-\zeta|^{2m}\right]\leq\mbb{E}_\haar\left[|f-\zeta|^{2m}\right]+\f{\varepsilon}{d^{\,\design}}\left(\alpha+|\zeta|\right)^{2m}.
    \label{eq: Low lemma 1}
\end{equation}
\end{lemma}

\begin{lemma}[5.2 of Ref.~\cite{Low_2009}]
\label{Lemma: Low 2}
Let $X$ be any non-negative random variable with probability concentration
\begin{equation}
    \mbb{P}(X\geq\delta+\gamma)\leq{C}\exp(-\mscr{I}\,\delta^2),
\end{equation}
where $\gamma\geq0$, then
\begin{equation}
    \mbb{E}[X^m]\leq{C}\left(\f{2m}{\mscr{I}}\right)^{m/2}+(2\gamma)^m,
    \label{eq: Low lemma 2}
\end{equation}
for any $m>0$.
\end{lemma}

So, in essence, to solve our problem of finding a large deviations bound for quantum processes, what we can do given the results from Ref.~\cite{Low_2009}, is to determine the right-hand sides of Eq.~\eqref{eq: Low lemma 1} and Eq.~\eqref{eq: Low lemma 2} by phrasing our measure of non-Markovianity and all the other relevant quantities in such terms.

\section{Large deviations on unitary designs around Markovian processes}
Let us now then revisit Theorem~\ref{Thm: typicality processes}. For concreteness, in Ref.~\cite{FigueroaRomero2020makovianization} we focused only in the random interaction picture ($\mc{U}_i\neq\mc{U}_j$). As we mention in the previous chapter, the choice of the non-Markovianity measure as the trace distance between a quantum process and the closest Markovian one in Ref.~\cite{FigueroaRomero2019almostmarkovian} was motivated mainly due to its relation with the equilibration and state typicality results; strictly speaking this measure gives the distinguishability between explicitly constructed Choi states of corresponding process tensors and has no operational meaning. In this case a more suitable choice is the so-called \emph{diamond norm}.

While trace distance is a natural metric for differentiating two quantum states, the natural distance for differentiating two quantum channels is the diamond norm, which allows for the use of additional ancillas, as in discussed in Ref.~\cite{PhysRevA.71.062310}. We are interested in optimally distinguishing between a non-Markovian process from a Markovian one, which leads to the multitime diamond distance:
\begin{gather}
    \mc{N}_\bdiamond:=\f{1}{2}\min_{\Upsilon^\markov}\|\Upsilon-\Upsilon^\markov\|_\bdiamond,
 \label{Eq: def nM diamond}
\end{gather}
where
\begin{equation}
    \|\Upsilon\|_\bdiamond:=\sup_{\mc{J}=\{\mc{O}_i\}}\left\|\sum_i\tr[\mc{O}_i\Upsilon\otimes\mbb1]|i\rangle\!\langle{i}|\right\|_1
\end{equation} is a generalized diamond norm~\cite{Sacchi_2005,Phil_MemStr,Chiribella_2008_mem, Chiribella_2009, Gutoski_2012}, with supremum over the instrument $\mc{J}=\{\mc{O}_i\}$. This definition generalizes the diamond norm for quantum channel distinguishability~\cite{Kitaev} (also called cb-norm in Ref.~\cite{Paulsen_diamond} or completely bounded trace norm in Ref.~\cite{Watrous}), reducing to it for a single step process tensor, and similarly being interpreted as the optimal probability to discriminate a process from the closest Markovian one in a single shot, given any set of measurements together with an ancilla.

As we saw in the previous two chapters, this choice of non-Markovianity measure is not unique; more generally for any Schatten $p$-norm $\|X\|_p:=\tr(|X|^p)^{\frac{1}{p}}$ we can define a family of non-Markovianity measures as
\begin{gather}
    \mc{N}_p:= \f{1}{2} \min_{\Upsilon^\markov} \|\Upsilon-\Upsilon^\markov\|_p,
    \label{eq: Schatten}
\end{gather}
as done with $p=1$ in Theorem~\ref{Thm: typicality processes}, whenever $\Upsilon$ is normalized such that $\tr[\Upsilon]= \tr[\Upsilon^\markov]=1$. Then, we have the hierarchy $\mc{N}_1\geq\mc{N}_2\geq\ldots$, induced by that of the Schatten norms in Eq.~\eqref{eq: Schatten hierarchy}. In particular, we have the relation $d_\mathsf{S}^{-2k-1} \mc{N}_\bdiamond \leq \mc{N}_1\leq\mc{N}_\bdiamond$. This implies that the result in Theorem~\ref{Thm: typicality processes}, can be written equivalently as
\begin{gather}
    \mbb{P}_\haar\left[\mc{N}_\bdiamond\geq{d}_\mathsf{S}^{2k+1}\mc{B}+\delta\right]\leq \exp\left\{ -4\,\mscr{C}\,\delta^2 d_\mathsf{S}^{-2(2k+1)} \right\},
    \label{eq: non-Markov concentration Haar}
\end{gather}
where
\begin{gather}
    \mscr{C}=\f{d_\mathsf{SE}(k+1)}{16} \left(\f{d_\mathsf{S}-1}{d_\mathsf{S}^{k+1}-1}\right)^2,
    \label{eq: Lipschitz Haar}
\end{gather}
is the constant related to the Lipschitz constant of $\mc{N}_1$ in Theorem~\ref{Thm: typicality processes} up to a dimensional multiplicative factor of $4$, and where we now simply denote $\mc{B}\geq\mbb{E}[\mc{N}_1]$ the upper bound on the non-Markovianity $\mc{N}_1$ given in Eq.~\eqref{eq: bound on nM} with average purity in the random interaction case given in Eq.~\eqref{average purity ergodic}. In Ref.~\cite{FigueroaRomero2020makovianization} we refer to $\mscr{C}$ as \emph{the} Lipschitz constant of $\mc{N}_1$, but strictly speaking this is an inverse and rescaled version of $\mscr{L}$ in Eq.~\eqref{eq: N1 Lipschitz}. We recall that, holding everything else constant, we have the limiting cases $\mc{B}=0$ when $d_\mathsf{E}\to\infty$ and  $\mc{B}=1$ when $k\to\infty$, so that the expected non-Markovianity vanishes in the small subsystem limit and becomes loosest in the long time limit case.

\begin{theorem}[Markovianization with approximate unitary designs~\cite{FigueroaRomero2020makovianization}]
\label{Thm: Large Dev Markov}
Given a $k$-step process $\Upsilon$ on a $d_\mathsf{S}$ dimensional subsystem, generated from global unitary $d_\mathsf{SE}$ dimensional \gls{syst-env} dynamics distributed according to an $\varepsilon$-approximate unitary $\design$-design $\mu_{\,\appdesign}$, the likelihood that its non-Markovianity exceeds any $\delta>0$ is bounded as
\begin{gather}
  \mbb{P}_{\appdesign}\left[\,\mc{N}_\bdiamond\geq\delta\,\right]\leq\mathsf{B},
  \label{eq: largedev_design}
\end{gather}
where $\mathsf{B}$ is defined as
\begin{gather}
    \mathsf{B}:=\f{d_\mathsf{S}^{3m(2k+1)}}{\delta^{2m}}\left[\left(\f{m}{\mscr{C}}\right)^{m}\hspace{-0.05in}+(2\mc{B})^{2m}+\f{\varepsilon}{d_\mathsf{SE}^{\,\mathsf{t}}}\eta^{2m}\right],
    \label{eq: main markovianization thm}
\end{gather}
for any $m\in(0,\mathsf{t}/4]$ and
\begin{gather}
    \eta:=(d_{\mathsf{SE}}^4d_\mathsf{S}^{2k}+d_\mathsf{S}^{-(2k+1)})/4,
    \label{eq: def eta}
\end{gather}
where $\mscr{C}$ is defined in Eq.~\eqref{eq: Lipschitz Haar} and $\mc{B}$ an upper bound on the expected norm-1 non-Markovianity $\mbb{E}_\haar[\mc{N}_1]$, defined in Eq.~\eqref{eq: bound on nM} with average noisiness given by Eq.~\eqref{average purity ergodic}.
\end{theorem}

\begin{proof}
The overall strategy is as in Ref.~\cite{Low_2009}: a bound on the moments $\mbb{E}_{\,\appdesign}[\mc{N}_\bdiamond^{\,m}]$ is given in terms of $\mc{B}$, $\mscr{C}$ and $\eta$, followed by Markov's inequality. As we mention by the end of the previous section, this amounts to determining the right-hand sides of Eq.~\eqref{eq: Low lemma 1} and Eq.~\eqref{eq: Low lemma 2}. For us the relevant quantity is the non-Markovianity $\mc{N}_\bdiamond$, which is a fairly hard quantity to work with, however, we can use the relationship of this norm with the family of Schatten-norm non-Markovianity measures to compute relevant bounds.

Unsurprisingly, the hardest quantity to compute is the sum of moduli of coefficients of the non-Markovianity expressed as a polynomial in the unitaries, so the most accessible way to do this is to turn to the non-Markovianity $\mc{N}_2$.

In general, $\|X\|_1\geq\|X\|_2$ for Schatten norms, so given the concentration result for $\mc{N}_1$ in Theorem~\ref{Thm: typicality processes} and the upper bound $\mc{B}\geq\mbb{E}[\mc{N}_1]$ given in Eq.~\eqref{eq: bound on nM}, this also implies
\begin{gather}
    \mbb{P}_\haar\left[\mc{N}_2\geq\mc{B}+\delta\right]\leq\mathrm{e}^{-
\mscr{C}\delta^2},
\end{gather}
so that in turn Eq.~\eqref{eq: Low lemma 2} implies that
\begin{align}
	\mbb{E}_\haar\left[\mc{N}_2^{\,2m}\right]&\leq\left(\f{m}{\mscr{C}}\right)^{m}+(2\mc{B})^{2m}=\left[\f{16m}{(k+1)d_\mathsf{SE}}\left(\f{d_\mathsf{S}^{k+1}-1}{d_\mathsf{S}-1}\right)^2\,\right]^{m}+(2\mc{B})^{2m},
\end{align}
for any $m>0$, and where we have absorbed a multiplicative factor of $4$ in the definition of $\mscr{C}$, which is now as in Eq.~\eqref{eq: Lipschitz Haar}.

For the case of all unitaries at each step being independently sampled, $\mc{N}_2^{\,2}$ is a polynomial of degree $p=2$ when the unitaries are all distinct (random interaction type). We can thus take $\mc{N}_2^{\,2}$ and apply Lemma~\ref{Lemma: Low 1} of Ref.~\cite{Low_2009} for a unitary $\design$-design $\mu_{\,\appdesign}$ with $\design\geq{4m}$, which holds for real $m>0$,\footnote{ The proof of Lemma~\ref{Lemma: Low 1} in Ref.~\cite{Low_2009} requires $m$ to be an integer through the multinomial theorem; in the notation of the cited paper, this can be relaxed to be a real number when $\mu=0$ and applying the multinomial theorem for a real power: convergence will require (an ordering such that) \unexpanded{$|\alpha_t\mbb{E}M_t|>2^{1-n}|\alpha_{t-n}\mbb{E}M_{t-n}|$} for each $n=1,\ldots,t-1$ for both the approximate design and Haar expectations.} as
\begin{align}
    \mbb{E}_{\,\appdesign}\left[\mc{N}_2^{\,2m}\right]\leq\mbb{E}_\haar\left[\mc{N}_2^{\,2m}\right]+\f{\varepsilon}{d_\mathsf{SE}^{\,\design}}\,\eta^{2m},
\end{align}
where $\eta$ is the sum of the moduli of the coefficients of
\begin{align}
    \mc{N}_2^{\,2}&=\left(\f{1}{2}\min_{\Upsilon^\markov}\|\Upsilon-\Upsilon^\markov\|_2\right)^2\leq\f{1}{4}\|\Upsilon-\f{\mbb1}{d_\mathsf{S}^{2k+1}}\|_2^2=\f{1}{4}\left[\tr\left(\Upsilon^2\right)-d_\mathsf{S}^{-(2k+1)}\right],
    \label{eq: appendix eta}
\end{align}
so our problem now boils down to computing the sum of the moduli of the coefficients of the purity, or noisiness, of the process.

Let us explicitly write the process $\Upsilon$ as a function of the set of unitary maps $\mathfrak{U}:=\{\mc{U}_i\}_{i=0}^k$, i.e.
\begin{align}
    &\Upsilon[\mathfrak{U}]=\tr_\mathsf{E}\left[\,\mc{U}_k\,\mc{S}_k\,\mc{U}_{k-1}\mc{S}_{k-1}\cdots\mc{U}_1\mc{S}_1\left(\rho\otimes\mathsf{\Psi}^{\otimes\,k}\right)\,\right]\nonumber\\
    &=\!\!\sum_{\alpha,\ldots,\delta}\!\!\tr_\mathsf{E}\left[U_k\FS_{\alpha_k\beta_k}\!\cdots{U}_1\FS_{\alpha_1\beta_1}\rho\FS_{\delta_1\gamma_1}U_1^\dg\cdots\FS_{\delta_k\gamma_k}U_k^\dg\right]\otimes|\beta_1\alpha_1\!\cdots\beta_k\alpha_k\rangle\!\langle\delta_1\gamma_1\!\cdots\delta_k\gamma_k|,
    \label{eq: pt design proof}
\end{align}
where the sum runs over all Greek indices from $1$ to $d_\mathsf{S}$, with $\mathsf\Psi$ being a maximally entangled state acting in the respective $d_\mathsf{S}$-dimensional ancillary spaces $\mathsf{A}_i\mathsf{B}_i$, where $\mc{S}_i$ are swaps between \gls{syst} and ancillary system $\mathsf{A}_i$ at time-step $i$, and where $\FS_{\alpha\beta}=\mbb{1}_\mathsf{E}\otimes|\alpha\rangle\!\langle\beta|$. Here as well, full detail can be revisited around the definition in Eq.~\eqref{eq: process tensor Choi state}.

Now, the standard approach to compute the sum of the moduli of the coefficients of a given polynomial is to evaluate on an argument (here a $d_\mathsf{SE}\times{d}_\mathsf{SE}$ matrix) full of ones (so that all single monomials equal to one) and take each summand to the corresponding modulus. We follow this approach, however, we first notice that the environment part in Eq.~\eqref{eq: pt design proof} is just a product of the environment parts of all unitaries and initial state.\footnote{ Let \unexpanded{$U=\sum{U}^{es}_{e^\prime{s}^\prime}|es\rangle\!\langle{e}^\prime{s}^\prime|$} where \unexpanded{$|e\rangle$ and $|s\rangle$} are $\mathsf{E}$ and $\mathsf{S}$ bases. Unitarity then implies \unexpanded{$\sum\overline{U}_{es}^{ab}U_{\epsilon\sigma}^{ab}=\delta_{e\epsilon}\delta_{s\sigma}$}, and so this means that \unexpanded{$\tr_\mathsf{E}[V\mf{S}_{\alpha\beta}U\rho{U}^\dg\mf{S}_{\gamma\delta}V^\dg]=\sum
{V}^{es}_{e^\prime{s}^\prime}\overline{V}^{e\sigma}_{e^\prime\sigma^\prime}{U}^{e^\prime{s}_2}_{b{s}_2^\prime}
\overline{U}^{e^\prime\sigma_2}_{b\sigma_2^\prime}\rho^{br}_{bt}\,\phi(S)$} where \unexpanded{$\phi(S)$} stands for the system $\mathsf{S}$ part; for each $b$ index the rest of the terms are summed over $e$. This generalizes similarly for any number of unitaries.} This implies that at most $d_\mathsf{E}$ terms need to be set to one and we can evaluate $\Upsilon$ in a set of matrices $\mscr{J}=\{\mbb1_\mathsf{E}\otimes{J}_\mathsf{S},\cdots,\mbb1_\mathsf{E}\otimes{J}_\mathsf{S},J_\mathsf{E}\otimes{J}_\mathsf{S}\}$ with $J$ a matrix with each element equal to one in the respective \gls{env} or \gls{syst} systems: let $\rho= \sum \rho_{ese^\prime{s}^\prime}|es\rangle\!\langle{e^\prime{s}^\prime}|$, then
\begin{align}
    &\Upsilon[\mscr{J}]=d_\mathsf{S}^{-k}\sum\rho_{ese^\prime{s}^\prime}\tr\left[d_\mathsf{E}J_\mathsf{E}|e\rangle\!\langle{e}^\prime|\right]{J}_\mathsf{S}|\alpha_k\rangle\!\langle\beta_k|\cdots|\alpha_1\rangle\!\langle\beta_1|J_\mathsf{S}|s\rangle\!\langle{s}^\prime|J_\mathsf{S}|\delta_1\rangle\!\langle\gamma_1|\cdots|\delta_k\rangle\!\langle\gamma_k|J_\mathsf{S}\nonumber\\
    &\qquad\qquad\qquad\qquad\qquad\qquad\otimes|\beta_1\alpha_1\cdots\beta_k\alpha_k\rangle\!\langle\delta_1\gamma_1\cdots\delta_k\gamma_k|\nonumber\\
    &=\f{d_\mathsf{E}}{d_\mathsf{S}^k}\sum\rho_{ese^\prime{s}^\prime}\,{J}_\mathsf{S}|\alpha_k\rangle\!\langle\beta_k|\cdots|\alpha_1\rangle\!\langle\beta_1|J_\mathsf{S}|s\rangle\!\langle{s}^\prime|J_\mathsf{S}|\delta_1\rangle\!\langle\gamma_1|\cdots|\delta_k\rangle\!\langle\gamma_k|J_\mathsf{S}\nonumber\\
    &\qquad\qquad\qquad\qquad\qquad\qquad\otimes|\beta_1\alpha_1\cdots\beta_k\alpha_k\rangle\!\langle\delta_1\gamma_1\cdots\delta_k\gamma_k|,
\end{align}
and hence (we now omit the subindex $\mathsf{S}$ on the $J$ matrices for simplicity),
\begin{align}
    &\left(\f{d_\mathsf{S}^k}{d_\mathsf{E}}\right)^2
    \tr[\Upsilon^2(\mscr{J})]=
    \sum\rho_{ese^\prime{s}^\prime}\rho_{\epsilon\sigma\epsilon^\prime\sigma^\prime}\tr[
    {J}|\alpha_k\rangle\!\langle\beta_k|\cdots|\alpha_1\rangle\!\langle\beta_1|J|s\rangle\!\langle{s}^\prime|J|\delta_1\rangle\!\langle\gamma_1|\cdots\nonumber\\
    &\qquad\qquad\qquad J|\delta_k\rangle\!\langle\gamma_k|J^2|\gamma_k\rangle\!\langle\delta_k|J\cdots|\gamma_1\rangle\!\langle\delta_1|J|\sigma\rangle\!\langle\sigma^\prime|J|\beta_1\rangle\!\langle\alpha_1|J\cdots|\beta_k\rangle\!\langle\alpha_k|J
    ]\nonumber\\
    &=
    d_\mathsf{S}^2\sum\rho_{ese^\prime{s}^\prime}\rho_{\epsilon\sigma\epsilon^\prime\sigma^\prime}\tr[J|\alpha_k\rangle\!\langle\beta_k|\cdots|\alpha_1\rangle\!\langle\beta_1|J|s\rangle\!\langle{s}^\prime|J|\delta_1\rangle\!\langle\gamma_1|\cdots\nonumber\\
    &\qquad\qquad\qquad \langle\gamma_{k-1}|J|\delta_k\rangle\!\langle\delta_k|J|\gamma_{k-1}\rangle\cdots|\gamma_1\rangle\!\langle\delta_1|J|\sigma\rangle\!\langle\sigma^\prime|J|\beta_1\rangle\!\langle\alpha_1|J\cdots|\beta_k\rangle\!\langle\alpha_k|J]\nonumber\\
    &=
    d_\mathsf{S}^{2k+1}\sum\rho_{ese^\prime{s}^\prime}\rho_{\epsilon\sigma\epsilon^\prime\sigma^\prime}
    \,\tr[J|\alpha_k\rangle\!\langle\beta_k|\cdots|\alpha_1\rangle\!\langle\beta_1|J|s\rangle\!\langle{s}^\prime|J|\sigma\rangle\!\langle\sigma^\prime|J|\beta_1\rangle\!\langle\alpha_1|J\cdots|\beta_k\rangle\! \langle\alpha_k|J]\nonumber\\
    &=
    d_\mathsf{S}^{2k+3}\sum\rho_{ese^\prime{s}^\prime}\rho_{\epsilon\sigma\epsilon^\prime\sigma^\prime}
     \langle\beta_k|J|\alpha_{k-1}\rangle\cdots|\alpha_1\rangle\!\langle\beta_1|J|s\rangle\!\langle{s}^\prime|J|\sigma\rangle\!\langle\sigma^\prime|J|\beta_1\rangle\!\langle\alpha_1|\cdots \langle\alpha_{k-1}|J|\beta_k\rangle\nonumber\\
    &=
    d_\mathsf{S}^{2k+5}\sum\rho_{ese^\prime{s}^\prime}\rho_{\epsilon\sigma\epsilon^\prime\sigma^\prime}
    \langle\beta_{k-1}|J|\alpha_{k-2}\rangle\cdots|\alpha_1\rangle\!\langle\beta_1|J|s\rangle\!\langle{s}^\prime|J|\sigma\rangle\!\langle\sigma^\prime|J|\beta_1\rangle\!\langle\alpha_1|\cdots \langle\alpha_{k-1}|J|\beta_{k-1}\rangle\nonumber\\
    &=
    d_\mathsf{S}^{2(2k+1)}\sum\rho_{ese^\prime{s}^\prime}\rho_{\epsilon\sigma\epsilon^\prime\sigma^\prime},
\end{align}
where to obtain the second line we used the fact that $J^n=d_\mathsf{S}^{n-1}J$ for positive integers $n$, here applied for $n=2$, together with the trace over system $\mathsf{S}$ given by $\sum\langle\gamma_k|\cdot|\gamma_k\rangle$. This is similarly done to get the third line by $\sum|\delta_k\rangle\!\langle\delta_k|=\mbb1_\mathsf{S}$, and taking the trace summing over $|\gamma_{k-1}\rangle$, which can subsequently be done for all $|\gamma_i\rangle$ and $|\delta_i\rangle$. For the fourth line the cyclicity of the trace was used, followed by taken an identity summing up over $|\alpha_k\rangle$, using $J^2=d_\mathsf{S}J$ and taking the trace. This can be done through all remaining steps, giving the last equality.
This, together with Eq.~\eqref{eq: appendix eta}, implies that (now writing simply $i$, $j$ for $\mathsf{SE}$ indices),
\begin{align}
    4\eta&\leq d_\mathsf{E}^2d_\mathsf{S}^{2(k+1)}\left(\sum|\rho_{ij}|\right)^2+\f{1}{d_\mathsf{S}^{2k+1}}
    \nonumber\\&
    \leq{d}_\mathsf{E}^4d_\mathsf{S}^{2(k+2)}\sum|\rho_{ij}|^2+\f{1}{d_\mathsf{S}^{2k+1}}
    \nonumber\\&
    \leq{d}_\mathsf{E}^4d_\mathsf{S}^{2(k+2)}+\f{1}{d_\mathsf{S}^{2k+1}},
\end{align}
where in the second line we used $\|X\|_1^2\leq{d}\|X\|_2^2$ for element-wise norms $\|X\|_p^p=(\sum|x_{ij}|^p)$ and in the third line we used $\|\rho\|_2^2\leq1$.

We can finally put everything together as follows. As $d_\mathsf{S}^{-2k1-1}\mc{N}_\bdiamond\leq\mc{N}_1\leq\sqrt{d_\mathsf{S}^{2k+1}}\mc{N}_2$, also for $0<m\leq{t}/4$,
\begin{align}
    &\mbb{P}_{\appdesign}[\mc{N}_\bdiamond\geq\delta] \leq \ \mbb{P}_{\appdesign} \left[ \sqrt{d_\mathsf{S}^{3(2k+1)}}\,\mc{N}_2\geq\delta \right]
    =\mbb{P}_{\appdesign} \left[ \mc{N}_2^{\,2m}\geq\f{\delta^{2m}}{d_\mathsf{S}^{3m(2k+1)}} \right]\nonumber\\
    \leq& \f{d_\mathsf{S}^{3m(2k+1)}\,\mbb{E}_{\,\appdesign}\mc{N}_2^{\,2m}}{\delta^{2m}}
    \leq\left(\f{d_\mathsf{S}^{3(2k+1)}}{\delta^2}\right)^m\left[\left(\f{4m}{\mscr{C}}\right)^{m}+(2\mc{B})^{2m}+\f{\epsilon}{d_\mathsf{SE}^{\,t}}\eta^{2m}\right]\nonumber\\
    =&\left(\f{d_\mathsf{S}^{3(2k+1)}}{\delta^2}\right)^m\left\{\left[\f{16m}{(k+1)\,d_\mathsf{SE}}\left(\f{d_\mathsf{S}^{k+1}-1}{d_\mathsf{S}-1}\right)^2\right]^{m}
    \!\!+(2\mc{B})^{2m}+\f{\epsilon}{16^md_\mathsf{SE}^{\,t}}\left(d_\mathsf{E}^4d_\mathsf{S}^{2(k+2)}+\f{1}{d_\mathsf{S}^{2k+1}}\right)^{2m}\right\}\label{largedev_design},
\end{align}
where in the third line we used Markov's inequality.
\end{proof}

The choice of $0<m\leq\design/4$ can be made to optimize the right-hand-side of the inequality, which ideally should be small whenever $\delta$ is. The term $d_\mathsf{S}^{3(2k+1)}/\delta^2$ arises from bounding $\mc{N}_\bdiamond$ and Markov's inequality, while the three summands within square brackets will be small provided \emph{i}) $\mscr{C}$ is large, \emph{ii}) $\mc{B}$ is small and \emph{iii}) the unitary design sufficiently small $\varepsilon$ and large $\design$ {is well-approximate and high enough}. For conditions \emph{i}) and \emph{ii}), as in Chapter~\ref{sec:typicality}, we require a fixed $k$ such that $d_\mathsf{E}\gg{d}_\mathsf{S}^{2k+1}$: this implies $\mc{B}\approx0$ so that
\begin{align}
    \mbb{P}_{\appdesign}\left[\mc{N}_\bdiamond\geq\delta\right]
    &\lesssim\left(\f{d_\mathsf{S}^{3(2k+1)}}{\delta^2}\right)^m\left\{\left[\f{16m}{(k+1)\,d_\mathsf{SE}}\left(\f{d_\mathsf{S}^{k+1}-1}{d_\mathsf{S}-1}\right)^2\right]^{m}\hspace{-0.1in}+\f{\varepsilon}{16^md_\mathsf{SE}^{\,\design}}\left({d}_\mathsf{E}^4d_\mathsf{S}^{2(k+2)}+\f{1}{d_\mathsf{S}^{2k+1}}\right)^{2m}\right\} \nonumber\\[0.1in]
    &\approx\left\{\left[\f{16m}{\delta^2(k+1)}\f{d_\mathsf{S}^{2(4k+1)}}{d_\mathsf{E}}\right]^{m} +\varepsilon\f{d_\mathsf{E}^{8m-\design}d_\mathsf{S}^{m(10k+11)-\design}}{\delta^{2m}16^m}\right\}
    \label{LargeDev_BigE}.
\end{align}

Now, supposing the $\design$-design is exact, i.e. $\varepsilon=0$, we require $m\leq\delta^2\f{(k+1)d_\mathsf{E}}{16\,d_\mathsf{S}^{6k}}$, together with $m\leq\design/4$. On the other hand if $\varepsilon$ is non-zero, we require
\begin{align}
    \varepsilon\ll\left[\delta^2\left(\f{2}{d_\mathsf{E}^2d_\mathsf{S}^{(10k+11)/4}}\right)^4\right]^md_\mathsf{SE}^{\,\design}.
\end{align}

The choice of real $m$ is only restricted by $0<m\leq\design/4$, but is otherwise arbitrary. The right-hand side of Eq.~\eqref{eq: main markovianization thm} is not monotonic in $m$ over all the remaining parameters, so it won't always be optimal for some fixed choice. One may thus optimize the choice of $m$ numerically for each particular case.

\begin{figure}[t]
\centering
\begin{minipage}{\textwidth}
\centering
\begin{minipage}{0.84\textwidth}
\begin{tikzpicture}
\node[anchor=south west, inner sep=0] (image) at (0,0) {\includegraphics[width=\textwidth]{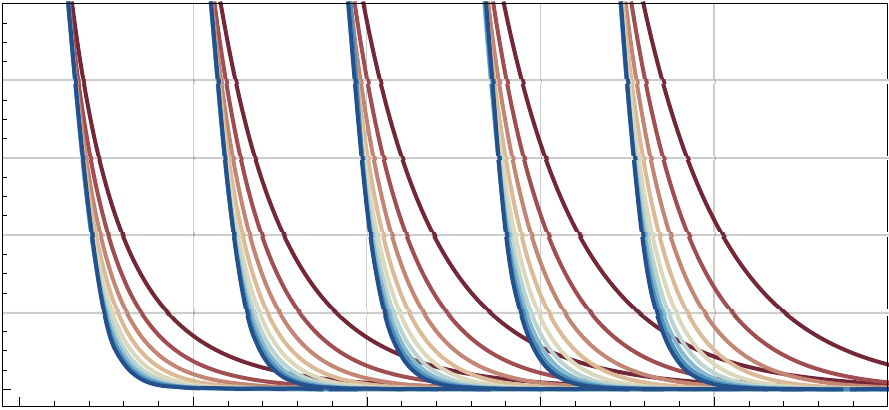}};
\begin{scope}[x={(image.south east)},y={(image.north west)}]
  \node[right] at (0,1.075) {\large{$\mathsf{B}$}};
  \node at (0.5,-0.175)
  {\large$\log_2(d_\mathsf{E})$};
  \node[below, right] at (0,-0.05)
  {\large{10}};
  \node[below] at (0.5,0)
  {\large{35}};
  \node[below,left] at (1,-0.05)
  {\large{60}};
  \node[rotate=90, above] at (0,1) {\large{1}};
  \node[rotate=90, above] at (0,0.51) {\large{0.5}};
  \node[rotate=90, above] at (0,0) {\large{0}};
  \node[below] at (0.168,0.8) {\large$k=0$};
  \node[below] at (0.338,0.8) {\large$k=1$};
  \node[below] at (0.493,0.8) {\large$k=2$};
  \node[below] at (0.65,0.8) {\large$k=3$};
  \node[below] at (0.82,0.8) {\large$k=4$};
  \end{scope}
\end{tikzpicture}
\end{minipage}%
\begin{minipage}{0.15\textwidth}
\hspace*{0.25in}
\vspace*{0.325in}
\begin{tikzpicture}
\node[anchor=south west, inner sep=0] (image) at (0,0) {\includegraphics[width=0.25\textwidth]{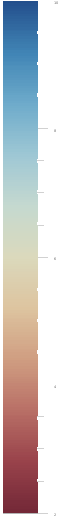}};
\begin{scope}[x={(image.south east)},y={(image.north west)}]
  \node[above] at (0.425,1) {\large$\mathsf{t}$};
  \node[rotate=90,above] at (0,0) {\large{2}};
  \node[rotate=90,above] at (0,0.975) {\large{10}};
  \end{scope}
\end{tikzpicture}
\end{minipage}
\end{minipage}
\caption[Large deviations bound on non-Markovianity with unitary designs]{\textbf{Large deviations bound on non-Markovianity with unitary designs:} Upper bound $\mathsf{B}$ on $\mbb{P}_\appdesign[\,\mc{N}_\bdiamond\geq0.1]$ defined by Eq.~\eqref{eq: largedev_design} against $\log_2(d_\mathsf{E})$ for a subsystem qubit undergoing a joint closed approximate unitary design interaction at each step. We fix an $\varepsilon=10^{-12}$ approximate unitary $\design$-design for different values $2\leq\design\leq10$ and fixed values of timesteps $k$, optimizing $m$ for each case.}
  \label{Fig: plot Prob-dE}
\end{figure}
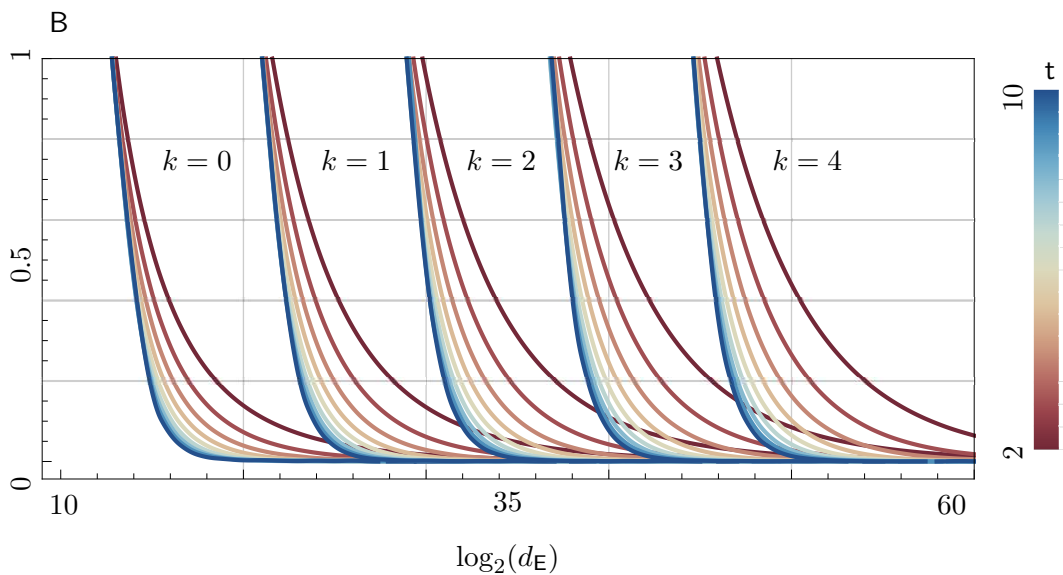

Overall, the bound in Eq.~\eqref{eq: largedev_design} approaches concentration whenever $d_\mathsf{E}$ is large relative to $d_\mathsf{S}$ and $k$, together with large enough $\design$, as shown in Fig.~\ref{Fig: plot Prob-dE}. Therefore, the vast majority of processes sampled from such a $\design$-design are indistinguishable from Markovian ones in this limit. Despite the final bound in Theorem~\ref{Thm: Large Dev Markov} being seemingly complicated, we can apply it to state-of-the-art efficient costructions of simple circuits that generate unitary designs. We now show how these processes can be modelled in terms of random circuits.

\section{Markovianization by circuit design}
While no explicit sets forming unitary $\design$-designs for $\design\ge4$ are known to date, several efficient constructions generating approximate unitary designs by quantum circuits are known. Using these constructions we can highlight the physical implications of the theorem above. We begin by discussing the details of one such construction. As suggested in Fig.~\ref{Fig: Design box and process}$\mathsf{(b)}$, this construction only requires two-qubit interactions and, under certain conditions, yields an approximate unitary design, from which we can use Eq.~\eqref{eq: largedev_design} in Theorem~\ref{Thm: Large Dev Markov} to verify that Markovianization emerges.

In our manuscript we focused specifically on Result 2 of Ref.~\cite{Winter_HamDesign}, where a circuit with interactions mediated by two-qubit diagonal gates with three random parameters is introduced. To begin with, an efficient approximation for a unitary design on a system composed of $n$-qubits is shown in Ref.~\cite{Winter_HamDesign} for a circuit labeled $\mathrm{RDC}(\mc{I}_2)$, where the name stands for \emph{Random Diagonal Circuit}, and refers to a circuit where $\mc{I}_2=\{I_i\}$ is a set of subsets of qubit labels $I_i\subset\{1,\ldots,n\}$, such that $|I_i|=2$, i.e., at step $i$, $I_i$ picks a pair of qubits, to which a Pauli-$Z$-diagonal gate with three random parameters is applied. The same idea follows in general for an arbitrary number of qubits $|I_i|$, but here we focus on the case of two-qubit interactions only. This construction can already be seen in Ref.~\cite{2designsXZ} as arising from only two types of random diagonal interactions, which can be simplified into a product of $Z$-diagonal ones. The brilliance in this construction lies in the intuition that repeated alternate applications of these diagonal gates quickly randomizes the system.

This idea now fully captures the gas scenario depicted in Fig.~\ref{Fig: Design box and process}$\mathsf{(a)}$, where we only have two types of random two-body interactions repeatedly occurring, and we focus on one of the particles of the gas. We can more concretely illustrate this idea in Fig.~\ref{fig: RDC}, where we depict an $n$-qubit \gls{syst-env} composite with $k$ interventions on one of the qubits, with the unitary interactions within the circuit being only between pairs of qubits and of only two kinds.

\begin{figure}[t]
\centering
\begin{tikzpicture}
  \begin{scope}[xscale=0.575, yscale=0.55]
 \draw[-, C2, thick, rounded corners, fill=C2!1!white](1.25,-0.5) -- (1.25,-6.75) -- (4.75,-6.75) -- (4.75,-5.35) -- (5.35,-5.35) -- (5.35,-0.5) -- cycle;
 
 \draw[-,C2, thick, rounded corners, fill=C2!1!white](5.55,-0.5) -- (10.75,-0.5) -- (10.75,-5.35) -- (9.75,-5.35) -- (9.75,-6.75) -- (6.75,-6.75) -- (6.75,-5.35) -- (5.55,-5.35)  -- cycle;
 
\draw[-,C2, thick, rounded corners, fill=C2!1!white](15.5,-0.5) -- (20.75,-0.5) -- (20.75,-5.35) -- (19.8,-5.35) -- (19.8,-6.75) -- (16.75,-6.75) -- (16.75,-5.35) --  (15.5,-5.35) -- cycle;
 
 \node[below] at (3,-6.76) {\Large$\mc{W}_{\ell_0}$};
 \node[below] at (8.3,-6.76) {\Large$\mc{W}_{\ell_1}$};
 \node[below] at (18.35,-6.76) {\Large$\mc{W}_{\ell_k}$};
 
 \foreach \i in {1,...,6}{
 \node at (0,-\i) {\Large$|0\rangle$};
 \draw[] (0.75,-\i) -- (24,-\i);
 }
 
 \draw[thick] (1.75,-1.25) -- (1.75,-1.75);
\shade[inner color=white, outer color=C1, draw=black, rotate around={45:(1.75,-1)},thick] (1.5,-0.75) rectangle (2,-1.25);
\shade[inner color=white, outer color=C1, draw=black, rotate around={45:(1.75,-2)}, thick] (1.5,-1.75) rectangle (2,-2.25);
 
 \draw[thick] (1.75,-4.25) -- (1.75,-5.75);
 \shade[bottom color=C2, top color=white, draw=black] (1.5,-3.75) rectangle (2,-4.25);
 \shade[bottom color=C2, top color=white, draw=black] (1.5,-5.75) rectangle (2,-6.25);
 
 \draw[thick] (2.75,-3.25) -- (2.75,-4.75);
 \shade[inner color=white, outer color=C1, draw=black, rotate around={45:(2.75,-3)}] (2.5,-2.75) rectangle (3,-3.25);
 \shade[inner color=white, outer color=C1, draw=black, rotate around={45:(2.75,-5)}] (2.5,-4.75) rectangle (3,-5.25);
 
 \draw[thick] (3.75,-4.25) -- (3.75,-5.75);
 \shade[bottom color=C2, top color=white, draw=black] (3.5,-3.75) rectangle (4,-4.25);
 \shade[bottom color=C2, top color=white, draw=black] (3.5,-5.75) rectangle (4,-6.25);
 
 \draw[thick] (4.75,-2.25) -- (4.75,-4.75);
 \shade[inner color=white, outer color=C1, draw=black, rotate around={45:(4.75,-2)}] (4.5,-1.75) rectangle (5,-2.25);
 \shade[inner color=white, outer color=C1, draw=black, rotate around={45:(4.75,-5)}] (4.5,-4.75) rectangle (5,-5.25);
 
  \shade[outer color=C3!60!white, inner color=white, draw=black, rounded corners, thick] (5,-5.5) rectangle (6.5,-7);
  \node at (5.75,-6.25) {\Large$\mc{A}_0$};
 
 \draw[thick] (6.1,-1.25) -- (6.1,-3.75);
 \shade[bottom color=C2, top color=white, draw=black] (5.85,-0.75) rectangle (6.35,-1.25);
 \shade[bottom color=C2, top color=white, draw=black] (5.85,-3.75) rectangle (6.35,-4.25);
 
 \draw[thick] (7.25,-2.25) -- (7.25,-2.75);
 \shade[inner color=white, outer color=C1, draw=black, rotate around={45:(7.25,-2)}] (7,-1.75) rectangle (7.5,-2.25);
 \shade[inner color=white, outer color=C1, draw=black, rotate around={45:(7.25,-3)}] (7,-2.75) rectangle (7.5,-3.25);
 
 \draw[thick] (8.25,-1.25) -- (8.25,-4.75);
  \shade[bottom color=C2, top color=white, draw=black] (8,-0.75) rectangle (8.5,-1.25);
 \shade[bottom color=C2, top color=white, draw=black] (8,-4.75) rectangle (8.5,-5.25);
 
 \draw[thick] (9.25,-2.25) -- (9.25,-3.75);
 \shade[inner color=white, outer color=C1, draw=black, rotate around={45:(9.25,-2)}] (9,-1.75) rectangle (9.5,-2.25);
 \shade[inner color=white, outer color=C1, draw=black, rotate around={45:(9.25,-4)}] (9,-3.75) rectangle (9.5,-4.25);
 
 \shade[outer color=C3!60!white, inner color=white, draw=black, rounded corners, thick] (10,-5.5) rectangle (11.5,-7);
  \node at (10.75,-6.25) {\Large$\mc{A}_1$};
  
 \draw[thick] (10.25,-3.25) -- (10.25,-4.75);
  \shade[bottom color=C2, top color=white, draw=black] (10,-2.75) rectangle (10.5,-3.25);
 \shade[bottom color=C2, top color=white, draw=black] (10,-4.75) rectangle (10.5,-5.25);
 
 \draw[thick] (11.25,-1.25) -- (11.25,-3.75);
 \shade[inner color=white, outer color=C1, draw=black, rotate around={45:(11.25,-1)}] (11,-0.75) rectangle (11.5,-1.25);
 \shade[inner color=white, outer color=C1, draw=black, rotate around={45:(11.25,-4)}] (11,-3.75) rectangle (11.5,-4.25);
 
 \draw[thick] (12.25,-3.25) -- (12.25,-5.75);
 \shade[bottom color=C2, top color=white, draw=black] (12,-2.75) rectangle (12.5,-3.25);
 \shade[bottom color=C2, top color=white, draw=black] (12,-5.75) rectangle (12.5,-6.25);
 
  \shade[outer color=C3!60!white, inner color=white, draw=black, rounded corners, thick] (14.75,-5.5) rectangle (16.25,-7);
  \node at (15.25,-6.25) {\Large$\mc{A}_{k-1}$};
 
 \draw[thick] (16.25,-1.25) -- (16.25,-2.75);
 \shade[inner color=white, outer color=C1, draw=black, rotate around={45:(16.25,-1)}] (16,-0.75) rectangle (16.5,-1.25);
 \shade[inner color=white, outer color=C1, draw=black, rotate around={45:(16.25,-3)}] (16,-2.75) rectangle (16.5,-3.25);
 
 \draw[white, fill=white, path fading= east] (15,0.5) -- (17,0.5) -- (17,-7) -- (15,-7);
 \draw[white, fill=white] (13,0.5) -- (15,0.5) -- (15,-7) -- (13,-7);
 \draw[white, fill=white, path fading= west] (11,-7) -- (11,0.5) -- (13,0.5) -- (13,-7) ;
 
 \node at (14,-3.5) {$\cdots$};
 
 \draw[thick] (17.25,-2.25) -- (17.25,-3.75);
 \shade[bottom color=C2, top color=white, draw=black] (17,-1.75) rectangle (17.5,-2.25);
 \shade[bottom color=C2, top color=white, draw=black] (17,-3.75) rectangle (17.5,-4.25);
 
 \draw[thick] (17.25,-5.25) -- (17.25,-5.75);
 \shade[bottom color=C2, top color=white, draw=black] (17,-4.75) rectangle (17.5,-5.25);
 \shade[bottom color=C2, top color=white, draw=black] (17,-5.75) rectangle (17.5,-6.25);
 
 \draw[thick] (18.25,-4.25) -- (18.25,-4.75);
 \shade[inner color=white, outer color=C1, draw=black, rotate around={45:(18.25,-4)}] (18,-3.75) rectangle (18.5,-4.25);
 \shade[inner color=white, outer color=C1, draw=black, rotate around={45:(18.25,-5)}] (18,-4.75) rectangle (18.5,-5.25);
 
 \draw[thick] (19.25,-3.25) -- (19.25,-5.75);
 \shade[bottom color=C2, top color=white, draw=black] (19,-2.75) rectangle (19.5,-3.25);
 \shade[bottom color=C2, top color=white, draw=black] (19,-5.75) rectangle (19.5,-6.25);
 
 \shade[outer color=C3!60!white, inner color=white, draw=black, rounded corners, thick] (20,-5.5) rectangle (21.5,-7);
  \node at (20.75,-6.25) {\Large$\mc{A}_k$};
  
 \draw[thick] (20.25,-2.25) -- (20.25,-3.75);
  \shade[inner color=white, outer color=C1, draw=black, rotate around={45:(20.25,-2)}] (20,-1.75) rectangle (20.5,-2.25);
 \shade[inner color=white, outer color=C1, draw=black, rotate around={45:(20.25,-4)}] (20,-3.75) rectangle (20.5,-4.25);
 
 \draw[thick] (21.25,-1.25) -- (21.25,-4.75);
 \shade[bottom color=C2, top color=white, draw=black] (21,-0.75) rectangle (21.5,-1.25);
 \shade[bottom color=C2, top color=white, draw=black] (21,-4.75) rectangle (21.5,-5.25);
  
 \draw[thick] (22.25,-3.25) -- (22.25,-5.75);
  \shade[inner color=white, outer color=C1, draw=black, rotate around={45:(22.25,-3)}] (22,-2.75) rectangle (22.5,-3.25);
 \shade[inner color=white, outer color=C1, draw=black, rotate around={45:(22.25,-6)}] (22,-5.75) rectangle (22.5,-6.25);
 
 \draw[thick] (23.25,-2.25) -- (23.25,-4.75);
 \shade[bottom color=C2, top color=white, draw=black] (23,-1.75) rectangle (23.5,-2.25);
 \shade[bottom color=C2, top color=white, draw=black] (23,-4.75) rectangle (23.5,-5.25);
  
  \draw[white, fill=white, path fading= west] (22,0.5) -- (24,0.5) -- (24,-6.5) -- (22,-6.5);
  \end{scope}
  \end{tikzpicture}
  \caption[Markovianization by circuit design]{\textbf{Markovianization by circuit design:} Cartoon of a quantum process which can Markovianize under only two different types of 2-qubit interaction dynamics. For an $n$-qubit system, the unitaries $\mc{W}_\ell$ defined by Eq.~\eqref{eq: circuit W} generate an $\varepsilon$-approximate unitary $\design$-design whenever $\ell\geq\design-\log_2(\varepsilon)/n$, as found in Ref.~\cite{Winter_HamDesign}. This can be thought as stemming from repeated alternate applications of random 2-qubit gates diagonal in only two Pauli bases (rectangles and diamonds). A qubit probed with a set of operations $\{\mc{A}_i\}$ on a system undergoing $\varepsilon$-approximate unitary $\design$-design dynamics $\mc{W}_\ell$ on a large environment will Markovianize for small design error $\varepsilon$ and large complexity $\design$ as specified in the main text.}
  \label{fig: RDC}
  \end{figure}
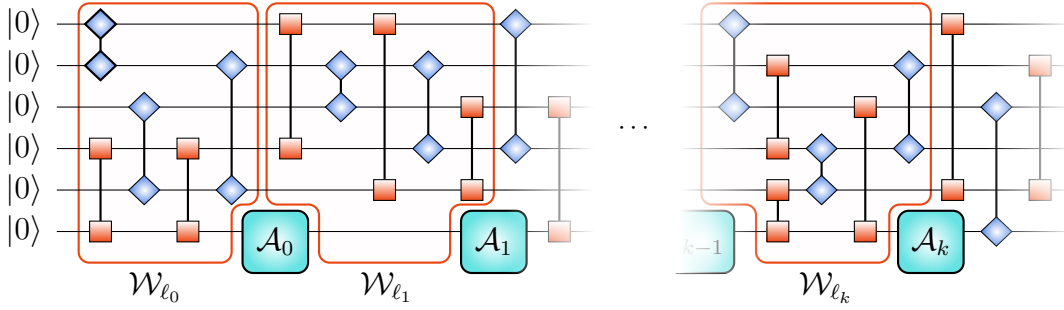
  
 A particular case which further simplifies things by writing all two-qubit gates in a diagonal form is denoted by $\mathrm{RDC}_\text{disc}^{(\design)}(\mc{I}_2)$, where the subscript ``$\text{disc}$'' refers to \emph{discrete} sets from which the parameters of the two-qubit gates will be sampled, and the superscript $\design$ is a natural number which determines this set. Specifically, all gates in $\mathrm{RDC}_\text{disc}^{(\design)}(\mc{I}_2)$ have the simplified form 
\begin{gather}
(\mathrm{diag}\{1,e^{i\phi_1}\}\otimes\mathrm{diag}\{1,e^{i\phi_2}\}) \ \mathrm{diag}\{1,1,1,e^{i\vartheta}\},
\end{gather}
where $\mathrm{diag}$ denotes Pauli-$Z$ diagonal, and with
\begin{equation}
    \phi_1,\phi_2\sim\{2\pi\,m/(\design+1):m\in\{0,\ldots,\design\}\},
\end{equation} chosen independently from such discrete set, and similarly \begin{equation}
    \vartheta\sim\{2\pi\,m/(\lfloor\design/2\rfloor+1):m\in\{0,\ldots,\lfloor\design/2\rfloor\}\}.
\end{equation}

Notice that despite the apparent complexity of this construction, it is still just a circuit comprised only of 2-qubit diagonal gates with only three random parameters each, and therein lies its simplicity. Let us then state the main Result of Ref.~\cite{Winter_HamDesign} that we applied on our result for Markovianization.

\begin{figure}[t]
\centering
\begin{minipage}{\textwidth}
\centering
\begin{minipage}{0.84\textwidth}
\begin{tikzpicture}
\node[anchor=south west, inner sep=0] (image) at (0,0) {\includegraphics[width=\textwidth]{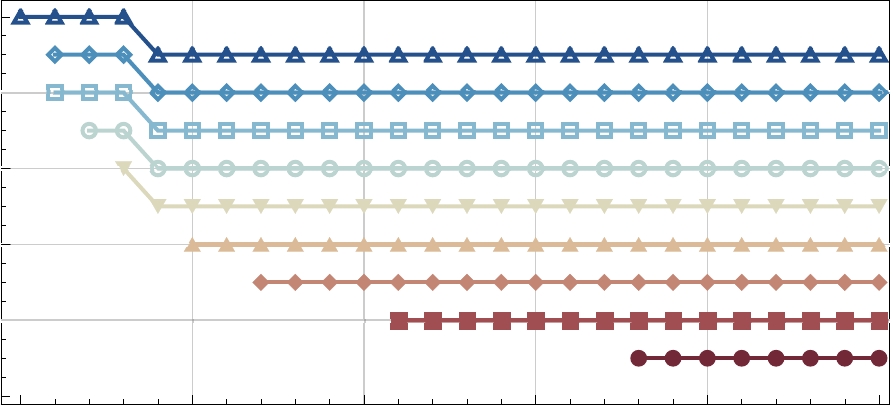}};
\begin{scope}[x={(image.south east)},y={(image.north west)}]
  \node[right] at (-0.05,1.075) {\large{Non-commuting gate depth $\mscr{D}$}};
  \node at (0.5,-0.175)
  {\large{$\log_2(d_\mathsf{E})$}};
  \node[below, right] at (-0.0075,-0.05)
  {\large{35}};
  {\large{40}};
  \node[below] at (0.42,0)
  {\large{45}};
  {\large{50}};
  \node[below] at (0.795,0)
  {\large{55}};
  {\large{60}};
  \node[rotate=90,above] at (0,0.95) {\large{12}};
  \node[rotate=90, above] at (0,0.58) {\large{8}};
  \node[rotate=90, above] at (0,0.22) {\large{4}};
  \node at (0.15,0.15) {\large$k=2$};
  \end{scope}
\end{tikzpicture}
\end{minipage}%
\begin{minipage}{0.15\textwidth}
\hspace{0.25in}
\vspace{0.325in}
\begin{tikzpicture}
\node[anchor=south west, inner sep=0] (image) at (0,0) {\includegraphics[width=0.25\textwidth]{Fig_bar.pdf}};
\begin{scope}[x={(image.south east)},y={(image.north west)}]
  \node[above] at (0.45,1) {\large{$\mathsf{t}$}};
  \node[rotate=90,above] at (0,0) {\large{2}};
  \node[rotate=90,above] at (0,0.975) {\large{10}};
  \end{scope}
\end{tikzpicture}
\end{minipage}
\end{minipage}
\caption[Markovianization by circuit design: non-commuting gate depth]{\textbf{Scaling of the non-commuting gate depth $\mscr{D}$}, as in Eq.~\eqref{eq:depth}, equivalent to that of the minimum amount of repetitions $\ell$ in the $n=n_\mathsf{E}+1$ qubit circuit $\mc{W}_\ell$ on \gls{syst-env}, plotted against the environment qubits, $n_\mathsf{E}=\log_2(d_\mathsf{E})$, to generate an $\varepsilon=10^{-12}$ approximate unitary $\design$-design for $2\leq\design\leq10$, such that for a single-qubit system undergoing a process with $k=2$ timesteps, $\mbb{P}_{\appdesign}\left[\mc{N}_\bdiamond\geq0.1\right]\leq\mathsf{B}\leq0.01$.}
\label{Plot: nqubits}
\end{figure}

\begin{theorem}[Main Result 2 of Ref.~\cite{Winter_HamDesign}]
Let $\mathsf{H}_n = \mathsf{H}^{\otimes{n}}$ be $n$ copies of the Hadamard gate,\footnote{ In the computational basis the Hadamard gate is given by $\mathsf{H}=(|0\rangle\!\langle0|+|1\rangle\!\langle0|+|0\rangle\!\langle1|-|1\rangle\!\langle1|)/\sqrt{2}$.} then for an $n$-qubit system, when $\design$ is of order $\sqrt{n}$, a circuit of the form
\begin{gather}
\mc{W}_\ell:= \left( \mathrm{RDC}_\text{disc}^{(\design)}(\mc{I}_2) \ \mathsf{H}_n \right)^{2\ell} \ \mathrm{RDC}_\text{disc}^{(\design)}(\mc{I}_2),
\label{eq: circuit W}
\end{gather}
yields an $\varepsilon$-approximate unitary $\design$-design if \begin{gather}
    \ell\geq\design - \f{\log_2(\varepsilon)}{n},
\end{gather}
up to leading order in $n$ and $\design$.
\end{theorem}

Furthermore, of great relevance in this result is the fact that all the 2-qubit gates in each repetition of $\mc{W}_\ell$, except those in $\mathsf{H}_n$, can be applied simultaneously because they commute~\cite{2designsXZ,Nakata2017decouplingrandom}. Therefore, if $\mc{W}_\ell$ yields an approximate unitary design as above, the order of the non-commuting gate depth $\mscr{D}$, defined in Ref.~\cite{Nakata2017decouplingrandom} as the circuit depth when each commuting part of the circuit is counted as a single part, will coincide with the bound on the order of the number of repetitions $\ell$. That is, the non-commuting gate depth asymptotes to
\begin{gather}\label{eq:depth}
    \mscr{D}\sim\design-\f{\log_2(\varepsilon)}{n}.
\end{gather}

We can now take the system from the toy model of Fig.~\ref{Fig: Design box and process}$\mathsf{(a)}$ as given by a spin locally interacting with a large, $n_\mathsf{E}$-qubit environment via a random time-independent Hamiltonian, with Eq.~\eqref{eq: largedev_design} statistically predicting under which conditions memory effects can be neglected. In Fig.~\ref{Plot: nqubits} we take such a system for a single qubit and demand a bound $\mathsf{B}\leq0.01$ on the probability $\mbb{P}_{\appdesign}\left[\mc{N}_\bdiamond\geq0.1\right]$ for a $k=2$ timestep process; with this, we plot the scaling of the non-commuting gate depth $\mscr{D}$ required to achieve an $\epsilon=10^{-12}$ approximate unitary $\design$-design using $\mc{W}_\ell$ circuits for different values of $2\leq\design\leq10$. While the number of 2-qubit gates is on the order of $10^4$, the number of repetitions $\ell$ is at most 12 for an approximate 10-design and stays mostly constant as the number of environment qubits increases.

This construction naturally accommodates the cartoon example in Fig.~\ref{Fig: Design box and process}: as long as the two interactions in the example together generate the necessary level of complexity, Markovianization will emerge. This shows, in principle, how simple dynamics described by approximate unitary designs can Markovianize under the right conditions, but, moreover, taking the physical interpretation of a qubit locally interacting through two-qubit diagonal unitaries with a large environment, it also hints at how macroscopic systems can display Markovianization of small subsystem dynamics in circuits requiring just a small gate depth. Furthermore, for macroscopic systems with coarse observables, the same Markovianization behaviour would remain resilient to a much larger number of interventions.

\section{Conclusions}
The results in this Chapter, based on Ref.~\cite{FigueroaRomero2020makovianization} show that quantum processes with physically motivated interactions can Markovianize, in the sense of becoming Markovian with very high probability in suitable limits. Beyond foundational considerations, our results have direct consequences for the study of open systems using standard tools, such as master equations and dynamical maps. As we saw in Chapter~\ref{sec:processes}, these can be seen as a family of one-step process tensors, where in the presence of initial \gls{syst-env} correlations a minimum of two steps must be considered~\cite{Ringbauer_2015, Vega_2020}. Specifically, for the case of $k\le 2$, our result can be used to estimate the time scale, using gate depth as a proxy, on which an approximate unitary design's open dynamics can be described (with high probability) with a truncated memory kernel~\cite{breuer2002theory, Cerrillo2014, Pollock2018tomographically}, or even a Markovian master equation.

Conversely, for larger $k$, our results would have implications for approximations made in computing higher order correlation functions, such as the quantum regression theorem~\cite{Guarnieri2014}. These higher order approximations are independent of those at the level of dynamical maps, which can, e.g., be divisible, even when the process is non-Markovian~\cite{Milz_PRL2019}. This is reflected in the loosening behaviour of the bound in Eq.~\eqref{eq: largedev_design} as the number of timesteps increases, which can be interpreted as a growing potential for temporal correlations to become relevant when more information about the process is accessible.

This breadth of applicability is in contrast with the results presented in Chapter~\ref{sec:typicality} from Ref.~\cite{FigueroaRomero2019almostmarkovian}, which has two main drawbacks: first, as stated above, Haar random interactions do not exist in nature and hence the relevance of the result is limited. Second, the rate of Markovianization is far too strong. Almost all processes, sampled according to the Haar measure, will look highly Markovian even for a large $k$. This, unlike our current result, misses almost all interesting physical dynamical processes.

While the behaviour of the large deviations bound in Theorem~\ref{Thm: Large Dev Markov} is polynomial, rather than exponential, thus not exhibiting concentration per-se, we have nevertheless exemplified how, with modestly large environments and relatively simple interactions, almost Markovian processes can come about with high probability. Physical macroscopic environments will be far larger than the scale shown in Figs.~\ref{Fig: plot Prob-dE}~and~\ref{Plot: nqubits}.

Finally, despite the fundamental relevance of our result, it is well known that typicality arguments can have limited reach. For instance, the exotic Hamiltonians, introduced in Ref.~\cite{Gemmer_2020}, which lead to \emph{strange} relaxation, may not Markovianize even though the \gls{syst-env} process is highly complex with a large \gls{env}. There is also still significant scope for further addressing physical aspects, such as the question of whether, and how, a time-independent Hamiltonian can give rise to an approximate unitary designs~\cite{Winter_HamDesign}, the real-world time scales of Markovianization, or the potential role of different approaches to pseudo-randomness such as that in Ref.~\cite{PhysRevX.7.041015}, where it is shown that driven quantum systems can converge rapidly to the uniform distribution. 

Furthermore, a renewed wave of interest in thermalization has come along with the so-called eigenstate thermalization hypothesis, which is a stronger and seemingly more fundamental condition on thermalization~\cite{Deutsch_ETH,Srednicki_1994, Srednicki_1999, DAlessio_2016, Murthy_2019, Brenes_2020}, and we would thus expect a deep connection in the sense of ETH between Markovianization, thermalization, and/or dynamical equilibration to be forthcoming. In any case, it is clear that many physical systems Markovianize at some scale, and it only remains to discover how.
    \part{Envoi}
    \chapter{Conclusions}
Within this thesis, we have investigated the concepts of equilibration and typicality in the context of quantum stochastic processes, appealing to the motivation of advancing an understanding of the emergence of Markovianity purely from quantum mechanical laws. We have seen that questions regarding this emergence can be phrased in analogous terms to those related to the emergence of statistical mechanics from quantum mechanics.

Specifically, with the original results of Part II of this thesis we have:
\hypersetup{linkcolor=black, linktoc=all}{
\begin{description}
 \item[\cref{sec:equilibration}] Extended the notion of equilibration on average on expectation values of observables to a multitime scenario, establishing sufficient conditions for the multitime statistics due to a sequence of observations on a quantum process to equilibrate.
 \item[\cref{sec:typicality}] Formally shown---without resorting to the Born-Markov assumption or making any approximations---that quantum processes undergone within a finite large dimensional environment concentrate around Markovian ones.
 \item[\cref{sec: Markovianization by design}] Obtained a large deviation bound for quantum processes described by approximate unitary designs, showing that quantum processes can Markovianize in the sense that they can converge to a class of Markovian processes whenever they are undergone within a large environment and under complex enough dynamics.
\end{description}
}
These results were made possible due to the developments in the topics of equilibration and typicality, which relate to the emergence of statistical mechanics solely from quantum mechanics, and the process tensor framework for quantum stochastic processes, both presented in Part I.

In {\bfseries\sffamily\cref{sec:statmech}} we focused on the results on equilibration on average by Ref.~\cite{Short_finite} and typicality by Ref.~\cite{Popescu2006}, which we consider were pivotal regarding the emergence of statistical mechanics purely from quantum mechanics. Topics related to equilibration and typicality are currently highly active ones in fronts such as equilibration time-scales~\cite{Gogolin_2016, Wilming2018} and the eigenstate thermalization hypothesis~\cite{Srednicki_1994,Srednicki_1999,Rigol_2012,DAlessio_2016,Murthy_2019}.

In {\bfseries\sffamily\cref{sec:processes}} we presented the mathematical framework of quantum maps and some of their representations, which allowed us to approach the problem of temporal correlations in open quantum systems. We saw that this is resolved by the process tensor framework for quantum stochastic processes, introduced in Ref.~\cite{Pollock_2018} (and equivalent frameworks with different perspectives in Ref.~\cite{Hardy_2012, Hardy_2016, Werner_2005,Caruso_2014, Chiribella_2008, Chiribella_2009, Chiribella_2013, Oreshkov2012, Costa_2016, Oreshkov_2016_2,Portmann_2017}), which establishes an unambiguous condition for Markovianity capturing all multitime phenomena and memory effects~\cite{Pollock_2018_Markov}. Not unexpectedly, the topic of quantum stochastic processes and the process tensor framework are highly active research areas, with investigations in machine learning~\cite{Luchnikov2020MachineLN}, resource theories~\cite{berk2019resource}, open dynamics simulations~\cite{Mathias_2019} or the causal structure of quantum processes~\cite{Milz_2018}, to name but a few examples. Our work is rather a contribution to the area of quantum processes in the spirit of the approaches of equilibration and thermalization of Chapter 1.

We now discuss the outlook for future research related to this thesis and some possible ways to approach it.

\section{Outlook}
The results within this thesis make a significant step towards the understanding of the emergence of Markovianity in quantum processes. There are, nevertheless, several ways in which this understanding can be further advanced, as well as applied.

In the case of multitime equilibration, two main issues stand out. The first relates to the role that memory has, both in the dynamical process itself and within the external operations. While we obtained sufficient conditions under which quantum processes will equilibrate, the interplay with memory effects, both through the environment and the ancillary space in the interventions is as yet not entirely clear, e.g. under which circumstances finite-temporal resolution equilibration can occur without the dynamics being Markovian, or if the temporal correlations among interventions through the ancillary space can display a departure from equilibration within a finite-time.

The relation between multitime equilibration and Markovianity is somewhat akin to the link between equilibration and thermalization, in the sense that stronger conditions and a deeper characterization of the equilibrium process might be needed in order to determine whether it is generically indeed an almost Markovian process. This could be achieved e.g. by investigating the non-Markovianity of the equilibrium process $\Omega$ and the limits that make it exactly Markovian. In the case of the memory on the external operations, while it is expected that having access to a memory space that can keep track of the multitime statistics in a process would necessarily serve to witness a departure from equilibration, attempts to show this proved elusive during this PhD. Possible ways to move forward could be e.g. employing ideas topics such as state transfer~\cite{Giovannetti_2006} or algorithmic cooling~\cite{PhysRevLett.123.170605}, where it would find relevant technological applications as well.

A further issue in multitime equilibration will be setting time-scales in which equilibration occurs. This, however, is a longstanding problem even in the standard scenario, so tackling the problem in general quantum processes could be a very ambitious goal.

In the case of Markovianization, there are some open questions that might be possible to address in a reasonable time-frame. In particular, further constraining Markovianization to general physical and experimentally realizable processes, is a goal that could be achieved in many ways with recent theoretical developments; one such way could be developing the notion of time-independent Hamiltonian unitary designs~\cite{Winter_HamDesign}, or another could be looking at driven quantum systems that can quickly increase their complexity~\cite{PhysRevX.7.041015}.

Finally, two topics in the border with the ones discussed in this thesis are the eigenstate thermalization hypothesis (ETH) and Randomized Benchmarking.

The ETH relates to the energy eigenstates of large and chaotic systems behaving as random variables and displaying the corresponding equilibrium statistical properties~\cite{Srednicki_1994, Deutsch_ETH}. The framework in which the ETH is set naturally overlaps with many of the technical concepts employed in this thesis and many questions regarding higher-order correlations could be approached with the process tensor and techniques from random matrix theory employed to study the typicality of Markovian processes.

On the other hand, Randomized Benchmarking refers to the estimation of average error rates in random circuits by implementing random quantum gates (elementary unitary operations) that would amount to identities if there was no error present. The mathematical formalism employed in the derivation for the results on the typicality of Markovian processes and Markovianization is mostly the same as the one in randomized benchmarking: the idea is to estimate the so-called fidelity, i.e. a measure of noisiness, of a class of sets of noisy random quantum gates entering a quantum process with respect to the corresponding noise-free ones by looking at the average over the noisy gates. Whenever the noise is Markovian, so that there is both no time-dependence and no dependence on the particular choice of gates, randomized benchmarking has been widely studied~\cite{Knill_2008, Wallman_2014, Proctor_2017, Dirkse_2019, Harper_2019, Helsen_2019}. However, randomized benchmarking for non-Markovian noise, i.e. due to memory effects, is still an open problem and the process tensor is the natural framework to approach it. This is a problem whose solution could represent a significant contribution with high impact in the field of quantum computing.

We remain confident that the work presented in this thesis will serve as a stepping stone to tackle these problems, as well as an inspiration to approach others, old or new, naturally connected to quantum stochastic processes and quantum information science.
    \appendix           
    \appendixpage       
    \hypersetup{colorlinks=true, linkcolor=C2, citecolor=C2, urlcolor=C2}
    \chapter{Notation}
\label{appendix: notation}
\setlength\LTleft{0pt}
\setlength\LTright{0pt}
An effort has been made to maintain consistency in the notation throughout this thesis. Here we present the standard notation that is employed throughout the text.

Many symbols are only distinguished by their style, e.g. $\mbb{P}$, $\mc{P}$, $\mscr{P}$, $\wp$ and $P$ all stand for different concepts, so we hope this stylistic distinction, together with the relevant context, is enough to easily tell them apart.

Keep in mind that we only work with finite dimensional quantum systems. Most abbreviated terminology, such as referring to density matrices as quantum states, Hilbert space as space, or stochastic process as process, is sufficiently clear by context and we do not reproduce it here. Similarly, some simplifications such as positivity of a quantum states meaning positive semi-definiteness, or removing the subindex for the number of time-steps in a process tensor is usually done when sufficiently clear by context.

We have also committed some abuse of notation when relevant for simplicity's sake, e.g. omitting identity operators when an operator acts only on a subpart of the whole system in question. When this is done it has been duly pointed out.

We have also been somewhat lax when it comes to the usual mathematical structure of introducing Definitions, Remarks or Notation in a bullet point style and opted for a more narrative structure; we have nevertheless used these when we have considered relevant, in particular with the original theorems from this PhD which are followed by their proof.

Finally, we made an effort to employ a consistent color scheme on the different figures, mainly using red to denote dynamical objects, blue/teal on experimental operations and/or maps and purple on process tensors.

\begin{longtable}{l l}  
\toprule
      \textsc{General}                                 \\
      \toprule
      \\
      $\mbb{R}, \mbb{R}^+, \mbb{R}_0^+$ & Set of real, positive real and non-negative real numbers\\
      $\mbb{C}$ & Set of complex numbers\\
      $\mbb{U}(d)$ & Unitary group of dimension $d$ \\
      $\mscr{H}$ & Hilbert space\\
      $(\cdot)^*, (\cdot)^\mathrm{T}, (\cdot)^\dg$ & Conjugate, Transpose and Conjugate transpose\\
      $|\psi\rangle, \langle\psi|$ & Vector on a Hilbert space (ket) and Adjoint vector on a Hilbert space (bra)\\
      $\langle\,\cdot\,\rangle$ & Expectation value\\
      $[\cdot,\cdot]$ & Commutator, $[a,b]=ab-ba$\\
      $\rho$ & Density matrix on a Hilbert space\\
      $\tr$ & Trace\\
      $\mathsf{A,B,\ldots}$ & Systems are labelled with sans-serif capital letters\\
      $\mathsf{in}$, $\mathsf{out}$ & Inputs and outputs\\
      $\mathsf{in'}$ & An apostrophe distinguishes an input that has been acted on with a map\\
      $\mscr{H}_\mathsf{A}$ & Hilbert space associated to system $\mathsf{A}$\\
      $\rho_\mathsf{A}$ & Density matrix on Hilbert space $\mathsf{A}$\\
      $\tr_\mathsf{A}$ & Partial trace over system $\mathsf{A}$\\
      $\tilde{\Psi}$ & Unnormalized maximally entangled state $\tilde{\Psi}=\sum|ii\rangle\!\langle{jj}|$\\
      $\mathsf{\Psi}$ & Normalized to unity maximally entangled state, $\mathsf{\Psi}:=\tilde{\Psi}/d$\\
      $d_\mathsf{A}$ & Standard for the dimension of the Hilbert space of a system, $\dim(\mscr{H}_\mathsf{A})$\\
      $\mscr{T}_k$ & Set of time-steps $\{t_0,t_1,\ldots,t_{k-1}\}$ \\
      $k:0$ & Set of ordered time-steps $t_0<t_1<\ldots<{t}_{k-1}$
\end{longtable}
    
\begin{longtable}{l l}  
\toprule
      \textsc{Operations}                                 \\
      \toprule
      \\
      $\mscr{B}(\mscr{H})$ & Space of bounded linear operators acting on the Hilbert space $\mscr{H}$\\
      $\$(\mscr{H})$ & Space of density matrices acting on the Hilbert space $\mscr{H}$\\
      $A,B,\ldots$ & Bounded linear operators are denoted with standard capital letters\\
      $\mathrm{M}_i$ & \gls{POVM} element\\
      $H$ & Hamiltonian operator\\
      $\Pi$ & Projection operator\\
      $U$ & Unitary operator\\
      $\mbb1$ & Identity operator\\
      $\swap$ & Swap operator, $\swap:=\sum|ij\rangle\!\langle{ji}|$\\
      $\FS_{\alpha\beta}$ & Defined as $\FS_{\alpha\beta}:=\mbb1_\mathsf{E}\otimes|\alpha\rangle\!\langle\beta|$ with $|\alpha\rangle,|\beta\rangle$ being \gls{syst} basis vectors\\
      $\Phi$ & Generic map between bounded linear operator spaces\\
      $\mc{U}$ & Unitary map\\
      $\mc{I}$ & Identity map\\
      $\mscr{Z}_t,\mscr{Z}_{j:i}$ & Dynamical map on time parameter $t$, dynamical map from time-step $i$ to $j$ \\
      $\mc{S}_i$ & Swap map between system \gls{syst} and an ancilla at time-step $t_i$\\
      $\Xi^{(n)}$ & $n$-fold twirl map\\
      $\Upsilon_\Phi$ & Choi state of a map $\Phi$\\
      $\mc{A}_0,\mc{A}_1,\ldots$ & \gls{CPTNI} maps at time-steps $t_0,t_1,\ldots$\\
      $\mc{M}$ & Superchannel\\
      $\mc{T}_{k:0}$ & $k$-step process tensor; the subindex $k:0$ is dropped when clear by context\\
      $\Upsilon_{k:0}$ & Choi state of a process tensor\\
      $\mc{J}$ & Instrument - a collection of \gls{CP} maps that sums to a \gls{CPTP} map\\
      $\Lambda_{k:0}$ & Choi state of a sequence of \gls{CP} operations (tester when over instruments)\\
\end{longtable}

\begin{longtable}{l l}  
\toprule
      \textsc{Equilibration}\hfill                \\
      \toprule
\\
     $\overline{(\cdot)}$ & Infinite time-average\\
     $\overline{(\cdot)}^T$ & Uniform time-average over an interval of width $T$\\
     $\overline{(\cdot)}^{\mscr{P}_T}$ & Time-average over a distribution with \gls{pdf} $\mscr{P}_T$ with fuzziness $T$\\
     $\omega$ & Equilibrium state\\
     $\varrho$, $\varpi$ & Initial \gls{syst-ancilla} state and \gls{syst-ancilla} equilibrium state\\
     $d_\text{eff}$ & Effective dimension (a.k.a. participation ration)\\
     $E_n$ & Energy levels (Hamiltonian eigenvalues)\\
     $P_n$ & Projector onto the eigenspace of energy $E_n$\\
     $\mc{E}_{nm}$ & Energy difference $\mc{E}_{nm}\propto E_n-E_m$\\
     $\mathfrak{D}$ & Number of distinct energies, $\mathfrak{D}=|\mathrm{spec}(H)|$\\
     $N(\varepsilon)$ & Max number of energy gaps in any interval of size $\epsilon>0$\\
     $\mc{D}$ & Dephasing map with respect to a Hamiltonian\\
     $\mc{G}_T$ & Partial dephasing map within a time interval of width $T$\\
\end{longtable}

\begin{longtable}{l l}  
\toprule
      \textsc{Norms / distances}\hfill                \\
      \toprule
\\
     $\|\cdot\|_p$ & Schatten $p$-norm\\
     $\|\cdot\|_1$ & (Schatten) 1-norm or trace norm\\
     $\|\cdot\|$ & Operator norm (max singular value)\\
     $D(\cdot,\cdot)$ & Trace distance\\
     $S(\cdot)$ & von-Neumann entropy\\
     $S(\cdot\,\|\,\cdot)$ & Relative entropy
\end{longtable}

\begin{longtable}{l l}  
\toprule
      \textsc{Probability / Statistics}\hfill                \\
      \toprule
\\
      $\mu$ & Probability measure\\
      $\mbb{P}$ & Probability distribution\\
      $\mathfrak{P}_{j:i}$ & Propagator stochastic matrix from time $t_i$ to time $t_j$\\
      $\mbb{E}$ & Expectation\\
      $\haar$ & Haar measure\\
      $\haarrand$ & Haar measure in random interaction ($U_i\neq{U}_j$)\\
      $\mathsf{t}$ & Unitary $\mathsf{t}$-design\\
      $\mathsf{t}_\epsilon$ & $\epsilon$-approximate unitary $\mathsf{t}$-design\\
      $\mscr{L}$ & Lipschitz constant\\
      $\mscr{P}$ & Probability density function\\
      $\mathfrak{G}_n$ & Symmetric group on $n$-elements \\
      $\wp$ & Permutation operator
\end{longtable}
    \chapter{Haar distributed process tensors}
\label{appendix - Haar measure}
Here we present the derivation of the expressions presented in Chaper~\ref{sec:typicality} for the expectation of quantum processes $\Upsilon$ and the expectation of the purity, or noisiness, of quantum processes $\tr\left(\Upsilon^2\right)$, according to the Haar measure.

\section{Average processes - Constant interaction}
\label{appendix - Average process constant}
We are considering here $\Upsilon$ to be unnormalized (or more precisely, normalized to $\tr(\Upsilon)=d_\mathsf{S}^{k+1}$), as presented first in the main text in Eq.~\eqref{average state Ui=Uj}. Also, here $\Upsilon$ stands for a $k$-step process $\Upsilon_{k:0}$ as done in the main text as well.

By definition,
\begin{align}
\mbb{E}_\haar(\Upsilon)&=\tr_\mathsf{E}\sum_{\alpha,\ldots,\gamma=1}^{d_\mathsf{S}}\,\int_{\mbb{U}(d_\mathsf{SE})}U\FS_{\alpha_k\beta_k}
\cdots{U}\FS_{\alpha_1\beta_1}U\,\rho\,
{U}^\dg\FS_{\gamma_1\delta_1}^\dg{U}^\dg\cdots\FS_{\gamma_k\delta_k}^\dg{U}^\dg\,d\muhaar(U)\nonumber\\
&\qquad\qquad\otimes|\beta_1\alpha_1\ldots\beta_k\alpha_k\rangle\!\langle\delta_1\gamma_1\ldots\delta_k\gamma_k|,
\end{align}
where $\FS_{\alpha\beta}:=\mbb1_\mathsf{E}\otimes|\alpha\rangle\!\langle\beta|$. We now decompose the unitaries as $U=\sum{U}_{ab}|a\rangle\!\langle{b}|$ and $U^\dg=\sum{U}^*_{a^\prime{b}^\prime}|b^\prime\rangle\!\langle{a}^\prime|$ --notice that the $a$ and $b$ labels refer to the whole \gls{syst-env} space--, introducing a resolution of the identity on \gls{env} within the $\FS$ operators as $\FS_{ab}=\sum_e|ea\rangle\!\langle{eb}|$, and then evaluating the $k+1$ moments of $\mbb{U}(d_\mathsf{SE})$ by means of the moments of the unitary group in Eq.~\eqref{k moments of U},
\begin{align}
    &\int_{\mbb{U}(d_\mathsf{SE})}U\FS_{\alpha_k\beta_k}\cdots{U}\FS_{\alpha_1\beta_1}U\rho{U}^\dg\FS_{\gamma_1\delta_1}^\dg{U}^\dg\cdots\FS_{\gamma_k\delta_k}^\dg{U}^\dg\,d\muhaar(U)\nonumber\\
    &=\sum\int_{\mbb{U}(d_\mathsf{SE})}U_{i_0j_0}\cdots{U}_{i_kj_k}U^*_{i^\prime_0j^\prime_0}\cdots{U}^*_{i^\prime_kj^\prime_k}\,d\muhaar(U)\,|i_k\rangle\!\langle{j}_0|\rho|j^\prime_0\rangle\!\langle{i}^\prime_k|\nonumber\\
    &\qquad\qquad\qquad\prod_{\ell=1}^{k}\delta_{j_\ell(e\alpha)_\ell}\delta_{(e\beta)_\ell{i}_{\ell-1}}\delta_{j^\prime_\ell(e^\prime\gamma)_\ell}\delta_{(e^\prime\delta)_\ell{i^\prime}_{\ell-1}}\nonumber\\
    &=\sum\sum_{\sigma,\tau\in\mathfrak{G}_{k+1}}\langle{j}_0|\rho|j^\prime_0\rangle\delta_{i_0i^\prime_{\sigma(0)}}\delta_{j_0j^\prime_{\tau(0)}}\prod_{\ell=1}^k\delta_{i_\ell{i}^\prime_{\sigma(\ell)}}\delta_{j_\ell{j}^\prime_{\tau(\ell)}}\delta_{j_\ell(e\alpha)_\ell}\delta_{(e\beta)_\ell{i}_{\ell-1}}\delta_{j^\prime_\ell(e^\prime\gamma)_\ell}\delta_{(e^\prime\delta)_\ell{i^\prime}_{\ell-1}}\nonumber\\
    &\qquad\qquad\qquad\qquad\Wg(\tau\sigma^{-1})\,|i_k\rangle\!\langle{i}^\prime_k|\nonumber\\
    &=\sum\sum_{\sigma,\tau\in\mathfrak{G}_{k+1}}\langle{j}^\prime_{\tau(0)}|\rho|j^\prime_0\rangle\delta_{i_ki^\prime_{\sigma(k)}}\prod_{\ell=1}^k\delta_{(e\beta)_\ell{i}^\prime_{\sigma(\ell-1)}}\delta_{(e^\prime\delta)_\ell{i^\prime}_{\ell-1}}\delta_{(e\alpha)_\ell{j}^\prime_{\tau(\ell)}}\delta_{(e^\prime\gamma)_\ell{j}^\prime_\ell}\,\Wg(\tau\sigma^{-1})\,|i_k\rangle\!\langle{i}^\prime_k|,
\end{align}
where here $\mathfrak{G}_{k+1}$ denotes the symmetric group on $\{0,1,\ldots,k\}$, $\Wg$ is the Weingarten function with implicit dimensional argument $d_\mathsf{SE}$, and which taking $i\to\epsilon\varsigma$ and $j\to\epsilon^\prime\varsigma^\prime$ to recover each \gls{env} and \gls{syst} part explicitly, turns into
\begin{align}
    &\int_{\mbb{U}(d_\mathsf{SE})}U\FS_{\alpha_k\beta_k}\cdots{U}\FS_{\alpha_1\beta_1}U\,\rho\,{U}^\dg\FS_{\gamma_1\delta_1}^\dg{U}^\dg\cdots\FS_{\gamma_k\delta_k}^\dg{U}^\dg\,d\mu(U)\nonumber\\
    &=\sum\sum_{\sigma,\tau\in\mathfrak{G}_{k+1}}\langle\epsilon^\prime_{\tau(0)}\varsigma^\prime_{\tau(0)}|\rho|\epsilon^\prime_0\varsigma^\prime_0\rangle\prod_{\ell=1}^k\delta_{\epsilon_{\sigma(\ell-1)}\epsilon^\prime_{\tau(\ell)}}\delta_{\epsilon_{\ell-1}\epsilon^\prime_\ell}\delta_{\beta_\ell\varsigma_{\sigma(\ell-1)}}\delta_{\delta_\ell\varsigma_{\ell-1}}\delta_{\alpha_\ell\varsigma^\prime_{\tau(\ell)}}\delta_{\gamma_\ell\varsigma^\prime_\ell}\,\Wg(\tau\sigma^{-1})\nonumber\\
    &\qquad\qquad\qquad\qquad|\epsilon_{\sigma(k)}\varsigma_{\sigma(k)}\rangle\!\langle\epsilon_k\varsigma_k|,
\end{align}
and thus
\begin{align}
    &\mbb{E}_\haar(\Upsilon)=\sum\sum_{\sigma,\tau\in\mathfrak{G}_{k+1}}\hspace{-0.1in}\langle\epsilon^\prime_{\tau(0)}\varsigma^\prime_{\tau(0)}|\rho|\epsilon^\prime_0\varsigma^\prime_0\rangle\,\Wg(\tau\sigma^{-1})\,\delta_{\epsilon_{\sigma(k)}\epsilon_k}\prod_{\ell=1}^k\delta_{\epsilon_{\sigma(\ell-1)}\epsilon^\prime_{\tau(\ell)}}\delta_{\epsilon_{\ell-1}\epsilon^\prime_\ell}\nonumber\\
&\qquad\qquad\qquad\qquad|\varsigma_{\sigma(k)}\varsigma_{\sigma(0)}\varsigma^\prime_{\tau(1)}\cdots\varsigma_{\sigma(k-1)}\varsigma^\prime_{\tau(k)}\rangle\!\langle\varsigma_k\varsigma_0\varsigma^\prime_1\cdots\varsigma_{k-1}\varsigma^\prime_k|,
\end{align}
as stated by Eq.\eqref{average state Ui=Uj}.

\subsection{Superchannel case}\label{appendix: average state Ui=Uj superchannel}
For the superchannel case, $k=1$, we have $\mathfrak{G}_2=\{(0,1),(0)(1)\}$ where the elements are permutations stated in \emph{cycle} notation, representing the assignments $(0,1):0\to1\to0$ and $(0)(1)=\boldsymbol{1}^2=0\to0;1\to1$, then (we write $\alpha,\beta,\gamma,\delta$ for $\varsigma_0,\varsigma_1,\varsigma_0^\prime,\varsigma_1^\prime$, respectively to ease the notation)
\begin{align}
    \mbb{E}_\haar(\Upsilon_{1:0})=\f{\sum}{d_\mathsf{S}}\bigg\{&\left[\langle\delta|\rho_\mathsf{S}|\gamma\rangle|\alpha\beta\gamma\rangle\!\langle\beta\alpha\delta|+d_\mathsf{E}^2|\beta\alpha\delta\rangle\!\langle\beta\alpha\delta|\right]\Wg[\boldsymbol{1}^2]\nonumber\\
    &\qquad\qquad+d_\mathsf{E}\left[\langle\delta|\rho_\mathsf{S}|\gamma\rangle|\beta\alpha\gamma\rangle\!\langle\beta\alpha\delta|+|\alpha\beta\delta\rangle\!\langle\beta\alpha\delta|\right]\Wg[(0,1)]\bigg\},
\end{align}
so that with $\Wg[\boldsymbol{1}^2,d]=\f{1}{d^2-1}$ and $\Wg[(0,1),d]=\f{-1}{d(d^2-1)}$~\cite{GuMoments} we get
\begin{align}\mbb{E}_\haar(\Upsilon_{1:0})=\f{1}{d_\mathsf{E}^2d_\mathsf{S}^2-1}\bigg[\f{d_\mathsf{E}^2}{d_\mathsf{S}}\mbb1_\mathsf{SAB}+\f{\swap}{d_\mathsf{S}}\otimes\rho_\mathsf{S}^\mathrm{T}-\f{\swap}{d_\mathsf{S}}\otimes\f{\mbb1_\mathsf{B}}{d_\mathsf{S}}-\f{\mbb1_\mathsf{SA}}{d_\mathsf{S}^2}\otimes\rho_\mathsf{S}^\mathrm{T}\bigg],\label{avg state Ui=Uj superchannel}\end{align} where $\swap=\sum_{i,j}|ij\rangle\!\langle{j}i|$ is the usual swap operator, and hence for the corresponding purity one may verify that
\begin{gather}
    \tr[\mbb{E}_\haar(\Upsilon_{1:0})^2]=\f{2}{(d_\mathsf{E}^2d_\mathsf{S}^2-1)^2}\left[\f{1}{d_\mathsf{S}^3}+\tr(\rho_S^2)\f{d_\mathsf{S}^2-d_\mathsf{S}-1}{2d_\mathsf{S}^2}-\f{d_\mathsf{E}^2}{d_\mathsf{S}}+\f{d_\mathsf{E}^4d_\mathsf{S}}{2}\right].
    \label{purity TI superchannel}
\end{gather}

\section{Average purity - Random interaction}
\label{appendix - average purity random}
Here $\Upsilon$ is normalized to unit trace, in accordance with the main text. Let $\Theta=\rho\otimes\Psi^{\otimes{k}}$, where we assume $\rho$ is pure, $\tr\rho=1$, and $\mathsf{U}_{k:0}:=U_k\mc{S}_k\cdots{U}_1\mc{S}_1U_0$, we first (following the approach in section 2 of~\cite{ZhangXiang}) write the trace as
\begin{align}
    \tr\left(\Upsilon^2\right)&=\tr\left[\left(\tr_\mathsf{E}(\mathsf{U}_{k:0}\Theta\,\mc{U}^\dg_{k:0})\right)^2\right]\nonumber\\
    &=\tr\left[\tr_\mathsf{E}\left((\mbb{1}_E\otimes\Upsilon)\,\mathsf{U}_{k:0}\Theta\,\mc{U}^\dg_{k:0}\right)\right]\nonumber\\
    &=\tr\left[\Gamma_\mathsf{E}(\mathsf{U}_{k:0}\Theta\,\mc{U}^\dg_{k:0})\,\mathsf{U}_{k:0}\Theta\,\mc{U}^\dg_{k:0}\right]\nonumber\\
    &=\tr\left[\Theta\,\mathsf{U}_{k:0}^\dg\Gamma_\mathsf{E}(\mathsf{U}_{k:0}\Theta\,\mc{U}^\dg_{k:0})\,\mathsf{U}_{k:0}\right],
\end{align}
where $\Gamma_\mathsf{E}(\cdot)\equiv\mbb{1}_E\otimes\tr_\mathsf{E}(\cdot)$ is a CP map and can be expressed as an operator-sum, i.e. for any $\mc{X}$ acting on $ESA_1B_1\cdots{A}_kB_k$ we have $\Gamma_\mathsf{E}(\mc{X})=\sum_{i,j=1}^{d_\mathsf{E}}(K_{ij}\otimes\mbb1)\mc{X}(K_{ij}\otimes\mbb1)^\dg$ where the identities are on the ancillary system and $K_{ij}=|i\rangle\!\langle{j}|\otimes\mbb{1}_{S}$ are the ($ES$ system) Kraus operators with $\{|i\rangle\}_{i=1}^{d_\mathsf{E}}$ a given basis for $E$. And so we need to compute the integrals
\begin{align}
    \Omega_{k:0}:=\int_{\mbb{U}(d_\mathsf{SE})}\mathsf{U}_{k:0}^\dg\Gamma_\mathsf{E}(\mathsf{U}_{k:0}\Theta\,\mc{U}^\dg_{k:0})\,\mathsf{U}_{k:0}\,d\muhaar(U_0)\cdots\,d\muhaar(U_k),
    \label{full purity integrals}
\end{align}
and in fact these are really in the \gls{syst-env} part only, which has the form
\begin{align}
    \omega_{k:0}^{(\rho)}:=\int_{\mbb{U}(d_\mathsf{SE})}&U_0^\dg\FS_{x_1\alpha_1}U_1^\dg\cdots{\FS}_{x_k\alpha_k}U_k^\dg{K}_{\tilde\imath\tilde\jmath}U_k\FS_{\beta_kx_k}\nonumber\\
    &\qquad\cdots{U_1}\FS_{\beta_1x_1}U_0\rho{U}_0^\dg\FS_{y_1\lambda_1}U_1^\dg\cdots\FS_{y_k\lambda_k}U_k^\dg{K}^\dg_{\tilde\imath\tilde\jmath}{U}_k\FS_{\theta_ky_k}\nonumber\\
    &\qquad\qquad\qquad\qquad\cdots{U}_1\FS_{\theta_1y_1}U_0\,d\muhaar(U_0)\,d\muhaar(U_1)\cdots d\muhaar(U)_k,
\label{ES part purity integ}
\end{align}
with summation over repeated indices $x,y$ and $i,j$ on the Kraus operators implied, and the ancillary part takes the form
\begin{align}
    \Omega_{k:0}^{(\mathsf{AB})}&:=(|\alpha_1\rangle\!\langle\beta_1|\otimes\mbb1_{\mathsf{B}_1}\otimes\cdots\otimes|\alpha_k\rangle\!\langle\beta_k|\otimes\mbb1_{\mathsf{B}_k})\mathsf{\Psi}^{\otimes\,k}(|\lambda_1\rangle\!\langle\theta_1|\otimes\mbb1_{\mathsf{B}_1}\otimes\cdots\otimes|\lambda_k\rangle\!\langle\theta_k|\otimes\mbb1_{\mathsf{B}_k}).
\end{align}

We may then evaluate each of the integrals using
\begin{gather}\int_{\mbb{U}(d)}U^\dg{A}UX{U}^\dg{B}U\,d\muhaar(U)=\f{d\tr(AB)-\tr(A)\tr(B)}{d(d^2-1)}\tr(X)\mbb1+\f{d\tr(A)\tr(B)-\tr(AB)}{d(d^2-1)}X,
\label{intmp}
\end{gather}
which can be seen to follow from the second moment of the unitary group (an explcit derivation can be seen e.g. in Ref.~\cite{ZhangMInt}), here with $d=d_\mathsf{SE}$, which from now on we employ, then we have
\begin{align}
    &\omega_{k:0}^{(\rho)}=\f{1}{d(d^2-1)}\nonumber\\
    &\bigg\{\delta_{x_1y_1}\int_{\mbb{U}(d)}\tr\bigg[\FS_{\theta_1\alpha_1}U_1^\dg\cdots{\FS}_{x_k\alpha_k}U_k^\dg{K}_{ij}U_k\FS_{\beta_kx_k}\nonumber\\
    &\qquad\qquad\qquad\qquad\cdots{U}_1\FS_{\beta_1\lambda_1}U_1^\dg\cdots\FS_{y_k\lambda_k}U_k^\dg{K}^\dg_{ij}{U}_k\FS_{\theta_ky_k}\cdots{U}_1\bigg](d\mbb1-\rho)\,d\muhaar(U_1)\cdots{d}\muhaar(U_k)\nonumber\\
    &+\delta_{x_1x_1}\delta_{y_1y_1}\int\tr\bigg[\FS_{\beta_1\alpha_1}U_1^\dg\cdots{\FS}_{x_k\alpha_k}U_k^\dg{K}_{ij}U_k\FS_{\beta_kx_k}\cdots{U}_1\bigg]\nonumber\\
    &\qquad\qquad\qquad\qquad\tr\bigg[\FS_{\theta_1\lambda_1}U_1^\dg\cdots\FS_{y_k\lambda_k}U_k^\dg{K}^\dg_{ij}{U}_k\FS_{\theta_ky_k}\cdots{U}_1\bigg](d\rho-\mbb1)\,\,d\muhaar(U_1)\cdots{d}\muhaar(U_k)\bigg\}\nonumber\\
    &=\f{1}{d(d^2-1)}\bigg[d_\mathsf{S}\langle{e}_1\alpha_1|\omega_{k:1}^{(\FS_{\beta_1\lambda_1})}|e_1\theta_1\rangle(d\mbb1-\rho)+d_\mathsf{S}^2\langle{e}_1\alpha_1|\omega_{k:1}^{(|e_1\beta_1\rangle\!\langle{e}^\prime_1\lambda_1|)}|e^\prime_1\theta_1\rangle(d\rho-\mbb1)\bigg],
\end{align}
again with sum implied over $e_i$'s and Greek indices. Now, let us notice that
\begin{align}
    \omega_{k:i-1}^{(X)}&=\f{1}{d(d^2-1)}\bigg[d_\mathsf{S}\langle{e}_i\alpha_i|\omega_{k:i}^{(\FS_{\beta_i\lambda_i})}|e_i\theta_i\rangle(d\tr(X)\mbb1-X)\nonumber\\
    &\qquad\qquad\qquad\qquad+d_\mathsf{S}^2\langle{e}_i\alpha_i|\omega_{k:i}^{(|e_i\beta_i\rangle\!\langle{e}^\prime_i\lambda_i|)}|e^\prime_i\theta_i\rangle(dX-\tr(X)\mbb1)\bigg],
\end{align}
for $k>i-1\geq0$. Then we get
\begin{align}
    &\omega_{k:0}^{(\rho)}=\f{1}{d^2(d^2-1)^2}\nonumber\\
    &\bigg\{d_\mathsf{S}\langle{e}_2\alpha_2|\omega_{k:2}^{(\FS_{\beta_2\lambda_2})}|e_2\theta_2\rangle\bigg[d_\mathsf{S}(dd_\mathsf{E}^2\delta_{\alpha_1\theta_1}\delta_{\beta_1\lambda_1}-d_\mathsf{E}\delta_{\alpha_1\beta_1}\delta_{\lambda_1\theta_1})(d\mbb1-\rho)\nonumber\\
    &\qquad\qquad\qquad\qquad+d_\mathsf{S}^2(dd_\mathsf{E}\delta_{\alpha_1\theta_1}\delta_{\beta_1\lambda_1}-d_\mathsf{E}^2\delta_{\alpha_1\beta_1}\delta_{\lambda_1\theta_1})(d\rho-\mbb1)\bigg]\nonumber\\
    &\qquad\qquad+d_\mathsf{S}^2\langle{e}_2\alpha_2|\omega_{k:2}^{(|e_2\beta_2\rangle\!\langle{e}_2^\prime\lambda_2|)}|e_2^\prime\theta_2\rangle\bigg[d_\mathsf{S}(dd_\mathsf{E}\delta_{\alpha_1\beta_1}\delta_{\lambda_1\theta_1}-d_\mathsf{E}^2\delta_{\alpha_1\theta_1}\delta_{\beta_1\lambda_1})(d\mbb1-\rho)\nonumber\\
    &\qquad\qquad\qquad\qquad+d_\mathsf{S}^2(dd_\mathsf{E}^2\delta_{\alpha_1\beta_1}\delta_{\lambda_1\theta_1}\hspace*{-0.05in}-d_\mathsf{E}\delta_{\alpha_1\theta_1}\delta_{\beta_1\lambda_1})(d\rho-\mbb1)\bigg]\bigg\}.
\label{int1_k:0}
\end{align}

Before carrying on, let us notice that the case $\omega_{k:k}$ is special because here one evaluates the terms on the Kraus operators $\tr(K_{ij}K_{ij}^\dg)$ and $\tr(K_{ij})\tr(K_{ij}^\dg)$ summed over their indices, which leaves
\begin{align}
    \langle{e}_k\alpha_k|\omega_{k:k}^{(\FS_{\lambda_k\beta_k})}|e_k\theta_k\rangle&=\f{1}{d(d^2-1)}\bigg[dd_\mathsf{E}^2(dd_\mathsf{E}\delta_{\alpha_k\theta_k}\delta_{\beta_k\lambda_k}-\delta_{\alpha_k\beta_k}\delta_{\lambda_k\theta_k})\\
    &\qquad\qquad\qquad\qquad+d^2d_\mathsf{E}(d_\mathsf{S}\delta_{\alpha_k\beta_k}\delta_{\lambda_k\theta_k}-\delta_{\alpha_k\theta_k}\delta_{\beta_k\lambda_k})\bigg],
\end{align}
and
\begin{align}
    \langle{e}_k\alpha_k|\omega_{k:k}^{(|e_k\beta_k\rangle\!\langle{e}^\prime_k\lambda_k|)}|e^\prime_k\theta_k\rangle&=\f{1}{d(d^2-1)}\bigg[dd_\mathsf{E}^3(d_\mathsf{S}\delta_{\alpha_k\theta_k}\delta_{\beta_k\lambda_k}-\delta_{\alpha_k\beta_k}\delta_{\lambda_k\theta_k})\nonumber\\
    &\qquad\qquad\qquad\qquad+d^2(dd_\mathsf{E}\delta_{\alpha_k\beta_k}\delta_{\lambda_k\theta_k}-\delta_{\alpha_k\theta_k}\delta_{\beta_k\lambda_k})\bigg].
\end{align}

Notice that we can get rid of the Kronecker deltas easily in the full $\Omega_{k:0}$ when summing over Greek indices, as $\delta_{\alpha\theta}\delta_{\beta\lambda}$ give rise to maximally mixed states in the ancillas while terms $\delta_{\alpha\beta}\delta_{\lambda\theta}$ give rise to identities, and there are only deltas of this kind. Thus we may simply assign $\delta_{\alpha\theta}\delta_{\beta\lambda}\to1/d_\mathsf{S}$ and $\delta_{\alpha\beta}\delta_{\lambda\theta}\to1$ when plugging the corresponding expressions in $\Omega_{k:0}$.\footnote{ In particular this implies that the terms $d_\mathsf{S}\delta_{\alpha\theta}\delta_{\beta\lambda}-\delta_{\alpha\beta}\delta_{\lambda\theta}$ and $dd_\mathsf{E}\delta_{\alpha\beta}\delta_{\lambda\theta}-d_\mathsf{E}^2\delta_{\alpha\theta}\delta_{\beta\lambda}$ won't contribute to the average purity.} Furthermore, a direct consequence of this is that $\Omega_{k:0}$ will be a linear combination of $2^{k+1}$ tensor products between $\mbb1$, $\rho$ and $\mathsf{\Psi}$, implying that (as $\rho$ is pure, $\mathsf{\Psi}$ is idempotent with trace one and the trace of an outer product is the product of traces) the average purity $\mbb{E}_\haarrand[\tr(\Upsilon^2)]$ will be a sum of the scalar terms in $\omega^{(\rho)}_{k:0}$ together with an extra overall factor of $d-1$. Let us denote by
\begin{gather}
    A=d_\mathsf{E}(d_\mathsf{E}^2-1),\quad
    B=d_\mathsf{E}^2-d_\mathsf{E}^2=0,\quad
    C=d_\mathsf{E}^2(d_\mathsf{S}-1/d_\mathsf{S}),\quad
    D=d_\mathsf{E}(dd_\mathsf{E}-1/d_\mathsf{S}),
\end{gather}
the terms that appear in expression~\eqref{int1_k:0} after having taken $\delta_{\alpha\theta}\delta_{\beta\lambda}\to1/d_\mathsf{S}$ and 
$\delta_{\alpha\beta}\delta_{\lambda\theta}\to1$. Doing similarly for the $\omega_{k:k}$ case,
\begin{gather}
    \mc{A}=dd_\mathsf{E}^2(d_\mathsf{E}^2+d_\mathsf{S}^2-2),\quad
    \mc{B}=dd_\mathsf{E}(d^2-1),
\end{gather}
we find that\footnote{ A way to deduce this is to sub-label each factor $A, B,\cdots$ by the $i$ index of the term ${\langle{e}_i\alpha_i|\omega^{X_i}|e_i^{(\prime)}\theta_i\rangle}$ that they come from and then substituting recursively in Eq.~\eqref{int1_k:0} as if one were already evaluating the whole purity. The vanishing of $B$ simplifies this process a great deal.}
\begin{align}
    \mbb{E}_\haarrand[\tr(\Upsilon^{2})]&=\frac{d_\mathsf{S}^k(d-1)}{[d(d^2-1)]^{k+1}}\left[\mc{A}\,A^{k-1}+d_\mathsf{S}\mc{B}\left(CA^{k-2}+C\sum_{i=1}^{k-2}d_\mathsf{S}^iD^iA^{k-i-2}+d_\mathsf{S}^{k-1}D^{k-1}\right)\right],
\end{align}
where the series has to be expanded before this can be evaluated; by means of the geometric series, $\sum_{i=1}^nx^i=\f{x(x^n-1)}{x-1}$, i.e. with
\begin{align}
    A^{k-2}\sum_{i=1}^{k-2}\left(\f{d_\mathsf{S}D}{A}\right)^i=Dd_\mathsf{S}\,\left(\f{d_\mathsf{S}^{k-2}D^{k-2}-A^{k-2}}{d_\mathsf{S}D-A}\right),
\end{align}
then plugging in the dimensional values (with $d=d_\mathsf{SE}$) of all the constants we can simplify the purity to
\begin{align}
    \mbb{E}_\haarrand[\tr(\Upsilon^{2})]&=\frac{(d_\mathsf{SE}+1) \left( d_\mathsf{SE}^2-1\right)^k+\left(d_\mathsf{E}^2-1\right)^{k+1}}{d_\mathsf{E} (d_{\mathsf{SE}}+1) \left(d_\mathsf{SE}^2-1\right)^k},
\end{align}
which simplifies to the expression in Eq.~\eqref{average purity ergodic}.

\section{Average purity - Constant interaction}
\label{appendix - average purity constant}
As in the previous Appendix~\ref{appendix - average purity random}, here $\Upsilon$ is normalized to unit trace. We may write the integral in the $ES$ part, as in Eq.\eqref{ES part purity integ} now with same unitary (throughout we take sums over repeated indices),
\begin{align}
    \omega_{k:0}^{(\rho)}:=\int_{\mbb{U}(d_\mathsf{SE})}&U^\dg\FS_{x_1\alpha_1}U^\dg\cdots{\FS}_{x_k\alpha_k}U_k^\dg{K}_{\tilde\imath\tilde\jmath}U\FS_{\beta_kx_k}\cdots{U}\FS_{\beta_1x_1}U\rho{U}^\dg\FS_{y_1\lambda_1}U^\dg\nonumber\\
    &\qquad\qquad\qquad\qquad\qquad\cdots\FS_{y_k\lambda_k}U^\dg{K}^\dg_{\tilde\imath\tilde\jmath}{U}\FS_{\theta_ky_k}\cdots{U}\FS_{\theta_1y_1}U\,d\muhaar(U),
\label{ES integral k step}
\end{align}

We decompose the unitaries in the whole \gls{syst-env} space as $U=\sum{U}_{ab}|a\rangle\!\langle{b}|$ and $U^\dg=\sum{U}^*_{a^\prime{b}^\prime}|b^\prime\rangle\!\langle{a}^\prime|$ (i.e. the $a$, $b$ labels here refer to the whole \gls{syst-env} space) and we enumerate the labels of the ones to the left of $\rho$ in Eq.~\eqref{ES integral k step} from $0$ to $k$ (priming, ${}^\prime$, the adjoint ones) and the remaining from $k+1$ to $2k+1$, we also do this in increasing order for the adjoint unitaries and decreasing order for the original unitaries so that they match the order of the $\FS$ operators, i.e. the unitary components will originally appear as $U^*_{i^\prime_0j^\prime_0}\cdots{U}^*_{i^\prime_kj^\prime_k}U_{i_kj_k}\cdots{U_{i_0j_0}}
U^*_{i^\prime_{k+1}j^\prime_{k+1}}\cdots{U}^*_{i^\prime_{2k+1}j^\prime_{2k+1}}U_{i^\prime_{2k+1}j^\prime_{2k+1}}\cdots
{U_{i^\prime_{k+1}j^\prime_{k+1}}}$, and we then rearrange them to the form of Eq.~\eqref{k moments of U} just keeping track of the correct order in the operator part, for which we introduce resolutions of identity on \gls{env} into each $\FS$ operator as $\FS_{ab}=\sum|ea\rangle\!\langle{eb}|$ labeled by $e$ for the ones to the left of $\rho$ in Eq.~\eqref{ES integral k step}, priming the ones between adjoint unitaries, and by $\epsilon$ the corresponding ones to the right, also priming the ones between unitaries.

This then leads to
\begin{align}
    \omega_{k:0}^{(\rho)}&=\int_{\mbb{U}(d_{\mathsf{SE}})}\hspace{-0.1in}U_{i_0j_0}\cdots{U}_{i_{2k+1}j_{2k+1}}U^*_{i^\prime_0j^\prime_0}\cdots{U}^*_{i^\prime_{2k+1}j^\prime_{2k+1}}\,d\muhaar(U)\,\langle{i}^\prime_k|K_{\tilde\imath\tilde\jmath}|i_k\rangle\!\langle{j}_0|\rho|j^\prime_{k+1}\rangle\!\langle{i}^\prime_{2k+1}|K^\dg_{\tilde\imath\tilde\jmath}|i_{2k+1}\rangle\nonumber\\
    &\qquad\prod_{\ell=1}^k\delta_{i_{\ell-1}(ex)_\ell}\delta_{j_\ell(e\beta)_\ell}\delta_{i_{\ell+k}(\epsilon{y})_\ell}\delta_{j_{\ell+k+1}(\epsilon\theta)_\ell}\delta_{i^\prime_{\ell-1}(e^\prime{x})_\ell}\delta_{j_\ell^\prime(e^\prime\alpha)_\ell}\delta_{i^\prime_{\ell+k}(\epsilon^\prime{y})_\ell}\delta_{j_{\ell+k+1}^\prime(\epsilon^\prime\lambda)_\ell}|j_0^\prime\rangle\!\langle{j}_{k+1}|\nonumber\\
    &=\langle{i}^\prime_k|K_{\tilde\imath\tilde\jmath}|i_k\rangle\!\langle{j}_0|\rho|j^\prime_{k+1}\rangle\!\langle{i}^\prime_{2k+1}|K^\dg_{\tilde\imath\tilde\jmath}|i_{2k+1}\rangle\sum_{\sigma,\tau\in\mathfrak{G}_{2k+2}}\!\!\Wg(\tau\sigma^{-1})\prod_{n=0}^{2k+1}\delta_{i_ni^\prime_{\sigma(n)}}\delta_{j_nj^\prime_{\tau(n)}}\,|j_0^\prime\rangle\!\langle{j}_{k+1}|\nonumber\\
    &\qquad\prod_{\ell=1}^k\delta_{i_{\ell-1}(ex)_\ell}\delta_{j_\ell(e\beta)_\ell}\delta_{i_{\ell+k}(\epsilon{y})_\ell}\delta_{j_{\ell+k+1}(\epsilon\theta)_\ell}\delta_{i^\prime_{\ell-1}(e^\prime{x})_\ell}\delta_{j_\ell^\prime(e^\prime\alpha)_\ell}\delta_{i^\prime_{\ell+k}(\epsilon^\prime{y})_\ell}\delta_{j_{\ell+k+1}^\prime(\epsilon^\prime\lambda)_\ell}\nonumber\\
    &=\sum_{\sigma,\tau\in\mathfrak{G}_{2k+2}}\Wg(\tau\sigma^{-1})\langle{i}^\prime_k|K_{\tilde\imath\tilde\jmath}|i^\prime_{\sigma(k)}\rangle\!\langle{j}^\prime_{\tau(0)}|\rho|j^\prime_{k+1}\rangle\!\langle{i}^\prime_{2k+1}|K^\dg_{\tilde\imath\tilde\jmath}|i^\prime_{\sigma(2k+1)}\rangle|j_0^\prime\rangle\!\langle{j}^\prime_{\tau(k+1)}|\nonumber\\
    &\qquad\prod_{\ell=1}^k\delta_{i^\prime_{\sigma(\ell-1)}(ex)_\ell}\delta_{j^\prime_{\tau(\ell)}(e\beta)_\ell}\delta_{i^\prime_{\sigma(\ell+k)}(\epsilon{y})_\ell}\delta_{j^\prime_{\tau(\ell+k+1)}(\epsilon\theta)_\ell}\delta_{i^\prime_{\ell-1}(e^\prime{x})_\ell}\delta_{j_\ell^\prime(e^\prime\alpha)_\ell}\delta_{i^\prime_{\ell+k}(\epsilon^\prime{y})_\ell}\delta_{j_{\ell+k+1}^\prime(\epsilon^\prime\lambda)_\ell}.
\end{align}

We now replace $i^\prime\to\varepsilon\varsigma$ and $j^\prime\to\varepsilon^\prime\varsigma^\prime$ in order to split the \gls{env} and \gls{syst} parts explicitly, leaving
\begin{align}
    &\omega_{k:0}^{(\rho)}\nonumber\\&=\hspace{-0.1in}\sum_{\sigma,\tau\in\mathfrak{G}_{2k+2}}\hspace{-0.15in}\Wg(\tau\sigma^{-1})\langle\varepsilon_k\varsigma_k|K_{\tilde\imath\tilde\jmath}|\varepsilon_{\sigma(k)}\varsigma_{\sigma(k)}\rangle\!\langle\varepsilon^\prime_{\tau(0)}\varsigma^\prime_{\tau(0)}|\rho|\varepsilon^\prime_{k+1}\varsigma^\prime_{k+1}\rangle\!\langle\varepsilon_{2k+1}\varsigma_{2k+1}|K^\dg_{\tilde\imath\tilde\jmath}|\varepsilon_{\sigma(2k+1)}\varsigma_{\sigma(2k+1)}\rangle\nonumber\\
    &\qquad\qquad\qquad\prod_{\ell=1}^k\delta_{(\varepsilon\varsigma)_{\sigma(\ell-1)}(ex)_\ell}\delta_{(\varepsilon^\prime\varsigma^\prime)_{\tau(\ell)}(e\beta)_\ell}\delta_{(\varepsilon\varsigma)_{\sigma(\ell+k)}(\epsilon{y})_\ell}\delta_{(\varepsilon^\prime\varsigma^\prime)_{\tau(\ell+k+1)}(\epsilon\theta)_\ell}\delta_{(\varepsilon\varsigma)_{\ell-1}(e^\prime{x})_\ell}\nonumber\\
    &\qquad\qquad\qquad\qquad\qquad\delta_{(\varepsilon^\prime\varsigma^\prime)_\ell(e^\prime\alpha)_\ell}\delta_{(\varepsilon\varsigma)_{\ell+k}(\epsilon^\prime{y})_\ell}\delta_{(\varepsilon^\prime\varsigma^\prime)_{\ell+k+1}(\epsilon^\prime\lambda)_\ell}|\varepsilon^\prime_0\varsigma^\prime_0\rangle\!\langle\varepsilon^\prime_{\tau(k+1)}\varsigma^\prime_{\tau(k+1)}|\nonumber\\
    &=\sum_{\sigma,\tau\in\mathfrak{G}_{2k+2}}\hspace{-0.15in}\Wg(\tau\sigma^{-1})\langle\varepsilon^\prime_{\tau(0)}\varsigma^\prime_{\tau(0)}|\rho|\varepsilon^\prime_{k+1}\varsigma^\prime_{k+1}\rangle|\varepsilon^\prime_0\varsigma^\prime_0\rangle\!\langle\varepsilon^\prime_{\tau(k+1)}\varsigma^\prime_{\tau(k+1)}|\delta_{\varepsilon_k\varepsilon_{\sigma(2k+1)}}\delta_{\varepsilon_{\sigma(k)}\varepsilon_{2k+1}}\nonumber\\
    &\qquad\qquad\qquad\prod_{\ell=1}^k\delta_{\varepsilon_{\sigma(\ell-1)}\varepsilon^\prime_{\tau(\ell)}}\delta_{\varepsilon_{\sigma(\ell+k)}\varepsilon^\prime_{\tau(\ell+k+1)}}\delta_{\varepsilon_{\ell-1}\varepsilon^\prime_\ell}\delta_{\varepsilon_{\ell+k}\varepsilon^\prime_{\ell+k+1}}\delta_{\varsigma^\prime_\ell\alpha_\ell}\nonumber\\
    &\qquad\qquad\qquad\qquad\qquad\delta_{\varsigma^\prime_{\tau(\ell)}\beta_\ell}\delta_{\varsigma^\prime_{\tau(\ell+k+1)}\theta_\ell}\delta_{\varsigma^\prime_{\ell+k+1}\lambda_\ell}\prod_{n=0}^{2k+1}\delta_{\varsigma_{\sigma(n)}\varsigma_n}.\end{align}

The ancillary part is simply $\Omega_{k:0}^{(\mathsf{AB})}=d_\mathsf{S}^{-k}|\alpha_1\beta_1\ldots\alpha_k\beta_k\rangle\!\langle\theta_1\lambda_1\ldots\theta_k\lambda_k|$, so we get for the analogue of Eq.~\eqref{full purity integrals},
\begin{align}
    \Omega_{k:0}&=\omega_{k:0}^{(\rho)}\otimes\Omega_{k:0}^{(AB)}\nonumber\\
    &=\f{1}{d_\mathsf{S}^k}\sum_{\sigma,\tau\in\mathfrak{G}_{2k+2}}\hspace{-0.15in}\Wg(\tau\sigma^{-1})\delta_{\varepsilon_k\varepsilon_{\sigma(2k+1)}}\delta_{\varepsilon_{\sigma(k)}\varepsilon_{2k+1}}\nonumber\\
    &\qquad|\varepsilon^\prime_0\varsigma^\prime_0\varsigma^\prime_1\varsigma^\prime_{\tau(1)}\cdots\varsigma^\prime_k\varsigma^\prime_{\tau(k)}\rangle\langle\varepsilon^\prime_{\tau(k+1)}\varsigma^\prime_{\tau(k+1)}\varsigma^\prime_{\tau(k+2)}\varsigma^\prime_{k+2}\cdots\varsigma^\prime_{\tau(2k+1)}\varsigma^\prime_{2k+1}|\nonumber\\
    &\qquad\qquad\langle\varepsilon^\prime_{\tau(0)}\varsigma^\prime_{\tau(0)}|\rho|\varepsilon^\prime_{k+1}\varsigma^\prime_{k+1}\rangle\prod_{n=0}^{2k+1}\delta_{\varsigma_{\sigma(n)}\varsigma_n}\prod_{\ell=1}^k\delta_{\varepsilon_{\sigma(\ell-1)}\varepsilon^\prime_{\tau(\ell)}}\delta_{\varepsilon_{\sigma(\ell+k)}\varepsilon^\prime_{\tau(\ell+k+1)}}\delta_{\varepsilon_{\ell-1}\varepsilon^\prime_\ell}\delta_{\varepsilon_{\ell+k}\varepsilon^\prime_{\ell+k+1}},
\end{align}
and thus
\begin{align}
    \mbb{E}_\haar[\tr(\Upsilon^2)]&=\tr[(\rho\otimes\mathsf{\Psi}^{\otimes{k}})\,\Omega_{k:0}]\nonumber\\
    &=\f{1}{d_\mathsf{S}^{2k}}\sum_{\sigma,\tau\in\mathfrak{G}_{2k+2}}\hspace{-0.15in}\Wg(\tau\sigma^{-1})\langle\varepsilon^\prime_{\tau(0)}\varsigma^\prime_{\tau(0)}|\rho|\varepsilon^\prime_{k+1}\varsigma^\prime_{k+1}\rangle\!\langle\varepsilon^\prime_{\tau(k+1)}\varsigma^\prime_{\tau(k+1)}|\rho|\varepsilon^\prime_{0}\varsigma^\prime_{0}\rangle\delta_{\varepsilon_k\varepsilon_{\sigma(2k+1)}}\delta_{\varepsilon_{2k+1}\varepsilon_{\sigma(k)}}\nonumber\\
&\hspace{0.5in}\prod_{\substack{\ell=1\\\ell\neq{k+1}}}^{2k+1}
\delta_{\varepsilon_{\ell-1}\varepsilon^\prime_\ell}\delta_{\varepsilon_{\sigma(\ell-1)}\varepsilon^\prime_{\tau(\ell)}}
\delta_{\varsigma^\prime_\ell\varsigma^\prime_{\tau(\ell)}}\prod_{n=0}^{2k+1}\delta_{\varsigma_{\sigma(n)}\varsigma_n},\end{align}
as stated in Eq.~\eqref{average time independent purity}, where we make the definitions
\begin{align}
    &\rho_{\tau(0);k+1}\rho_{\tau(k+1);0}=\langle\varepsilon^\prime_{\tau(0)}\varsigma^\prime_{\tau(0)}|\rho|\varepsilon^\prime_{k+1}\varsigma^\prime_{k+1}\rangle\!\langle\varepsilon^\prime_{\tau(k+1)}\varsigma^\prime_{\tau(k+1)}|\rho|\varepsilon^\prime_{0}\varsigma^\prime_{0}\rangle
    \label{appendix: purity const rho}\\
    &\tilde\Delta_{k,\sigma,\tau}^{(d_\mathsf{E})}=\delta_{e_ke_{\sigma(2k+1)}}\delta_{e_{2k+1}e_{\sigma(k)}}
    \prod_{\substack{\ell=1\\\ell\neq{k+1}}}^{2k+1}\delta_{e_{\ell-1}e^\prime_\ell}\delta_{e_{\sigma(\ell-1)}e^\prime_{\tau(\ell)}},\quad
    \tilde{\tilde\Delta}_{k,\sigma,\tau}^{(d_\mathsf{S})}=\prod_{\substack{\ell=1\\\ell\neq{k+1}}}^{2k+1}\delta_{s^\prime_\ell{s}^\prime_{\tau(\ell)}}\prod_{n=0}^{2k+1}\delta_{s_ns_{\sigma(n)}}\nonumber\\
    &\Delta_{k,\sigma,\tau}^{(d_\mathsf{E},d_\mathsf{S})}=\tilde\Delta_{k,\sigma,\tau}^{(d_\mathsf{E})}\tilde{\tilde\Delta}_{k,\sigma,\tau}^{(d_\mathsf{S})}
    \label{TIpuritydef}
\end{align}

\subsection{Small subsystem limit}\label{appendix: purity Ui=Uj asympt}
In the $d_\mathsf{E}\to\infty$ limit the only term that remains is the one with
\begin{gather}
    \sigma=\tau=(0,k+1)(1,k+2)\cdots(k,2k+1),
\end{gather}
as expressed in cycle notation, which simply means $\sigma(0)=k+1=\tau(0)$, $\sigma(k+1)=0=\tau(k+1)$ for the first $(0,k+1)$, and similarly for the rest of assignments. This then leads to the correct limit, as
\begin{align}
    \tilde{\tilde\Delta}_{k,\sigma,\tau}^{(d_\mathsf{S})}\to\sum_{\substack{\varsigma^{(\prime)}_i=1\\\varsigma^\prime\neq\{\varsigma^\prime_0,\varsigma^\prime_{k+1}\}}}^{d_\mathsf{S}}\prod_{\substack{\ell=1\\\ell\neq{k+1}}}^{2k+1}\delta_{\varsigma^\prime_\ell\varsigma^\prime_{\tau(\ell)}}\prod_{n=0}^{2k+1}\delta_{\varsigma_{\sigma(n)}\varsigma_n}&={d_\mathsf{S}}^{2k+1},\\
    \tilde\Delta_{k,\sigma,\tau}^{(d_\mathsf{E})}\to\sum_{\substack{\varepsilon^{(\prime)}_i=1\\\varepsilon^\prime\neq\{\varepsilon^\prime_0,\varepsilon^\prime_{k+1}\}}}^{d_\mathsf{E}}\delta_{\varepsilon_k\varepsilon_{\sigma(2k+1)}}\delta_{\varepsilon_{2k+1}\varepsilon_{\sigma(k)}}\prod_{\substack{\ell=1\\\ell\neq{k+1}}}^{2k+1}\delta_{\varepsilon_{\ell-1}\varepsilon^\prime_\ell}\delta_{\varepsilon_{\sigma(\ell-1)}\varepsilon^\prime_{\tau(\ell)}}&={d_\mathsf{E}}^{2k+2},
\end{align}
while keeping $\sigma\tau^{-1}=\mathbf{1}^{2k+2}$, i.e. with the argument in the $\Wg$ function an identity (or fixed point), which generates the least denominator powers because $\#\boldsymbol{1}^n=\#[(1)(2)\cdots(n)]=n$, i.e. the identity generates the most number of cycles. When considering other permutations for $\sigma$ and $\tau$ we see that the $d_\mathsf{E}$ powers in the numerator can only decrease while those in the denominator can only increase (as the identity is the permutation that generates the most cycles) so indeed in the large subsystem limit this is the only pair of permutations that survive; the terms in $\rho$ yield traces when summed over basis vectors and thus give factors of one,
\begin{gather}\sum_{\varepsilon^\prime_{k+1},\varsigma^\prime_{k+1}}\langle\varepsilon^\prime_{k+1}\varsigma^\prime_{k+1}|\rho|\varepsilon^\prime_{k+1}\varsigma^\prime_{k+1}\rangle=\sum_{\varepsilon^\prime_0,\varsigma^\prime_0}\langle\varepsilon^\prime_0\varsigma^\prime_0|\rho|\varepsilon^\prime_{0}\varsigma^\prime_{0}\rangle=\tr\rho=1.\end{gather}
 This leads then to the conclusion
\begin{align}\mbb{E}_\haar[\tr(\Upsilon^2)]&\sim\f{d_\mathsf{S}^{2k+1}d_\mathsf{E}^{2k+2}}{d_\mathsf{S}^{2k}}\Wg(\mathbf{1}^{2k+2})=d_\mathsf{S}d_\mathsf{E}^{2k+2}\f{1}{(d_\mathsf{E}d_\mathsf{S})^{2k+2}}=\f{1}{d_\mathsf{S}^{2k+1}},\hspace{0.1in}\text{when}\,\,d_\mathsf{E}\to\infty.\end{align}

\subsection{Long time limit}\label{appendix: purity Ui=Uj asympt time}
For the case when $k\to\infty$, the $\Wg$ function behaves dominantly as in the small subsystem limit and the only pair of permutations that identify all dimension powers in the numerator without dependence of $k$ are identities $\sigma=\tau=\mathbf{1}^{2k+2}$. In this case
\begin{align}\sum_{\substack{\varsigma^{(\prime)}_i=1\\\varsigma^\prime\neq\{\varsigma^\prime_0,\varsigma^\prime_{k+1}\}}}^{d_\mathsf{S}}
\prod_{\substack{\ell=1\\\ell\neq{k+1}}}^{2k+1}\delta_{\varsigma^\prime_\ell\varsigma^\prime_\ell}
\prod_{n=0}^{2k+1}\delta_{\varsigma_n\varsigma_n}
&={d_\mathsf{S}}^{4k+2},\\
\sum_{\substack{\varepsilon^{(\prime)}_i=1\\\varepsilon^\prime\neq\{\varepsilon^\prime_0,\varepsilon^\prime_{k+1}\}}}^{d_\mathsf{E}}
\delta_{\varepsilon_k\varepsilon_{2k+1}}\delta_{\varepsilon_{2k+1}\varepsilon_{k}}
\prod_{\substack{\ell=1\\\ell\neq{k+1}}}^{2k+1}
\delta_{\varepsilon_{\ell-1}\varepsilon^\prime_\ell}\delta_{\varepsilon_{\ell-1}\varepsilon^\prime_{\ell}}
&={d_\mathsf{E}}^{2k+1},\end{align}
giving
\begin{align}\mbb{E}_\haar[\tr(\Upsilon^2)]&\sim\f{d_\mathsf{S}^{4k+2}d_\mathsf{E}^{2k+1}}{d_\mathsf{S}^{2k}}\tr(\rho^2)\Wg(\mathbf{1}^{2k+2})=\f{d_\mathsf{S}^{2k+2}d_\mathsf{E}^{2k+1}}{(d_\mathsf{E}d_\mathsf{S})^{2k+2}}=\f{1}{d_\mathsf{E}},\hspace{0.1in}\text{when}\,\,k\to\infty,\end{align}
as the fiducial state $\rho$ is pure, $\tr(\rho^2)=1$, by assumption.
    \chapter{Sampling from the unitary group}
\label{appendix - numerics Almost}
Here we present detail for how the numerical sampling and therefore the plot in Fig.~\ref{numPlot} of Chapter~\ref{sec:typicality} was obtained. We follow Ref.~\cite{Mezzadri}, which contains a complete discussion and reasoning behind the algorithm to numerically sample unitaries from the Haar measure.

Specifically, we computed the bound on the non-Markovianity $\mc{N}_1$ given by the trace distance $D$ between a process $\Upsilon$ sampled at random and the maximally noisy process $\mbb1/d_\mathsf{S}^{2k+1}$, i.e. $D\left(\Upsilon,\mbb1/d_\mathsf{S}^{2k+1}\right)\geq\mc{N}_1$. This amounts to sampling unitaries $U_i\sim\mu_\haar$ of dimension $d_\mathsf{ES}\times{d}_\mathsf{ES}$, with which the full process is defined, and then computing its trace distance with respect to the maximally mixed state of dimension $d_\mathsf{S}^{2k+1}$. This is repeated a given number of times $n$ in order to obtain an arithmetic average, $\mathrm{E}_n$, which we compare with the upper-bound $\mc{B}_k\geq\mbb{E}_\haarrand[\mc{N}_1]$. These must satisfy $\mathrm{E}_n\leq\mc{B}_k$, given that $\mc{B}_k$ is in turn an upper bound on the uniform (Haar) average trace distance between a process and the maximally noisy process.

The algorithm in Ref.~\cite{Mezzadri} to obtain a uniformly distributed $d\times{d}$ unitary matrix is as follows, with each step explained below:
\begin{compactenum}[\itshape i.]
    \item Generate a $d\times{d}$ complex matrix $Z$ with independent identically distributed (i.i.d.) standard normal complex random variables.
    \item Perform a $QR$ decomposition of $Z$.
    \item Define a diagonal matrix $\Lambda=\mathrm{diag}\left(\f{r_1}{|r_1|},\cdots,\f{r_d}{|r_d|}\right)$, where $r_i$ are the diagonal elements of $R$.
    \item The matrix $U=Q\Lambda$ is distributed according to the Haar measure.
\end{compactenum}

The first step is straightforward, where here ``\emph{standard normal complex}'' refers to a complex random variable with real and imaginary parts being independent normally distributed random variables with mean zero and variance $1/2$. This means the components $Z_{ij}$ are normally distributed with a \gls{pdf} given by $\mscr{P}(Z_{ij})=\f{1}{\pi}\exp\left(-|Z_{ij}|^2\right)$. Since these components are statistically independent, the joint distribution for $Z$ is simply the product of distributions so that $\mscr{P}(Z)=\f{1}{\pi^{d^2}}\exp\left[-\tr(ZZ^\dg)\right]$. This distribution describes the so-called \emph{Ginibre ensemble}, and the variable $Z$ induces a probability measure $d\mu_\mathsf{G}$ on $\mathrm{GL}(\mbb{C},d)$, the set of $d\times{d}$ invertible matrices with ordinary matrix multiplication. Importantly, this measure on the Ginibre ensemble is left and right invariant under unitary transformations, $d\mu_\mathsf{G}(UZ)=d\mu_\mathsf{G}(ZV)=d\mu_\mathsf{G}(Z)$ for any fixed $U,V\in\mbb{U}(d)$.

For step two, any matrix $Z\in\mathrm{GL}(\mbb{C},d)$ can be decomposed as $Z=QR$ with $Q\in\mbb{U}(d)$ and $R$ an upper-triangular and invertible matrix~\cite{Horn_2012}. This is called the $QR$ decomposition of $Z$. This means that, because the matrix $R$ is invertible, we can write $Q=ZR^{-1}$ and we get a unitary random matrix. The $QR$ decomposition is a standard routine in most symbolic softwares such as Matlab or Mathematica or in packages like NumPy for Python.

The last step is related to the non-uniqueness of the $QR$ decomposition. As discussed in Ref.~\cite{Mezzadri}, the unitary $Q$ is not quite Haar-distributed because of this non-uniqueness. That is, for any diagonal unitary matrix $\Lambda\in\mbb{U}(d)$, we have $QR=(Q\Lambda)(\Lambda^\dg{R})=Q'R'$ with $Q'$ unitary and $R'$ upper-triangular. This is remedied with $\Lambda=\mathrm{diag}\left(\f{r_1}{|r_1|},\cdots,\f{r_d}{|r_d|}\right)$ because it forces $R$ to have positive diagonal entries.

Thus, for a $k$-step random interaction process, i.e. with $U_i\neq{U}_j$, in the $i\textsuperscript{th}$ run, we generate $k+1$ Haar-distributed $d_\mathsf{SE}\times{d}_\mathsf{SE}$ unitary matrices, $\left\{U^{(i)}_0,U^{(i)}_1,\ldots,U^{(i)}_k\right\}$, according to the algorithm above and compute the arithmetic average
\begin{equation}
    \mathrm{E}_n=\f{1}{n}\sum_{i=1}^nD\left(\Upsilon^{(i)},\mbb1/d_\mathsf{S}^{2k+1}\right),
    \label{eq: numerical average}
\end{equation}
for a given number $n$ of trials, where $\Upsilon^{(i)}=\tr_\mathsf{E}\left[\mc{U}_k^{(i)}\mc{S}_k\cdots\mc{U}_1^{(i)}\mc{S}_1\mc{U}_0^{(i)}(\rho\otimes\mathsf{\Psi})\right]$ with $\mc{U}_\ell^{(i)}(\cdot)=U_\ell^{(i)}(\cdot)U_\ell^{(i)\,\dg}$, is the random process in the $i\textsuperscript{th}$ trial for any fiducial pure state $\rho$.

Finally, the greatest problem to be overcome is the large-dimensional nature of the Choi state, which prior to partial tracing of the environment, is a $d_\mathsf{E}d_\mathsf{S}^{2k+1}$ dimensional square matrix. While we directly computed Eq.~\eqref{eq: numerical average}, depending on the purpose of the calculation,it might be worth it to try and render this problem efficient through other computational or numerical techniques.
    \backmatter         
    \printbibliography

\end{document}